\documentclass[12pt]{book}

\usepackage{booktabs}

\usepackage[T1, OT1]{fontenc}
\DeclareTextSymbolDefault{\DH}{T1}
\usepackage{amssymb,latexsym}
\usepackage{graphicx,amsmath,amsfonts}
\numberwithin{equation}{chapter}
\usepackage{comment}
\usepackage{tikz}
\usetikzlibrary{arrows}
\usetikzlibrary{calc}
\usepackage{pgfplots}
\usepackage{chngcntr}
\usepackage{apptools}

\usepackage{url}
\usepackage{mathrsfs}
\usepackage{colortbl}
\definecolor{LightCyan}{rgb}{0.88,1,1}
\definecolor{celadon}{rgb}{0.67, 0.88, 0.69}
\definecolor{columbiablue}{rgb}{0.61, 0.87, 1.0}
\usepackage{paralist}
\usepackage{hyperref}
\usepackage[sort&compress,nameinlink]{cleveref}
\usepackage{nicefrac}
\usepackage{booktabs}
\usepackage{setspace}
\usepackage{etoolbox}
\usepackage{wrapfig}
\usepackage[square,numbers,sort&compress,sectionbib]{natbib}
\usepackage{chapterbib}  
\usepackage{amsthm}
\usepackage{makecell}
\allowdisplaybreaks
\usepackage[nottoc,numbib]{tocbibind}
\usepackage{enumitem}
\usepackage{pifont}
\usepackage{listings}

\counterwithin{figure}{chapter}
\counterwithin{table}{chapter}

\usepackage{xcolor}

\definecolor{codegreen}{rgb}{0,0.6,0}
\definecolor{codegray}{rgb}{0.5,0.5,0.5}
\definecolor{codepurple}{rgb}{0.58,0,0.82}
\definecolor{backcolour}{rgb}{0.95,0.95,0.92}

\lstdefinestyle{mystyle}{
    backgroundcolor=\color{backcolour},   
    commentstyle=\color{codegreen},
    keywordstyle=\color{magenta},
    numberstyle=\tiny\color{codegray},
    stringstyle=\color{codepurple},
    basicstyle=\ttfamily\footnotesize,
    breakatwhitespace=false,         
    breaklines=true,                 
    captionpos=b,                    
    keepspaces=true,                 
    numbersep=5pt,                  
    showspaces=false,                
    showstringspaces=false,
    showtabs=false,                  
    tabsize=2
}

\makeatletter

\patchcmd{\NAT@citex}
  {\@citea\NAT@hyper@{%
     \NAT@nmfmt{\NAT@nm}%
     \hyper@natlinkbreak{\NAT@aysep\NAT@spacechar}{\@citeb\@extra@b@citeb}%
     \NAT@date}}
  {\@citea\NAT@nmfmt{\NAT@nm}%
   \NAT@aysep\NAT@spacechar\NAT@hyper@{\NAT@date}}{}{}

\patchcmd{\NAT@citex}
  {\@citea\NAT@hyper@{%
     \NAT@nmfmt{\NAT@nm}%
     \hyper@natlinkbreak{\NAT@spacechar\NAT@@open\if*#1*\else#1\NAT@spacechar\fi}%
       {\@citeb\@extra@b@citeb}%
     \NAT@date}}
  {\@citea\NAT@nmfmt{\NAT@nm}%
   \NAT@spacechar\NAT@@open\if*#1*\else#1\NAT@spacechar\fi\NAT@hyper@{\NAT@date}}
  {}{}

\makeatother

\renewenvironment{description}
  {\list{}{\labelwidth=15pt \leftmargin=15pt
   }}
  {\endlist}

\hypersetup{
colorlinks=true,citecolor=blue!80!black,linkcolor=red!60!black,urlcolor=blue!80!black,
pagebackref=true
}

\usepackage{pifont}

\newtheorem{theorem}{Theorem}

\newtheorem{remark}{Remark}
\newtheorem{definition}{Definition}
\newtheorem{example}{Example}
\newtheorem{proposition}{Proposition}

\newtheorem{abcrule}{Rule}
\newtheorem{apportionmentrule}{Apportionment Rule}

\counterwithin{lemma}{chapter}
\counterwithin{theorem}{chapter}
\counterwithin{proposition}{chapter}
\counterwithin{example}{chapter}
\counterwithin{definition}{chapter}

\usepackage{etoolbox}
\newcommand{\addQEDstyle}[2]{\AtBeginEnvironment{#1}{\pushQED{\qed}\renewcommand{\qedsymbol}{#2}}
\AtEndEnvironment{#1}{\popQED}}
\addQEDstyle{example}{$\dashv$}

\newcommand{\score}[1]{{{\mathrm{score}_{#1}}}}
\newcommand{\loss}[1]{{{\mathrm{loss}_{#1}}}}
\newcommand{\reals}{\mathbb R}
\newcommand{\naturals}{\mathbb N}
\newcommand{\expected}{\mathbb E}
\newcommand{\powerset}{\mathcal{P}}

\newcommand{\calR}{\mathcal{R}}

\newcommand{\pav}{{{\mathrm{PAV}}}}
\newcommand{\np}{{{\mathrm{NP}}}}
\newcommand{\ptime}{{{\mathrm{P}}}}
\newcommand{\conp}{{{\mathrm{coNP}}}}

\newcommand{\cc}{{{{\mathrm{CC}}}}}

\newcommand{\monroe}{{{{\mathrm{Monroe}}}}}

\newcommand{\welf}{\mathrm{welf}}

\newcommand{\phragmen}{Phragm\'{e}n}
\newcommand{\lexphrag}{leximax-\phragmen}
\newcommand{\enephrag}{Enestr\"om--Phragm\'{e}n}

\DeclareMathOperator*{\argmax}{arg\,max}
\DeclareMathOperator*{\argmin}{arg\,min}

\makeatletter
\newtheorem*{rep@theorem}{\rep@title}
\newcommand{\newreptheorem}[2]{%
\newenvironment{rep#1}[1]{%
 \def\rep@title{#2 \ref{##1}}%
 \begin{rep@theorem}}%
 {\end{rep@theorem}}}
\makeatother
\newreptheorem{theorem}{Theorem}
\newreptheorem{proposition}{Proposition}
\newreptheorem{lemma}{Lemma}
\newreptheorem{corollary}{Corollary}

\newcommand{\av}{{{\mathrm{AV}}}}
\newcommand{\sav}{{{\mathrm{SAV}}}}
\newcommand{\hamming}{\mathrm{ham}}

\renewcommand{\H}{{{\mathrm{h}}}}

\definecolor{OKgreen}{HTML}{035925}
\definecolor{NOred}{HTML}{8F061F}
\newcommand*\cmark{\textcolor{OKgreen}{\ding{51}}}
\newcommand*\ccand{\textcolor{OKgreen}{cand}}
\newcommand*\xmark{\textcolor{NOred}{$\times$}}
\newcommand*\strong{\textcolor{OKgreen}{\bf strong}}
\newcommand*\weak{\textcolor{OKgreen}{\bf weak}}

\newcommand{\notstarted}

\iffalse

\fi

\def\N{\mathbb{N}}

\newcommand{\intro}[1]{}

\usepackage[leftbars]{changebar}
\crefname{abcrule}{Rule}{Rules}

\begin{document}

\title{Multi-Winner Voting with Approval Preferences}

\author{Martin Lackner\\
  TU Wien\\
  Vienna, Austria\\
  {\small \texttt{lackner@dbai.tuwien.ac.at}}
  \and 
Piotr Skowron\\
  University of Warsaw\\
  Warsaw, Poland\\
  {\small \texttt{p.skowron@mimuw.edu.pl}}
}
\date{Version 1.0 (\today)}

\maketitle

\chapter*{Preface}
Multi-winner voting is the process of selecting a fixed-size set of representative candidates based on voters' preferences.
It occurs in applications ranging from politics (parliamentary elections) to the design of modern computer applications (collaborative filtering, dynamic Q\&A platforms, diversifying search results).
All these applications share the problem of identifying a representative subset of alternatives---and the study of multi-winner voting is the principled analysis of this task.

This book provides a thorough and in-depth look at multi-winner voting based on approval preferences. One speaks of approval preferences if voters express their preferences by providing a set of candidates they approve.
Approval preferences thus separate candidates in approved and disapproved ones, a simple, binary classification.
The corresponding multi-winner voting rules are called approval-based committee (ABC) rules.
Due to the simplicity of approval preferences, ABC rules are widely suitable for practical use.

Recent years have seen a rising interest in ABC voting.
While multi-winner voting has been originally a topic studied by economists and political scientists, 
a significant share of recent progress has occurred in the field of \emph{computational social choice}. This
discipline is situated in the intersection of artificial intelligence, computer science, economics, and (to a lesser degree) political science,
combining insights and methods from these distinct fields.
The goal of this book is to present fundamental concepts and results for ABC voting
and to discuss the recent advances in computational social choice.
The main focus is on axiomatic analysis, algorithmic results, and relevant applications.

\bigskip

\begin{flushright}
\textbf{Martin Lackner}\\
Vienna, Austria
\end{flushright}

\begin{flushright}
\textbf{Piotr Skowron}\\
Warsaw, Poland
\end{flushright}

\renewcommand{\bibname}{References}

\chapter*{Acknowledgments}

First and foremost, we would like to thank Piotr Faliszewski for extensive feedback and discussions that significantly improved this book.
We are further very thankful to Markus Brill, Svante Janson, J{\'e}r\^{o}me Lang, Marie-Louise Lackner, Jean-Fran{\c{c}}ois Laslier, Jan Maly, Dominik Peters, and Luis S{\'a}nchez-Fern{\'a}ndez for providing valuable feedback and comments.

\medskip

Martin Lackner was supported by the Austrian Science Fund (FWF): P31890.
Piotr Skowron was supported by Poland's National Science Center grant UMO-2019/35/B/ST6/02215.

\setcounter{tocdepth}{2}
\tableofcontents

\chapter{Approval-Based Committee Voting}

\newcommand{\checkedbox}{{\Large\mbox{\ooalign{\ding{55}\cr\hidewidth$\square$\hidewidth\cr}}}}
\newcommand{\emptybox}{{\Large\mbox{\ooalign{\phantom{\ding{55}}\cr\hidewidth$\square$\hidewidth\cr}}}}
\begin{wrapfigure}[18]{r}{6.2cm}
\ \ \parbox{5.8cm}{
\vspace{-2em}
\begin{center}
{\def\arraystretch{1.25}
\setlength{\extrarowheight}{6pt}
\begin{tabular}{|lcc|}
\hline
Victor D'Hondt &  & \checkedbox \\ 
Gustaf Enestr\"om&  & \emptybox \\ 
Vilfredo Pareto &  & \emptybox \\ 
L.~Edvard\ \phragmen  &  & \checkedbox \\ 
Thorvald N.\ Thiele &  & \emptybox \\ 
\hline
\end{tabular} }
\caption{An approval ballot. Here, the voter decided to approve two of the five candidates.
In this hypothetical election, the five candidates are 19th-century academics who are relevant to this book.}\label{fig:approval-ballot}
\end{center}}
\end{wrapfigure}

\intro{We give an overview of multi-winner voting with approval ballots, its applications, main themes, advantages and disadvantages. We discuss the structure of the book and necessary prerequisites.
Furthermore, we provide a pointer to a Python library supplementing this book.}

\section{Introduction}

What is multi-winner voting?
In a multi-winner election, we are given a set of candidates, a set of voters, the preferences that each voter has over these candidates, and a desired size $k$ of the committee to be elected. The goal is to select a committee of exactly~$k$ candidates based on the voters' preferences.

Using this broad understanding of what multi-winner elections are, we encounter them in many vastly different scenarios ranging from everyday life to technical applications. A prototypical multi-winner election is the democratic selection of a representative body, such as a parliament\footnotemark, a faculty council, or a board of trustees. Moreover, selecting finalists in a competition, based on judgements of experts, is also an instance of multi-winner elections---here the experts act as voters and the contestants as candidates. Other possible applications of multi-winner election rules have been identified in the artificial intelligence, economics, and broader computer science literature:
\begin{enumerate}
\item finding group recommendations~\citep{owaWinner,budgetSocialChoice,bou-lu:c:value-directed}, where the possible recommendations can be thought of as candidates and individual group members as voters,
\item collaborative filtering~\citep{ChakrabortyPGGL19,gawronusing}, where, for example, related movies are recommended based on large data collections,
\item diversifying search results~\citep{proprank}, where users sending a search query can be interpreted as voters and the possible search results correspond to candidates,
\item locating public facilities~\citep{far-hek:b:facility-location,owaWinner}, where the candidates are possible locations in which facilities can be built, 
\item the design of dynamic Q\&A platforms \citep{israel2021dynamic}, where participants propose and upvote questions to be asked in a Q\&A session,
\item selecting validators in consensus protocols (blockchain)~\cite{cevallos2020verifiably,burdges2020overview}, with the users of the protocol corresponding to both voters and candidates, and
\item genetic programming~\cite{FaliszewskiSSS17}, a technique to solve global  optimisation  problems.
\end{enumerate}

\footnotetext{Most countries use legislatures based on political parties for electing parliaments. However, in some countries open-list systems are used (e.g., in Austria, Belgium, Finland, Latvia, Luxembourg, Netherlands, Sweden, and Switzerland); these systems (also) allow voters to vote for individual candidates rather than only for political parties.
Indeed, a few important arguments for allowing to vote for individual candidates have been raised.
For example, when voting for individual candidates, the elected candidates are more committed to the electorate rather than to their political parties. 
At the same time, open-list systems allow the candidates to focus on campaigning for the citizens' votes rather than on gaining influence within their party~\cite{personalRepresentation, audrey2015electoral, ames1995ElectoralStrategies,chang2005ElectoralIncentives}.
For a more general, comparative analysis of different electoral systems, we refer the reader to the relevant political science literature \cite{grofmanChoosingElectoral,grofmanProspectives,Far11,renwick16}.}

The outcome of a multi-winner election should clearly depend on the available preference information.
In a political election, this preference information is typically elicited with (paper) ballots; in a computer application, this information is shaped by the design of the user interface.
In this book, we are looking at the approval-based model of multi-winner elections.
The approval-based model is based on the assumption that the available preference information for each voter is a separation between approved and non-approved candidates, as illustrated in \Cref{fig:approval-ballot}.
That is, each voter submits \emph{approval preferences} via a subset of candidates---this subset consists of the candidates approved by the voter.%
\footnote{The other main variant of multi-winner elections is based on rankings, where each voter orders the candidates from the most to the least preferred one. We only briefly consider ranking-based multi-winner elections in this book (\Cref{sec:multiwinner-rank})---for a more substantial overview we refer the reader to a book chapter by \citet{FSST-trends}.}
The main object of this book are approval-based committee voting rules (ABC voting rules), i.e., functions that select one or more committees given an approval-based multi-winner election.
Importantly, we require that ABC voting rules are deterministic (and not randomised).

To illustrate a multi-winner election with approval preferences, consider the following simple example.
There are 100 voters and 5 candidates $a,b,c,d,e$: 66 voters approve the set $\{a,b,c\}$, 33 voters approve $\{d\}$, and one voter approves $\{e\}$.
Assume we want to select a committee of size three. 
If we count by how many voters each candidate is approved, we see that $a$, $b$, and $c$ are approved most often (66 times). This can be seen as a good reason to choose the committee $\{a, b, c\}$ based on these preferences; this committee contains the ``strongest'' candidates.
Note, however, that this committee essentially ignores the preferences of 34 voters.
Instead, one could choose the committee $\{a, d, e\}$, in which every voter finds one approved candidate.
A more proportional committee would be $\{a, b, d\}$: here, the 66 voters approving $\{a,b,c\}$ (which are roughly two-thirds of the population) have two approved candidates in the committee (two-thirds of the committee).
All of these committees are sensible and it is easy to find arguments for and against them.
For now, let us just observe that we need a principled way in which we can distinguish the properties of committees and ABC voting rules.

In recent years, much progress has been made in the field of ABC elections.
This can be seen when comparing the content of this book with
the comprehensive overviews by \citet{kil-handbook} and \citet{kil-mar:j:minimax-approval}, published in 2010 and 2012, respectively.
Indeed, multi-winner elections have been extensively studied from the perspective of economics (in the field of \emph{social choice theory} \citep{kil-handbook,kil-mar:j:minimax-approval}), political science (in the context of political elections and voting systems \citep{Far11,renwick16}) and artificial intelligence (in the field of \emph{computational social choice} \cite{Handbook-COMSOC}).
The goal of this book is to provide an up-to-date summary of the state of the art. A particular focus is put on axiomatic and algorithmic analysis; this line of work is prevalent in social choice theory and computational social choice.

Broadly speaking, we want to answer two main questions  in this book:
\begin{enumerate}[label=(\arabic*)]
\item What are the main properties of established ABC rules? Based on which principles can one choose a good ABC rule for a given application? (The answer to this question usually depends on the types of properties that the reader considers particularly important for his or  her application.) 
\item What are the practical limitations of using a particular rule, and how can one deal with these limitations? This question encompasses, e.g., algorithmic questions regarding computational complexity, and the possibility of conflicting axiomatic properties.
\end{enumerate}

Before we delve into ABC voting rules, let us first take a step back and discuss advantages and disadvantages of making collective decisions based  on approval preferences.

\section{Advantages and Drawbacks of Approval Ballots}
There are several arguments for using approval ballots in multi-winner elections (i.e., to work with approval preferences). Compared to the ranking-based model, where voters provide complete rankings of candidates (i.e., linear orders), providing approval preferences requires much less cognitive effort from the voters.
Thus this kind of voting is often more practical and preferable due to its clear meaning. \citet{bramsScience} and \citet{aragones2011making} discuss positive effects of using approval ballots on voters' participation, and \citet{bramsScience} further argue that using approval ballots can reduce negative campaigning; \citet{brams1983approval} discuss other possible positive implications of using approval ballots in political elections. In fact, approval ballots are often used for voting in  scientific societies (see, e.g., the work of \citet{brams2010}). 
Further experimental studies explore the possibility of using approval ballots in political elections and their conclusions are largely positive \cite{laslier2016StrategicVoting, alos2012two, LasStr17, las-str:j:approval-experiment}. Approval ballots are widely used in participatory budgeting~\cite{goel2015knapsack}. These are elections where the citizens decide through voting how to spend a municipal budget (we discuss this setting in \Cref{sec:part-budg}).

In general, the approval-based model has the advantage of a simple yet expressive preference model.
This simplicity grants definitions of more complex concepts within this model  (e.g., proportionality, strategyproofness, etc.) a solid intuitive grounding.

\medskip

The simplicity of approval ballots necessarily also has downsides.
An important underlying assumption is that the preferences of voters are separable, i.e., voters are not given the possibility to specify relations between candidates. For example, it is not possible for a voter to indicate that she believes that a certain group of candidates would work particularly well together in the elected committee or that she thinks that two candidates should never be elected together. We discuss several related models that allow voters to specify this kind of information in \Cref{sec:combinatorial_voting}.

Approval ballots imply a dichotomy between candidates: approved and disapproved candidates.
While it is generally clear how to interpret the set of approved candidates, it is less
clear how to interpret the set of disapproved candidates, i.e., its complement.
Generally, it can be assumed that a voter prefers approved candidates to be included in the winning committee, that is, adding an approved candidate to the committee will increase
a voter's satisfaction.
For disapproved candidates, the situation is less clear as the voter may be either neutral
about whether these candidates are included or opposed to their inclusion (or a mixture of these two cases).
As the ABC model does not allow to distinguish between neutral and negative candidates, this
information cannot be taken into account by ABC rules.
We discuss the trichotomous (three-valued) model in \Cref{sec:multiwinner-trich}.
Generally, moving from dichotomous preferences (the ABC model) to trichotomous preferences results in a vastly different model, with its own advantages and disadvantages.

A much more elaborate discussion of the approval-based model can be found in the Handbook of Approval Voting \cite{av-handbook}.

\section{Python Code}
This book is closely connected with the \textsf{abcvoting} Python library \cite{abcvoting}. The ABC rules discussed in this book are available as Python code at \url{https://github.com/martinlackner/abcvoting} and are directly usable, e.g., in numerical experiments. 
To give a flavour how \textsf{abcvoting} looks like, we show here the code to compute winning committees for Proportional Approval Voting (PAV), an important ABC rule.

\begin{lstlisting}[language=Python]
from abcvoting.preferences import Profile
from abcvoting import abcrules

# a preference profile with 5 candidates (0, 1, 2, 3, 4)
profile = Profile(5)

# add six voters, specified by the candidates that they approve;
# the first voter approves candidates 0, 1, and 2,
# the second voter approves candidates 0 and 1, etc.
profile.add_voters([{0,1,2}, {0,1}, {0,1}, {1,2}, {3,4}, {3,4}])

# compute winning committees
committees = abcrules.compute_pav(profile, committeesize=3)
\end{lstlisting}

Many examples from this book are also available in the \textsf{abcvoting} library, including the counterexamples from \Cref{app:proofs}.
If the reader prefers a Python-based hands-on approach, this library can be a useful tool.

\section{Mathematical Notation and Prerequisites}

We use the following basic notation.
We write $\N$ to denote the set of non-negative integers and $\reals$ to denote the set of real numbers.
Given a real number $x$, the \emph{floor function} $\lfloor x \rfloor$ returns the largest integer $\leq x$.
Similarly, the \emph{ceiling function} $\lceil x \rceil$ returns the smallest integer $\geq x$.
For each $t \in \naturals$, we let $[t]$ denote the set $\{1, \ldots, t\}$. For a set $X$, we write $|X|$ to denote its cardinality.
We further write $\powerset(X)$ to denote the \emph{powerset} of $X$, i.e., the set of all subsets of $X$.

A \emph{weak order} is a binary relation on a set~$X$ which is complete and transitive.
A \emph{linear order} is a weak order that is antisymmetric; we refer to linear orders also as \emph{rankings}.
Observe that weak orders may contain ties between elements (in contrast to linear orders).

We use the standard \emph{asymptotic notation} $O(.)$, $o(.)$ and $\Theta(.)$, denoting upper, lower, and tight bounds up to constant factors, respectively.

\medskip

We assume that the reader is familiar with basic concepts regarding algorithms (such as polynomial-time vs exponential-time algorithms, the concepts of fixed-parameter and approximation algorithms) and computational complexity theory (such as $\np$-hardness, $\np$-completeness, reductions). These concepts are, however, only required for \Cref{sec:algorithms}.

\section{Structure of the Book}
This book is structured as follows. 
In \Cref{sec:abc_rules}, we give detailed descriptions and examples for many approval-based  committee rules. Only parts of this chapter are required for understanding the remainder of the book; these parts are marked with a bar on the side of the page.
\Cref{sec:basic} provides an overview of basic axiomatic properties of ABC rules. We discuss which of these properties are satisfied by the rules introduced in the previous chapter.
In \Cref{sec:proportionality}, we focus on a major topic in recent years: proportional representation.
We discuss concepts of proportionality (but also concepts of non-proportionality) and their relation to other axiomatic properties.
\Cref{sec:algorithms} discusses the computational results concerning the complexity of computing winning committees, and algorithmic questions related to proportionality and strategyproofness.
In \Cref{sec:related_formalisms}, we provide an overview of related formalisms and their connection to ABC rules.
Finally, in \Cref{sec:outlook}, we provide an outlook on important research directions and list some specific open questions.
This book contains a technical appendix, \Cref{app:proofs}, with proofs and counterexamples that we were not able to find in the published literature.

\bibliographystyle{abbrvnat}
\bibliography{main}

\chapter{Dramatis Personae: ABC Rules}\label{sec:abc_rules}

\intro{In this chapter, we define the basic ingredients of approval-based committee (ABC) voting: candidates, voters, preferences, and committees. 
Most importantly, we present the main characters of this book: ABC voting rules.
We introduce and define the most important ABC rules and discuss the main classes they belong to.
These include Thiele methods and their sequential variants, Monroe's rule, \phragmen's rules and its derivatives, as well as non-standard ABC rules.
}

In this chapter, we define the basic ingredients of approval-based committee (ABC) voting: candidates, voters, preferences, and committees. 
Most importantly, we present the main characters of this book: ABC voting rules.
We introduce and define the most important ABC rules and discuss the main classes they belong to.
These include Thiele methods and their sequential variants, Monroe's rule, \phragmen's rules and its derivatives, as well as non-standard ABC rules.

\section{The Formal Model}\label{sec:model}


We now define the basic ingredients of approval-based committee (ABC) voting: candidates, voters, preferences, committees, and ABC rules.

\subsection*{Candidates, Voters, and Preferences}
Let $C$ be a finite set of available \emph{candidates} (also called \emph{alternatives}).
We assume that voters' preferences are available in the form of \emph{approval preferences}, i.e., voters distinguish between alternatives they approve and those that they disapprove---due to this dichotomy such preferences are also called \emph{dichotomous preferences}.
Hence a voter's preference over candidates can be represented by a set of approved alternatives.
Let $N\subseteq \N$ denote a finite set of \emph{voters}.

An \emph{approval profile} is the collection of all voters' preferences; formally it is
a function $A: N \to \powerset(C)$.
We say that $A(i) \subseteq C$ is \emph{voter $i$'s approval ballot}.
Throughout the book, we use $n=|N|$ to denote the number of voters and $m=|C|$ to denote the number of alternatives.
Further, we write $N(c)$ to denote the subset of voters that approve candidate $c$, i.e., $N(c)=\{i\in N: c\in A(i)\}$.

\begin{example}\label{ex:running}
An academic society chooses a steering committee. 
Such a committee consists of four persons ($k=4$) and 
there are seven candidates competing for these positions, $C=\{a,b,c,d,e,f,g\}$.
All members of the society are eligible to vote and may provide approval ballots to indicate their
preference. In total, 12 ballots have been submitted ($N=[12]$):
	\begin{align*}
		&A(1) \colon \{ a,b\} & & A(2) \colon \{ a,b\} & & A(3) \colon \{ a,b\} && A(4) \colon \{ a,c\}\\*
		  & A(5) \colon \{ a,c\} && A(6) \colon \{ a,c\} 
		&&  A(7)\colon  \{ a, d\} && A(8)\colon  \{ a, d\} \\*
		&   A(9)\colon \{ b, c, f\}	&& A(10)\colon \{ e\} &  & A(11)\colon \{ f\}   & & A(12)\colon  \{ g\} \text{.}
	\end{align*}
\Cref{fig:runningexample} shows a graphical representation of this approval profile. In this figure, each column correspond to one voter (one approval set) and each candidate appears in only one row---each candidate is approved by the voters that appear below the boxes that represent the candidate. 
Colours are used to distinguish different candidates.
\begin{figure}
\begin{center}
\begin{tikzpicture}
[yscale=1.0,xscale=2.0]
\newcommand{\width}{.6}

\filldraw[fill=green!10!white, draw=black] (-\width/2,0) rectangle (-\width/2 + 8 * \width,0.5)
node[pos=.5] {$a$};
\filldraw[fill=green!30!white, draw=black] (-\width/2,0.5) rectangle (-\width/2 + 3 * \width,1.0)
node[pos=.5] {$b$};
\filldraw[fill=green!30!white, draw=black] (-\width/2+ 8 * \width,0.5) rectangle (-\width/2 + 9 * \width,1.0)
node[pos=.5] {$b$};
\filldraw[fill=red!10!white, draw=black] (-\width/2+ 3 * \width,1) rectangle (-\width/2 + 6 * \width,1.5)
node[pos=.5] {$c$};
\filldraw[fill=red!30!white, draw=black] (-\width/2+ 6 * \width,0.5) rectangle (-\width/2 + 8 * \width,1.)
node[pos=.5] {$d$};
\filldraw[fill=red!10!white, draw=black] (-\width/2+ 8 * \width,1) rectangle (-\width/2 + 9 * \width,1.5)
node[pos=.5] {$c$};
\filldraw[fill=blue!20!white, draw=black] (-\width/2+ 8 * \width,0) rectangle (-\width/2 + 9 * \width,0.5)
node[pos=.5] {$f$};
\filldraw[fill=cyan!10!white, draw=black] (-\width/2+ 9 * \width,0) rectangle (-\width/2 + 10 * \width,0.5)
node[pos=.5] {$e$};
\filldraw[fill=blue!20!white, draw=black] (-\width/2+ 10 * \width,0) rectangle (-\width/2 + 11 * \width,0.5)
node[pos=.5] {$f$};
\filldraw[fill=cyan!30!white, draw=black] (-\width/2+ 11 * \width,0) rectangle (-\width/2 + 12 * \width,0.5)
node[pos=.5] {$g$};

\foreach \x in {1,...,12}
\node at (-\width + \x*\width, -0.35) {$\x$};

\draw [decorate,decoration={brace,amplitude=10pt},yshift=0pt]
(11.5*\width,-0.6) -- (-0.5*\width,-0.6) node [black,midway,below,yshift=-8pt] 
{voters};
\draw [decorate,decoration={brace,amplitude=10pt},yshift=0pt]
(-0.6*\width,-0.1) -- (-0.6*\width,1.5) node [black,midway,below,xshift=-27pt, rotate=90] 
{candidates};
\end{tikzpicture}
\caption{Graphical representation of the approval profile from \Cref{ex:running}. Each candidate is represented by one or several boxes that appear in a single row in the figure, and that are marked with a candidate-specific colour. A voter approves those candidates whose corresponding boxes appear above the voter. For example, voter $1$ approves candidates $a$ and $b$ and voter 4 approves candidates $a$ and $c$.}
\label{fig:runningexample}
\end{center}
\end{figure}

Sometimes, we are only interested in how often a specific approval set occurs in an approval profile and thus ignore the names (identifiers) of the voters who cast the approval ballots.
In such cases, we do not specify the concrete mapping from $N$ to approval sets but use the following notation:
	\begin{align*}
		&3 \times \{ a,b\} & & 3 \times \{ a,c\} &&  2 \times \{ a, d\} &&  1 \times \{ b, c, f\}\\
		 & 1 \times \{ e\} &  & 1 \times \{ f\}   & & 1 \times \{ g\} \text{.}
	\end{align*}

The reader may ponder which steering committee of size $k=4$ should be selected given this approval profile---there 
is certainly more than one sensible choice.
In the following chapter, we will see how different voting rules decide in this situation.
\end{example}

We do not make assumptions about the size of approval ballots, as we assume that it is the voters' decision how many candidates she approves.
In applications, however, there is sometimes an upper limit on how many candidates can be approved (often the desired committee size). Such a requirement has hardly any effect on the results presented in this book.
In a richer model where voters have underlying, non-dichotomous (i.e., non-binary) preferences, such a restriction would become more relevant; this effect has been analysed by \citet{xiao2019optimal} and \citet{god-bat-sko-fal:c:2d-abc}. The main conclusion is that it is typically better to give the voters freedom in choosing how many candidates they wish to approve. 

\subsection*{Committees and ABC Rules}

As we have seen in \Cref{ex:running}, committees are sets of candidates. Typically, we are interested in committees of a specific size, which we denote by $k$.
The input for choosing such a committee is an \emph{election instance} $E=(A, k)$ consisting of a preference profile $A$ and a desired committee size $k$.
Note that given $A$, we can derive $N$ and $C$ from this function: $N$ is the domain of $A$  and---under the mild assumption that all candidates are approved by at least one voter---$C$ is the union of all function values, i.e., $C=\bigcup_{i\in N} A(i)$. Thus we do not mention $N$ and $C$ in this notation.

Let us now define the key concept of this book: \emph{approval-based committee voting rules} (short: \emph{ABC rules}).
An ABC rule is a voting method for choosing committees, i.e., an ABC rule takes an election instance as input and outputs one or more size-$k$ subsets of candidates.
We refer to these size-$k$ subsets as \emph{winning committees}.

If an ABC rule outputs more than one committee, we say that these committees are \emph{tied}. 
An ABC rule is \emph{resolute} if it always outputs exactly one committee.
In practical settings, it is often undesirable to have more than one winning committee.
Consequently, in many concrete voting systems a tiebreaking method is included so that a resolute outcome is guaranteed.
This tiebreaking method is typically a random process. 
As we assume that an ABC rule is a deterministic process, we further assume that all randomisation is done before the election (or at least before the ABC rule is applied).
Under this assumption, a randomised tiebreaking method corresponds to a fixed (linear) tiebreaking order over committees; if more than one committee is winning, this tie is resolved by picking the winning committee that is maximal in the tiebreaking order. 
In this sense, our model incorporates voting systems that rely on randomised tiebreaking.\footnote{For a more careful study of randomised tie-breaking, one would have to 
model the outcome of a randomised ABC rule as a probability distribution over potentially winning committees.
Note that this distribution is not necessarily uniform.}

Some of the ABC rules defined in the following are resolute, i.e., they always return a single winning committee, and some are irresolute. Most rules can be defined either way; we have chosen the more natural definition for each rule.
%



\begin{changebar}
We are now giving an overview of major ABC voting rules.
For readers who are less interested in this overview, we have marked the most relevant ABC rules with a bar on the side of the page (as shown here); one can largely follow the book with knowledge of these rules only.
\end{changebar}

For the following definitions, we assume that we are given an election instance $E=(A,k)$ with
a voter set~$N$ and a candidate set~$C$.

\section{Thiele Methods}\label{subsec:thiele}

In the single-winner setting, i.e., if $k=1$, there are few reasonable voting rules
when presented with approval ballots. The arguably most natural rule is Approval Voting.
Approval Voting selects those alternatives that are approved by the maximum number of voters, all of which are \mbox{(co-)}winners according to this rule.
Most ABC rules introduced in this chapter are equivalent to Approval Voting for the case $k=1$ (we discuss notable exceptions in \Cref{sec:furtherabc}).
There is, however, one ABC rule that extends the reasoning of Approval Voting
to $k>1$ in the most natural manner; this rule is therefore called
Multi-Winner Approval Voting (short: AV).%
\footnote{Let us briefly mention variants of AV that are widely used in political settings: Block Voting, where voters may not approve more than $k$ candidates (or sometimes exactly $k$), Limited Voting, where voters may approve at most $s$ candidates with $s<k$, and Single Non-Transferable Vote (SNTV), which is Limited Voting for $s=1$. 
Note that properties of AV do not necessarily transfer to these input-restricted variants and vice-versa. For example, forcing voters to approve exactly $k$ candidates appears to have severe negative consequences, as demonstrated by \citet{2dpictures} in numerical experiments.
In this book we consider only AV, which allows arbitrary approval ballots as input.}

\begin{abcrule}[Multi-Winner Approval Voting, AV]
\begin{changebar}
This ABC rule selects the $k$ candidates which are approved by most voters. Formally, the AV-score of an alternative $c\in C$ is defined as $\score{\av}(A, c) = |N(c)| = |\{i\in N: c\in A(i)\}|$ and AV selects committees $W$ that maximise $\score{\av}(A, W) = \sum_{c\in W}\score{\av}(A, c)$.  \end{changebar}\label{rule:av}
\end{abcrule}

\begin{example}
Let us consider the instance of \Cref{ex:running}:
	\begin{align*}
		&3 \times \{ a,b\} & & 3 \times \{ a,c\} &&  2 \times \{ a, d\} &&  1 \times \{ b, c, f\}
		 & 1 \times \{ e\} &  & 1 \times \{ f\}   & & 1 \times \{ g\} \text{.}
	\end{align*}
To compute winning committees according to AV, we count how often each alternative is approved:
$a$: $8$ times, $b$: $4$, $c$: $4$, $d$: $2$, $e$: $1$, $f$: $2$ and $g$: $1$.
We want to select the four most-approved alternatives.
These are $a$, $b$, $c$, and there is a tie between $d$ and $f$ (both having the fourth highest number of approvals).
Hence, AV returns
two tied committees: the sets $W_1=\{a,b,c,d\}$ and $W_2=\{a,b,c,f\}$.
It is noteworthy that $W_1$ leaves three voters completely unsatisfied with the chosen
alternatives, whereas $W_2$ results in only two completely unsatisfied voters.
\end{example}

We continue with an ABC rule that can be seen as the exact opposite of AV. Whereas AV disregards whether some voters completely disagree with a committee, the Approval Chamberlin--Courant rule grants as many voters as possible at least one approved alternative in the committee.
This rule was first mentioned by Thiele\footnotemark~\citep{Thie95a}, and then independently introduced in a different context by Chamberlin and Courant~\citep{ccElection}.

\footnotetext{Thorvald Nicolai Thiele (1838--1910) was a Danish astronomer and mathematician. He was professor of astronomy at the University of Copenhagen and director of the Copenhagen University Observatory. He is most known for his work in mathematics, in particular in statistics~\cite{thiele1903theory,thiele1909interpolationsrechnung,lauritzen2002thiele,thiele-history}. The contributions of Thiele to voting theory are discussed in detail by \citet{Janson16arxiv}.}

\begin{abcrule}[Approval Chamberlin--Courant, CC]\begin{changebar}
The CC rule outputs
all committees $W$ that maximise $\score{\cc}(A, W) = |\{i\in N: A(i)\cap W\neq \emptyset\}|$.
\end{changebar}\end{abcrule}

\begin{example}
Considering again the instance of \Cref{ex:running},
there is exactly one committee that grants each voter
(at least) one approved candidate: $W=\{a, e, f, g\}$.
This is the winning committee according to Approval Chamberlin--Courant.
While this committee indeed provides some satisfaction for every voter,
it includes alternatives ($e$ and $g$) that are approved only by single voters.
\end{example}

The two ABC rules we discussed so far---AV and CC---can be seen as extreme points in the spectrum of ABC rules captured by the class of \emph{Thiele methods}.
This class, introduced by Thiele in the late 19th century~\citep{Thie95a},
encompasses all rules that maximise the sum of the voters' individual satisfaction, subject to a chosen definition of how satisfaction is measured.
The unifying assumption is that a voters' satisfaction with a committee~$W$ is solely determined by the number of approved candidates in this committee, i.e., voter~$i$'s satisfaction is determined by a function $w\left(|W \cap A(i)|\right)$.
By choosing different $w$-functions, a very broad spectrum of ABC rules can be covered.

\begin{abcrule}[Thiele methods, $w$-Thiele\footnotemark]
\begin{changebar}
A Thiele method is parameterized by a
non-decreasing function $w:\naturals\to\reals$ with $w(0)=0$. The score of a committee $W$ given a profile $A$ is defined as \[\score{w}(A, W) = \sum_{i \in N} w\left(|W \cap A(i)|\right);\] the $w$-Thiele method returns committees with maximum score.
\end{changebar}
\end{abcrule}

\footnotetext{The class of Thiele methods is sometimes also referred to as \emph{weighted
PAV rules}~\cite{justifiedRepresentation}; we prefer the term \emph{Thiele methods} as only few rules in this class are actually proportional. \citet{kil-mar:j:minimax-approval} refer to this class as \emph{generalised approval procedures}.}

Indeed, AV is the $w$-Thiele method with $w(x)=x$, and CC is the $w$-Thiele method with $w(x)=\min(1, x)$.
This is an immediate consequence of the respective definitions.

The following Thiele method is arguably one of the most important: Proportional Approval Voting, in short PAV. Also this rule was defined in Thiele's original paper~\cite{Thie95a}.
The definition (and properties) of PAV crucially depend on the sequence of harmonic numbers.

\begin{abcrule}[Proportional Approval Voting, PAV]\begin{changebar}
Let $\H(x) = \sum_{j=1}^{x} \nicefrac{1}{j}$ denote the \emph{sequence of harmonic numbers}.
PAV is $\H$-Thiele, i.e., it is the $w$-Thiele rule with $w(x)=\H(x)$.
In other words, PAV assigns to each committee $W$ the PAV-score, $\score{\pav}(A, W) = \sum_{i \in N} \H\left(|W \cap A(i)|\right)$, and returns all committees with maximum score.
\end{changebar}\end{abcrule}

\begin{figure}
\begin{center}
\begin{tikzpicture}[y=.6cm, x=.8cm,font=\sffamily]

\pgfdeclarelayer{bg}
\pgfsetlayers{bg,main}
    \draw (0,0) -- coordinate (x axis mid) (8.5,0);
        \draw (0,0) -- coordinate (y axis mid) (0,8.5);
        \foreach \x in {0,...,8}
             \draw (\x,1pt) -- (\x,-3pt)
            node[anchor=north] {\x};
        \foreach \y in {0,2,...,8}
             \draw (1pt,\y) -- (-3pt,\y) 
                 node[anchor=east] {\y}; 
    \node[below=0.8cm] at (x axis mid) {number of approved candidates in committee ($x$)};
    \node[rotate=90, above=0.8cm] at (y axis mid) {satisfaction $w(x)$};

    \draw plot[mark=*, mark options={fill=white}] 
        file {av.data};
    \draw plot[mark=triangle*, mark options={fill=white} ] 
        file {pav.data};
    \draw plot[mark=square*]
        file {cc.data}; 
    
    \node[right=0.2cm] at (8,1) {CC};
    \node[right=0.2cm] at (8,2.82441700960219) {PAV};
    \node[right=0.2cm] at (8,8) {AV};        
    
\end{tikzpicture}
\caption{Defining $w$-functions for three Thiele methods: Multi-Winner Approval Voting (AV), Proportional Approval Voting (PAV), and Approval Chamberlin--Courant (CC).}
\label{fig:countingfcts}
\end{center}
\end{figure}

By using the sequence of harmonic numbers $\H(\cdot)$, 
we introduce a flattening satisfaction function for voters, akin to the \emph{law of diminishing returns}.
As a consequence, PAV balances the (justified) demands of large groups with the conflicting goal of satisfying small groups.
Indeed, as we will see in \Cref{sec:proportionality}, Proportional Approval Voting achieves this balance in a proportional fashion. 
\Cref{fig:countingfcts} shows a visualisation of the defining $w$-functions of different Thiele methods: 
\begin{align*}
w_{\av}(x) = x & &  w_{\pav}(x) = \sum_{i=1}^{x} \nicefrac{1}{i} & & w_{\cc}(x) = \begin{cases}0 & \text{if }x = 0,\\1 & \text{if }x\geq 1 \text{.}\end{cases}
\end{align*}
Note that also visually the function defining PAV is ``in between'' AV and CC.

\begin{example}
Given the instance of \Cref{ex:running}:
	\begin{align*}
		&3 \times \{ a,b\} & & 3 \times \{ a,c\} &&  2 \times \{ a, d\} &&  1 \times \{ b, c, f\}
		 & 1 \times \{ e\} &  & 1 \times \{ f\}   & & 1 \times \{ g\} \text{,}
	\end{align*}
PAV selects the committee $W=\{a,b,c, 	f\}$.
For one voter (the one that approves $\{b,c,f\}$) this committee contains three approved alternatives,
for six voters this committee contains two approved alternatives,
for three voters $W$ contains one approved alternative,
and two voters are not at all satisfied with $W$.
Thus, we have $\score{\pav}(A, W) = (1+\nicefrac 1 2+\nicefrac 1 3) + 6 \cdot (1+\nicefrac 1 2) + 3 \cdot 1 = \nicefrac{83}{6}$ and this value is optimal.
Coincidentally, $W$ is one of the two committees produced by AV, namely the one with fewer dissatisfied voters.
It appears that PAV strives for a compromise between AV and CC---this is an intuition that we will discuss 
in more detail later (\Cref{sec:degressive_and_regressive_proportionality}).
\end{example}

Other Thiele methods that have been studied in the literature are the class of $p$-geometric rules~\citep{owaWinner}, threshold procedures~\cite{fishburn2004approval,kil-mar:j:minimax-approval}, and Sainte-Lagu\"e Approval Voting (SLAV)~\cite{lac-sko2019}.


Thiele methods pick committees that maximise a certain welfare of the voters and thereby belong to a broader class of welfarist rules.

\begin{definition}\label{def:welfarist_rules}
A \emph{welfare vector} induced by a committee $W$ specifies, for each voter, her satisfaction from $W$ (measured as the number of candidates she approves in $W$):
\begin{align*}
 \welf(W) = (|A(1) \cap W|, |A(2) \cap W|, \ldots, |A(n) \cap W|) \text{.}
\end{align*}
An ABC rule $\calR$ is \emph{welfarist} if there is a function $f\colon \naturals^N \to \reals$, mapping welfare vectors to scores, such that for each instance $(A, k)$ we have \[\calR(A, k) = \argmax_{W\subseteq C\text{ with } |W|=k} f(\welf(W)).\]
\end{definition}

In this definition, $f(\welf(W))$ can be viewed as the welfare that voters gain from $W$.
For Thiele methods, $f(\welf(W)) = \score{w}(A, W)$, i.e., welfare is the sum of the voters' $w$-scores.
The class of welfarist rules also allows for an aggregation other than summation.
For example, one can define $f(\welf(W))$ as the satisfaction of the least-satisfied voter---akin to egalitarian aggregation~\cite{moulinAxioms}.
Another example of a welfarist rule is a dictatorial rule which compares welfare vectors lexicographically given a fixed order of voters: the first voter in this order is a dictator and only if the dictator is indifferent between two outcomes, the second-in-place may decide, and so on.

These other forms of aggregation have been studied in the context of multi-winner elections with ranking-based preferences (for the egalitarian aggregation see the work of \citet{AFGST-egalitarian} and \citet{sko-fal-sli:j:multiwinner}; for OWA-based aggregation see the work of~\citet{elk-ism:c:owa-egalitarian-utilitarian} and \citet{FaliszewskiSST17}). For approval ballots, we are aware of only two works that consider such aggregations. Computational properties of CC and Monroe rules based on the egalitarian aggregation are considered by \citet{fullyProportionalRepr}. \citet{amanatidis2015multiple} consider OWA-based aggregation but for other types of welfare of individual voters.
Specifically, the satisfaction of voters with a committee is measured via the Hamming distance, which is in contrast to the definition of $\welf(W)$.
The most important rule based on the Hamming distance is Minimax Approval Voting, which we discuss in \Cref{sec:furtherabc}.

\section{Sequential Variants of Thiele Methods}\label{sec:seq-rules}

Thiele methods are defined via optimisation statements: given an objective function, Thiele methods return all committees that maximise this function.
Instead of computing the true optimum (which is computationally hard, as we will see in \Cref{sec:algorithms}), one can define sequential procedures that construct an approximate solution.
We define here two classes of sequential procedures: sequential and reverse sequential Thiele methods.
Both classes have been introduced in Thiele's original paper~\cite{Thie95a} (see Janson's survey for further historical remarks~\cite{Janson16arxiv}).
Furthermore, both classes can be seen as greedy approximation algorithms for Thiele methods; we return to this analogy in \Cref{sec:compute-approx}.

Let us begin with sequential Thiele methods: starting with an empty committee, they add committee members one by one, in each step the one that increases the objective function the most.

\begin{abcrule}[Sequential $w$-Thiele, seq-$w$-Thiele]\begin{changebar}
For each $w$-Thiele method, we define its sequential variant, seq-$w$-Thiele, as follows. We start with an empty committee $W_0 = \emptyset$.
In each round $r\in\{1,\dots,k\}$, we compute $W_r=W_{r-1} \cup \{c\}$, where $c$ is a candidate that maximises $\score{w}(A, W_{r-1} \cup \{c\})$, i.e., the candidate that improves the committee's score the most. If more than one candidate yields a maximum score, we break ties according to some given tie-breaking order. The seq-$w$-Thiele rule returns $W_k$.
\end{changebar}\end{abcrule}

Two sequential Thiele methods will be of particular interest here: sequential $w_\pav$-Thiele and sequential $w_\cc$-Thiele. We refer to these two rules as \emph{seq-PAV} and \emph{seq-CC}.
In contrast, the sequential variant of AV (seq-$w_\av$-Thiele) is not relevant to us as it is equivalent to AV. This is because the AV-score ($\score{\av}$) of candidates is not influenced by the other candidates in the committee.


\begin{example}\label{ex:seqpav}
Since the instance of \Cref{ex:running} yields the same result for PAV and seq-PAV (and also for CC and seq-CC), we take a look at a different profile:
	\begin{align*}
		&3 \times \{ a,b \} & & 6 \times \{ a, d\} && 4 \times \{ b \} & & 5 \times \{c\} & & 5 \times \{c, d\} \text{.}
	\end{align*}
For $k=2$, PAV selects the committee $\{a,c\}$ with a PAV-score of $19$.
(Each voter except those that approve only candidate $b$ has exactly one approved candidate in the committee.)
Let us contrast this result with seq-PAV. 
All sequential Thiele methods with $w(1)>0$, including seq-PAV, select the candidate with the largest number of approvals in the first round---the winner according to (single-winner) Approval Voting.
Thus, $d$ is selected in the first round as it gives an AV-score of $11$.
In the second round, we choose between $a$ (increasing the score by $6$) and $b$
(increasing the score by $7$) and $c$ (increasing the score by $7.5$). Hence, seq-PAV returns the committee $\{c, d\}$ with a PAV-score of $18.5$.
\end{example}

Similarly to sequential Thiele methods, \emph{reverse sequential Thiele methods} build committees sequentially, but here one starts with the set of all candidates and sequentially removes the candidate that contributes the least to the committee's score.\footnote{This idea of removing candidates with the lowest score can also be found in ranking-based voting rules such as STV or Baldwin \cite{Handbook-comsocintro}.}

\begin{abcrule}[Reverse Sequential $w$-Thiele, rev-seq-$w$-Thiele] For each $w$-Thiele method, we define its reverse sequential variant, rev-seq-$w$-Thiele, as follows. We start with $W_m = C$, the set of all candidates. 
Each round, the candidate with the least marginal contribution to the score is removed.
To be precise, in each round $r$ from $m-1$ down to $k$, we compute $W_{r}=W_{r+1} \setminus \{c\}$, where $c$ is a candidate that maximises $\score{w}(A, W_{r+1} \setminus \{c\})$, i.e., the candidate whose removal decreases the committee's score the least. 
If more than one candidate does that, we break ties according to some given tie-breaking order. The rev-seq-$w$-Thiele rule returns $W_k$.
\end{abcrule}

In the remainder of the book, we will only encounter \emph{reverse sequential PAV (rev-seq-PAV)} from the class of Reverse Sequential $w$-Thiele methods.

\addtocounter{example}{-1}
\begin{example}[continued]
For rev-seq-PAV, 
we start with the full set of candidates $W_4=\{a,b,c,d\}$ and remove the candidate with the least marginal contribution:
removing $a$ decreases the score by $4.5$, removing $b$ decreases the score by $5.5$, $c$ by $7.5$, and $d$ by $5.5$. Thus, $a$ is removed and $W_3=\{b,c,d\}$. Now, we again compute the marginal contributions: for $b$ it is $7$, for $c$ it is $7.5$, and for $d$ it is $8.5$. We obtain $W_2=\{c,d\}$, which is the winning committee.
We see that for this instance seq-PAV and rev-seq-PAV yield the same winning committee. This does not hold in general.

An election instance where PAV, seq-PAV, and rev-seq-PAV all yield different winning committees can be found in Janson's survey~\cite[Example 13.3]{Janson16arxiv}. The example is due to Thiele~\cite{Thie95a} and is significantly larger than the examples presented here.
\end{example}

As we have mentioned in \Cref{subsec:thiele}, most ABC rules coincide with Approval Voting for $k=1$. Reverse Sequential PAV is an exception.
This is, however, not a consequence of the underlying assumptions how ballots are interpreted, but a consequence of how the rule is computed (i.e., in a reverse fashion).

\begin{example}
To see that rev-seq-PAV is a non-standard method, consider the profile:
	\begin{align*}
		&1 \times \{ a, b\} & & 1 \times \{ a, b,c\} & & 1 \times \{ a, b,d\}& & 2 \times \{ a, c,d\} & & 1 \times \{ b\} & & 1 \times \{ c\} & & 1 \times \{ d\} \text{.}
	\end{align*}
In the first round, the marginal contribution of $a$ is $\nicefrac 1 2 + 4\cdot \nicefrac 1 3$; the marginal contribution from the other candidates is at least $2$. Thus, candidate $a$ is removed in the first round, even though it has the highest approval score.
\end{example}

Finally, let us mention a paper by \citet{FLPT-effective} which considers and compares several heuristic algorithms for approximating multi-winner rules (e.g., via simulated annealing).
This line of work has not yet been extended specifically to Thiele methods, though the ideas in their work can be applied to the ABC setting.

\section{Monroe's Rule}

Monroe's rule~\citep{monroeElection} is an ABC rule\footnotemark{} related to the Chamberlin--Courant rule.
It also aims at maximising the number of voters who are represented by at least one candidate in the elected committee.
The main difference is that each committee member can represent at most $\nicefrac 1 k$-th of the voters.

\footnotetext{
Although Monroe defined his rule in the original paper primarily for linear preference orders~\citep{monroeElection},
he considered the modified version based on approval ballots the ``most promising option'' for actual (political) use. If the distinction between these two rules is necessary, the approval-based version is often denoted as $\alpha$-Monroe; we do not need this distinction as we focus solely on approval ballots.}

\begin{abcrule}[Monroe]
Given a committee $W$,
a Monroe assignment for $W$ is a function $\phi \colon N \to W$ such that each candidate $c \in W$ is assigned roughly the same number of voters,
i.e., for all $c\in W$ it holds that $\lfloor \nicefrac{n}{k} \rfloor \leq |\phi^{-1}(c)| \leq \lceil \nicefrac{n}{k} \rceil$.
The candidate $\phi(i)$ can be viewed as the representative of voter $i$.
Let $\Phi(W)$ be the set of all possible Monroe assignments for $W$.
The Monroe-score of a committee $W$ is defined as the number of voters that have a representative assigned that they approve (given an optimal Monroe assignment), i.e.,
$\score{\monroe}(A, W) = \max_{\phi \in \Phi(W)} |\{i\in N: \phi(i) \in A(i)\}|$.
Monroe returns all committees with a maximum Monroe score.
\end{abcrule}

\begin{example}\label{ex:monroe}
Consider again the profile of \Cref{ex:running}:
	\begin{align*}
		&A(1) \colon \{ a,b\} & & A(2) \colon \{ a,b\} & & A(3) \colon \{ a,b\} && A(4) \colon \{ a,c\}\\
		  & A(5) \colon \{ a,c\} && A(6) \colon \{ a,c\} 
		&&  A(7)\colon  \{ a, d\} && A(8)\colon  \{ a, d\} \\
		&   A(9)\colon \{ b, c, f\}	&& A(10)\colon \{ e\} &  & A(11)\colon \{ f\}   & & A(12)\colon  \{ g\} \text{.}
	\end{align*}
We first note that the desired committee size $k=4$ divides the number of voters ($n=12$) and hence  Monroe assigns exactly $3$ voters to each committee member.
One optimal Monroe assignment (among many) is shown in \Cref{fig:monroe} and given by $\phi^{-1}(a)=\{3, 7, 8\}$, $\phi^{-1}(b)=\{1, 2, 9\}$, $\phi^{-1}(c)=\{4, 5, 6\}$, $\phi^{-1}(e)=\{10, 11, 12\}$. The Monroe score of this assignment is $\score{\monroe}(A, W) = 10$, since only voters $11$ and $12$ are assigned to a representative (candidate $e$) that they do not approve. In total there are six winning committees; committee $\{a,b,c,e\}$ is one of them.
\end{example}

\begin{figure}
\begin{center}
\begin{tikzpicture}
[yscale=1.0,xscale=2.0]
\newcommand{\width}{.65}

\filldraw[fill=green!10!white, draw=black] (-\width/2,0) rectangle (-\width/2 + 8 * \width,0.5)
node[pos=.5] {$a$};
\filldraw[fill=green!30!white, draw=black] (-\width/2,0.5) rectangle (-\width/2 + 3 * \width,1.0)
node[pos=.5] {$b$};
\filldraw[fill=green!30!white, draw=black] (-\width/2+ 8 * \width,0.5) rectangle (-\width/2 + 9 * \width,1.0)
node[pos=.5] {$b$};
\filldraw[fill=red!10!white, draw=black] (-\width/2+ 3 * \width,1) rectangle (-\width/2 + 6 * \width,1.5)
node[pos=.5] {$c$};
\filldraw[fill=white, draw=black] (-\width/2+ 6 * \width,0.5) rectangle (-\width/2 + 8 * \width,1.)
node[pos=.5] {$d$};
\filldraw[fill=red!10!white, draw=black] (-\width/2+ 8 * \width,1) rectangle (-\width/2 + 9 * \width,1.5)
node[pos=.5] {$c$};
\filldraw[fill=white, draw=black] (-\width/2+ 8 * \width,0) rectangle (-\width/2 + 9 * \width,0.5)
node[pos=.5] {$f$};
\filldraw[fill=cyan!10!white, draw=black] (-\width/2+ 9 * \width,0) rectangle (-\width/2 + 10 * \width,0.5)
node[pos=.5] {$e$};
\filldraw[fill=white, draw=black] (-\width/2+ 10 * \width,0) rectangle (-\width/2 + 11 * \width,0.5)
node[pos=.5] {$f$};
\filldraw[fill=white, draw=black] (-\width/2+ 11 * \width,0) rectangle (-\width/2 + 12 * \width,0.5)
node[pos=.5] {$g$};

\foreach \x in {1,...,12}
\node at (-\width + \x*\width, -0.35) {$\x$};

\newcommand{\mylist}{1/b, 2/b, 3/a, 4/c, 5/c, 6/c, 7/a, 8/a, 9/b, 10/e, 11/e, 12/e}
\foreach \x/\name in \mylist
\node at (-\width + \x*\width, 2.) {\bf \name};
\end{tikzpicture}
\caption{An optimal Monroe assignment for \Cref{ex:monroe}: the top row shows the assigned representative for each voter. For example, the assigned representative of voter $1$ is $b$; voter 12 is dissatisfied with her assigned representative $e$.}
\label{fig:monroe}
\end{center}
\end{figure}

Monroe's rule has also a natural sequential version called \emph{Greedy Monroe}, which was introduced by~\citet{sko-fal-sli:j:multiwinner}.\footnote{Greedy Monroe is called \emph{Algorithm A} in the original paper~\citep{sko-fal-sli:j:multiwinner} and is defined therein only for instances where $k$ divides $n$. The first general definition was given in~\cite{elk-fal-sko-sli:c:multiwinner-rules}.}
We present Greedy Monroe here in a slightly simpler, more practical fashion, where dissatisfied voters are not assigned to groups.


\begin{abcrule}[Greedy Monroe] This ABC rule proceeds in $k$ rounds: In each round $r\in\{1,\dots, k\}$ Greedy Monroe assigns a candidate to a group of voters $G_r$ of size at most $n_r$ (defined below); this candidate is added to the committee.
The maximum size of a group, $n_r$, is defined as follows: 
for $d = n \mod k$, 
we set $n_1=\dots=n_d=\lceil \nicefrac{n}{k} \rceil$ and $n_{d+1}=\dots=n_k=\lfloor \nicefrac{n}{k} \rfloor$.
In round $r+1$, let $N_{r+1}$ denote the voters that have not yet an assigned committee member, i.e., $N_{r+1}=N\setminus (G_1\cup\dots\cup G_{r})$.
Candidate $c_{r+1}$ is chosen as the candidate~$c$ that maximises $|\{i\in N_{r+1}: c\in A(i)\}|$ among those not contained in the committee yet (using a tiebreaking order on candidates if necessary).
Now, if there are at most $n_{r+1}$ not yet assigned voters that approve $c_{r+1}$, then $G_{r+1}=\{i\in N_{r+1}: c_{r+1}\in A(i)\}$; if there are more than $n_{r+1}$ such voters, a tiebreaking order on voters is used to assign exactly $n_{r+1}$ from these voters to $G_{r+1}$.
Greedy Monroe outputs the committee $\{c_1,\dots,c_k\}$.
\end{abcrule}

\begin{example}\label{ex:greedymonroe}
In our running example (\Cref{ex:running}) given by
	\begin{align*}
		&A(1) \colon \{ a,b\} & & A(2) \colon \{ a,b\} & & A(3) \colon \{ a,b\} && A(4) \colon \{ a,c\}\\
		  & A(5) \colon \{ a,c\} && A(6) \colon \{ a,c\} 
		&&  A(7)\colon  \{ a, d\} && A(8)\colon  \{ a, d\} \\
		&   A(9)\colon \{ b, c, f\}	&& A(10)\colon \{ e\} &  & A(11)\colon \{ f\}   & & A(12)\colon  \{ g\} \text{,}
	\end{align*}
Greedy Monroe first picks candidate $a$ as it is approved by most voters. We assume that ties among voters are broken in increasing order, so $G_1=\{1, 2, 3\}$.
Now $c$ is chosen since it is the only candidate with four supporters among the remaining voters ($N_2=\{4,\dots,12\}$).
The corresponding group of voters is $G_2=\{4,5,6\}$ (again choosing voters with smaller indices first).
Now there are two candidates left that are approved by two voters in the remaining set ($N_3=\{7,\dots,12\}$): candidates~$d$ and~$f$.
We choose $d$ by alphabetic tiebreaking and so we set $G_3=\{7,8\}$.
Finally, there is one candidate that has two supporting voters in $N_4=\{9,\dots,12\}$: $f$ is approved by voters~$9$ and~$11$; thus $G_4=\{9, 11\}$.
A Monroe assignment corresponding to this committee $\{a,c,d,f\}$ is, e.g., given by $\phi^{-1}(a)=\{1,2,3\}$, $\phi^{-1}(c)=\{4,5,6\}$, $\phi^{-1}(d)=\{7,8,10\}$, and $\phi^{-1}(f)=\{9, 11, 12\}$. In this instance, Greedy Monroe was able to find a committee with an optimal Monroe score, but this does not hold in general.
\end{example}

\section{\phragmen's Rules}

\phragmen\footnotemark {} introduced a number of voting rules, most of which are based on a form of cost-sharing (or load balancing).
The core idea is that placing a candidate in the winning committee incurs a cost, or load, that has to be shouldered by the voters who approve this candidate.
The goal is to choose a committee that allows for as equal as possible a distribution of its cost.
In this way, the preferences of as many voters as possible are taken into account.

\footnotetext{Lars Edvard \phragmen\ (1863--1937)~\cite{carleman1938phragmen,phrag-history,Janson16arxiv,stubhaug2010gosta}
was a Swedish mathematician and an actuary. He was a professor of mathematics at Stockholm University and long-time editor of Acta Mathematica.
His best known mathematical work is the \phragmen-Lindel\"of principle in complex analysis~\cite{phragmen1908extension},
but he also published several works on election methods \cite{Phrag93,Phra94a, Phra95a, Phra96a, Phra99a} and was involved in Swedish electoral reforms; see Janson's survey \cite{Janson16arxiv} for a comprehensive summary of his work on election methods.}

\phragmen{}'s main proposal is called \emph{\phragmen's Sequential Rule} (seq-\phragmen).
Even though \phragmen's Sequential Rule can be considered one of the most appealing ABC rules,
it remained unknown to many social choice researchers until recently.
Few publications before 2017 mention \phragmen's methods; notable exceptions are a survey by \citet{Jans12a} (in Swedish) and a paper by \citet{MoOl15a} (in Catalan).
Since 2017 several papers have proven \phragmen's method to be a particularly strong ABC rule, in particular being a proportional ABC rule that is both polynomial-time computable and committee monotone.

We present two (equivalent) formulations of seq-\phragmen. The first is conceptually simpler, while the second gives a clearer picture how the rule is computed in practice.

\begin{abcrule}
[\phragmen's Sequential Rule, seq-\phragmen]\label{rule:seqphrag}\begin{changebar}
This ABC rule is based on the assumption that placing a candidate in the winning committee incurs a cost (or a load) of~$1$, which is distributed among the set of voters that approve this candidate.

\paragraph{Continuous formulation:}
We assume that each voter has a budget which constitutes his or her voting power. Voters start with a budget of $0$ and this budget continuously increases as time advances. At time $t$, the budget of each voter is $t$. As soon as a group of voters that jointly approve a candidate has a total budget of $1$, the joint candidate is added to the winning committee. Then the budget of all involved voters is reset to~$0$; only voters that do not approve the selected candidate keep their current budget. This process continues until the committee is filled. If at some point two candidates could be added to the committee at the same time, a tie-breaking order is used to decide which candidate is selected.
\end{changebar}

\paragraph{Discrete formulation:}
seq-\phragmen{} works in rounds; each round one candidate is added to the committee.
Let $y_r(v)$ denote the load assigned to (or cost contributed by) voter $v$ after round $r\leq k$.
We naturally start with $y_0(v)=0$ for all $v\in N$.
Let $\{c_1,\dots,c_{r-1}\}$ be the candidates added to the committee in rounds $1$ to $r-1$.
To determine the next candidate $c_{r}$ to add, we compute for each candidate $c\in C\setminus\{c_1,\dots,c_{r-1}\}$ the maximum load that would arise from adding $c_{r}$:
\[\ell_{r}(c)=\frac{1 + \sum_{i\in N(c)} y_{r-1}(i)}{|N(c)|};\]
the load of voters in $N(c)$ would increase
to this amount if $c$ were added to the committee.
Note that the load is distributed so that all voters approving $c$ end up with the same total load; this is so to minimise the maximum load.
Now, to keep the maximum load as small as possible, seq-\phragmen{} chooses
the candidate $c$ with a minimum $\ell_{r}(c)$, i.e.,
\[c_{r} = \argmin_{c\in C\setminus\{c_1,\dots,c_{r-1}\}} \ell_{r}(c).\]
If two or more candidates yield the same maximum load, a tie-breaking method is required (typically some fixed order on $C$).
After choosing $c_{r}$, the voter loads are adapted accordingly:
\[y_{r}(i) = \begin{cases} \ell_{r}(c_{r}) & \text{if }i\in N(c_{r}),\\
y_{r-1}(i) & \text{if }i\notin N(c_{r}).\end{cases}\]
The rule returns the winning committee $\{c_1,\dots,c_k\}$.
\end{abcrule}

To see that these two formulations are equivalent, note that for a winning committee $W=\{c_1,\dots,c_k\}$ (selected in this order) the maximum loads in each round $\ell_r(c_r)$ directly corresponds to the time points at which sufficient budget was available to pay for $c_r$. From this point of view, the discrete formulation is only the explicit calculation of time points at which sufficient budget is available to place a new candidate in the committee.

\begin{example}\label{ex:seqphragmen}
Let us again consider our running example (\Cref{ex:running}):
	\begin{align*}
		&3 \times \{ a,b\} & & 3 \times \{ a,c\} &&  2 \times \{ a, d\} &&  1 \times \{ b, c, f\}
		 & 1 \times \{ e\} &  & 1 \times \{ f\}   & & 1 \times \{ g\} \text{.}
	\end{align*}
We use the continuous formulation to describe the method, but it is easy to repeat the calculations using the discrete formulation.
\Cref{fig:seqphragmen} shows a visualisation of the procedure, which we will now explain step by step.
The first time sufficient budget is available to add a candidate to the committee is at time $t_1=\nicefrac 1 8$. At this point, voters $\{1,\dots, 8\}$ can jointly pay for candidate~$a$.
Now the budgets of voters $1$ to $8$ are reset to $0$; the remaining voters have a budget of $\nicefrac 1 8$ each.

A second candidate can be added to the committee at time $t_2=\nicefrac{11}{32}$. Voters 1, 2, 3, 9 approve candidate~$b$; their respective budgets are $(\nicefrac{7}{32},\nicefrac{7}{32},\nicefrac{7}{32},\nicefrac{11}{32})$ (note that voters 1, 2, and 3 have budgets that are by $\nicefrac 1 8$ lower than that of voter $9$).
At this time, also voters 4, 5, 6, 9 (who all approve candidate $c$) have a joint budget of~$1$. We use alphabetic tiebreaking and select~$b$.

Candidate~$c$ is then added as a third candidate at time $t_3=\nicefrac{55}{128}$. At this point, voters~4, 5, and 6 have budgets of $\nicefrac{39}{128}$, and voter 9 has a budget of $\nicefrac{11}{128}$; that's in total~$1$. Note that these numbers follow from the fact that voters 4--6 already paid $\nicefrac 1 8$ each for selecting candidate~$a$ and voter 9 paid $\nicefrac{11}{32}$ for selecting candidate~$b$.

Finally, at time $t_4=\nicefrac 5 8$ the last candidate, $d$, is added to the committee. At this point, the two voters approving $d$ (voters~7 and~8) have budgets of $\nicefrac 5 8 - \nicefrac 1 8 = \nicefrac 1 2$, in total $1$. Thus, seq-\phragmen{} returns the committee $\{a,b,c,d\}$.
When repeating this calculation using the discrete formulation, one obtains the final loads
$y_4 = (t_2, t_2, t_2, t_3, t_3, t_3, t_4, t_4, t_3, 0,0,0)$.
\end{example}

\begin{figure}
\begin{center}
\begin{tikzpicture}
[yscale=1.0,xscale=1.9]
\newcommand{\width}{.62}

\filldraw[fill=green!10!white, draw=black] (-\width/2,0) rectangle (-\width/2 + 8 * \width,0.5)
node[pos=.5] {$a$};
\filldraw[fill=green!30!white, draw=black] (-\width/2,0.5) rectangle (-\width/2 + 3 * \width,1.0)
node[pos=.5] {$b$};
\filldraw[fill=green!30!white, draw=black] (-\width/2+ 8 * \width,0.5) rectangle (-\width/2 + 9 * \width,1.0)
node[pos=.5] {$b$};
\filldraw[fill=red!10!white, draw=black] (-\width/2+ 3 * \width,1) rectangle (-\width/2 + 6 * \width,1.5)
node[pos=.5] {$c$};
\filldraw[fill=red!30!white, draw=black] (-\width/2+ 6 * \width,0.5) rectangle (-\width/2 + 8 * \width,1.)
node[pos=.5] {$d$};
\filldraw[fill=red!10!white, draw=black] (-\width/2+ 8 * \width,1) rectangle (-\width/2 + 9 * \width,1.5)
node[pos=.5] {$c$};
\filldraw[fill=white, draw=black] (-\width/2+ 8 * \width,0) rectangle (-\width/2 + 9 * \width,0.5)
node[pos=.5] {$f$};
\filldraw[fill=white, draw=black] (-\width/2+ 9 * \width,0) rectangle (-\width/2 + 10 * \width,0.5)
node[pos=.5] {$e$};
\filldraw[fill=white, draw=black] (-\width/2+ 10 * \width,0) rectangle (-\width/2 + 11 * \width,0.5)
node[pos=.5] {$f$};
\filldraw[fill=white, draw=black] (-\width/2+ 11 * \width,0) rectangle (-\width/2 + 12 * \width,0.5)
node[pos=.5] {$g$};

\foreach \x in {1,...,12}
\node at (-\width + \x*\width, -0.35) {$\x$};

\newcommand{\yoffset}{2.5}
\newcommand{\timelength}{2.5}
\newcommand{\timeone}{\timelength*1/8/5*8};
\newcommand{\timetwo}{\timelength*11/32/5*8};
\newcommand{\timethree}{\timelength*55/128/5*8};
\newcommand{\timefour}{\timelength*5/8/5*8};

\draw (-\width/2,\yoffset) rectangle (-\width/2 + 8 * \width,\yoffset+\timeone);
\filldraw[fill=green!10!white, draw=black] (-\width/2,\yoffset) rectangle (-\width/2 + 8 * \width,\yoffset+\timeone)
node[pos=.5] {$a$};
\filldraw[fill=red!10!white, draw=black] (-\width/2+ 3 * \width,\yoffset+\timeone) rectangle (-\width/2 + 6 * \width,\yoffset+\timethree)
node[pos=.5] {$c$};
\filldraw[fill=red!10!white, draw=black] (-\width/2+ 8 * \width,\yoffset+\timetwo) rectangle (-\width/2 + 9 * \width,\yoffset+\timethree)
node[pos=.5] {$c$};

\filldraw[fill=green!30!white, draw=black] (-\width/2,\yoffset+\timeone) rectangle (-\width/2 + 3 * \width,\yoffset+\timetwo)
node[pos=.5] {$b$};
\filldraw[fill=green!30!white, draw=black] (-\width/2+ 8 * \width,\yoffset) rectangle (-\width/2 + 9 * \width,\yoffset+\timetwo)
node[pos=.5] {$b$};

\filldraw[fill=red!30!white, draw=black] (-\width/2+ 6 * \width,\yoffset+\timeone) rectangle (-\width/2 + 8 * \width,\yoffset+\timefour)
node[pos=.5] {$d$};

\draw[thick,->] (-0.8*\width, \yoffset) -- (-0.8*\width, \yoffset+\timelength+.3);
\draw[thick, solid] (-0.8*\width, \yoffset+0) -- (11.6 * \width, \yoffset+0);

\draw[thick,dotted] (-0.95*\width, \yoffset+\timeone) -- (-0.5 * \width, \yoffset+\timeone);
\draw[thick,dotted] (8.5*\width, \yoffset+\timeone) -- (11.6 * \width, \yoffset+\timeone);
\draw[thick,dotted] (-0.95*\width, \yoffset+\timetwo) -- (-0.5 * \width, \yoffset+\timetwo);
\draw[thick,dotted] (8.5*\width, \yoffset+\timetwo) -- (11.6 * \width, \yoffset+\timetwo);
\draw[thick,dotted] (-0.95*\width, \yoffset+\timethree) -- (2.5 * \width, \yoffset+\timethree);
\draw[thick,dotted] (8.5*\width, \yoffset+\timethree) -- (11.6 * \width, \yoffset+\timethree);
\draw[thick,dotted] (-0.95*\width, \yoffset+\timefour) -- (5.5 * \width, \yoffset+\timefour);
\draw[thick,dotted] (7.5*\width, \yoffset+\timefour) -- (11.6 * \width, \yoffset+\timefour);
                                 
\node[inner sep=0,anchor=east] at (-1*\width, \yoffset+\timeone) {$\nicefrac 1 8$};
\node[inner sep=0,anchor=east] at (-1*\width, \yoffset+\timetwo-0.1) {$\nicefrac{11}{32}$};               
\node[inner sep=0,anchor=east] at (-1*\width, \yoffset+\timethree+0.1) {$\nicefrac{55}{128}$};               
\node[inner sep=0,anchor=east] at (-1*\width, \yoffset+\timefour) {$\nicefrac 5 8$};               
\end{tikzpicture}
\caption{A visualisation of seq-\phragmen{} (upper part) applied to the election instance of \Cref{ex:running} (lower part). In the upper part all regions of the same colour (corresponding to the same candidate) have an area of $1$, which is the budget spent on this candidate.}
\label{fig:seqphragmen}
\end{center}
\end{figure}

\phragmen{} also discussed optimisation-based analogues of seq-\phragmen{}.
These rules are based on choosing a committee that optimises an objective function (in a similar way as Thiele methods optimise an objective function).
We will discuss the most notable optimisation-based method: \lexphrag{}\footnote{\phragmen{} discusses optimisation variants of his rule in \cite{Phra96a} and proposes to minimise the maximum load (see \cite{Janson16arxiv}); this rule has been referred to as opt-\phragmen\ or max-\phragmen. \citet{aaai/BrillFJL17-phragmen} show that it is more sensible to use a lexicographic comparison of loads instead of only considering the maximum load. We thus only discuss \lexphrag{} (referred to as opt-\phragmen{} in \cite{aaai/BrillFJL17-phragmen}). Further optimisation variants exist, such as minimising the variance of loads \cite{Phra96a,aaai/BrillFJL17-phragmen,Janson16arxiv}.} \cite{Phra96a,aaai/BrillFJL17-phragmen,Janson16arxiv}.

\begin{abcrule}[\phragmen's Leximax Rule, \lexphrag{}]
Each candidate in the committee incurs a load (or cost) of $1$ which has to be distributed among voters approving this candidate.
Given a committee $W=\{c_1,\dots,c_k\}$, a \emph{valid load distribution for $W$} is a function
$\ell_W:W\times N\to[0,1]$ which satisfies \begin{inparaenum}[(1)]
\item if $\ell_W(c, i) > 0$ then voter $i$ approves~$c$, and
\item $\sum_{i \in N} \ell_W(c, i) = 1$ 
\end{inparaenum}
for all $c\in W$.
Let $\bar\ell_W=\left(\sum_{c\in W} \ell_W(c,i) \right)_{i\in N}$ denote the vector of total loads assigned to the voters.

To compare two (valid) load distributions, we use a lexicographic order.
Given a valid load distribution $\ell_W$ for $W$, let
$\mathit{sort}(\bar\ell_W)$ denote the tuple $\bar\ell_W$ sorted from largest to smallest.
Let $\ell_W$ and $\ell_{W'}$ denote two valid load distributions for committees $W$ and $W'$, respectively.
We say that $\ell_W$ is lexicographically smaller than $\ell_{W'}$ if there exists an index $j\leq |N|$ such that the first $j$ entries of $\mathit{sort}(\bar\ell_W)$ and $\mathit{sort}(\bar\ell_{W'})$ are equal and the $(j+1)$-st entry of $\mathit{sort}(\bar\ell_W)$ is strictly smaller than the $(j+1)$-st entry of $\mathit{sort}(\bar\ell_{W'})$.


Let $\ell_W^{\min}$ denote a lexicographically smallest valid load distribution for committee~$W$.
Then, \lexphrag{} returns all committees~$W$ for which $\ell_W^{\min}$ is lexicographically 
minimal in the set $\{\ell_{W'}^{\min} : W'\subseteq C\text{ and }|W'|=k \}$.
Note that if \lexphrag{} returns two committees $W_1$ and $W_2$, then $\mathit{sort}(\bar\ell_{W_1}^{\min} )=\mathit{sort}(\bar\ell_{W_2}^{\min})$.
\end{abcrule}

\begin{example}
In our running example, \lexphrag{} behaves differently than seq-\phragmen. When looking for a committee that has the lexicographically smallest load distribution, we find committee $W=\{a,b,c,f\}$ with $\bar\ell_W^{\min}=(\nicefrac 3 8, \nicefrac 3 8, \nicefrac 3 8, \nicefrac 3 8, \nicefrac 3 8, \nicefrac 3 8,  \nicefrac 3 8, \nicefrac 3 8, \nicefrac 1 2, 0, \nicefrac 1 2,  0)$. This load distribution is depicted in \Cref{fig:leximax-phragmen}. Committee~$W$ is the only winning committee; for example, committee $W'=\{a,b,c,d\}$ (the winning committee of seq-\phragmen) has $\bar\ell_{W'}^{\min}=(\nicefrac 3 7, \nicefrac 3 7, \nicefrac 3 7, \nicefrac 3 7, \nicefrac 3 7, \nicefrac 3 7, \nicefrac 1 2, \nicefrac 1 2, \nicefrac 3 7, 0, 0, 0)$, which is lexicographically larger.
\end{example}

\begin{figure}
\begin{center}
\begin{tikzpicture}
[yscale=1.0,xscale=1.9]
\newcommand{\width}{.62}

\filldraw[fill=green!10!white, draw=black] (-\width/2,0) rectangle (-\width/2 + 8 * \width,0.5)
node[pos=.5] {$a$};
\filldraw[fill=green!30!white, draw=black] (-\width/2,0.5) rectangle (-\width/2 + 3 * \width,1.0)
node[pos=.5] {$b$};
\filldraw[fill=green!30!white, draw=black] (-\width/2+ 8 * \width,0.5) rectangle (-\width/2 + 9 * \width,1.0)
node[pos=.5] {$b$};
\filldraw[fill=red!10!white, draw=black] (-\width/2+ 3 * \width,1) rectangle (-\width/2 + 6 * \width,1.5)
node[pos=.5] {$c$};
\filldraw[fill=white, draw=black] (-\width/2+ 6 * \width,0.5) rectangle (-\width/2 + 8 * \width,1.)
node[pos=.5] {$d$};
\filldraw[fill=red!10!white, draw=black] (-\width/2+ 8 * \width,1) rectangle (-\width/2 + 9 * \width,1.5)
node[pos=.5] {$c$};
\filldraw[fill=blue!20!white, draw=black] (-\width/2+ 8 * \width,0) rectangle (-\width/2 + 9 * \width,0.5)
node[pos=.5] {$f$};
\filldraw[fill=white, draw=black] (-\width/2+ 9 * \width,0) rectangle (-\width/2 + 10 * \width,0.5)
node[pos=.5] {$e$};
\filldraw[fill=blue!20!white, draw=black] (-\width/2+ 10 * \width,0) rectangle (-\width/2 + 11 * \width,0.5)
node[pos=.5] {$f$};
\filldraw[fill=white, draw=black] (-\width/2+ 11 * \width,0) rectangle (-\width/2 + 12 * \width,0.5)
node[pos=.5] {$g$};

\foreach \x in {1,...,12}
\node at (-\width + \x*\width, -0.35) {$\x$};

\newcommand{\yoffset}{2.5}
\newcommand{\timelength}{2.5}
\newcommand{\timezero}{\timelength*1/24/1*2};
\newcommand{\timeone}{\timelength*3/8/1*2};
\newcommand{\timetwo}{\timelength};

\draw (-\width/2,\yoffset) rectangle (-\width/2 + 8 * \width,\yoffset+\timeone);
\filldraw[fill=green!10!white, draw=black] (-\width/2,\yoffset) rectangle (-\width/2 + 8 * \width,\yoffset+\timeone);
\path[] (-\width/2 + 6 * \width,\yoffset) rectangle (-\width/2 + 8 * \width,\yoffset+\timeone)
node[pos=.5] {$a$};
\filldraw[fill=red!10!white, draw=black] (-\width/2+ 3 * \width,\yoffset+\timezero) rectangle (-\width/2 + 6 * \width,\yoffset+\timeone)
node[pos=.5] {$c$};
\filldraw[fill=blue!20!white, draw=black] (-\width/2+ 8 * \width,\yoffset) rectangle (-\width/2 + 9 * \width,\yoffset+\timetwo)
node[pos=.5] {$f$};
\filldraw[fill=blue!20!white, draw=black] (-\width/2+ 10 * \width,\yoffset) rectangle (-\width/2 + 11 * \width,\yoffset+\timetwo)
node[pos=.5] {$f$};
\filldraw[fill=green!30!white, draw=black] (-\width/2,\yoffset+\timezero) rectangle (-\width/2 + 3 * \width,\yoffset+\timeone)
node[pos=.5] {$b$};

\draw[thick,->] (-0.8*\width, \yoffset) -- (-0.8*\width, \yoffset+\timelength+.3);
\draw[thick, solid] (-0.8*\width, \yoffset+0) -- (11.6 * \width, \yoffset+0);

\draw[thick,dotted] (-0.95*\width, \yoffset+\timeone) -- (-0.5 * \width, \yoffset+\timeone);
\draw[thick,dotted] (8.5*\width, \yoffset+\timeone) -- (9.5 * \width, \yoffset+\timeone);
\draw[thick,dotted] (10.5*\width, \yoffset+\timeone) -- (11.6 * \width, \yoffset+\timeone);
\draw[thick,dotted] (-0.95*\width, \yoffset+\timetwo) -- (7.5 * \width, \yoffset+\timetwo);
\draw[thick,dotted] (8.5*\width, \yoffset+\timetwo) -- (9.5 * \width, \yoffset+\timetwo);
\draw[thick,dotted] (10.5*\width, \yoffset+\timetwo) -- (11.6 * \width, \yoffset+\timetwo);
                                 
\node[inner sep=0,anchor=east] at (-1*\width, \yoffset+\timeone) {$\nicefrac 3 8$};
\node[inner sep=0,anchor=east] at (-1*\width, \yoffset+\timetwo) {$\nicefrac 1 2$};               
\end{tikzpicture}
\caption{A visualisation of \lexphrag{} (upper part) applied to the election instance of \Cref{ex:running} (lower part). In the upper part all regions of the same colour (corresponding to the same candidate) have an area of $1$, which is the budget spent on this candidate.}
\label{fig:leximax-phragmen}
\end{center}
\end{figure}

\section{\phragmen-like Rules}

We now discuss a very recent addition to the zoo of ABC rules: the Method of Equal Shares~\cite{pet-sko:laminar,pet-pie-sko:c:participatory-budgeting-cardinal} (this method had been originally named Rule~X). This rule can be viewed as a variant of seq-\phragmen{}, where the voters are given some budget upfront, rather than receiving it continuously.
This rule is polynomial-time computable and even surpasses the proportionality guarantees of seq-\phragmen{}.

\begin{abcrule}[Method of Equal Shares]
The rule proceeds in two phases. The first phase consists of at most $k$ rounds; in each round one candidate is added to the committee. In the second phase the committee is completed in one of several possible ways.

For the first phase, we assume each voter is initially given 
a budget of $\nicefrac{k}{n}$.
Let $x_r(i)$ denote the budget of voter $i$ after round $r$; thus $x_0(i) = \nicefrac{k}{n}$. As with seq-\phragmen{}, putting a candidate in the committee incurs a cost of $1$. In round~$r+1$, we consider the set of candidates that have not yet been placed in the committee and whose supporters can afford to pay for them, i.e., all candidates $c$ for which $\sum_{i \in N(c)}x_r(i) \geq 1$. Let this set be $C_r\subseteq C$. If $C_r$ is empty, then we conclude the first phase and move to phase two. Otherwise, for each candidate $c\in C_r$ we ask what is the minimal budget $\rho(c)$ such that each voter approving $c$ pays at most $\rho(c)$ and all voters who approve $c$ pay $1$ in total, i.e., what is the minimal value $\rho(c)$ that satisfies:
\begin{align*} 
\sum_{i \in N(c)} \min(\rho(c), x_{r}(i)) = 1 \text{.}
\end{align*}
(Such a $\rho(c)$ always exists, since otherwise $c$ would not be contained in $C_r$.)
We select the candidate~$c$ that minimises $\rho(c)$ (using some fixed tiebreaking if necessary), and reduce the budget of voters who approve $c$ accordingly---for each $i \in N$ we set 
\begin{align*} 
x_{r+1}(i) = 
\begin{cases}
x_i(r) - \rho(c) & \text{if } c\in A(i) \text{ and } x_i(r) \geq \rho(c),\\
0 & \text{if } c\in A(i) \text{ and } x_i(r) < \rho(c),\\
x_i(r) & \text{if }c\notin A(i) \text{,}
\end{cases}
\end{align*}
i.e., voters who approve~$c$ either pay $\rho(c)$ or their remaining budget.

The second phase is only relevant if fewer than $k$~candidates have been put in the committee~$W$ so far.
If $|W| < k$, we have to add $k - |W|$ additional candidates to $W$. 
Many properties of the Method of Equal Shares do not depend on the specific way in which these $k - |W|$ candidates are selected.\footnotemark{}
A concrete and recommendable way to fill the committee is to use seq-\phragmen{} but with initial budgets defined in the following fashion: 
When using the continuous formulation, we set the starting budget of each voter to their budget after the first phase of the Method of Equal Shares; this starting budget increases as usual as time advances. Alternatively, we can use the discrete formulation of seq-\phragmen{}: if the first phase ends with round $r'$, the starting loads are $y_0(i)= - x_{r'}(i)$.
Then seq-\phragmen{} proceeds as usual until the desired committee size is reached.
\end{abcrule}

The name of the rule corresponds to the two elements of its definition. First, each voter is initially given an equal share of the budget that she can spend for ``buying'' candidates. Second, when a candidate is selected, its cost is split as equally as possible among the voters who approve the candidate (each voter covers an equal share of the cost of the candidate).

\footnotetext{An exception is the priceability axiom, see \Cref{sec:laminar}; this axiom is dependent on how to extend the committee to its full size. The proposed completion via seq-\phragmen{} fulfils priceability.}

\begin{example}\label{ex:rule_x_1}
Consider once again our running example. Each voter is initially given a budget of $\nicefrac{1}{3}$. In the first round candidate $a$ is selected and each of the first 8 voters pays $\nicefrac{1}{8}$ for this. In the second round, $C_2=\emptyset$ since no candidate has sufficiently endowed supporters. For example, the budget of voters who approve $b$ is in total \[3\cdot(\nicefrac{1}{3} - \nicefrac{1}{8}) + \nicefrac{1}{3} < 1\] and thus insufficient to pay for $b$. This ends the first phase of the rule.

In the second phase, the voters start receiving additional budget.
Voters $1$ to $8$ start with a budget of $\nicefrac 1 3 - \nicefrac 1 8$; voters $9$ to $12$ start with a budget of $\nicefrac 1 3$.
At time $t_2 = \nicefrac{1}{96}$, voters $1$ to $8$ have a budget of $\nicefrac{1}{3} - \nicefrac{1}{8} + t_2$ each and voters $9$ to $12$ have a budget of $\nicefrac{1}{8} + t_2$ each. Hence the voters who approve $b$ ($1$, $2$, $3$, $9$) have enough money to pay for~$b$: \begin{align*}
3\cdot(\nicefrac{1}{3} - \nicefrac{1}{8} + t_2) + (\nicefrac{1}{3} + t_2) = 1 \text{.}
\end{align*}
The same is true for the voters who approve $c$. Let us assume that we resolve the tie in favour of $b$: $b$ is selected and the voters $1, 2, 3$ and $9$ are left without budget. Next, at time $t_3 = \nicefrac{37}{384}$ candidate $c$ is selected (voters $4$--$6$ contribute $\nicefrac{1}{3}-\nicefrac{1}{8} + t_3$ and voter~$9$ contributes $t_3-t_2$, with the required total of $1$). Finally, at time $t_4 = \nicefrac{7}{24}$ we select $d$ ($2\cdot(\nicefrac{1}{3}-\nicefrac{1}{8}+t_4)=1$). Committee~$W=\{a,b,c,d\}$ is the only winning committee. In this example, the Method of Equal Shares returns the same committee as seq-\phragmen{}.
\end{example}

Since in \Cref{ex:rule_x_1} only one candidate is selected in the first phase of the Method of Equal Shares, we provide one additional example which better illustrates the first phase of this rule and also shows that seq-\phragmen\ and the Method of Equal Shares may produce different committees.

\begin{example}\label{ex:rule_x_2}
Consider the following approval profile given by
	\begin{align*}
		&A(1) = A(2) = A(3) = \{ c,d \} & & A(4) = A(5) = \{ a,b\} \\& A(6)= A(7) = \{ a,c \} & & A(8) =  \{b,d\}\text{.}
	\end{align*}
The goal is to select a committee of size $k = 3$.
Thus, voters start with a budget of $\nicefrac{3}{8}$.

In this example, candidate $c$ is selected in the first round with each approving voter ($1,2,3,6,7$) paying $\nicefrac 1 5$. Next, candidate~$a$ is selected. Voters~4 and~5 contribute $\nicefrac{13}{40}$, voters~6 and~7 contribute their remaining budget ($\nicefrac{7}{40}$).
None of the remaining candidates achieves a total budget of $1$ and thus the second phase starts.
The starting budgets for seq-\phragmen\ are $(\nicefrac{7}{40}, \nicefrac{7}{40}, \nicefrac{7}{40}, \nicefrac{1}{20}, \nicefrac{1}{20}, 0, 0, \nicefrac{3}{8})$.
At time $t=\nicefrac{1}{40}$ candidate~$d$ is selected: voters 1~to~3 can contribute $\nicefrac{7}{40}+t = \nicefrac 1 5$ each and voter~8 can contribute the remaining $\nicefrac{3}{8} + t = \nicefrac 2 5$. Hence, the Method of Equal Shares selects the committee $\{a,c,d\}$. The voters' payments in the two phases are illustrated in \Cref{fig:rule_x_example}.

In contrast, seq-\phragmen{} picks $\{b, c, d\}$.
These candidates are selected in order $c,b,d$ at time $t_1=\nicefrac{1}{5}$, $t_2=\nicefrac{1}{3}$, and $t_3=\nicefrac{29}{60}$, respectively.
\end{example}

\begin{figure}[t!]
\begin{center}
\begin{tikzpicture}
[yscale=0.9,xscale=2.3]
\newcommand{\width}{.48}

\filldraw[fill=green!30!white, draw=black] (-\width/2 + 3 * \width,0) rectangle (-\width/2 + 7 * \width,0.5) node[pos=.5] {$a$};

\filldraw[fill=white, draw=black] (-\width/2 + 3 * \width,0.5) rectangle (-\width/2 + 5 * \width,1.0) node[pos=.5] {$b$};
\filldraw[fill=white, draw=black] (-\width/2 + 7 * \width,0.5) rectangle (-\width/2 + 8 * \width,1.0) node[pos=.5] {$b$};

\filldraw[fill=red!10!white, draw=black] (-\width/2,1) rectangle (-\width/2 +3 * \width,1.5) node[pos=.5] {$c$};
\filldraw[fill=red!10!white, draw=black] (-\width/2 + 5 * \width,1) rectangle (-\width/2 +7 * \width,1.5) node[pos=.5] {$c$};

\filldraw[fill=blue!20!white, draw=black] (-\width/2,1.5) rectangle (-\width/2 + 3 * \width,2.) node[pos=.5] {$d$};
\filldraw[fill=blue!20!white, draw=black] (-\width/2 + 7 * \width,1.5) rectangle (-\width/2 + 8 * \width,2.) node[pos=.5] {$d$};

\foreach \x in {1,2,...,8}
\node at (\x*\width -\width, -0.35) {$\x$};
\newcommand{\yoffset}{3.5}
\newcommand{\timelength}{1.5}
\newcommand{\timezero}{\timelength};
\newcommand{\timeone}{\timelength * 0.533};
\newcommand{\timetwo}{\timelength * 0.866};
\newcommand{\timethree}{\timelength * 1.08};

\filldraw[fill=green!30!white] (-\width/2+ 3 * \width,\yoffset) rectangle (-\width/2 + 5 * \width,\yoffset+\timetwo);
\filldraw[fill=green!30!white] (-\width/2+ 5 * \width,\yoffset+\timeone) rectangle (-\width/2 + 7 * \width,\yoffset+\timelength);
\path[] (-\width/2 +  3 * \width,\yoffset) rectangle (-\width/2 + 5 * \width,\yoffset+\timetwo) node[pos=.5] {$a$};
\path[] (-\width/2 +  5 * \width,\yoffset + \timeone) rectangle (-\width/2 + 7 * \width,\yoffset+\timelength) node[pos=.5] {$a$};
\draw (-\width/2,\yoffset) rectangle (-\width/2 + 3 * \width,\yoffset+\timeone);
\filldraw[fill=red!10!white, draw=black] (-\width/2,\yoffset) rectangle (-\width/2 + 3 * \width,\yoffset+\timeone) node[pos=.5] {$c$};
\draw (-\width/2 + 5 * \width,\yoffset) rectangle (-\width/2 + 7 * \width,\yoffset+\timeone);
\filldraw[fill=red!10!white, draw=black] (-\width/2 + 5 * \width,\yoffset) rectangle (-\width/2 + 7 * \width,\yoffset+\timeone) node[pos=.5] {$c$};

\draw[thick,->] (-0.8*\width, \yoffset) -- (-0.8*\width, \yoffset+\timelength+.3);

\draw[thick, solid] (-0.8*\width, \yoffset+0) -- (8.6 * \width, \yoffset+0);

\draw[thick,dotted] (-0.95*\width, \yoffset+\timeone) -- (-0.5 * \width, \yoffset+\timeone);
\draw[thick,dotted] (-0.5*\width + 7 * \width, \yoffset+\timeone) -- (8 * \width, \yoffset+\timeone);

\draw[thick,dotted] (-2.05*\width, \yoffset+\timetwo) -- (-0.5 * \width + 3 * \width, \yoffset+\timetwo);
\draw[thick,dotted] (-0.5*\width + 7 * \width, \yoffset+\timetwo) -- (8 * \width, \yoffset+\timetwo);

\draw[thick,dotted] (-3.35*\width, \yoffset+\timelength) -- (-0.5*\width + 5 * \width, \yoffset+\timelength);
\draw[thick,dotted] (-0.5*\width + 7 * \width, \yoffset+\timelength) -- (8 * \width, \yoffset+\timelength);

\draw[thick,<->] (-1.1*\width, \yoffset + 0.02) -- (-1.1*\width, \yoffset+\timeone-0.02); 
\draw[thick,<->] (-2.1*\width, \yoffset + 0.02) -- (-2.1*\width, \yoffset+\timetwo-0.02);     
\draw[thick,<->] (-3.4*\width, \yoffset + 0.02) -- (-3.4*\width, \yoffset+\timelength-0.02);                                
                                                            
\node[inner sep=0,anchor=east] at (-1.3*\width, \yoffset+\timeone/2) {$\nicefrac{1}{5}$};
\node[inner sep=0,anchor=east] at (-2.3*\width, \yoffset+\timetwo/2) {$\nicefrac{13}{40}$}; 
\node[inner sep=0,anchor=east] at (-3.6*\width, \yoffset+\timelength/2) {$\nicefrac{3}{8}$}; 

\node[inner sep=0,anchor=east] at (4*\width, \yoffset - 0.7) {first phase of the Method of Equal Shares};

\newcommand{\yoffsettwo}{6.5}

\filldraw[fill=green!30!white] (-\width/2+ 3 * \width,\yoffsettwo) rectangle (-\width/2 + 5 * \width,\yoffsettwo+\timetwo);
\filldraw[fill=green!30!white] (-\width/2+ 5 * \width,\yoffsettwo+\timeone) rectangle (-\width/2 + 7 * \width,\yoffsettwo+\timelength);
\path[] (-\width/2 +  3 * \width,\yoffsettwo) rectangle (-\width/2 + 5 * \width,\yoffsettwo+\timetwo) node[pos=.5] {$a$};
\path[] (-\width/2 +  5 * \width,\yoffsettwo + \timeone) rectangle (-\width/2 + 7 * \width,\yoffsettwo+\timelength) node[pos=.5] {$a$};
\draw (-\width/2,\yoffsettwo) rectangle (-\width/2 + 3 * \width,\yoffsettwo+\timeone);
\filldraw[fill=red!10!white, draw=black] (-\width/2,\yoffsettwo) rectangle (-\width/2 + 3 * \width,\yoffsettwo+\timeone) node[pos=.5] {$c$};
\draw (-\width/2 + 5 * \width,\yoffsettwo) rectangle (-\width/2 + 7 * \width,\yoffsettwo+\timeone);
\filldraw[fill=red!10!white, draw=black] (-\width/2 + 5 * \width,\yoffsettwo) rectangle (-\width/2 + 7 * \width,\yoffsettwo+\timeone) node[pos=.5] {$c$};

\filldraw[fill=blue!20!white, draw=black] (-\width/2,\yoffsettwo+\timeone) rectangle (-\width/2 + 3 * \width,\yoffsettwo+\timethree) node[pos=.5] {$d$};
\filldraw[fill=blue!20!white, draw=black] (-\width/2 + 7 * \width,\yoffsettwo) rectangle (-\width/2 + 8 * \width,\yoffsettwo+\timethree) node[pos=.5] {$d$};

\draw[thick,->] (-0.8*\width, \yoffsettwo) -- (-0.8*\width, \yoffsettwo+\timelength+.3);

\draw[thick, solid] (-0.8*\width, \yoffsettwo+0) -- (8.6 * \width, \yoffsettwo+0);

\draw[thick,dotted] (-0.95*\width, \yoffsettwo+\timeone) -- (-0.5 * \width, \yoffsettwo+\timeone);
\draw[thick,dotted] (-0.5*\width + 8 * \width, \yoffsettwo+\timeone) -- (8 * \width, \yoffsettwo+\timeone);

\draw[thick,dotted] (-2.05*\width, \yoffsettwo+\timetwo) -- (-0.5 * \width, \yoffsettwo+\timetwo);
\draw[thick,dotted] (-0.5*\width + 8 * \width, \yoffsettwo+\timetwo) -- (8 * \width, \yoffsettwo+\timetwo);

\draw[thick,dotted] (-3.35*\width, \yoffsettwo+\timelength) -- (-0.5 * \width, \yoffsettwo+\timelength);
\draw[thick,dotted] (-0.5*\width + 3 *\width, \yoffsettwo+\timelength) -- (-0.5 * \width + 5 *\width, \yoffsettwo+\timelength);
\draw[thick,dotted] (-0.5*\width + 8 * \width, \yoffsettwo+\timelength) -- (8 * \width, \yoffsettwo+\timelength);

\draw[thick,dotted] (-4.35*\width, \yoffsettwo+\timethree) -- (-0.5 * \width, \yoffsettwo+\timethree);
\draw[thick,dotted] (-0.5*\width + 3 *\width, \yoffsettwo+\timethree) -- (-0.5 * \width + 7 *\width, \yoffsettwo+\timethree);
\draw[thick,dotted] (-0.5*\width + 8 * \width, \yoffsettwo+\timethree) -- (8 * \width, \yoffsettwo+\timethree);

\draw[thick,<->] (-1.1*\width, \yoffsettwo + 0.02) -- (-1.1*\width, \yoffsettwo+\timeone-0.02); 
\draw[thick,<->] (-2.1*\width, \yoffsettwo + 0.02) -- (-2.1*\width, \yoffsettwo+\timetwo-0.02);     
\draw[thick,<->] (-3.4*\width, \yoffsettwo + 0.02) -- (-3.4*\width, \yoffsettwo+\timelength-0.02);
\draw[thick,<->] (-4.4*\width, \yoffsettwo + 0.02) -- (-4.4*\width, \yoffsettwo+\timethree-0.02);                                 
                                                            
\node[inner sep=0,anchor=east] at (-1.3*\width, \yoffsettwo+\timeone/2) {$\nicefrac{1}{5}$};
\node[inner sep=0,anchor=east] at (-2.3*\width, \yoffsettwo+\timetwo/2) {$\nicefrac{13}{40}$}; 
\node[inner sep=0,anchor=east] at (-3.6*\width, \yoffsettwo+\timelength/2) {$\nicefrac{3}{8}$}; 
\node[inner sep=0,anchor=east] at (-4.6*\width, \yoffsettwo+\timethree/2) {$\nicefrac{3}{8} + t$}; 

\node[inner sep=0,anchor=east] at (7*\width, \yoffsettwo - 0.7) {second phase: completion by \phragmen};

\end{tikzpicture}
\caption{A visualisation of the Method of Equal Shares applied to the election instance of \Cref{ex:rule_x_2} (lower part). In the two upper figures, all regions of the same colour (corresponding to the same candidate) have an area of $1$, which is the budget spent on this candidate.}
\label{fig:rule_x_example}
\end{center}
\end{figure}

Let us discuss three further rules that are related to \phragmen's rules.
The first is the Expanding Approvals Rule~\cite{aziz2020expanding}.
This rule is defined for weak-order preferences and has favourable axiomatic properties in this setting. It is less convincing for approval preferences\footnotemark{} and thus we do not consider it further.
The second rule is the maximin support method~\cite{sanchez2021maximin}, which is similar to seq-\phragmen. It is an iterative rule based on a form of load balancing, but in contrast to seq-\phragmen{} all loads can be redistributed each round. A first analysis showed that the maximin support method and seq-\phragmen{} share many axiomatic properties~\cite{sanchez2021maximin}, and a recent manuscript by \citet{cevallos2020verifiably} shows that the maximin support method provides a constant factor approximation of leximax-\phragmen---in contrast to seq-\phragmen.
In the light of the latter paper, one may view the maximin support method as a polynomial-time approximation of leximax-\phragmen\ (in the same sense as seq-PAV approximates PAV), whereas seq-\phragmen{} can rather be viewed as a largely independent rule.
We focus in this book on seq-\phragmen{} as it is better studied and conceptually simpler.
Still, the maximin support method is an interesting ABC rule that should be analysed in more depth.

\footnotetext{For approval preferences, the Expanding Approvals Rule (EAR) can be rather indecisive. For example, in profiles where no candidate reaches a specified quota and every voter approves only one candidate, EAR selects an arbitrary committee and thus ignores the voters' preferences.
For a practical application, EAR would have to be augmented with an additional mechanism that handles such cases.}

Finally, \phragmen{} also introduced a method now referred to as either \phragmen's first method, Enestr\"om's method, or method of \enephrag{}\footnotemark~\cite{enestrom:1896,Janson16arxiv,camps2019method}.
This rule can be viewed as an analogue of 
Single Transferable Vote (STV) with approval ballots.
\footnotetext{It is not completely clear whether \phragmen{} or Gustaf Enestr\"om (1852--1923) should be credited with this method. However, it appears to be justifiable to simply credit both of them; see the historical summary provided by Janson \cite[Footnote 38]{Janson16arxiv}.}%
\begin{abcrule}[\enephrag{}]
This method is based on a quota~$q$, which is typically chosen to be either the \textit{Hare quota} $q=\frac{n}{k}$ or the \textit{Droop quota} $q  = \frac{n}{k+1}$.
Candidates are selected in a sequential fashion.
All voters start with a weight of~$1$.
In each round, we compute for each unselected candidate the total weight of approving voters, i.e., the score of an unselected candidate~$c$ is the sum of weights of all voters approving~$c$.
The candidate with the maximum score is added to the committee (using a tie-breaking if necessary); let this candidate be $c'$ and its score~$s$.
Now, the weights are adapted: 
If $s>q$,
then the weights of all voters in $N(c')$ are multiplied by $\frac{s-q}{s}$.
Thus, the total weight of voters in $N(c')$ is reduced by $q$.
If $s\leq q$, the weights of voters in $N(c')$ are set to~$0$.
This step is repeated until $k$~candidates are selected.
\end{abcrule}
As \enephrag{} is not as well studied as \phragmen's rules, we do not discuss it further,
but we note that further analysis could prove this rule to be of independent interest.\footnote{The most substantial analysis of \enephrag{} is due to \citet{camps2019method}. Most notably, it is not committee monotone (in contrast to seq-\phragmen{}, cf.~\Cref{sec:comm-mon}), but it satisfies proportional justified representation (as seq-\phragmen{} does, cf.~\Cref{def:pjr}).}

\section{Non-Standard ABC Rules}\label{sec:furtherabc}

As mentioned at the beginning of this chapter, most ABC rules coincide with (single-winner) Approval Voting for $k=1$. 
If we understand an approval ballot as indicating those alternatives that a voter likes, then for $k=1$ it is indeed very natural to select the most-approved alternative.
Thus, we refer to rules that differ from Approval Voting for $k=1$ as \emph{non-standard} ABC rules.
In addition to rev-seq-PAV, which we already showed to be non-standard, we present two further non-standard rules. The first one, Minimax Approval Voting (MAV) introduced by \citet{minimaxProcedure}, interprets approval ballots as the voter's exact description of the desired outcome.
If a voter approves a set $X$, then she indicates that \emph{all} these alternatives should be chosen; any sub- or superset is less desirable.
In addition, MAV is an egalitarian rule in the sense that it only pays attention to the least-satisfied voter.

To measure the distance between an approval set and a committee, we rely on the Hamming distance:
\begin{definition}
Given two sets $X, Y$, we define the Hamming distance between $X$ and $Y$ as the size of their symmetric difference:
$d_{\hamming}(X, Y) = |X \setminus Y| + |Y \setminus X|$.
\end{definition}

\begin{abcrule}[Minimax Approval Voting, MAV]  MAV selects committees~$W$ that minimise the largest Hamming distance among all voters, i.e., MAV minimises $\max_{i \in N}d_{\hamming}(A(i), W)$.
\end{abcrule}

\begin{example}\label{ex:mav}
To see that MAV does not correspond to Approval Voting for $k=1$, consider the following approval profile:
	\begin{align*}
		&99 \times \{ a\} & & 1 \times \{ b,c\}\text{.}
	\end{align*}
The Hamming distance $d_{\hamming}$ between the committee $W_1=\{a\}$ and the approval set $\{b,c\}$ is $3$.
In contrast, for the committee $W_2=\{b\}$ (or $\{c\}$) we have $d_{\hamming}(\{b,c\}, W_2)=1$ and $d_{\hamming}(\{a\}, W_2)=2$. Thus, MAV selects either $b$ or $c$, even though these alternatives are approved by only a single voter.
\end{example}

\begin{remark}
It is interesting to note that if we replace the $\max$ operator in the definition of MAV by a sum, we obtain the Multi-Winner Approval Voting rule (\Cref{rule:av}).
\end{remark}

\begin{remark}\label{rem:lexmav}
MAV, as defined, has a major shortcoming.
Consider the following slight modification of \Cref{ex:mav}:
    \begin{align*}
		&99 \times \{ a\} & & 1 \times \{a, b,c\}\text{.}
	\end{align*}
For all size-$1$ committees, the Hamming distance to $\{a, b,c\}$ is $2$.
Hence, all three committees are equally preferable according to MAV---even though candidate $a$ is approved by every voter (and $b$ and $c$ by only one voter).
We see that MAV might disregard a unanimous choice.
This problem can be remedied by also considering the second-least satisfied voter in case of ties, and the third-least in case there is still a tie, and so on until a difference between the committees is found.
More formally, for each committee $W$, we compute $d_{\hamming}(A(1), W), d_{\hamming}(A(2), W), \dots$ and sort this tuple of length $|N|$ in decreasing order; we denote this tuple of distances $D_W$.
Instead of considering only the first entry in these tuples, we could lexicographically sort them.
That is, a committee $W_1$ is preferred to a $W_2$ if there exists an index $i\leq n$ such that $D_{W_1}(i)<D_{W_2}(i)$ and $D_{W_1}(j)=D_{W_2}(j)$ for all $1\leq j< i$.
In our example, we have $D_{\{a\}}=(2, 0, 0, \dots)$ and $D_{\{b\}}=D_{\{c\}}=(2, 2, 2, \dots)$; with this modification $\{a\}$ is the only winning committee.
To the best of our knowledge this modification of MAV has not been studied in the context of voting. However, it is equivalent to the $\textsc{Gmax}$ belief merging operator for the Hamming distance~\cite{konieczny2011logic}.
\end{remark}

The second non-standard rule is Satisfaction Approval Voting\footnote{Satisfaction Approval Voting was introduced under this name by Brams and Kilgour~\cite{BrKi14a}, but the method has been discussed already in the 19th century (see Janson's survey~\cite{Janson16arxiv}, Section~E.1.5.). It is also known as Equal and Even Cumulative Voting.} (SAV).
SAV is a variation of AV where each voter has one point and distributes it evenly among all approved candidates.
As a consequence, voters who approve more candidates contribute a lesser score to the individual approved candidates.

\begin{abcrule}[Satisfaction Approval Voting, SAV] The SAV-score of a committee $W$ is defined as \[\score{\sav}(A,W)=\sum_{i \in N}\frac{|W \cap A(i)|}{|A(i)|}.\]
SAV returns all committees with a maximum SAV-score.
\end{abcrule}

Note that SAV is not a Thiele method since the total number of candidates that a voter approves influences the SAV-score.

\begin{example}
To see that SAV does not correspond to Approval Voting for $k=1$, consider
	\begin{align*}
		&1 \times \{ a\} & & 3 \times \{ b,c,d,e\}\text{.}
	\end{align*}
The SAV-score of $a$ is $1$ and for $b$, $c$, $d$, and $e$ it is $\nicefrac 3 4$. Thus, SAV selects $\{a\}$ even though it is approved by only one voter.
\end{example}

\bibliographystyle{abbrvnat}
\bibliography{main}

\chapter{Basic Properties of ABC Rules}\label{sec:basic}

\intro{In this chapter, we consider basic axiomatic properties of ABC rules.
These properties describe the behaviour of such rules and offer insights into the nature of specific ABC rules.
Important axiomatic properties include Pareto efficiency, committee monotonicity, candidate monotonicity, consistency as well as axioms pertaining to strategic voting.}

In the previous chapter we have seen a wide array of ABC rules. Considering how much they differ in their definitions, it can be expected that they differ also in the properties they exhibit.
In this chapter we consider basic properties of ABC rules.
These properties describe the behaviour of such rules and offer insights into the nature of specific ABC rules.
\Cref{tab:axioms_summary} offers an overview of most properties discussed in this chapter.
This table also includes a rough dichotomy of the rules concerning their computational complexity. 
Rules that are in~P can be computed efficiently, whereas rules that are $\np$-hard are computationally more demanding;
we discuss this dichotomy and further complexity results in \Cref{sec:computational-complexity}.

{\newcommand{\suppmoncell}{\multicolumn{3}{c}{\makecell{support monot.\\with / without\\add.~voters}}}
\renewcommand{\arraystretch}{1.4}%
\begin{table}[!t]
        \footnotesize
	\centering
	\makebox[\textwidth][c]{
	\begin{tabular}{lcccccccc}
		\toprule
                          & \makecell{Pareto\\efficiency} & \makecell{committee\\monoton.} & \suppmoncell & \makecell{consist.} & \makecell{inclusion-\\strategypr.} & \makecell{comput.\\complexity}\\
		\midrule
		AV                & \strong                       & \cmark                             & \cmark\ &/& \cmark                                     & \cmark      & \cmark                             & P  \\
		CC                & \weak                         & \xmark                             & \cmark\ &/& \ccand                                     & \cmark      & ?                                  & $\np$-hard  \\
		PAV               & \strong                       & \xmark                             & \cmark\ &/& \ccand                                     & \cmark      & \xmark                             & $\np$-hard \\
		seq-PAV           & \xmark                        & \cmark                             & \ccand\ &/& \ccand\                                    & \xmark      & \xmark                             & P  \\
		seq-CC            & \xmark                        & \cmark                             & \ccand\ &/& \ccand\                                    & \xmark      & \xmark                             & P  \\
		rev-seq-PAV       & \xmark                        & \cmark                             & \cmark\ &/& \ccand                                     & \xmark      & \xmark                             & P  \\
		Monroe            & \xmark                        & \xmark                             & \xmark\ &/& \ccand                                     & \xmark      & \xmark                             & $\np$-hard  \\
		Greedy Monroe     & \xmark                        & \xmark                             & \xmark\ &/& \ccand                                     & \xmark      & \xmark                             & P  \\
		seq-\phragmen     & \xmark                        & \cmark                             & \ccand\ &/& \ccand                                     & \xmark      & \xmark                             & P  \\
		\lexphrag{}     & \xmark                        & \xmark                             & \ccand\ &/& \ccand                                     & \xmark      & ?                                  & $\np$-hard  \\
		Method of Eq. Shares            & \xmark                        & \xmark                             & \xmark\ &/& \ccand                                     & \xmark      & \xmark                             & P  \\
		MAV               & \weak                         & \xmark                             & \cmark &/& \ccand\                                                & \xmark      & \xmark                             & $\np$-hard \\
		SAV               & \strong                       & \cmark                             & \cmark\ &/& \cmark                                     & \cmark      & \xmark                             & P  \\
		\bottomrule
	\end{tabular}
	}
	\caption{Basic properties of ABC rules.}\label{tab:axioms_summary}
\end{table}}

\section{Anonymity, Neutrality, and Resoluteness}\label{sec:anon-neutr-resol}

Anonymity and neutrality are two of the most basic properties in the social choice literature~\cite{mayAxiomatic1952,arrow1963,moulinAxioms}.
Anonymity states that the identity of voters should not influence the outcome: it should be irrelevant whether voter $i$ approves $A(i)$ and voter $j$ approves $A(j)$ or vice versa.
Formally, an ABC rule $\calR$ satisfies \emph{anonymity} if for all election instances $(A,k)$ with voter set $N$ and bijections $\pi:N\to N$ it holds that $\calR(A, k)=\calR(A\circ\pi, k)$.
All but one rule introduced in \Cref{sec:abc_rules} satisfy anonymity; the exception is Greedy Monroe which uses a fixed tiebreaking order on voters.\footnote{If we defined Greedy Monroe so that it returns all committees that can result from some tiebreaking, then the rule would be anonymous.} A typical example of a voting rule that fails anonymity is any dictatorial rule (a rule considering only the preferences of a single distinguished voter, e.g., of voter~$1$).

Neutrality is the counterpart to anonymity but applies to candidates: it states that all candidates should be treated equally. Formally, an ABC rule $\calR$ satisfies \emph{neutrality} if for all election instances $(A,k)$ with candidate set $C$ and bijections $\pi:C\to C$ it holds that $\calR(A, k)=\calR(\pi^*\circ A, k)$, where $\pi^*$ is the natural extension of $\pi$ to a bijection from $\powerset(C)$ to $\powerset(C)$ defined by $\pi^*(X)=\{\pi(c) : c\in X\}$ for each $X \subseteq C$.
The rules that fail neutrality are usually those that require some form of tiebreaking.

The third and equally fundamental property we discuss here is resoluteness.
Recall that an ABC rule is resolute if it always returns exactly one winning committee.
An ABC rule can either be resolute or neutral, but not both. To see this, consider an approval profile where all voters approve candidates $\{a,b\}$ and $k=1$: either a rule returns two winning committees or decides in favour of one of the two candidates.
Clearly, any rule can be made resolute by imposing a tiebreaking between winning committees. Conversely, if a resolute rule is defined by a tiebreaking order over candidates (this includes all rules in \Cref{sec:abc_rules} that fail neutrality), it can be made neutral by returning all committees that win according to \emph{some} tiebreaking order. In this way, one can trade neutrality against resoluteness. 

Finally, we mention that an in-depth treatment of the interplay between anonymity, neutrality, and resoluteness---albeit in the setting of single-winner elections---can be found in the work of \citet{ozkes2021anonymous} and \citet{campbell2015finer}. 

\section{Pareto Efficiency and Condorcet Committees}\label{sec:pareto}

Pareto efficiency\footnotemark{} is a very general concept to compare two outcomes given the preferences of individuals: outcome $Y$ dominates outcome $X$ if \begin{inparaenum}[(1)]
\item every individual weakly prefers outcome $Y$ to $X$ (i.e., everyone likes $Y$ at least as much as $X$), and 
\item there is at least one individual that strictly prefers $Y$ to $X$.
\end{inparaenum}
Pareto efficiency, broadly speaking, means that dominated outcomes are avoided.
This concept can be directly translated to our setting by
defining when a voter prefers committee $W_1$ to $W_2$. 
This requires a so-called \emph{set extension}, i.e., a way how to extend preferences over individual items to sets of items; we refer the reader to the survey of \citet{barbera2004ranking} for a comprehensive overview.
Here, we use the Pareto efficiency definition by \citet{lac-sko2019} and assume that $W_1$ is preferred to $W_2$ if $W_1$ contains more approved candidates.

\begin{definition}
A committee $W_1$ dominates a committee $W_2$ if \begin{enumerate}
\item every voter has at least as many approved candidates in $W_1$ as in $W_2$ (for $i\in N$ it holds that $|A(i)\cap W_1|\geq |A(i)\cap W_2|$), and 
\item there is one voter with strictly more approved candidates (there exists $j\in N$ with $|A(j)\cap W_1|> |A(j)\cap W_2|$).
\end{enumerate}
A committee that is not dominated by any other committee (of the same size) is called \emph{Pareto optimal}.

\footnotetext{Named after Vilfredo Pareto (1848--1923), an Italian economist \cite{eisermann2001pareto}.} 

An ABC rule $\calR$ satisfies \emph{strong Pareto efficiency} if
$\calR$ never outputs dominated committees.
An ABC rule $\calR$ satisfies \emph{weak Pareto efficiency} if for all election instances $(A,k)$ it holds that if
$W_2\in\calR(A,k)$ and $W_1$ dominates $W_2$, then $W_1\in\calR(A,k)$.
\end{definition}

\Cref{tab:axioms_summary} summaries which rules satisfy Pareto efficiency.\footnote{For details, in particular counterexamples, we refer the reader to~\cite{lac-sko2019}. Although this paper does not discuss the Method of Equal Shares, the counterexample for seq-\phragmen{} \cite[Example 2]{lac-sko2019} also works for this method.}
It may be surprising that rather few ABC rules satisfy this kind of Pareto efficiency.
Indeed, among the rules introduced in \Cref{sec:abc_rules} only Thiele rules, SAV, and MAV satisfy weak Pareto efficiency~\cite{lac-sko2019}, and among those, e.g., AV, PAV, and SAV satisfy strong Pareto efficiency (but not CC and MAV, for details see \Cref{prop:av-pav-sav-pareto}).
(Although, we recall that these results rely of course on our chosen set extension.)

To see an example how a rule may fail Pareto efficiency, it is instructive to consider Monroe's rule:
\begin{example}[{\citep[Example~3]{lac-sko2019}}]
Consider the approval profile
	\begin{align*}
		&2 \times \{ a \} & & 1 \times \{ a,c\} &&  1 \times \{ a,d\} && 10 \times \{ b,c \}  && 10 \times \{ b,d \}\text{.}
	\end{align*}
	For $k=2$, Monroe selects $\{c, d\}$ as the (only) winning committee with a Monroe-score of $22$.
	Committee $\{c, d\}$ is however dominated by $\{a,b\}$: every voter approves a candidate in~$\{a,b\}$ but only $22$ voters approve one in~$\{c, d\}$.
	Thus, every voter is either equally satisfied or better off with committee $\{a, b\}$.
	This example shows that Pareto efficiency clashes with Monroe's goal to assign representatives to groups of similar size.
\end{example}

One may wonder whether it is sensible to improve an ABC rule $\calR$ that is not Pareto efficient in the following way:
given an election instance $E$, if $W\in \calR(E)$ is dominated by another committee, then instead output all Pareto optimal committees that dominate~$W$.
There are two main objections against this idea:
First, this modification may destroy other axiomatic properties (e.g., Pareto efficiency and perfect representation, which is discussed in \Cref{sec:laminar}, are incompatible).
Second, finding Pareto improvements is a computationally hard task:

\begin{theorem}[{\citet[Theorem 2]{aziz2020computing}}]\label{thm:pareto-np}
Given an election instance $(A,k)$ and committee $W$, it is $\conp$-complete to determine whether $W$ is Pareto optimal.
\end{theorem}

As a consequence of \Cref{thm:pareto-np}, we cannot expect to obtain polynomial-time computable, Pareto efficient ABC rules by modifying existing rules as described above.
Note, however, that polynomial-time computable, Pareto efficient ABC rules exist, e.g., AV and SAV. Thus, \emph{finding} a Pareto optimal committee is possible in polynomial-time.


A related property to Pareto efficiency has been proposed by~\citet{Darmann13_condorcet_committees}: a committee $W$ is a \emph{Condorcet committee} if for every other committee $W'$, for a majority of voters $V\subseteq N$ ($|V|>\nicefrac{|N|}{2}$) it holds that $|A(i) \cap W| > |A(i) \cap W'|$ for $i\in V$. Similarly to \Cref{thm:pareto-np}, deciding whether a given committee $W$ is a Condorcet committee is $\conp$-complete.
However, in contrast to Pareto optimality, it is also $\conp$-complete to decide whether a Condorcet committee exists~\cite{Darmann13_condorcet_committees}. To the best of our knowledge, it has not been analysed which ABC rules output a Condorcet committee if it exists.

\section{Committee Monotonicity}\label{sec:comm-mon}

Committee monotonicity (also referred to as house monotonicity or committee enlargement monotonicity) is a property that is highly desirable in some settings: if the committee size $k$ is increased to $k+1$, then a winning committee of size $k$ should be a subset of a winning committee of size $k+1$.
Since this property is particularly useful for resolute rules, we define it exclusively for resolute rules.
Appropriate definitions for irresolute rules can be found, e.g., in papers of \citet{elk-fal-sko-sli:c:multiwinner-rules} and of
\citet{kil-mar:j:minimax-approval} (called upward- and downward-accretive in the latter work).

\begin{definition}\label{def:comm-mon}
A resolute ABC rule $\calR$ is \emph{committee monotone} if for all election instances $(A,k)$ it holds that $W \subseteq W'$, where $W$ is the single winning committee in $\calR(A,k)$, and $W'$ is the single winning committee in $\calR(A,k+1)$.
\end{definition}

To see why committee monotonicity can be an essential requirement in some applications, consider the following situation. A group can jointly acquire $k$ items and uses an ABC rule to fairly select those. Once these $k$ items are purchased, it turns out that one additional item can be afforded. If the used ABC rule is committee monotone, it is clear which item to acquire next. However, if the rule is not committee monotone, then the selection for $k+1$ items might contain several items that were not contained in the selection of $k$ items, a useless recommendation.

Another example is a hiring process where it is not determined up-front how many candidates are to be hired. Here it is useful that a committee monotone rule actually produces a \emph{ranking} of candidates: which one should be hired if only one position is available, which one if a second position is to be filled, etc. This connection between committee monotone ABC rules and rankings has been explored in-depth by \citet{proprank}.

However, committee monotonicity also reduces the flexibility of voting rules and thus comes 
at a price. For example, we will see in \Cref{sec:proportionality} that committee monotone rules 
are typically less proportional (although a formal proof for this statement is missing).
Thus, if the setting does not dictate committee monotonicity, it may be advantageous to 
set this axiom aside.
A more elaborate discussion of this topic can be found in the paper of \citet{elk-fal-sko-sli:c:multiwinner-rules}.

\Cref{tab:axioms_summary} shows which of the considered rules are committee monotone, assuming that these rules are made resolute by fixing a tiebreaking order among candidates.
AV, seq-PAV, seq-CC, rev-seq-PAV, seq-\phragmen{}, and SAV are committee monotone; this follows immediately from their corresponding definitions.
Counterexamples for the remaining rules can be found in \Cref{app:proofs}, \Cref{prop:comm-mon}.

\section{Candidate and Support Monotonicity}\label{sec:cand_supp_monotonicity}

Candidate monotonicity deals with a seemingly obvious requirement: if the support of a candidate increases (i.e., more voters approve this candidate), then this cannot harm the candidate's inclusion in a winning committee.
However, this property is not satisfied by some ABC rules, in particular, if we demand such a monotonicity to also hold for groups of candidates.
In addition, there is a difference depending on whether an existing voter changes her ballot, or if a new voter enters the election.

Candidate monotonicity axioms for ABC rules have been considered in a number of papers \cite{aziz2020expanding,Janson16arxiv,lac-sko:t:multiwinner-strategyproofness}, but the paper by \citet{sanchez2019monotonicity} should be highlighted for the most in-depth analysis.\footnote{Monotonicity is also studied in great detail by \citet{elk-fal-sko-sli:c:multiwinner-rules} and \citet{fal-sko-sli-tal-tal:j:hierarchy-committee}; these works, however, largely focus on multi-winner voting with voters' preferences given as rankings (cf.~\Cref{sec:multiwinner-rank}).}

Further, we write $A_{+X}$ to denote the profile $A$ with one additional voter approving $X$, i.e., $A_{+X}=(A(1),\dots,A(n),X)$, and $A_{i+X}$ to denote the profile $A$ where voter $i$ additionally approves the candidates from $X$.

\begin{definition}[\citet{sanchez2019monotonicity}]\label{def:monotonicity}
An ABC rule $\calR$ satisfies \emph{support monotonicity without additional voters} if for every election instance $(A,k)$, $i\in N$, and candidate set $X\subseteq C$ it holds that
\begin{enumerate}
\item if $X\subseteq W$ for all $W\in\calR(A,k)$, then $X\subseteq W'$ for all $W'\in\calR(A_{i+X},k)$, and
\item if $X\subseteq W$ for some $W\in\calR(A,k)$, then $X\subseteq W'$ for some $W'\in\calR(A_{i+X},k)$.
\end{enumerate}
An ABC rule $\calR$ satisfies \emph{support monotonicity with additional voters} if for any election instance $(A,k)$ and candidate set $X\subseteq C$
the properties above hold for~$A_{+X}$ instead of~$A_{i+X}$.
\end{definition}

If an ABC rule satisfies these axioms only for singleton sets ($X=\{c\}$),
we speak of \emph{candidate monotonicity with/without additional voters}.\footnote{\citet{sanchez2019monotonicity} further introduce \emph{weak support monotonicity with/without population increase}. These notions are slightly stronger than their candidate monotonicity counterparts (i.e., they imply candidate monotonicity with/without additional voters).}

The analysis of ABC rules with respect to these axioms is mostly due to 
\citet{Janson16arxiv}, \citet{sanchez2019monotonicity}, and \citet{MoOl15a}.
We summarise the results in \Cref{tab:axioms_summary}.
There, the symbol~\cmark{} means that support monotonicity is satisfied, ``\ccand{}'' means that candidate monotonicity is satisfied but not support monotonicity, and
\xmark{} means that the rule fails even candidate monotonicity.
Detailed counterexamples related to support monotonicity can be found in \Cref{prop:candmon} in the appendix.

If one is interested in ABC rules that are---in a sense---fair to candidates,
then candidate monotonicity (both with and without additional voters) is generally a desirable property. Hence, the fact that Monroe, Greedy Monroe, and the Method of Equal Shares fail the axiom can be seen as 
a serious argument against these rules.
Monroe and the Method of Equal Shares, however, have other distinguished advantages (discussed in \Cref{sec:proportionality}) that may override this downside.
In settings where a fair treatment of candidates is not necessary (e.g., because candidates represent 
inanimate objects to be chosen), candidate monotonicity should not be a concern.

\section{Consistency}\label{sec:consistency}

Consistency is an axiom describing whether a rule behaves \emph{consistently} with respect to disjoint groups: if the outcome of an election is the same for two disjoint groups, then a voting rule should arrive at this outcome also if these two groups are joined into a single electorate.
This axiom is a straightforward adaption of consistency as defined for single-winner rules by \citet{smi:j:scoring-rules} and \citet{young74} and was first discussed in the context of ABC rules by \citet{jet-consistentabc}.
In the following, for two profiles $A$ and $A'$ we write $A+A'$ to denote the joint profile where $A$ and $A'$ are concatenated.

\begin{definition}
  An ABC rule $\calR$ satisfies \emph{consistency} if for every $k\geq 1$ and two
  profiles $A: N \to \powerset(C)$ and $A': N' \to \powerset(C)$ with $N \cap N' = \emptyset$, if $\calR(A,k) \cap \calR(A',k) \neq \emptyset$ then $\calR(A + A',k) = \calR(A,k) \cap \calR(A',k)$.
\end{definition}

Monroe's rule, for example, does not satisfy consistency:

\begin{example}
Let profile $A$ be
	\begin{align*}
&A(1)\colon \{ a, y\}   &&  A(2)\colon \{ a, y\}     &&   A(3)\colon \{ b, y\}     && A(4)\colon \{ b, y\}
	\end{align*}
and profile $A'$ be
		\begin{align*}
&A(5)\colon \{ y\}   &&  A(6)\colon \{ a\}     &&   A(7) = A(8) = A(9) = A(10)\colon \{ a,x\} 
\\ & A(11)\colon \{ y\}  &&  A(12)\colon \{ b,y\}     &&   A(13)= A(14) = A(15) = A(16)\colon \{ b, x\}\text{.}
	\end{align*}
	For $k=2$, Monroe returns for profile $A$ the winning committees $\{a,b\}$, $\{a,y\}$, and $\{b,y\}$, all of which having a Monroe-score of $4$.
	For profile $A'$, Monroe returns the winning committee $\{a,b\}$, with a Monroe-score of $10$; the corresponding Monroe assignment groups voters $5$--$10$ and $11$--$16$.
	Now, let us consider the profile $A+A'$. Consistency would demand that $\{a,b\}$ is the unique winning committee, as it is the only committee winning in both $A$ and $A'$.
	Committee $\{a,b\}$ has a Monroe-score of $14$ in $A+A'$.
	This score, however, is not optimal: $\{x,y\}$ has a Monroe-score of $15$; the corresponding Monroe assignment groups voters $\{1,\dots,6,11,12\}$ and $\{7,\dots,10,13,\dots,16\}$.
	Thus, $\{a,b\}$ is not winning and consistency is violated.
\end{example}

Broadly speaking, the only rules satisfying consistency are so-called \emph{ABC scoring rules}~\cite{jet-consistentabc}.
These are defined similarly to Thiele methods but are more general, as the satisfaction of a voter may depend on the number of candidates approved by this voter:

\begin{definition}\label{def:abc-scoring-rules} A \emph{scoring function} is a function $f \colon \naturals\times\naturals \to \reals$ satisfying $f(x,y)\geq f(x',y)$ for $x\geq x'$.
Given such a scoring function, we
define the score of $W$ in $A$ as 
\begin{align*}
\score{f}(A, W) = \sum_{i \in N} f(|A(i) \cap W|, |A(i)|)\text{.}\label{eq:score}
\end{align*}
The \emph{ABC scoring rule} defined by a scoring function $f$ returns all committees with maximum score.
\end{definition}
By definition, each Thiele method is an ABC scoring rule, whereas SAV is an example of an ABC scoring rule that is not a Thiele method.
Further, it follows immediately from the definition of welfarist rules (\Cref{def:welfarist_rules}) that an ABC scoring rule is welfarist if and only if it is a Thiele method.
%
%

\citet{jet-consistentabc} axiomatically characterised the class of ABC scoring rules. This characterisation is in a slightly different model than the one we use in this book: the characterisation applies to ABC ranking rules instead of ABC rules (as defined in \Cref{sec:model}). ABC ranking rules output a weak order over committees (a ranking with ties over committees) instead of just distinguishing between winning and losing committees (as we assume here). However, note that every ABC ranking rule defines an ABC rule (top-ranked committees are winning).

The following characterisation uses two axioms we have not mentioned so far: weak efficiency and continuity. Both are rather weak axioms. Intuitively, weak efficiency requires that approved candidates are preferable to non-approved candidates, and continuity states that a sufficiently large majority can force a committee to win.

\newcommand{\thmcharacterizationWelfareFunctions}{An ABC ranking rule is an ABC scoring rule if and only if it satisfies anonymity, neutrality, consistency, weak efficiency, and continuity.}
\begin{theorem}[\citet{jet-consistentabc}]\label{thm:characterizationWelfareFunctions}
\thmcharacterizationWelfareFunctions
\end{theorem}%

As both weak efficiency and continuity are generally satisfied by sensible voting rules, one can conclude that ABC scoring rules essentially capture the class of consistent ABC ranking rules.\footnote{In the setting of single-winner rules a similar result holds: a social welfare function is a scoring rule if and only if it satisfies anonymity, neutrality, consistency, and continuity, as shown by \citet{smi:j:scoring-rules} and \citet{young74}.
Moreover, a similar characterisation holds for committee scoring rules, as shown by \citet{skowron2019axiomatic}. Committee scoring rules can be viewed as analogues of ABC scoring rules in the multi-winner model with preferences given as rankings (see \Cref{sec:multiwinner-rank}); the proof of \Cref{thm:characterizationWelfareFunctions} builds upon this result.}
In \Cref{sec:apportionment}, we will discuss how this result can be used to obtain further axiomatic characterisations, e.g., of PAV.

\section{Strategic Voting}\label{sec:strategic_voting}

Strategic voting is a phenomenon central to social choice theory. Sometimes, it is preferable for voters to misrepresent their preferences to change the outcome of an election; this is often referred to as ``manipulation''. The famous impossibility theorem by \citet{gib:j:polsci:manipulation} and \citet{sat:j:polsci:manipulation}, showing that all ``reasonable'' single-winner voting rules are susceptible to manipulation, is considered one of the main results in the field. The Gibbard--Satterthwaite theorem applies to elections where voters provide linear rankings over alternatives. As our approval-based setting uses a much more restricted form of preferences, strategyproofness is not completely out of the picture.

We are going to consider two forms of strategyproofness here: Cardinality-strategyproofness and inclusion-strategyproofness (taken from \citet{pet:prop-sp}, see the work of~\citet{gardenfors1979definitions} and \citet{taylor2005social} for more general discussions of strategyproofness in social choice). 
Cardinality-strategyproofness assumes that voters are concerned only about the number of approved candidates in the committee (and do not distinguish them), whereas inclusion-strategyproofness assumes that voters may have more complex preferences, so a successful manipulation must produce a committee including all approved candidates that were already included in the original committee.

To simplify the discussion, we assume resoluteness, i.e., we assume a (deterministic) tiebreaking order to resolve ties between committees.
To clarify what it means that a voter misrepresents their true preferences, we use the concept of $i$-variants: Given profiles $A$ and $A'$, both with the same set of voters $N$, we say that $A'$ is an $i$-variant of $A$ if $A(j)=A'(j)$ for all $j\in N \setminus \{i\}$ with $j\neq i$.
Let us first define both notions for resolute ABC rules.

\begin{definition}\label{def:cardstrategy}
A resolute ABC rule $\calR$ satisfies \emph{cardinality-strategyproofness} if for all profiles $A$ and $A'$ where $A'$ is an $i$-variant of $A$ and for all $k\geq 1$ it holds that $|\calR(A,k)\cap A(i)|\geq |\calR(A',k)\cap A(i)|$.
\end{definition}

\begin{definition}\label{def:incl-strategy}
A resolute ABC rule $\calR$ satisfies \emph{inclusion-strategyproofness} if for all profiles $A$ and $A'$ where $A'$ is an $i$-variant of $A$ and for all $k\geq 1$ it holds that $\calR(A,k)\cap A(i)$ 
is not a strict subset of $\calR(A',k)\cap A(i)$.
\end{definition}

Cardinality-strategyproofness is a stronger notion than inclusion-strategyproofness in the sense that 
all cardinality-strategyproof ABC rules are also inclusion-strategyproof.
This follows from the fact that $|\calR(A,k)\cap A(i)|\geq |\calR(A',k)\cap A(i)|$ (as required in \Cref{def:cardstrategy}) implies that $\calR(A,k)\cap A(i)$ cannot be
a strict subset of $\calR(A',k)\cap A(i)$ (as required in \Cref{def:incl-strategy}).

Among the rules considered in this book, only AV satisfies any of the mentioned strategyproofness axioms. Specifically, AV satisfies both inclusion-strategyproofness and cardinality-strategyproofness if AV is made resolute by any tiebreaking order on candidates (for details see \Cref{prop:strategyproofness}).
None of the other ABC rules considered in this paper satisfy these axioms, see \Cref{tab:axioms_summary} for an overview and \Cref{prop:strategyproofness} for details.
However, even AV is not strategyproof in a stronger sense when voters have underlying, non-dichotomous preferences (as discussed, e.g., by \citet{niemi1984problem}).

Both cardinality- and inclusion-strategyproofness can be generalised to irresolute ABC rules  via set extensions, i.e., by defining how voters compare sets of committees. For example, \citet{lac-sko:t:multiwinner-strategyproofness} propose a rather strong extension based on stochastic dominance. The resulting axiom, called  SD-strategyproofness, implies cardinality-strategyproofness. AV satisfies SD-strategyproofness  and
can even be characterised in the class of ABC scoring rules (\Cref{def:abc-scoring-rules}) as the only rule satisfying SD-strategyproofness \cite{lac-sko:t:multiwinner-strategyproofness}.
We note, however, that under more holistic models, e.g., models where voters have underlying non-dichotomous (non-binary) preferences, even AV is no longer strategyproof (see, e.g., \cite{laslier2016StrategicVoting,DeSinopoli2006,Dellis200647,mei-pro-ros-zoh:multiwinner_strategic}).
Another natural extension is the Kelly (or cautious) extension: 
a voter prefers $\calR(A',k)$ to $\calR(A,k)$ if every committee in $\calR(A',k)$ is preferable to every committee in $\calR(A,k)$. A more substantial discussion of strategyproofness of irresolute ABC rules can be found in the paper of \citet{KVVBE20:strategyproofness}.

We further discuss strategyproofness in \Cref{sec:proportionality_and_strategyproofness} in the context of proportionality. We will see that even weak forms of proportionality are incompatible with strategyproofness.

Finally, we note that \citet{sch-har-mat-ven:heuristic-strategic,ScheuermanHMV21} have conducted a behavioural experiment in which they analysed how the voters vote under non-dichotomous preferences, when they are uncertain about other voters' preferences, and when AV is used to select the winning candidates. These results suggest that the voters may use different (sometimes suboptimal) heuristics when making decisions which candidates they should approve. 
This shows that strategic voting in a practical setting can differ substantially from the axiomatic analysis we have presented here.

%

%
%
 
\bibliographystyle{abbrvnat}
\bibliography{main}
\chapter{Proportionality}\label{sec:proportionality}

\intro{The goal of this chapter is to discuss the many faces of proportional representation.
Proportionality, at its core, is a notion of fairness that grants smaller and larger groups of voters a fair consideration of their preferences.
The concrete definitions of what proportionality exactly means, however, differ.
In this chapter, we review the main approaches to proportionality and identify ABC rules which can be considered proportional.}

A key difference among ABC rules is how they treat minorities of voters, i.e., small groups with preferences different from larger groups. Let us illustrate this issue with the following simple example.

\begin{example}\label{ex:proportionality}
Consider the approval-based preference profile with $60$ voters approving $A = \{a_1, \ldots, a_{10}\}$, $20$ voters approving $B = \{b_1, \ldots, b_{6}\}$, 10 voters approving $C = \{c_1, c_2\}$, 8 voters approving $D = \{d_1, d_2, d_3, d_4\}$, and
2 voters approving $E = \{e_1, e_2, e_3\}$; assume our goal is to pick a committee of ten candidates. 
Given this instance AV returns committee $A$, and in some cases this is a reasonable choice (e.g., when the goal of the election is to select finalists of a contest). Yet, when the goal is to select a representative body that should reflect voters' preferences in a proportional fashion, this committee violates very basic principles of fairness. Indeed, the voters who approve committee $A$ constitute $60\%$ of the population, yet effectively they decide about the whole committee; at the same time the group of $20\%$ who approve $B$ is ignored. A committee that consists of six candidates from $A$, two candidates from $B$, one candidate from $C$, and one candidate from $D$ is, for example, a much more proportional choice.
\end{example}

In \Cref{ex:proportionality}, picking an outcome that is intuitively proportional is easy due to a very specific structure of voters' approval sets---each two approval sets are either the same or disjoint. Finding a proportional committee in the general case, when any two approval sets can arbitrarily overlap, is by far less straightforward, and to some extent ambiguous. Several approaches that allow one to formally reason about proportionality have been proposed in the literature. 

The goal of this chapter is to discuss the many faces of proportional representation.
Proportionality, at its core, is a notion of fairness that grants smaller and larger groups of voters a fair consideration of their preferences.\footnote{The concept of proportionality also finds application beyond voting, such as proportional clustering in machine learning \cite{DBLP:conf/icml/ChenFLM19,DBLP:conf/icalp/Micha020}.}
The concrete definitions of what proportionality exactly means, however, differ.
In this chapter, we review the main approaches to proportionality and identify ABC rules which can be considered proportional. \Cref{tab:proportionality_summary} and \Cref{fig:relation_between_proportionality_concepts} provide an overview of this analysis; the corresponding concepts are explained in this chapter.

But before we delve into this topic, let us answer the question why proportionality has such 
a prominent place in this book.
The main reason is that this reflects the attention this topic has received.
Since 2015, when \citet{justifiedRepresentation} first introduced (extended) justified representation (\Cref{sec:cohesive_groups}), there has been rapid progress in the understanding of proportionality in ABC elections. This progress has been along two trajectories: (i) defining stronger and stronger proportionality properties and (ii) finding (computationally tractable) ABC rules satisfying these properties.
In many situations, a proportional committee corresponds to a fair selection of candidates.
Thus, this line of research can be viewed as the search for a maximally fair ABC voting rule.
The following sections (\Cref{sec:apportionment,sec:cohesive_groups,sec:laminar,sec:the_core}) provide an overview of this exciting endeavour.

However, non-proportional rules are certainly also relevant and even necessary in many applications.
For example, when shortlisting candidates for a prize, we may want to select the ``best'' candidates without considerations of a proportional selection.
Or if we want to form a group that deliberates a topic, we would like to include as many diverse opinions as possible and thus we do not give a higher weight to popular opinions.
In general, much less work has been done on analysing and understanding non-proportional rules and this topic deserves much more attention.
In \Cref{sec:degressive_and_regressive_proportionality}, we summarise the existing literature and discuss concepts of ``non-proportionality''. 

The two final sections of this chapter are dedicated to the interplay of proportionality and strategyproofness (\Cref{sec:prop_external_attributes}) and considerations of proportionality when candidates have external attributes (\Cref{sec:proportionality_and_strategyproofness}).

\Crefname{example}{Ex.}{Ex.}
{\renewcommand{\arraystretch}{1.4}%
\begin{table}[!t]
        \footnotesize
	\centering
	\makebox[\textwidth][c]{
	\begin{tabular}{llllllll}
		\toprule
		& \makecell[l]{proportionality \\ degree} & EJR  & PJR & JR & \makecell[l]{laminar \\ prop.} & \makecell[l]{price-\\ability} & apportionment \\
		\midrule
		AV & 0~\cite{skowron:prop-degree} & & & & & & none \\
		\rowcolor{columbiablue}
		PAV & $\ell - 1$~\cite{AEHLSS18} & \checkmark~\cite{justifiedRepresentation} &\checkmark~\cite{justifiedRepresentation} &  \checkmark~\cite{justifiedRepresentation} & & & D'Hondt~\cite{bri-las-sko:c:apportionment} \\
		seq-PAV & \makecell[l]{$\approx 0.7\ell - 1$ \\ (for $k \leq 200$)~\cite{skowron:prop-degree}} & & & & & & D'Hondt~\cite{bri-las-sko:c:apportionment} \\
		rev-seq-PAV & ? & & & & & & D'Hondt~\cite{bri-las-sko:c:apportionment} \\
		CC & $\leq 1$~(\Cref{ex:pjr_and_ejr})  & & & \checkmark~\cite{justifiedRepresentation} & & & none \\
		seq-CC & $\leq 1$~(\Cref{ex:pjr_and_ejr}) & & &  \checkmark~\cite{justifiedRepresentation} & & & none \\
		\rowcolor{columbiablue}
		seq-\phragmen & $\nicefrac{(\ell-1)}{2}$~\cite{skowron:prop-degree} & & \checkmark~\cite{aaai/BrillFJL17-phragmen} & \checkmark~\cite{aaai/BrillFJL17-phragmen} & \checkmark~\cite{pet-sko:laminar} &  \checkmark~\cite{pet-sko:laminar} & D'Hondt~\cite{bri-las-sko:c:apportionment} \\
		\rowcolor{columbiablue}
		M. Equal Shares & $\nicefrac{(\ell \pm 1)}{2}$~(\ref{prop:prop_degree_of_sav_and_mav}) & \checkmark~\cite{pet-sko:laminar} & \checkmark~\cite{pet-sko:laminar} & \checkmark~\cite{pet-sko:laminar}  &  \checkmark~\cite{pet-sko:laminar} &  \checkmark~\cite{pet-sko:laminar} & D'Hondt~\cite{pet-sko:laminar} \\
		\lexphrag{} & 1~\cite{skowron:prop-degree} & & \checkmark~\cite{aaai/BrillFJL17-phragmen} & \checkmark~\cite{aaai/BrillFJL17-phragmen} &  &  \checkmark~\cite{pet-sko:laminar} & D'Hondt~\cite{bri-las-sko:c:apportionment} \\
		Monroe & $\leq 1$~(\Cref{ex:pjr_and_ejr}) & & $\dag$~\cite{pjr17} & \checkmark~\cite{justifiedRepresentation} & & & LRM $\dag$~\cite{bri-las-sko:c:apportionment} \\
		Greedy Monroe & $\leq 1$~(\Cref{ex:pjr_and_ejr}) & & $\dag$~\cite{pjr17} & \checkmark~(\ref{prop:greedy-monroe-jr}) & & & LRM $\dag$ (\ref{prop:greedy_monroe_apportionment})\\
		MAV     & 0~(\ref{prop:prop_degree_of_sav_and_mav}) & & & & & & none \\
		SAV     & 0~(\ref{prop:prop_degree_of_sav_and_mav}) & & & & & & none \\
		\bottomrule
	\end{tabular}
	}
	\caption{Proportionality of ABC rule. There are three rules which perform particularly well in terms of proportionality: PAV, \phragmen's sequential rule, and the Method of Equal Shares. The mark $\dag$ means that the result holds only when the number of voters $n$ is divisible by the committee size $k$. References of the form (A.x) refer to propositions in \Cref{app:proofs}.}
	\label{tab:proportionality_summary}
\end{table}}

\begin{figure}[t!]
\begin{center}
\begin{tikzpicture}[scale=0.8]

        \node at (0.0, 2.0) {stable priceability};
        \node at (0.0, 1.4) {(\Cref{sec:laminar})};
        
        \node at (0.0, 0.0) {the core};
        \node at (0.0, -0.6) {(\Cref{def:core})};
        
        \node at (0.0, -2.0) {FJR};
        \node at (0.0, -2.6) {(\Cref{def:fjr})};
        
        \node at (0.0, -4.0) {core subject to priceability};
        \node at (0.0, -4.6) {with equal payments}; 
        \node at (0.0, -5.2) {(\Cref{sec:core_restrictions})}; 
        
        \node at (6.0, -4.0) {priceability};
        \node at (6.0, -4.6) {(\Cref{def:priceability})}; 
        
        \node at (0.0, -6.6) {EJR};
        \node at (0.0, -7.2) {(\Cref{def:ejr})};
        
        \node at (0.0, -8.6) {PJR};
        \node at (0.0, -9.2) {(\Cref{def:pjr})};
        
        \node at (0.0, -10.6) {JR};
        \node at (0.0, -11.2) {(\Cref{def:jr})};
        
        \node at (6.0, -10.6) {lower quota};
        \node at (6.0, -11.2) {(\Cref{sec:apportionment})};
        
        \draw[thick,->] (0.0, 1.0) -- (0.0, 0.5);
        \draw[thick,->] (1.0, 1.0) -- (6.0, -3.5);
        \draw[thick,->] (6.0, -5.0) -- (1.0, -8.1);
        \draw[thick,->] (0.0, -1.0) -- (0.0, -1.5);
        \draw[thick,->] (0.0, -3.0) -- (0.0, -3.5);
        \draw[thick,->] (0.0, -5.6) -- (0.0, -6.1);
        \draw[thick,->] (0.0, -7.6) -- (0.0, -8.1);
        \draw[thick,->] (0.0, -9.6) -- (0.0, -10.1);
        \draw[thick,->] (1.9, -9.2) -- (4.3, -10.4);
        \draw[dashed] (-0.7, -6.6) -- node[above] {{\small incompatible}} (-5.3, -6.6);
        
        \node at (-8.0, -6.6) {perfect representation};
        \node at (-7.0, -7.2) {(\Cref{def:perfect_representation})};

\end{tikzpicture}
\end{center}
\caption{The relation between different proportionality axioms. An arrow from property $A$ to $B$ means that $A$ implies $B$.}\label{fig:relation_between_proportionality_concepts}
\end{figure}

\Crefname{example}{Example}{Examples}

\section{Apportionment}\label{sec:apportionment}

One approach to reasoning about proportionality of voting rules is to first identify a class of well-structured preference profiles where the concept of proportionality can be intuitively captured, and then to examine the behaviour of voting rules on such well-structured profiles. 
We focus here on so-called \emph{party-list profiles}, which are election instances of the form as we have seen in \Cref{ex:proportionality}.

\begin{definition}[Party-list profiles]\label{def:party_profiles}
We say that an approval profile $A = (A(1), \ldots, A(n))$ is a \emph{party-list profile} if for each two voters $i, j \in N$ we have that either $A(i) = A(j)$ or that $A(i) \cap A(j) = \emptyset$. 
We say that an election instance $(A, k)$ is a \emph{party-list instance} if 
\begin{inparaenum}[(i)]
\item $A$ is a party-list profile, and 
\item for each voter $i \in N$ we have that $|A(i)| \geq k$.
\end{inparaenum}
\end{definition}

Party-list profiles closely resemble political elections with political parties, hence the name of the domain. In such elections, voters are typically asked to vote for exactly one party. To see the connection to party-list profiles, note the following: If $A$ is a party-list profile, then the sets of voters and candidates can be divided into $p$ disjoint groups each, $N = N_1 \cup \ldots \cup N_p$ and $C \supseteq C_1 \cup \ldots \cup C_p$, so that all voters from group $N_i$, $i \in [p]$, approve exactly the candidates from $C_i$ (and no others). The candidates from $C_i$ can be thought of as members of some (virtual) party, and the voters from $N_i$ are those who cast their vote on party $C_i$.

In such elections, where the voters do not vote for individual candidates but rather each voter casts a single vote for one political party, the problem of distributing seats to political parties is called the \emph{apportionment problem}. The concept of proportionality in the apportionment setting has been extensively studied in the literature and is well understood---for a detailed overview we refer the reader to the comprehensive books by \citet{BaYo82a} and by \citet{Puke17}. 

We see from \Cref{def:party_profiles} that the apportionment problem can be viewed as a strict subdomain of approval-based multi-winner elections, and consequently ABC rules can be viewed as functions that extend apportionment methods to the more general setting of approval profiles. This connection was already known and referred to by \citet{Thie95a} and \citet{Phra95a}. In a more systematic fashion, \citet{bri-las-sko:c:apportionment} showed such relations between various ABC rules and methods of apportionment.
To properly explain this relation, let us first define three prominent apportionment methods,  used in parliamentary elections all over the world. 

In the following, we assume that there are $p$ political parties, consisting of the candidate sets $C_1, \ldots, C_p$. By $n_i$ we denote the number of votes cast on party $C_i$. Further, in line with our usual notation, $k$ denotes the number of committee seats that we want to distribute among the parties.

\begin{apportionmentrule}[D'Hondt method\footnote{Victor D'Hondt (1841--1901) was a Belgian professor of law and active proponent of proportional representation \citep{d1885expose,d1878question}. The D'Hondt method is also known as Jefferson method. Thomas Jefferson (1743--1826) was president of the United States, and proposed this method to allocate seats in the House of Representatives to states. D'Hondt's proposal was specifically meant for proportional representation in parliaments.
D'Hondt developed this method independently of Jefferson, even though Jefferson's proposal was earlier and largely similar.  The name ``Jefferson method'' is typically used in the U.S., while ``D'Hondt method'' is prevalent in Europe.\label{footnote:dhondt}
}]
The D'Hondt method proceeds in $k$ rounds, in each round allocating one seat to some party. Consider the $r$-th round, and let $s_i(r)$ be the number of seats that are currently assigned to party $C_i$; thus, $\sum_{i \in [p]}s_i(r) = r-1$. The D'Hondt method assigns the $r$-th seat to the party $C_i$ with the highest ratio $\frac{n_i}{s_i(r) + 1}$ (using a tiebreaking order between parties if necessary).
\end{apportionmentrule}

\begin{apportionmentrule}[Sainte-Lagu\"e\footnote{As it is the case with the D'Hondt/Jefferson method, this rule has been developed independently in Europe and in the U.S. and goes by different names: Sainte-Lagu\"e is used in Europe (in particular in the context of proportional representation in parliaments) and Webster is the name used in the U.S.\ literature. Sainte-Lagu\"e (1882--1950) was a French mathematician and proposed this method in 1910 \cite{sainte1910representation}. Daniel Webster (1782--1852) was a U.S.\ statesman and proposed this method in 1832 \cite{BaYo82a}.} method]
The Sainte-Lagu\"e method is defined analogously to the D'Hondt method, but in the $r$-th round it allocates the $r$-th seat to the party $C_i$ which maximises the ratio $\frac{n_i}{2s_i(r) + 1}$.  
\end{apportionmentrule}

Both the D'Hondt and the Sainte-Lagu\"e method belong to the class of divisor methods.
Divisor methods differ in the formula for the ratio used to distribute seats to parties. 
The aforementioned books by \citet{BaYo82a} and by \citet{Puke17} discuss this important class of apportionment methods in much more detail.

\begin{apportionmentrule}[Largest remainder method, LRM\footnote{The largest remainder method is also known as the Hamilton method, as it was proposed in the U.S.\ by Alexander Hamilton (1755--1804). His proposal was abandoned in favour of Jefferson's method \cite{BaYo82a}.}]
The largest remainder method first assigns to each party $\left\lfloor k \cdot \frac{n_i}{n} \right\rfloor$ seats---this way at least $k-p+1$ seats are assigned. Second, it assigns the remaining $r < p$ seats to the $r$ parties with the largest remainders $ k \cdot \frac{n_i}{n} - \left\lfloor k \cdot \frac{n_i}{n} \right\rfloor$, assigning each party at most one seat.
\end{apportionmentrule}

\begin{example}
Consider a party-list representation of the profile from \Cref{ex:proportionality}. We have five parties, $A$, $B$, $C$, $D$, and $E$, each getting, respectively, 60, 20, 10, 8, and 2 votes; the committee size is $k = 10$. The computation of the D'Hondt method can be followed in the left table below:

\begin{center}
\renewcommand{\arraystretch}{1.2}%
\minipage{0.45\textwidth}
\begin{center}
	\begin{tabular}{l|ccccc}
		\toprule
		  & $A$ & $B$ & $C$ & $D$ & $E$ \\
		\midrule
		$n_i$   & $\boldsymbol{60}$ & $\boldsymbol{20}$ & $\boldsymbol{10}$  & 8 & 2 \\
		$\nicefrac{n_i}{2}$   & $\boldsymbol{30}$ & $\boldsymbol{10}$ & 5  & 4 & 1 \\
		$\nicefrac{n_i}{3}$   & $\boldsymbol{20}$ & $6\,\nicefrac{2}{3}$ & $3\,\nicefrac{1}{3}$  & $2\,\nicefrac{2}{3}$ & $\nicefrac{2}{3}$ \\
		$\nicefrac{n_i}{4}$   & $\boldsymbol{15}$ & 5 & $2\,\nicefrac{1}{2}$  & $2$ & $\nicefrac{1}{2}$ \\
		$\nicefrac{n_i}{5}$   & $\boldsymbol{12}$ & 4 & 2  & $1\,\nicefrac{3}{5}$ & $\nicefrac{2}{5}$ \\
		$\nicefrac{n_i}{6}$   & $\boldsymbol{10}$ & $3\,\nicefrac{1}{3}$ & $1\,\nicefrac{2}{3}$  & $1\,\nicefrac{1}{3}$ & $\nicefrac{1}{3}$ \\
		$\nicefrac{n_i}{7}$   & $\boldsymbol{8}\,\pmb{\nicefrac{4}{7}}$ & $2\,\nicefrac{6}{7}$ & $1\,\nicefrac{3}{7}$  & $1\,\nicefrac{1}{7}$ & $\nicefrac{2}{7}$ \\
		$\nicefrac{n_i}{8}$   & $7\nicefrac{1}{2}$ & $2\,\nicefrac{1}{2}$ & $1\,\nicefrac{1}{4}$  & $1$ & $\nicefrac{1}{4}$ \\
		\bottomrule
	\end{tabular}
\end{center}
\endminipage\hfill
\minipage{0.45\textwidth}
\begin{center}
	\begin{tabular}{l|ccccc}
		\toprule
		  & $A$ & $B$ & $C$ & $D$ & $E$ \\
		\midrule
		$n_i$   & $\boldsymbol{60}$ & $\boldsymbol{20}$ & $\boldsymbol{10}$  & $\boldsymbol{8}$ & 2 \\
		$\nicefrac{n_i}{3}$   & $\boldsymbol{20}$ & $\boldsymbol{6}\,\pmb{\nicefrac{2}{3}}$ & $3\,\nicefrac{1}{3}$  & $2\,\nicefrac{2}{3}$ & $\nicefrac{2}{3}$ \\
		$\nicefrac{n_i}{5}$   & $\boldsymbol{12}$ & 4 & 2  & $1\,\nicefrac{3}{5}$ & $\nicefrac{2}{5}$ \\
		$\nicefrac{n_i}{7}$   & $\boldsymbol{8}\,\pmb{\nicefrac{4}{7}}$ & $2\,\nicefrac{6}{7}$ & $1\,\nicefrac{3}{7}$  & $1\,\nicefrac{1}{7}$ & $\nicefrac{2}{7}$ \\
		$\nicefrac{n_i}{9}$   & $\boldsymbol{6}\,\pmb{\nicefrac{2}{3}}$ & $2\,\nicefrac{2}{9}$ & $1\,\nicefrac{1}{9}$  & $1\,\nicefrac{8}{9}$ & $\nicefrac{2}{9}$ \\
		$\nicefrac{n_i}{11}$  & $\boldsymbol{5}\,\pmb{\nicefrac{5}{11}}$ & $1\,\nicefrac{9}{11}$ & $\nicefrac{10}{11}$  & $\nicefrac{8}{11}$ & $\nicefrac{2}{11}$ \\
		$\nicefrac{n_i}{13}$  & $4\,\nicefrac{8}{13}$ & $1\,\nicefrac{7}{13}$ & $\nicefrac{10}{13}$  & $\nicefrac{8}{13}$ & $\nicefrac{2}{13}$ \\
		$\nicefrac{n_i}{15}$  & $4$ & $1\,\nicefrac{1}{3}$ & $\nicefrac{2}{3}$  & $\nicefrac{8}{15}$ & $\nicefrac{2}{15}$ \\
		\bottomrule
	\end{tabular}
\end{center}
\endminipage
\end{center} 
In the subsequent rounds the D'Hondt method allocates seats to parties $A$, $A$, $A$ (by tie-breaking), $B$, $A$, $A$, $A$ (by tie-breaking), $B$ (by tie-breaking), $C$, and $A$. For example, in the fourth round, when $A$ is already allocated $3$ seats and $B$ is allocated none, the rule will give the next seat to $B$ rather than to $A$, because $\frac{20}{0 + 1} > \frac{60}{3 + 1}$. Summarising, seven seats will be allocated to party $A$, two seats to party $B$, and one seat to party $C$; the remaining parties will get no seats. In the diction of ABC rules, winning committees are exactly those that consist of seven candidates from $A$, two candidates from $B$ and one candidate from $C$.

The computation of the Sainte-Lagu\"e method is illustrated in the above right table. It will allocate six seats to $A$, two seats to $B$, one seat to $C$, and one seat to $D$. 

The largest remainder method first assigns to parties $A$, $B$, $C$, $D$, and $E$---respectively---6, 2, 1, 0, and 0 seats. Then, the remainders are considered:
\begin{center}
\renewcommand{\arraystretch}{1.2}%
\minipage{0.45\textwidth}
\begin{center}
	\begin{tabular}{l|ccccc}
		\toprule
		  & $A$ & $B$ & $C$ & $D$ & $E$ \\
		\midrule
		$n_i$   & ${60}$ & ${20}$ & ${10}$  & $8$ & $2$ \\
		$\left\lfloor k \cdot \frac{n_i}{n} \right\rfloor$  & $6$ & $2$ & $1$  & $0$ & $0$ \\
		remainder   & ${0}$ & $0$ & $0$  & $0.8$ & $0.2$ \\ \midrule
		seats & $6$ & $2$ & $1$  & $1$ & $0$ \\
		\bottomrule
	\end{tabular}
\end{center}
\endminipage
\end{center} 
There is one unassigned seat which will be given to the party with the largest remainder, namely to $D$. Thus, LRM will allocate six seats to $A$, two seats to $B$, one seat to $C$, and one seat to $D$.
\end{example}

The D'Hondt method, the Sainte-Lagu\"e method, and LRM exhibit particularly appealing properties. For example, the D'Hondt method satisfies \emph{lower quota}, which means that a party $i$ which receives $n_i$ out of $n$ votes must be allocated at least $\lfloor k\cdot \nicefrac{n_i}{n} \rfloor$ committee seats. 
The largest remainder method satisfies not only lower quota but also \emph{upper quota}: a party $i$ with $n_i$ out of $n$ votes must not receive more than $\lceil k\cdot \nicefrac{n_i}{n} \rceil$ seats.
However, the largest remainder method fails an important axiom called population monotonicity, which states that an increase in support must not harm a party.
In contrast, population monotonicity is satisfied by D'Hondt and Sainte-Lagu\"e. 
For further details, we refer the interested reader to the aforementioned books on apportionment methods~\cite{BaYo82a, Puke17}.

We are now ready to formulate the main results of \citet{bri-las-sko:c:apportionment}:

\begin{theorem}[\citet{bri-las-sko:c:apportionment}]\label{thm:apportionment}
PAV, sequential PAV, seq-\phragmen{}, and \lexphrag{} extend the D'Hondt method of apportionment. Phragm\'{e}n's variance-minimising rule\footnote{This rule is similar to \lexphrag{} but minimises the variance of loads instead of the maximum load, see \cite{aaai/BrillFJL17-phragmen,Janson16arxiv} for a precise definition.} extends the Sainte-Lagu\"e method of apportionment. If~$n$ is divisible by~$k$, then Monroe's rule extends the largest remainders method.
\end{theorem}

\Cref{thm:apportionment} lists ABC rules
that behave proportionally on party-list profiles and thus these rules can be considered good contenders for being proportional in the general ABC model.
In addition, we show in the appendix that also Greedy Monroe extends the largest remainder method when $n$ is divisible by $k$
 (\Cref{prop:greedy_monroe_apportionment}),
but both Monroe's rule and Greedy Monroe do not if $n$ is not divisible by~$k$ (\Cref{prop:greedy_monroe_apportionment_k_does_not_divide_n}). 

\citet{jet-consistentabc} strengthened the results of \citet{bri-las-sko:c:apportionment}, providing a strong argument in favour of PAV:

\begin{theorem}[\citet{jet-consistentabc}]\label{thm:pav_characterisation}
PAV is the unique extension of the D'Hondt method that satisfies neutrality, anonymity, consistency, and continuity.\label{thm:pav-characterisation}
\end{theorem}

\citet{jet-consistentabc} further show that this result can be generalised to 
arbitrary divisor-based apportionment methods. For example, the Sainte-Lagu\"e method yields the $w$-Thiele method with $w(x)=\sum_{j=1}^{x} \frac{1}{2j-1}$.

\section{Cohesive Groups}\label{sec:cohesive_groups}

In party-list profiles (\Cref{def:party_profiles}), voters can be arranged in groups with identical preferences. Then, proportionality requires that a large-enough group of voters with identical preferences deserves a certain number of representatives in the elected committee (proportional to the size of the group). 
This approach can be generalised to groups with non-identical but similar preferences.
We now discuss axioms that relax the requirements for groups of voters to be entitled to representatives. These axioms are based on the concept of $\ell$-cohesiveness:

\begin{definition}\label{def:cohesiveness}
For $\ell\geq 1$, a group $V \subseteq N$ is \emph{$\ell$-cohesive} if:
\begin{enumerate}[label=(\roman*)]
\item $|V| \geq \ell\cdot \frac{n}{k}$, and
\item $\left|\bigcap_{i \in V} A(i) \right| \geq  \ell$.
\end{enumerate}
\end{definition}

An $\ell$-cohesive group consists of an $\nicefrac{\ell}{k}$-th fraction of voters, thus, intuitively, such a group should be able to control at least $\nicefrac{\ell}{k} \cdot k = \ell$ committee seats. Further, an $\ell$-cohesive group agrees on $\ell$ candidates, so one can ensure each member of the group gets $\ell$ representatives by selecting only $\ell$ candidates. It is, hence, tempting to require that for each $\ell$-cohesive group $V$, each voter from $V$
should be given at least $\ell$ representatives in the elected committee. Unfortunately, this would be too strong---there exists no rule that would satisfy this property.

\begin{example}[\citet{AEHLSS18}]\label{ex:cohesive}
Consider a profile $A$ with four candidates ($a, b, c, d$) and 12 voters, with the following approval sets:
\begin{align*}
&A(1)\colon \{ a, d\}   &&  A(4)\colon \{ a, b\}     &&   A(7)\colon \{ b, c\}     && A(10)\colon \{ c, d\} \\
&A(2)\colon \{ a \}     &&  A(5)\colon \{ b\}        &&   A(8)\colon \{ c\}        && A(11)\colon \{ d\}     \\
&A(3)\colon \{ a \}     &&  A(6)\colon \{ b\}        &&   A(9)\colon \{ c\}        && A(12)\colon \{ d\}\text{.} 
\end{align*}
Let $k=3$. The group $\{1, 2, 3, 4\}$ is 1-cohesive, as it has a commonly approved candidate ($a$) and is of size $\frac{12}{3}=4$. If we want to give each voter in this group a representative, candidate $a$ has to be in the winning committee (voters $2$ and $3$ only approve $a$). Now observe that also the groups $\{4,5,6,7\}$, $\{7,8,9,10\}$, and $\{10,11,12,1\}$ are 1-cohesive. Thus, also candidates $b$, $c$, and $d$ have to be in every winning committee. This is impossible as we are interested in committees of size $3$.
We see that it is impossible to satisfy \emph{every} voter in 1-cohesive groups.
\end{example} 

We see from this example that the requirement that each voter from an $\ell$-cohesive group should have at least $\ell$ representatives in the elected committee is simply too strong.\footnote{In a very recent work, \citet{individual-representation} explore this intuitive (but unachievable) requirement---called individual representation---in much more depth. In particular, they show that all ABC rules presented in this book sometimes fail individual representation even for elections where such a committee exists.
In addition, they study conditions under which individual representation can be satisfied.} However, it can be weakened a bit without losing much of its intuitive appeal. We start our discussion with \emph{extended justified representation (EJR)}~\citep{justifiedRepresentation} and \emph{proportionality degree}~\citep{pjr17,AEHLSS18,proprank,skowron:prop-degree}.\footnote{The concept of proportionality degree was initially referred to as \emph{average satisfaction of $\ell$-cohesive groups}~\cite{pjr17,AEHLSS18}. \citet{proprank} called an almost equivalent property $\kappa$-group representation.}
The former concept is formulated as an axiom, the latter as a proportionality guarantee specified by a function.

\begin{definition}[Extended justified representation, EJR]\label{def:ejr} An ABC rule $\calR$ satisfies \emph{extended justified representation (EJR)} if for each election instance $E = (A, k)$, each winning committee $W \in \calR(E)$, and each $\ell$-cohesive group of voters $V$ there exists a voter $i\in V$ with at least $\ell$~representatives in $W$, i.e., $|A(i) \cap W| \geq \ell$.
\end{definition}

\begin{example}
Let us revisit \Cref{ex:cohesive}. The committee $\{a,b,c\}$ satisfies the condition of EJR: every $1$-cohesive group contains at least one voter with one representative in $\{a,b,c\}$. For example, for the 1-cohesive group $\{10,11,12,1\}$, the voters $10$ and $1$ have a representative in the committee. Note that in this example actually all size-$3$ committees satisfy the EJR condition; also there are no $\ell$-cohesive groups for $\ell\geq2$.
\end{example}

\begin{definition}[Proportionality degree]
Fix a function $f\colon \naturals \to \reals$. An ABC rule $\calR$ has a \emph{proportionality degree} of $f$ if for each election instance $E = (A, k)$,  each winning committee $W \in \calR(E)$, and each $\ell$-cohesive group of voters $V$, the average number of representatives that voters from $V$ get in $W$ is at least $f(\ell)$, i.e.,
\begin{align*}
\frac{1}{|V|} \cdot \sum_{i \in V} \left| A(i) \cap W \right| \geq f(\ell) \text{.}
\end{align*}
\end{definition}


At first, it might appear that even for large cohesive groups, EJR gives a guarantee only to a single voter within this group. However, the EJR property applies to any group of agents: Let $V$ be an $\ell$-cohesive group. If we remove a voter with $\ell$ representatives (who, by EJR, is guaranteed to exist), the resulting group will be at least $(\ell-1)$-cohesive. Consequently, in such a group there must exist a voter with at least $\ell-1$ representatives, etc. 
As a consequence of this argument, EJR implies a proportionality degree of at least $f_{\calR}(\ell) = \frac{\ell-1}{2}$~\cite{pjr17}.
The other direction does not hold: even an ABC rule with a proportionality degree of $f_{\calR}(\ell) = \ell-1$ may fail EJR (cf. \Cref{prop:ejr+proprank}).

\Cref{ex:cohesive} also shows that there exists no rule with a proportionality degree of $f(\ell) = \ell$:

\begin{example}
Consider again the profile of \Cref{ex:cohesive}.
Assume, there exists a rule~$\calR$ with a proportionality degree of $f_{\calR}(\ell) = \ell$ and let $k = 3$. The group $\{1, 2, 3, 4\}$ is 1-cohesive, so in order to ensure that these voters get on average one representative, candidate $a$ must be selected. By applying the same reasoning to $\{4, 5, 6, 7\}$ we infer that $b$ must be selected. Analogously, we conclude that $c$ and $d$ must be selected. However, there are only three seats in the committee, a contradiction.
\end{example} 

\citet{AEHLSS18} generalise the above example and prove that there exists no rule with a proportionality degree of $f(\ell) = \ell - 1 + \epsilon$ for $\epsilon>0$. 
PAV matches this bound, and thus has an optimal proportionality degree. Below we include the proof of this result, since a similar idea is often used in the analysis of proportionality properties of Thiele methods.

\begin{theorem}[\citet{justifiedRepresentation,AEHLSS18}]\label{thm:prop-of-pav}
PAV has a proportionality degree of $\ell -1$. It also satisfies EJR.
\end{theorem}
\begin{proof}
Consider an election $E = (A, k)$ and let $W$ be a winning committee according to PAV. Let $N$ and $C$ denote the sets of voters and candidates in $E$, respectively. We will show that for each $\ell$-cohesive group of voters $V$ it holds that $\frac{1}{|V|} \cdot \sum_{i \in V} \left| A(i) \cap W \right| > \ell-1$. This proves that PAV has the proportionality degree of $\ell -1$. We can further conclude that there exists a voter $i \in V$ with $\left|A(i) \cap W \right| > \ell-1$, and hence PAV also satisfies EJR. 

Towards a contradiction assume there exists an $\ell$-cohesive group of voters $V$ with $\frac{1}{|V|} \cdot \sum_{i \in V} \left| A(i) \cap W \right| \leq \ell-1$. We will show that there exists a pair of candidates, $c \in W$ and $c' \notin W$, such that $\score{\pav}(A, (W \cup \{c'\}) \setminus \{c\}) > \score{\pav}(A, W)$. This would indicate that we can replace one member of $W$ with another not-selected candidate so that the new winning committee has a higher PAV-score than $W$. This would contradict the fact that $W$ is a winning committee.

For convenience, for a set of candidates $X$ and a candidate $y$ we will use the notation:
\begin{align*}
\Delta(X, y) = \score{\pav}(X \cup \{y\}) - \score{\pav}(X) \text{,}
\end{align*}
i.e., $\Delta(X, y)$ is the marginal contribution of $y$ given $X$.

Since $\frac{1}{|V|} \cdot \sum_{i \in V} \left| A(i) \cap W \right| \leq \ell-1$ and $V$ is $\ell$-cohesive, there exists a not-selected candidate $c' \in C$ that is approved by all the voters from $V$. If we add this candidate to the committee $W$, the PAV-score will increase by:
\begin{align*}
\Delta(W, c') = \sum_{i \in N(c')} \frac{1}{|A(i) \cap W| + 1} \geq \sum_{i \in V} \frac{1}{|A(i) \cap W| + 1}\text{.}
\end{align*}
From the inequality between the arithmetic and harmonic means we further get that:
\begin{align*}
\Delta(W, c') \geq \frac{|V|^2}{\sum_{i \in V}(|A(i) \cap W| + 1)} \geq \frac{|V|^2}{|V|(\ell - 1) + |V|} = \frac{|V|}{\ell} \geq \frac{n}{k}\text{.}
\end{align*}
The last inequality follows from $\ell$-cohesiveness.

Now, consider a committee $W' = W \cup\{c'\}$, and observe that
\begin{align*}
\sum_{c \in W'} \Delta(W'\setminus \{c\}, c) &= \sum_{c \in W'} \sum_{i \in N(c)} \frac{1}{|A(i) \cap W'|} = \sum_{i \in N} \sum_{c \in A(i) \cap W'} \frac{1}{|A(i) \cap W'|} \\
                                             &= \sum_{i \in N \colon A(i) \cap W' \neq \emptyset}  |A(i) \cap W'| \cdot \frac{1}{|A(i) \cap W'|} \leq n \text{.}
\end{align*}
As a result, there exists $c \in W'$ such that $\Delta(W'\setminus \{c\}, c) \leq \frac{n}{k+1}$. Consequently:
\begin{align*}
\score{\pav}(A, (W \cup \{c'\}) \setminus \{c\}) &=  \score{\pav}(A, W) + \Delta(W, c') - \Delta(W'\setminus{c}, c) \\
                                                                        &\geq \score{\pav}(A, W) + \frac{n}{k} - \frac{n}{k+1} > \score{\pav}(A, W) \text{.}
\end{align*}
This yields a contradiction and completes the proof.
\end{proof}

In contrast to PAV, the two sequential variants of PAV, seq-PAV and rev-seq-PAV, do not satisfy EJR.
However, the proportionality guarantees of \Cref{thm:prop-of-pav} also hold for a local-search variant of PAV~\citep{AEHLSS18}, which---in contrast to PAV itself---runs in polynomial time.
Thus, EJR and a proportionality degree of $\ell-1$ are achievable in polynomial time.
\citet{AEHLSS18} also construct a second polynomial-time computable (but rather involved) rule that satisfies EJR.  More recently, \citet{pet-sko:laminar} prove that the Method of Equal Shares, which is also computable in polynomial time, satisfies EJR.
Among the rules introduced in \Cref{sec:abc_rules}, PAV and the Method of Equal Shares are the only ones that satisfy EJR.
An overview of the proportionality degree of rules can be found in \Cref{tab:proportionality_summary}.

Let us now consider two properties that are weaker than EJR.

\begin{definition}[Proportional justified representation, PJR~\citep{pjr17}]\label{def:pjr}
An ABC rule $\calR$ satisfies \emph{proportional justified representation (PJR)} if for each election $E = (A, k)$, each winning committee $W \in \calR(E)$, and each $\ell$-cohesive group of voters $V$ it holds that $\left| W\cap \left(\bigcup_{i\in V} A(i) \right)\right| \geq \ell$.
\end{definition}

\begin{definition}[Justified representation, JR~\citep{justifiedRepresentation}] \label{def:jr}An ABC rule $\calR$ satisfies \emph{justified representation (JR)} if for each election $E = (A, k)$, each $W \in \calR(E)$, and each $1$-cohesive group of voters $V$ there exists a voter $i\in V$ who is represented by at least one member of~$W$, i.e., $\left| W\cap A(i) \right| \geq 1$.
\end{definition}

PJR and JR are much weaker properties than EJR; in particular EJR implies PJR, which in turn implies JR. \Cref{ex:pjr_and_ejr}, below, illustrates that the stronger of the two axioms, PJR, can be satisfied even by rules that could be considered very bad from the perspective of proportionality degree (and, thus, also from the perspective of approximating EJR). On the other hand, there exist rules with good proportionality degree that do not satisfy even JR---this happens, e.g., when a rule does not provide sufficient guarantees for 1-cohesive groups (although it might satisfy EJR for $\ell\geq 2$). 
Generally, justified representation cannot be viewed as a proportionality axiom as it grants even large group only a single representative in the selected committee.
In contrast, PJR can be viewed as a moderate proportionality requirement, significantly weaker than EJR but stronger than, e.g., lower quota on party-list profile. 
We refer to \Cref{tab:proportionality_summary} for an overview which rules satisfy JR and PJR.

\begin{example}\label{ex:pjr_and_ejr}
Fix $k$ and consider the following instance:
\[
\begin{tikzpicture}
[yscale=1.0,xscale=2.0]
\filldraw[fill=green!10!white, draw=black] (0,0.5) rectangle (4.0,1.0);
\node at (2.0, 0.75) {$c_{k+1}$};
\filldraw[fill=green!10!white, draw=black] (0,1.0) rectangle (4.0,1.5);
\node at (2.0, 1.25) {$c_{k+2}$};
\filldraw[fill=green!10!white, draw=black] (0,1.5) rectangle (4.0,2.0);
\node at (2.0, 1.75) {$\cdots$};
\filldraw[fill=green!10!white, draw=black] (0,2.0) rectangle (4.0,2.5);
\node at (2.0, 2.25) {$c_{2k}$};

\filldraw[fill=blue!10!white, draw=black] (0,0) rectangle (0.8,0.5);
\node at (0.4, 0.25) {$c_1$};
\node at (0.4, -0.35) {$V_1$};
\filldraw[fill=blue!10!white, draw=black] (0.8, 0) rectangle (1.6,0.5);
\node at (1.2, 0.25) {$c_2$};
\node at (1.2, -0.35) {$V_2$};
\filldraw[fill=blue!10!white, draw=black] (1.6, 0) rectangle (2.4,0.5);
\node at (2.0, 0.25) {$c_3$};
\node at (2.0, -0.35) {$V_3$};
\filldraw[fill=blue!10!white, draw=black] (2.4, 0) rectangle (3.2,0.5);
\node at (2.8, 0.25) {$\cdots$};
\node at (2.8, -0.35) {$\cdots$};
\filldraw[fill=blue!10!white, draw=black] (3.2, 0) rectangle (4.0,0.5);
\node at (3.6, 0.25) {$c_{k}$};
\node at (3.6, -0.35) {$V_{k}$};
\end{tikzpicture}
\]
There are $2k$ candidates. The voters can be divided into $k$ equal-size groups so that the voters from the $i$-th group, in the diagram denoted as $V_i$, approve $c_i$ and $\{c_{k+1}, \ldots, c_{2k}\}$.
Committee $\{c_1, \ldots, c_k\}$ (marked blue) satisfies PJR, but clearly, $\{c_{k+1}, \ldots, c_{2k}\}$ (marked green) is a much better choice from the perspective of proportionality degree. Also, $\{c_{k+1}, \ldots, c_{2k}\}$ satisfies the EJR condition while $\{c_1, \ldots, c_k\}$ does not.
This example shows that PJR implies no better proportionality degree than $f(\ell)=1$.
\end{example} 

Given that there are rather few rules satisfying EJR, \citet{bre-fal-kac-nie2019:experimental_ejr} performed computer simulations for several distributions of voters' preferences and verified how hard it is \emph{on average} to find a committee that satisfies the condition imposed by EJR. They concluded that $\ell$-cohesive groups for $\ell \geq 2$ are quite rare, and that a random committee among those that satisfy the much weaker condition of JR is quite likely to satisfy EJR as well. Their second conclusion was that JR, PJR, and EJR, while highly desired, do not guarantee on their own a sensible selection of committees, and one needs to put forward additional criteria. Specifically, they showed that there are often many committees satisfying these conditions, and these committees may vary significantly. \citet{bre-fal-kac-nie2019:experimental_ejr} derived their conclusions from the analysis of specific distributions of voters' preferences; it would be desirable to analyse this phenomenon more broadly, e.g., for other types of distributions.



Recently, \citet{pet-pie-sko:c:participatory-budgeting-cardinal} introduced an even stronger axiom, called \emph{fully justified representation (FJR)}, where the precondition of $\ell$-cohesiveness is relaxed. In EJR we say that a group of voters $V$ is $\ell$-cohesive if $|V| \geq \ell \cdot \nicefrac{n}{k}$ and if there exists a set $T$ of $\ell$ candidates such that each voter from $V$ approves all $\ell$ candidates from $T$. In the definition of FJR, on the other hand, we only require that there must exist an integer $\beta$ such that each voter from $V$ approves at least $\beta$ candidates from $T$. FJR enforces that at least one member of $V$ must have at least $\beta$ representatives in the elected committee. Note that EJR corresponds to FJR with a fixed value of $\beta = \ell$.

\begin{definition}[Fully justified representation~\cite{pet-pie-sko:c:participatory-budgeting-cardinal}]\label{def:fjr}
Given an integer value $\beta$ and a subset of candidates $T \subseteq C$, we say that a group of voters $V$ is weakly $(\beta, T)$-cohesive if $|V| \geq |T| \cdot \nicefrac{n}{k}$ and if for each voter $i \in V$ it holds that $|A(i) \cap T| \geq \beta$. An ABC rule $\calR$ satisfies \emph{fully justified representation (FJR)} if for each election $E = (A, k)$, each winning committee $W \in \calR(E)$, each integer $\beta$ and $T \subseteq C$, and each weakly $(\beta, T)$-cohesive group of voters $V$, there exists a voter $i \in V$ such that $|W \cap A(i)| \geq \beta$.
\end{definition}
For the time being, the only known rule that satisfies FJR is rather artificial and specifically tailored to the definition of the axiom~\cite{pet-pie-sko:c:participatory-budgeting-cardinal}. It is an open question whether there exists a natural ABC rule which satisfies FJR together with other desirable properties. 

All proportionality concepts discussed in this section ensure that cohesive groups are guaranteed a certain representation in the elected committee. \citet{cevallos2020verifiably} argue that in some contexts---for example when using ABC rules for selecting validators in the blockchain protocol---it is equally important to ensure that groups are not over-represented. To the best of our knowledge formal axioms capturing this intuitive requirement are still missing.\footnote{We note that the upper quota axiom in the apportionment setting can be viewed as such an axiom.}

To sum up, when considering proportionality axioms based on cohesive groups, PAV stands out as the most proportional rule. The Method of Equal Shares comes at a close second (its proportionality degree is lower) but it is computable in polynomial time.
If we desire a committee monotone rule, then seq-\phragmen{} is a very good choice: it has a proportionality degree of $f_{\mathrm{Phrag}}(\ell) = \frac{\ell-1}{2}$~\cite{skowron:prop-degree}, i.e., the proportionality degree that is implied by EJR, and satisfies PJR \cite{aaai/BrillFJL17-phragmen}. Also seq-PAV is a good choice: for reasonable sizes of committees seq-PAV has a better proportionality degree than seq-\phragmen{}; on the other hand, it satisfies neither PJR nor JR.

\section{Laminar Proportionality and Priceability}\label{sec:laminar}

The properties that we discussed in \Cref{sec:cohesive_groups} (extended justified representation and the proportionality degree) and the axiomatic characterisation given in \Cref{thm:pav_characterisation} all indicate that PAV provides particularly strong proportionality guarantees. Specifically, one could interpret these results as suggesting that PAV is a better rule---in terms of proportionality---than \phragmen's sequential rule and the Method of Equal Shares. However, drawing such a conclusion based on the so-far presented results would be too early.
In the following we explain that proportionality can be understood in at least two different ways and that the axioms we discussed so far capture and formalise only one specific form of proportionality. We explain that \phragmen's sequential rule and the Method of Equal Shares provide very strong proportionality guarantees, but with respect to an interpretation of  proportionality that is not captured by properties based on cohesive groups, and which is---to some extent---incomparable with the type of proportionality guaranteed by PAV.

Let us start by illustrating the difference in how PAV and \phragmen's sequential rule (and the Method of Equal Shares) operate with the following example. 

\begin{example}[\cite{pet-sko:laminar}]\label{ex:pav_phragmen_difference}
There are 15 candidates and 6 voters---the voters' approval sets are depicted in the diagram below. The committee shaded in blue in the left-hand side picture is the one that is selected by the \phragmen's sequential rule and by the Method of Equal Shares. The committee shaded in the right-hand side picture is chosen by PAV.

\[
\begin{tikzpicture}
[yscale=0.43,xscale=0.78,voter/.style={anchor=south, yshift=-7pt}, select/.style={fill=blue!10}, c/.style={anchor=south, yshift=1.5pt, inner sep=0}]
	\draw[select] (0,0) rectangle (3,1);
	\draw[select] (0,1) rectangle (3,2);
	\draw[select] (0,2) rectangle (3,3);
	\draw[select] (0,3) rectangle (1,4);
	\draw[select] (1,3) rectangle (2,4);
	\draw[select] (2,3) rectangle (3,4);
	\node at (1.5,0.42) {$c_1$};
	\node at (1.5,1.42) {$c_2$};
	\node at (1.5,2.42) {$c_3$};
	\node at (0.5,3.42) {$c_4$};
	\node at (1.5,3.42) {$c_5$};
	\node at (2.5,3.42) {$c_6$};
	\foreach \x in {3,4,5}
		{
		\foreach \y in {0,1}
			{
			\draw[select] (\x,\y) rectangle (\x+1,\y+1);
			}
		\foreach \y in {2}
			{
			\draw (\x,\y) rectangle (\x+1,\y+1);
			}
		}
	\node at (3.5,0.42) {$c_{7}$};
	\node at (3.5,1.42) {$c_{8}$};
	\node at (3.5,2.42) {$c_{9}$};
	\node at (4.5,0.42) {$c_{10}$};
	\node at (4.5,1.42) {$c_{11}$};
	\node at (4.5,2.42) {$c_{12}$};
	\node at (5.5,0.42) {$c_{13}$};
	\node at (5.5,1.42) {$c_{14}$};
	\node at (5.5,2.42) {$c_{15}$};
	\foreach \i in {1,...,6}
		\node[voter] at (\i-0.5,-1) {$\i$};
		
	\node at (3, -2.5) {(a) seq-\phragmen and Equal Shares};
\end{tikzpicture}
\qquad\quad
\begin{tikzpicture}
[yscale=0.43,xscale=0.78,voter/.style={anchor=south, yshift=-7pt}, select/.style={fill=blue!10}, c/.style={anchor=south, yshift=1.5pt, inner sep=0}]
	\draw[select] (0,0) rectangle (3,1);
	\draw[select] (0,1) rectangle (3,2);
	\draw[select] (0,2) rectangle (3,3);
	\draw (0,3) rectangle (1,4);
	\draw (1,3) rectangle (2,4);
	\draw (2,3) rectangle (3,4);
	\node at (1.5,0.42) {$c_1$};
	\node at (1.5,1.42) {$c_2$};
	\node at (1.5,2.42) {$c_3$};
	\node at (0.5,3.42) {$c_4$};
	\node at (1.5,3.42) {$c_5$};
	\node at (2.5,3.42) {$c_6$};
	\foreach \x in {3,4,5}
		{
		\foreach \y in {0,1,2}
			{
			\draw[select] (\x,\y) rectangle (\x+1,\y+1);
			}
		}
	\node at (3.5,0.42) {$c_{7}$};
	\node at (3.5,1.42) {$c_{8}$};
	\node at (3.5,2.42) {$c_{9}$};
	\node at (4.5,0.42) {$c_{10}$};
	\node at (4.5,1.42) {$c_{11}$};
	\node at (4.5,2.42) {$c_{12}$};
	\node at (5.5,0.42) {$c_{13}$};
	\node at (5.5,1.42) {$c_{14}$};
	\node at (5.5,2.42) {$c_{15}$};
	\foreach \i in {1,...,6}
		\node[voter] at (\i-0.5,-1) {$\i$};
		
	\node at (3, -2.5) {(b) PAV};
\end{tikzpicture}
\]

The approval sets of voters $1, 2$, and $3$ are disjoint from those of voters $4, 5$, and~$6$. It seems intuitive that the first three voters, who together form half of the society, should be able to decide about half of the elected candidates. \phragmen's sequential rule and the Method of Equal Shares select committees where the first three voters approve in total half of the members, thus the behaviour of these rules is consistent with the aforementioned understanding of proportionality. PAV follows a different principle: In the committee depicted in~(a), each of the first three voters approves 4 candidates; each of the remaining three voters approves only 2 committee members. PAV notices that this is the case, and tries to reduce the societal inequality of voters' satisfaction by removing one representative of voter $1$ and adding one to $4$; similarly, PAV considers that it is more fair to remove the representatives of $2$ and $3$, and add the candidates liked by $5$ and $6$. On the one hand, PAV prefers to pick a committee that minimises the societal inequality in the voters' satisfactions (measured as the number of approved committee members). On the other hand, it punishes voters $1, 2$, and $3$ for being agreeable and ``easy to satisfy'' with fewer committee members---PAV allows them to decide only about one quarter of the committee.
\end{example}

\Cref{ex:pav_phragmen_difference} illustrates that PAV and \phragmen's sequential rule (and the Method of Equal Shares) follow two different types of proportionality. PAV implements a \emph{welfarist} type of proportionality which is primarily concerned with the welfare (satisfaction) of the voters. This type of proportionality is captured, e.g., by the properties discussed in \Cref{sec:cohesive_groups}. PAV also satisfies the Pigou--Dalton principle of transfers, which says that given an election $(A, k)$ and two committees, $W$ and $W'$, which in total get the same numbers of approvals ($\score{\av}(A, W) = \score{\av}(A, W')$), the one which minimises the societal inequality should be preferred \cite{pet-sko:laminar}. \phragmen's sequential rule and the Method of Equal Shares, on the other hand, implement proportionality with respect to power, which---informally speaking---says that a group consisting of an $\alpha$ fraction of voters should be given a voting power that enables to decide about an $\alpha$ fraction of the committee. In other words, the type of proportionality of \phragmen-like rules is not mainly concerned with the welfare of groups but with the justification of welfare, achieved by endowing each voter with the same amount of virtual budget that represents the voting power.
  

\citet{pet-sko:laminar} discuss two properties---laminar proportionality and priceability---which aim at formally capturing the high-level idea of proportionality with respect to power.\footnote{Laminar proportionality and priceability are similar in spirit but are logically independent (neither implies the other).}
The first of the two properties---\emph{laminar proportionality}---is very similar in spirit to proportionality on party-list profiles. The corresponding axiom identifies a class of well-structured election instances---called \emph{laminar elections}---and specifies how a laminar proportional rule should behave on these profiles. Laminar profiles are more general than party-list profiles and are defined by a recursive structure, similar to the election from \Cref{ex:pav_phragmen_difference}.

The second property, which we will discuss in more detail, is \emph{priceability}. Intuitively, we say that a committee $W$ is priceable if we can endow each voter with the same fixed budget and if for each voter there exists a payment function such that:
\begin{inparaenum}[(1)]
\item voters do not spend more than their allotted budget,
\item voters pay only for the candidates they approve, 
\item each elected candidate gets a total payment of 1; candidates that are not elected receive no payments, and 
\item there is no group of voters who approve a non-elected candidate, and who in total have more than one unit of unspent budget. 
\end{inparaenum}
Priceability is a notion of proportionality as it distributes power to groups of sufficient size; a large enough group receives enough collective budget to afford one or more candidates in the committee.

Formally, we obtain the following definition:

\begin{definition}[Priceability]\label{def:priceability}
Given an election instance $(A, k)$, a committee $W$ is priceable if there exists a per-voter budget $p \in \reals_{+}$ and $p_i \colon C \to [0, 1]$ for each voter $i \in N$ such that:
\begin{enumerate}[label=(\arabic*)]
\item $\sum_{c \in C}p_i(c) \leq p$ for each $i \in N$,
\item $p_i(c) = 0$ for each $i \in N$ and $c \notin A(i)$,
\item $\sum_{i \in N}p_i(c) = \begin{cases} 1 & \text{if $c \in W$,}\\
 0 & \text{otherwise.}\end{cases}$
\item $\sum_{i \in N(c)}\left(p - \sum_{c' \in W}p_i(c')\right) \leq 1$ for each $c \notin W$.
\end{enumerate}
An ABC rule is priceable if it returns only priceable committees.\footnote{We remark that this definition builds on the assumption that at least $k$ candidates are approved at least once. If there are fewer than $k$ such candidates, one may define each committee containing these candidates as priceable.}
\end{definition}

\begin{example}\label{ex:priceability}
Consider the election instance from \Cref{ex:pav_phragmen_difference}. The committees returned by \phragmen's sequential rule and by the Method of Equal Shares are priceable. For example, consider $W_1 = \{c_1, \ldots, c_{6}, c_{7}, c_{8}, c_{10}, c_{11}, c_{13}, c_{14}\}$ (the committee shaded blue in the left figure in \Cref{ex:pav_phragmen_difference}). This committee is priceable as witnessed by the following price system: the voters' budget is $p = 2$, and the payment functions are as follows (we only specify the non-zero payments):
$p_1(c_i) = p_2(c_i) = p_3(c_i) = \nicefrac 1 3 $ for $i\in \{1,2,3\}$ and
$p_1(c_4) = p_2(c_5) = p_3(c_6) = p_4(c_7) = p_4(c_8) = p_5(c_{10}) = p_5(c_{11}) = p_6(c_{13}) = p_6(c_{14}) = 1$. Each voter fully spends their budget of $2$.

On the other hand, the committee $W_2 = \{c_1, c_2, c_3, c_7, \ldots, c_{15}\}$ returned by PAV (the one shaded blue in the right figure in \Cref{ex:pav_phragmen_difference}) is not priceable. Indeed, if the voters' budget $p$ were $\leq 2$, then the voters $4, 5, 6$ could not afford to pay for 9 candidates $c_7, \ldots, c_{15}$. If $p > 2$, then some of the voters $1, 2, 3$, say voter $1$, would have a remaining budget of more than $1$. Hence, this voter would have more budget than needed to buy a candidate outside of $W_2$ (e.g., $c_4$), which contradicts condition~(4) in \Cref{def:priceability}.
\end{example}

\citet{pet-sko:laminar} generalised  \Cref{ex:priceability} and showed that no welfarist rule (see \Cref{def:welfarist_rules}) is priceable. This shows that priceability is inherently not a welfarist concept. The same is true for laminar proportionality.

\begin{theorem}[\citet{pet-sko:laminar}]\label{thm:laminar_proportionality_and_welfarism}
\phragmen's sequential rule and the Method of Equal Shares are laminar proportional and priceable. No welfarist rule is laminar proportional nor priceable. No rule satisfying the Pigou--Dalton principle of transfers is laminar proportional nor priceable.
\end{theorem}




While priceability is not a welfarist concept, it implies proportional justified representation. Further, all priceable rules must be equivalent to the D'Hondt method of apportionment on party-list profiles (cf. \Cref{thm:apportionment}). A price system provides an explicit and easily verifiable evidence explaining that the voters can use their power (represented through virtual money) to ensure that the candidates from the committee are selected. 
This intuitively explains that priceability captures the idea of proportionality with respect to power---proportionality follows from the fact that each voter is initially endowed with the same amount of virtual money.

Priceability itself puts rather mild constraints on the payment functions $\{p_i\}_{i \in N}$. Recently, \citet{pet-pie-sha-sko:c:stable-priceability} introduced a stronger version of the axiom: we say that a price system $(p, \{p_i\}_{i \in N})$ is \emph{stable} if it satisfies conditions (1)--(3) from \Cref{def:priceability} and the following strengthening of condition (4):
\begin{description}
\item[\emph{(4*) Condition for Stability:}] There exists no non-empty group of voters $V \subseteq N$, no subset $W' \subseteq C \setminus W$, and no collections $\{p'_i\}_{i\in V}$ ($p'_i\colon W' \to [0, 1]$) and $\{R_i\}_{i\in V}$ (with $R_i \subseteq W$ for all $i \in V$) such that all the following conditions hold:
\begin{enumerate}
    \item For each $c \in W'$: $\sum_{i \in V} p'_i(c) > 1$.
    \item For each $i \in V$: $p_i(W \setminus R_i) + p'_i(W') \leq p$.
    \item Each voter $i \in V$ approves more candidates in $W \setminus R_i \cup W'$ than in $W$, or $i$ approves as many candidates in $W \setminus R_i \cup W'$ as in $W$ but $\sum_{c \in W \setminus R_i} p_i(c) + \sum_{c \in W'} p'_i(c) < \sum_{c \in W} p_i(c)$.
\end{enumerate}
\end{description}
In words, it should not be possible for the voters from $V$ to propose a set of candidates $W'$ such that if each voter $i \in V$ transferred her money from $R_i \subseteq W$ to the candidates from $W'$, then these candidates would garner more than enough money to be elected, and each voter from $i \in V$ would be happier with $W \setminus R_i \cup W'$ than with $W$.

Stable priceability is a strong condition: stable-priceable committees do not always exist, and if so, they belong to the core (see \Cref{sec:the_core}). On the other hand, one can check in a polynomial time whether a committee is stable-priceable, and such committees often exist in practice. \citet{pet-pie-sha-sko:c:stable-priceability} also introduced the concept of \emph{balanced stable-priceability}, which additionally requires that each two voters must pay the same amount of virtual money for the same candidate. They proved that balanced stable-priceable committees can be characterised as outputs of slightly modified version of the Method of Equal Shares.

We mention one more property---\emph{perfect representation}~\citep{pjr17}---which is loosely related to priceability. It also requires an explanation how voters can distribute their support/power in a way that justifies electing a committee; however, the axiom applies only in very specific situations.

\begin{definition}[\citet{pjr17}]\label{def:perfect_representation}
We say that a committee~$W$ satisfies \emph{perfect representation} if the set of voters can be divided into $k$ equally-sized disjoint groups $N = N_1 \cup \ldots \cup N_k$ ($|N_i| = \nicefrac{n}{k}$ for each $i \in k$) and if we can assign a distinct candidate from $W$ to each of these groups in a way that for each $i \in k$ the voters from $N_i$ all approve their assigned candidate. An ABC rule $\calR$ satisfies perfect representation if $\calR$ returns only committees satisfying perfect representation whenever such committees exist.
\end{definition}

Perfect representation is incompatible with EJR~\citep{pjr17} and with weak (and strong) Pareto efficiency (\Cref{prop:pr-pareto}), and it is not implied by (nor implies) priceability. Among the rules considered in this paper, only Monroe \cite{pjr17} and \lexphrag{} \cite{aaai/BrillFJL17-phragmen} satisfy perfect representation, as does the variance-based rule by \phragmen{} mentioned in \Cref{thm:apportionment} \cite{aaai/BrillFJL17-phragmen}.

To sum up, if we are mainly interested in the welfarist interpretation of proportionality, as captured by axioms that specify how cohesive groups of voters should be treated, then PAV is the best among the considered rules. Yet, sequential PAV, seq-\phragmen{}, and the Method of Equal Shares perform also reasonably well with respect to these criteria, and they are computable in polynomial time. Sequential PAV does not satisfy JR, and so it might discriminate small cohesive groups of voters. On the other hand, for reasonably small committees sequential PAV has better proportionality degree than seq-\phragmen{}, and the Method of Equal Shares.  The axioms that well describe the welfarist type of proportionality are EJR and proportionality degree, and to a lesser extent PJR and JR.
If we are interested in proportionality with respect to power, then we shall also consider the axioms of priceability and laminar proportionality.
In this case the Method of Equal Shares and \phragmen's sequential rule are the two superior rules. It is not entirely clear which one of the two rules is better. On the one hand, the Method of Equal Shares satisfies the appealing axiom of EJR; on the other hand,  \phragmen's sequential rule is committee monotone (see~\Cref{sec:comm-mon}). In \Cref{tab:proportionality_summary}, we highlighted the three rules that---with the current state of knowledge---we consider the best ABC rules in terms of proportionality.

\section{The Core}\label{sec:the_core}

An important concept of group fairness that has been extensively studied in the context of ABC rules is the \emph{core}. This notion of proportionality is adopted from cooperative game theory\footnote{Specifically, the definition used in the literature on multi-winner voting is based on the definition of the core for cooperative games with non-transferable payoffs~\cite{RePEc:mtp:titles:0262650401,chalkiadakis2011computational}.}, and was first introduced in the context of multi-winner voting by \citet{justifiedRepresentation}.

\begin{definition}\label{def:core}
Given an instance $(A, k)$, we say that a committee $W$ is in the \emph{core} if for each non-empty $V \subseteq N$ and each $T \subseteq C$ with 
\begin{align}
\frac{|T|}{k} \leq \frac{|V|}{n},\label{eq:core}
\end{align} there exists a voter $i \in V$ such that $|A(i) \cap T| \leq |A(i) \cap W|$, i.e., voter $i$ is at least as satisfied with $W$ as with $T$. We say that an ABC rule $\calR$ satisfies the core property if for each instance $(A,k)$ each winning committee $W \in \calR(A, k)$ is in the core. 
\end{definition}

Informally speaking, the core property requires that a group $V$ constituting an $\alpha$ fraction of voters should be able to control an $\alpha$ fraction of the committee. If such a group can propose a set $T$ of $\lfloor \alpha k \rfloor$ candidates such that each voter in $V$ is more satisfied with the proposed set $T$ than with the winning committee $W$, then the group $V$ would have an incentive to deviate, hence would witness that committee $W$ is not stable (and, in some sense, also not fair). If a winning committee is in the core, then no such deviation is possible.

The core property implies extended justified representation (\Cref{def:ejr}): Assume an ABC rule $\calR$ satisfies the core property and consider an instance $(A, k)$, a winning committee $W$, and an $\ell$-cohesive group of voters $V$. Let $T$ be the set of $\ell$ candidates that are approved by all the voters in $V$ (such candidates exist because $V$ is $\ell$-cohesive). Since $W$ is in the core, there must exist a voter $i \in V$ such that $|A(i) \cap W| \geq |A(i) \cap T| = \ell$, hence the condition of EJR must be satisfied. While the notion of core strictly generalises EJR and thus implies strong satisfaction guarantees for cohesive groups, it can also be viewed as a concept formalising the idea of proportionality with respect to power (cf.\ \Cref{sec:laminar}).

It is an important open question whether there exists an ABC rule that satisfies the core property, or---equivalently---whether the core is always non-empty. For the time being only partial answers to this intriguing question are known: 
\begin{enumerate}
\item None of the rules mentioned in \Cref{sec:abc_rules} satisfies the property. Since a rule satisfying the core must satisfy EJR, only PAV and the Method of Equal Shares come into consideration.
However, counterexamples for both are known \cite{justifiedRepresentation,pet-sko:laminar}. For PAV, the instance from \Cref{ex:pav_phragmen_difference} shows a violation of the core.
A simple example for the Method of Equal Shares can be found in \cite[Example 4]{abs-2108-01987}.

\item No welfarist rule (\Cref{def:welfarist_rules}) can satisfy the core property \cite{pet-sko:laminar}.

\item If one restricts the attention to a special subclass of approval profiles, so-called \emph{approval-based party-list profiles} as introduced by \citet{BGPSW19}, the situation changes. 
Approval-based party-list profiles are approval profiles where each candidate appears with at least $k$ copies, i.e., for every candidate~$c$ it holds that $\left|\{c' \in C \colon N(c) = N(c')\}\right| \geq k$. 
Approval-based party-list profiles are thus \emph{more general} than party-list profiles (cf.~\Cref{def:party_profiles})---intuitively each voter can approve one or more parties.
%
%
\citet{BGPSW19} prove that PAV satisfies the core property on approval-based party-list profiles. As mentioned before, PAV does not satisfy the core property in the general case.

\item It is known that the core can be empty in settings that are related to the ABC model but are more expressive. This is the case, e.g., in committee elections with ranking-based preferences \citep{cheng2019group,abs-2108-01987} and in participatory budgeting with additive utilities \citep[Appendix~C]{fain2018fair}; these two settings are discussed \Cref{sec:multiwinner-rank} and in \Cref{sec:part-budg}, respectively.

\end{enumerate}

As it remains unclear whether an ABC rule satisfying the core property is an achievable goal, several works in the most recent literature analysed relaxed notions of the core. We review these notions in the following.

\subsubsection{Relaxation by Randomisation}

The first type of relaxation that we consider is a probabilistic variant of the notion, i.e., the question becomes: ``can core-like properties be guaranteed in expectation (ex-ante)?''
\citet{cheng2019group} prove that there always exists a lottery over committees that satisfies the core property in expectation. 
Let $\expected_{X \sim \Delta}(X)$ denote the expected value of a random variable $X$ distributed according to a lottery (probability distribution)~$\Delta$.

\begin{theorem}[\citet{cheng2019group}]\label{thm:core_in_expectation}
For each election instance $(A, k)$ there exists a lottery over committees~$\Delta$ such that for each group of candidates $T \subseteq C$ it holds that
\begin{align}
\frac{|T|}{k} > \frac{\expected_{W \sim \Delta}\left(N(T, W)\right)}{n}\text{,}\label{eq:exp-core}
\end{align}
where $N(T, W)$ is the set of voters who prefer $T$ over $W$:
\begin{align*}
N(T, W) = \left\{ i \in N \colon |A(i) \cap T| > |A(i) \cap W| \right\} \text{.}
\end{align*}
\end{theorem}

Note that \Cref{eq:exp-core} is indeed a negated, probabilistic version of \Cref{eq:core}, showing that in expectation there are too few voters to propose a different committee.
While it is not known whether such a lottery $\Delta$ can be found in a polynomial time, \citet{cheng2019group} prove that if we restrict our attention only to sets $T$ of size bounded by a constant, then for each $\epsilon > 0$ there is a polynomial-time algorithm that computes $\Delta$ such that $(1 + \epsilon) \cdot \frac{|T|}{k} > \frac{\expected_{W \sim \Delta}\left(N(T, W)\right)}{n}$.

\subsubsection{Relaxation by Deterministic Approximation}

Another approach is to ask whether the core property can be well approximated. A few notions of approximation have been proposed; \Cref{def:core_approx1} below unifies the definitions considered in the literature.

\begin{definition}\label{def:core_approx1}
We say that an ABC rule $\calR$ provides a $\gamma$-multiplicative-$\eta$-additive-satisfaction $\beta$-group-size approximation to the core if for each instance $(A,k)$, each winning committee $W \in \calR(A,k)$, each non-empty subset of voters $V \subseteq N$, and each subset of candidates $T \subseteq C$ with \[\beta \cdot \frac{|T|}{k} \leq \frac{|V|}{n}\] there exists a voter $i \in V$ such that $|A(i) \cap T| \leq \gamma\cdot |A(i) \cap W| + \eta$.
\end{definition}

There are two components in \Cref{def:core_approx1}: The satisfaction-approximation component says that a voter~$i$ has an incentive to deviate towards $T$ only if her gain in satisfaction is sufficiently large, that is, if $i$'s satisfaction in $T$ is greater at least by a multiplicative factor of $\gamma$ \emph{and} an additive factor of $\eta$ than her satisfaction in~$W$. The group-size-approximation component prohibits deviations towards sets $T$ which are (by a multiplicative factor of $\beta$) smaller than $k \cdot \frac{|V|}{n}$, as imposed by the core.
If $\gamma = 1$, then we omit the term ``$\gamma$-multiplicative'' from the name of the property. Similarly, if $\eta = 0$ we omit the term ``$\eta$-additive'', and if $\beta = 1$, then we omit the term ``$\beta$-group-size''. The satisfaction-approximation and the group-size approximation are incomparable.

When considering the problem of approximating the core, we distinguish two classes of algorithms. The first class contains dedicated approximation algorithms, which are mostly based on dependent rounding of fractional committees. The second class consists of established rules, such as PAV or the Method of Equal Shares, which can be shown to approximate the core (to some degree).  

\citet{jiang2019approx} present an algorithm that provides $32$-group-size approximation to the core. Their approach is based on dependent rounding of lotteries that are in expectation in the core (the existence of such lotteries is guaranteed by \Cref{thm:core_in_expectation}). Notably, the approach of \citet{jiang2019approx} extends much beyond the approval-based preferences; for cardinal utilities they round a  lottery that in expectation $2$-approximates the core and obtain a discrete committee with the $32$-group-size approximation guarantee. 

\citet{fain2018fair} present a family of algorithms based on dependent rounding of fractional committees (returned by a linear program that closely resembles the formulation of PAV as an integer linear program). For each $\lambda \in (1,2]$ they provide an algorithm that guarantees a $\lambda$-multiplicative-$\eta$-additive-satisfaction \mbox{$\frac{1}{2 - \lambda}$-group}-size approximation to the core, where $\eta = O\left(\frac{1}{\lambda^4} \log\left(\frac{k}{\lambda}\right)\right)$. Their algorithm naturally extends to a more general model related to participatory budgeting. 

The result of \citet{fain2018fair} has recently been improved. \citet{mun-she-wan-wan:approximate-core} presented a polynomial time algorithm that guarantees $67.37$-multiplicative-$1$-additive-satisfaction approximation to the core. They also presented an algorithm that offers a $9.27$-multiplicative-$1$-additive-satisfaction approximation to the core, yet running in exponential time. These algorithms, which are based on dependent rounding, can be also applied to more general types of voters' preferences.

For commonly known rules the following results are known: \citet{cheng2019group} prove that PAV does not guarantee $\beta$-group-size approximation to the core even for $\beta = \Theta(\sqrt k)$. On the other hand, \citet{pet-sko:laminar} prove that PAV gives $2$-multiplicative-satisfaction approximation to the core. Further, for each $\epsilon > 0$ no rule that satisfies the Pigou--Dalton principle can provide a $(2-\epsilon)$-multiplicative-satisfaction approximation to the core. Thus, PAV can be viewed as giving the strongest multiplicative-satisfaction approximation to the core subject to satisfying the Pigou--Dalton principle of transfers. Finally, they show that the Method of Equal Shares provides $O(\log(k))$-multiplicative-$1$-additive-satisfaction approximation to the core.




\subsubsection{Relaxation by Constraining the Space of Deviations}\label{sec:core_restrictions}

Yet another approach to relaxing the core property is to prohibit only certain types of deviations.
As we have already explained at the beginning of this section, EJR can be viewed as a restricted variant of the core property: It prohibits the deviations of groups of voters towards outcomes $T$ on which the deviating voters unanimously agree. Intuitively, if a group $V$ agrees on all candidates from $T$, then it is easier for such a group to synchronise and to deviate, thus EJR can be viewed as the minimal restricted variant of the core. Motivated by the same arguments, \citet{pet-sko:laminar} considered other restricted variants of the core property.

A \emph{committee property} is a set of triples $(A', k', W')$, where $(A', k')$ is an election instance and $W'$ is a size-$k'$ committee. We write $A|_V$ for profile $A$ restricted to voters in $V\subseteq N$.

\begin{definition}[\citet{pet-sko:laminar}]
Let $\mathcal{P}$ be a committee property. Given an instance $(A, k)$, we say that a pair $(V, T)$, with $V\neq \emptyset$, $V \subseteq N$, $T \subseteq C$, is an \emph{allowed deviation} from a committee $W$ if
 \begin{inparaenum}[(1)]
\item $\frac{|T|}{k} \leq \frac{|V|}{n}$,
\item $|A(i) \cap T| > |A(i) \cap W|$ for each $i \in V$, and
\item $T$ has property $\mathcal{P}$, i.e., $(A|_V, |T|, T) \in \mathcal{P}$.
\end{inparaenum}
An ABC rule $\calR$ satisfies the \emph{core subject to $\mathcal{P}$} if for each instance $(A, k)$ and each winning committee $W \in \calR(A, k)$ there exists no allowed deviation.
\end{definition}

For example, let $\mathcal{P}_{\mathrm{coh}}$ be a committee property such that $(A', k', W') \in \mathcal{P}_{\mathrm{coh}}$ if and only if $W' \subseteq A'(i)$ for all voters $i$ in the domain of $A'$; we call $\mathcal{P}_{\mathrm{coh}}$ cohesiveness (cf.\ \Cref{def:cohesiveness}). Then, EJR can be equivalently defined as the core subject to cohesiveness. 

The Method of Equal Shares satisfies core subject to priceability with equal payments, which is a variant of  priceability that additionally requires that voters must pay the same amount of virtual budget for the same candidate (cf. \Cref{def:priceability}); priceability with equal payments is thus stronger than priceability, yet weaker than cohesiveness~\cite{pet-sko:laminar}. It is an open question whether the core subject to weaker (yet still natural) types of constraints is always non-empty.

\section{Degressive and Regressive Proportionality}\label{sec:degressive_and_regressive_proportionality}

The notions of proportionality that we discussed in \Cref{sec:apportionment,sec:cohesive_groups,sec:laminar,sec:the_core} aimed at capturing the following intuitive idea: An $\alpha$ fraction of voters should be able to decide about an $\alpha$ fraction of the committee---in this approach the relation between the size of the group and its eligibility is linear. In this section we discuss two alternative concepts: degressive and regressive proportionality. These two concepts should be viewed more as high-level ideas than formal properties. We first explain them intuitively, providing an illustrative example, and next we will discuss a few formal approaches to reasoning about degressive and regressive proportionality. 

According to \emph{degressive} proportionality, smaller groups of voters should be favoured, i.e., be eligible to more representatives in the elected committee than in the case of linear proportionality.\footnote{Degressive proportional apportionment is often used for distributing parliamentary seats among geographical regions, e.g., in the division of the European Parliament seats among EU countries (see the book of \citet{rose2013RepresentingEuropeans} for a discussion of arguments and negotiations that resulted in a degressive apportionment rule being used for assembling the European Parliament).} An extreme form of degressive proportionality is  \emph{diversity}~\cite{FSST-trends}---there, if possible, each voter should be represented by at least one candidate in the elected committee. At the other end is the idea of \emph{regressive} proportionality, where the emphasis is put on well-representing large groups. An extreme form of regressive proportionality is \emph{individual excellence}~\cite{FSST-trends}, where it is assumed that only the candidates with the highest total support from the voters should be elected. In fact, these two notions---diversity and individual excellence---are extreme to the extent that they can no longer be considered notions of proportionality.
\Cref{ex:degressive_proportionality}, below, illustrates the ideas of degressive and regressive proportionality, and the two extreme variants of them---diversity and individual excellence.  

\begin{example}\label{ex:degressive_proportionality}
Consider the approval-based preference profile from \Cref{ex:proportionality}:
\begin{align*}
&\text{60 voters}\colon && \{a_1, \ldots, a_{10}\}  &&  \text{20 voters}\colon  && \{b_1, \ldots, b_{6}\}  &&  \text{10 voters}\colon  && \{c_1, c_{2}\} \\
&\text{8 voters}\colon   && \{d_1, \ldots, a_{4}\}    &&  \text{2 voters}\colon    && \{e_1, e_2, e_3\}  \text{.}  &&&&
\end{align*}
A linearly-proportional committee $W_1$ could consist of six candidates from $A$, two candidates from $B$, one candidate from $C$, and one candidate from $D$ (this is the committee selected by the Sainte-Lagu\"e apportionment method).
Another linearly-proportional committee could consist of seven candidates candidates from $A$, two from $B$, one from $C$, but none from $D$ (this is the committee selected by the D'Hondt apportionment method).

In contrast, a degressive-proportional committee $W_2$ could, for example, consist of four candidates from $A$, three candidates from $B$, two candidates from $C$, and one candidate from $D$. Another example of a degressive-proportional committee would be $W_3$ with three candidates from each of the sets $A$, $B$, and $C$, and one from $D$. Committees $W_2$ and $W_3$, however, are not diverse, since two voters who support $E = \{e_1, e_2, e_3\}$ are not represented at all.
A diverse committee could consist of, e.g., four candidates from $A$, three candidates from $B$, one candidates from $C$, one candidate from $D$, and one candidate from $E$. 
A regressive-proportional committee would include more candidates from the set $A = \{a_1, \ldots, a_{10}\}$ at the cost of groups supported by less voters. For example, a committee that consists of eight candidates from $A$ and two candidates from $B$ would be regressive-proportional.
\Cref{tab:disprop} shows the example relations between a size of a group and its number of representatives for different forms of proportionality:
\end{example}

\begin{table}
{\renewcommand{\arraystretch}{1.2}%
\begin{center}
	\begin{tabular}{lccccc}
		\toprule
		\# votes & 60 & 20  & 10 & 8 & 2 \\
		\midrule
		example of linear proportionality (Sainte-Lagu\"e)        & 6 & 2   & 1  & 1 & 0 \\
		a different example of linear proportionality (D'Hondt)                     & 7 & 2   & 1  & 0 & 0 \\
		an example of degressive proportionality     & 4 & 3  &  2 & 1 & 0 \\
 		another example of degressive proportionality   & 3 & 3  &  3 & 1 & 0 \\
		an example of diversity                             & 4 & 3  & 1 & 1 & 1 \\
		another example of diversity                             & 2 & 2  & 2 & 2 & 2 \\
		an example of regressive proportionality     & 8 & 2  &  0 & 0 & 0 \\
		individual excellence                            & 10 & 0  & 0 & 0 & 0 \\
		\bottomrule
	\end{tabular}
\end{center}}
\caption{Flavors of (dis)proportionality}\label{tab:disprop}
\end{table}

The arguments in favour of degressive proportionality usually come from the analysis of probabilistic models describing how the decisions made by the elected committee map to the satisfaction of individual voters participating in the process of electing the committee (for party-list preferences, an excellent exposition is given by~\citet{RePEc:ucp:jpolec:doi:10.1086/670380}; see also \citep{laslier2012WhyNotProportional,MT12}). An interesting concrete example of degressive proportionality is square-root proportionality devised by \citet{Penr46a} (see also \citep{SlZy06a}),
where the idea is that the groups of voters should be represented proportionally to the square-roots of their sizes.\footnote{This method has been proposed
for the United Nations Parliamentary Assembly~\citep{Bumm10a} and for allocating voting weights in the Council of the European Union~\citep{slomczynski2017degressive}.} Further, degressive proportionally in general, and diversity in particular, are particularly appealing ideas in the context of deliberative democracy---there, the goal is to select a committee that should discuss and deliberate on various issues rather than make majoritarian decisions. It is argued that for deliberative democracy it is particularly important to represent as many various opinions in the committees as possible~\cite{ccElection,monroeElection}, which can be achieved by maximising the number of voters who are represented in the elected committee. 

On the other hand, the idea of regressive proportionality is particularly appealing when the goal is to select a committee of candidates based on their individual merits, for example when the goal of an election is to select finalists in a contest or to choose a set of grants that should be funded (then, the voters act as judges/experts).   

In the remaining part of this section we discuss two approaches to formalising the ideas of degressive and regressive proportionality: axiomatic approaches and a quantitative approach.

\subsubsection{Axiomatic Approaches to Diversity and Individual Excellence}

The axiomatic approach generally applies only to the extreme forms of the degressive and regressive proportionality, i.e., to diversity and individual excellence, respectively. This approach is similar to the one we discussed in \Cref{sec:apportionment}: by formalising the concepts of diversity and individual excellence on party-list profiles (\Cref{def:party_profiles}), we obtain axiomatic characterisations for the more general domain of ABC rules.

Intuitively, disjoint diversity requires that in party-list profiles as many voters as possible have at least one representative in the elected committee. Disjoint equality says that each approval carries the same strength, and so all candidates that are approved once have the same right for being elected.   

\begin{definition}[Disjoint diversity]\label{def:disjoint_diversity}
An ABC rule $\calR$ satisfies \emph{disjoint diversity} if for each party-list instance $(A, k)$ with voter sets $(N_1, \ldots, N_p)$ and $|N_1| \geq |N_2| \geq \ldots \geq |N_p|$,  there exists a winning committee $W \in \calR(A, k)$ that contains one candidate for each of the $k$ largest parties, i.e., for each $r \leq \min(p, k)$ and each $i \in N_r$ we have that $A(i) \cap W \neq \emptyset$.
\end{definition} 

\begin{definition}[Disjoint equality]\label{def:disjoint_equality}
An ABC rule $\calR$ satisfies \emph{disjoint equality} if for each election instance $(A, k)$ where each candidate is approved at most once and the number of approved candidates is at least $k$ (i.e., $|\bigcup_{i \in N}A(i)| \geq k$), a committee $W$ is winning if and only if it contains only approved candidates, $W \subseteq \bigcup_{i \in N}A(i)$.
\end{definition} 
Intuitively, disjoint equality is aimed at capturing the idea of individual excellence---the candidates that are approved exactly once are virtually indistinguishable from the perspective of the support coming from the voters; thus all such candidates should have equal rights to be selected. 

The following theorems show that, similarly to the case of D'Hondt proportionality (\Cref{thm:pav-characterisation}), the concepts of disjoint diversity and disjoint equality uniquely extend to the full domain of approval-based preferences if one assumes the natural axioms of anonymity, neutrality, and consistency (and a few more technical axioms).

\begin{theorem}[\citet{jet-consistentabc}]\label{thm:cc_and_av_characterisation}
The Approval Chamberlin--Courant rule is the only non-trivial ABC ranking rule that satisfies anonymity, neutrality, consistency, weak efficiency, continuity, and disjoint diversity. Multi-Winner Approval Voting is the only ABC ranking rule that satisfies anonymity, neutrality, consistency, weak efficiency, continuity, and disjoint equality.
\end{theorem}

\citet{jet-consistentabc} provided a similar analysis for intermediate notions of degressive and regressive proportionality. They conclude that $w$-Thiele methods based on $w$-scoring functions that have a larger slope than the $w$-function of PAV are more oriented towards regressive proportionality, whereas $w$-functions that have a smaller slope are closer in spirit to the idea of degressive proportionality. This relation is symbolically visualised in \Cref{fig:countingfcts_and_proportionality}.

\citet{jaw-sko:phragmen-degressive-regressive} constructed a class of rules that generalise \phragmen{}'s  rule. Intuitively, a degressive variant of  seq-\phragmen{} is obtained by assuming that the voters who already have more representatives earn money at a slower rate than those that have fewer. Regressive proportionality is implemented by assuming that the candidates who are approved by more voters cost less than those that garnered fewer approvals.

\begin{figure}
\begin{center}
\begin{tikzpicture}[y=.6cm, x=.8cm,font=\sffamily]

\pgfdeclarelayer{bg}
\pgfsetlayers{bg,main}
    \draw (0,0) -- coordinate (x axis mid) (8.5,0);
        \draw (0,0) -- coordinate (y axis mid) (0,8.5);
        \foreach \x in {0,...,8}
             \draw (\x,1pt) -- (\x,-3pt)
            node[anchor=north] {\x};
        \foreach \y in {0,2,...,8}
             \draw (1pt,\y) -- (-3pt,\y) 
                 node[anchor=east] {\y}; 
    \node[below=0.8cm] at (x axis mid) {number of approved candidates in committee ($x$)};
    \node[rotate=90, above=0.8cm] at (y axis mid) {satisfaction $w(x)$};

    \draw plot[mark=*, mark options={fill=white}] 
        file {av.data};
    \draw plot[mark=triangle*, mark options={fill=white} ] 
        file {pav.data};
    \draw plot[mark=square*]
        file {cc.data}; 
    
    \node[right=0.2cm] at (8,0.8) {CC (diversity)};
    \node[right=0.2cm] at (8,1.9) {\footnotesize degressive proportionality};
    \node[right=0.2cm] at (8,5.5) {\footnotesize regressive proportionality};
    \node[right=0.2cm] at (8,3) {PAV (linear proportionality)};
    \node[right=0.2cm] at (8,8) {AV (individual excellence)};   

    \draw[->,>=stealth',thick] (8, 2.6) arc[radius=5.4, start angle=10, end angle=-5];     
    \draw[->,>=stealth',thick] (8, 2.9) arc[radius=14.0, start angle=-10, end angle=10];  
    
\end{tikzpicture}
\caption{A diagram illustrating the relation between defining $w$-functions of Thiele methods and the type of proportionality these Thiele rules implement.}
\label{fig:countingfcts_and_proportionality}
\end{center}
\end{figure}

\citet{FaliszewskiSST17} discuss three specific classes of rules that span the spectrum between individual excellence and diversity. They analyse these rules in the ranking-based model, that is when voters rank the candidates instead of approving some of them (see \Cref{sec:multiwinner-rank}). These classes of rules can be analogously defined for approval ballots.
\citet{bri-fal-som-tal:balanced-cc,fal-tal:balancing-cc} extend Monroe's rule so that it can implement the idea of regressive proportionality; this is also done in the ranking-based framework. It would be interesting to see whether their techniques can be successfully applied to the ABC model.

Finally, \citet{subiza2017representative} propose an axiom called $\alpha$-unanimity (parameterized with $\alpha\in[0,1]$), which can be seen as a strong diversity axiom. The authors propose a voting rule (Lexiunanimous Approval Voting) that satisfies this axiom; this rule is a refined version of CC. Thiele methods (including CC itself) do not satisfy this axiom for any $\alpha$.

\subsubsection{Quantifying Degressive and Regressive Proportionality}\label{sec:degressive_and_regressive_proportionality_quantifying}

The second approach to formally reason about degressive and regressive proportionality is quantitative in nature. \citet{lac-sko2019} define two measures---the utilitarian guarantee and the representation guarantee---that can be used to quantify how well a given rule performs in terms of individual excellence and diversity.

Recall that $\score{\av}(A, W)$ denotes the total number of approvals a given committee receives in profile $A$ and $\score{\cc}(A, W)$ denotes the number of voters who approve at least one member of $W$.

\begin{definition}[Utilitarian and Representation Guarantee~\citep{lac-sko2019}]
The utilitarian guarantee of an ABC rule $\calR$ is a function $\kappa_{\av} \colon \naturals \to [0, 1]$ that takes as input an integer~$k$, representing the committee size, and is defined as:
\begin{align*}
\kappa_{\av}(k) = \inf_{A} \frac{\min_{W \in \calR(A, k)}(\score{\av}(A, W))}{\max_{W \colon |W| = k} (\score{\av}(A, W))} \text{.}
\end{align*}
The representation guarantee of an ABC rule $\calR$ is a function $\kappa_{\cc} \colon \naturals \to [0, 1]$ defined as:
\begin{align*}
\kappa_{\cc}(k) = \inf_{A} \frac{\min_{W \in \calR(A, k)}(\score{\cc}(A, W))}{\max_{W \colon |W| = k} (\score{\cc}(A, W))} \text{.}
\end{align*}
\end{definition}

Note that the utilitarian and the representation guarantee of an ABC rule $\calR$ measure how well rule $\calR$ approximates Multi-Winner Approval Voting and the Approval Chamberlin--Courant rule, respectively. These two rules embody the principles of diversity and individual excellence (cf.\ \Cref{thm:cc_and_av_characterisation}). 

\citet{lac-sko2019} show that the utilitarian guarantee of PAV, sequential PAV, and seq-\phragmen{} is $\Theta(\nicefrac{1}{\sqrt{k}})$; their representation guarantee is $\nicefrac{1}{2} + \Theta(\nicefrac{1}{k})$. CC and seq-CC achieve a better representation guarantee (of 1 and $1 - \nicefrac{1}{e}$, respectively), but their utilitarian guarantee is only $\Theta(\nicefrac{1}{k})$.
In that sense, these three proportional rules (PAV, sequential PAV, and seq-\phragmen{}) can be viewed as a desirable compromise between the two guarantees. On the other, the authors also show that proportional rules are never an \emph{optimal} compromise.
Finally, $p$-geometric rules---the Thiele rules defined by $w_{p\text{-}\mathrm{geom}}(x) = \sum_{i = 1}^{x} \left(\nicefrac{1}{p}\right)^i$---for different values of the parameter $p$ span the whole spectrum from AV to CC. By adjusting the parameter $p$, one can obtain any desired compromise between the utilitarian and representation goals.

\citet{corr/abs-2112-05994} extend this work by considering the ``price'' of justified representation axioms:
what are the optimal utilitarian and representation guarantees when requiring justified representation (\Cref{def:jr}) or extended justified representation (\Cref{def:ejr})?
Their results show that already justified representation implies a utilitarian guarantee of no better than $\nicefrac{2}{\sqrt{k}}$; the same holds for EJR.
The consequences for the representation guarantee are less pronounced: JR does not restrict the representation guarantee (e.g., CC satisfies JR and has a representation guarantee of 1) and
EJR is compatible with a representation guarantee of $\frac{3}{4}$.

\subsubsection{An Experimental View on Degressive and Regressive Proportionality}

\citet{god-bat-sko-fal:c:2d-abc} visualised the structure of the committees produced by various ABC rules on histograms. They performed computer simulations in which the candidates and the voters were represented as points in the two-dimensional Euclidean space. Intuitively, a point corresponding to a voter or a candidate might represent their position in the spectrum of possible opinions regarding various issues. In each simulation the candidates and the voters were drawn from a given distribution, and a preference profile was constructed from the positions of the voters and the candidates. The main idea was that the voters are more likely to approve candidates whose corresponding points are closer to them, since their opinions resemble views of such candidates. Given a preference profile, a specific ABC rule was used to find a winning committee, and the points corresponding to the selected candidates were marked with red dots on the histogram of the respective rule. The experiment was repeated multiple times, and each time the dots were put on the same histogram. Thus, the density of red dots in a given area represent the probabilities that the candidates from this area are chosen for the winning committee. This idea was first proposed by \citet{2dpictures} in the context of ranking-based elections.

Such histograms give valuable insights into the nature of voting rules. We depict several of them in \Cref{fig:histograms}. In the left column of the figure, we depict distributions of the points representing the voters and the candidates: red areas correspond  to the candidates, green areas to the voters, and olive areas correspond to both. The subsequent columns depict distributions of the elected candidates for six ABC rules. These histograms already illustrate the very different natures of the considered rules.  For example, the distributions obtained for PAV and the sequential Phragm\'en's rule closely resemble distributions of the voters (which is exactly what one would expect from proportional rules), CC puts more emphasis on representing as diverse a spectrum of voters as possible, and AV selects candidates that are in the centres of the distributions---the choice that corresponds to individual excellence. The Method of Equal Shares induces histograms that are in some sense between PAV and AV. Finally, the behaviour of Minimax Approval Voting (MAV) is inconsistent with our intuitive interpretation of proportionality in the Euclidean model. 

The conclusions from this experimental exercise are to a large extent consistent with the conclusions coming from the axiomatic analysis. For a more detailed discussion we refer to the original work~\cite{god-bat-sko-fal:c:2d-abc}.  

\begin{figure}[t]
\centering
\begin{minipage}[t]{2.0cm}
\centering
\includegraphics[width=2.0cm]{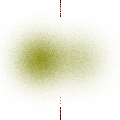}
\end{minipage}
\begin{minipage}[t]{2.0cm}
\centering
\includegraphics[width=2.0cm]{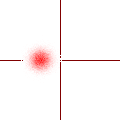}
\end{minipage}
\begin{minipage}[t]{2.0cm}
\centering
\includegraphics[width=2.0cm]{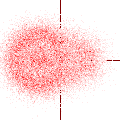}
\end{minipage}
\begin{minipage}[t]{2.0cm}
\centering
\includegraphics[width=2.0cm]{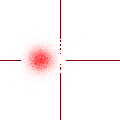}
\end{minipage}
\begin{minipage}[t]{2.0cm}
\centering
\includegraphics[width=2.0cm]{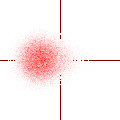}
\end{minipage}
\begin{minipage}[t]{2.0cm}
\centering
\includegraphics[width=2.0cm]{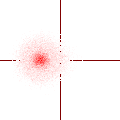}
\end{minipage}
\begin{minipage}[t]{2.0cm}
\centering
\includegraphics[width=2.1cm]{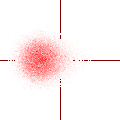}
\end{minipage}\\ \vspace{0.3cm}
\begin{minipage}[t]{2.0cm}
\centering
\includegraphics[width=2.0cm]{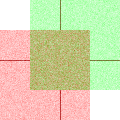}
\footnotesize distribution
\end{minipage}
\begin{minipage}[t]{2.0cm}
\centering
\includegraphics[width=2.0cm]{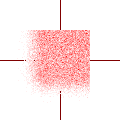}
\footnotesize AV
\end{minipage}
\begin{minipage}[t]{2.0cm}
\centering
\includegraphics[width=1.8cm]{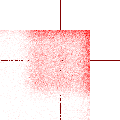}
\footnotesize CC
\end{minipage}
\begin{minipage}[t]{2.0cm}
\centering
\includegraphics[width=1.8cm]{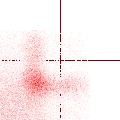}
\footnotesize MAV
\end{minipage}
\begin{minipage}[t]{2.0cm}
\centering
\includegraphics[width=2.0cm]{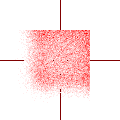}
\footnotesize PAV
\end{minipage}
\begin{minipage}[t]{2.0cm}
\centering
\includegraphics[width=2.0cm]{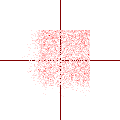}
\footnotesize Eq. Shares
\end{minipage}
\begin{minipage}[t]{2.0cm}
\centering
\includegraphics[width=2.0cm]{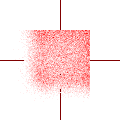}
\footnotesize Phragm\'en
\end{minipage}
\caption{Visualising the outcomes of some selected ABC rules (from~\cite{god-bat-sko-fal:c:2d-abc}).}\label{fig:histograms}
\end{figure}

\section{Proportionality and Strategic Voting}\label{sec:proportionality_and_strategyproofness}

The ABC rules that we have considered in the context of proportionality are all prone to manipulations (cf.\ \Cref{sec:strategic_voting}). In this section we explain that this is not a coincidence---achieving proportionality and strategyproofness at the same time is inherently impossible. This impossibility was first proven by \citet{pet:prop-sp,dominikthesis} for resolute rules (rules that always return a single winning committee), even for very weak formulations of the desired axioms.
(Earlier work by \citet{azi-gas-gud-mac-mat-wal:c:multiwinner-approval} and \citet{Janson16arxiv} already showed that certain proportional rules---such as PAV, seq-PAV, and seq-\phragmen---are not strategyproof.)

\begin{theorem}[\citet{pet:prop-sp,dominikthesis}]
Suppose $k \geq 3$, the number $n$ of voters is divisible by $k$, and $m \geq k + 1$. Then there exists no resolute ABC rule $\calR$ which satisfies the following three axioms:
\begin{enumerate}
\item weak proportionality: for each party-list election $(A, k)$ where some singleton ballot~$\{c\}$ appears at least $\nicefrac{n}{k}$ times ($|\{i \colon A(i) = \{c\}\}| \geq \nicefrac{n}{k}$), candidate $c$ must belong to the winning committee, i.e., $c \in \calR(A, k)$,
\item weak efficiency: a candidate who is approved by no voter may not be part of the winning committee, unless fewer than $k$ candidates receive at least one approval,
\item inclusion-strategyproofness\footnote{This axiom can be further weakened to allow voters only to manipulate by reporting subsets of their true approval sets.} (as defined in \Cref{sec:strategic_voting}).
\end{enumerate}
\end{theorem}
\citet{KVVBE20:strategyproofness} prove a similar result for irresolute rules (i.e., when rules are allowed to output multiple tied winning committees), using cardinality-strategyproofness and Pareto efficiency.
Further, \citet{duddy2014electing} proves a related impossibility result for irresolute rules using slightly different axioms; this result also requires a form of Pareto efficiency.

\citet{lac-sko:t:multiwinner-strategyproofness} showed that AV is the only ABC scoring rule (\Cref{sec:consistency}) that satisfies SD-strategyproofness; this result can also be seen as an impossibility result concerning proportionality and strategyproofness within the class of  ABC scoring rules.
Further, they quantified the trade-off between strategyproofness and proportionality. For various ABC rules they empirically measured their level of strategyproofness by assessing the fraction of profiles, for which there exists a voter who has an incentive to misreport her approval set. They concluded that rules which
are more similar to AV (i.e., rules that follow the principle
of regressive proportionality) are less manipulable than proportional rules. The rules that follow the principle of degressive proportionality are the most manipulable. A similar conclusion was obtained by \citet{bar-lan-mak:multiwinner_manipulation}, but there the authors analysed a different class of rules---namely those based on the Hamming distance, and spanning the spectrum from AV to Minimax Approval Voting.

Since in the general case, there exist no proportional strategyproof ABC rule, \citet{botan2021manipulability} restricted the analysis to three specific types of manipulations: 
\begin{inparaenum}[(1)]
\item subset manipulations, where a voter can manipulate only by submitting a subset of her true approval set,
\item superset manipulations, where each voter can only send a superset of her true preferences, and
\item disjoint manipulations, where a manipulation can be performed only by submitting a subset  of candidates disjoint from the true approval set of the voter.
\end{inparaenum}
They  showed that for party-list preference profiles (see \Cref{def:party_profiles}) all Thiele methods are cardinality-strategyproof\footnote{Formally, \citet{botan2021manipulability} defines strategyproofness for irresolute rules and states their results for the general class of G\"ardenfors preference extensions~\cite{gardenfors1979definitions}. These extensions define preference relations over sets of winning committees and thus can be applied to irresolute rules.} against subset, superset, and disjoint manipulations.

\section{Proportionality with Respect to External Attributes}\label{sec:prop_external_attributes}

In \Cref{sec:apportionment,sec:cohesive_groups,sec:laminar,sec:the_core,sec:degressive_and_regressive_proportionality,sec:proportionality_and_strategyproofness}, we have considered formal concepts that capture, in various ways, what it means that the structure of the elected committee proportionally reflects the (approval-based) preferences of the voters. In other words, we have considered proportionality with respect to the preferences given by the voters. In this section, we briefly consider a framework that approaches the concept of proportionality quite differently: we analyse proportionality with respect to external attributes of the candidates.\footnote{A noteworthy real-world example is the Lebanese Parliament, where an equal representation of Christians and Muslims (64 seats each) is mandated \cite{diss2017distribution}.}

Let us start by recalling the apportionment setting that we discussed in \Cref{sec:apportionment}. In the apportionment model we are given a set of candidates, each candidate belonging to a single political party; for each political party we are given a desired fraction of seats the party should ideally get in the elected committee (typically, this is the fraction of votes cast on the party). The goal is to pick the committee that matches the desired fractions as closely as possible. Thus, one can say that in the apportionment setting there is one external attribute, which is the party affiliation, each candidate has a certain value of this attribute, and the goal is to pick the committee where the different values of the attribute are represented proportionally to the given desired fractions. 

Now, assume that there are two attributes---each candidate has a political affiliation and a geographic region that she represents. For each value of each attribute we are given a desired fraction of seats that the candidates with this attribute value should get. This setting is called bi-apportionment, and it is discussed in detail in a book chapter by \citet{Puke14a:biApportionment} (several articles study the bi-apportionment setting from a computational perspective~\cite{LariRS14,PukelsheimRSSS12,SerafiniS12}). 
The model of bi-apportionment has been further extended to an arbitrary number of attributes by \citet{LS16:multi-attribute}.\footnote{The multi-attribute model finds its application, e.g., in the process of sortition. In sortition one needs to select a committee of ordinary people who will discuss certain controversial matters, and come up with recommendations helping the governments make decisions. In this process it is important to select a committee consisting of people who are representative for the whole society. Currently, randomised algorithms are mostly used for such selections~\cite{ben-gol-pro:stratification}. The multi-attribute model provides alternative methods that take advantage of information regarding attributes of the potential committee members.}
There, the authors analysed axiomatically and algorithmically two rules that extend the D'Hondt method and the largest remainder method to the multi-attribute apportionment.
Further work on this topic was done by \citet{cembrano2022multidimensional}.

The desired fractions in the (multi-attribute) apportionment model can be based on the voters preferences, or they might be given exogenously, e.g., by imposing certain quotas, specifying how many candidates with given attribute values should be included in the winning committee. Taking one specific interpretation, namely assuming the voters are asked to approve attribute values, \citet{kagita2021committee} proposed several other rules for selecting committees. They formulated axioms, requiring that the selected committee should consist of candidates whose attribute values proportionally represent voters' preferences. Unfortunately, none of the rules they propose satisfies any of these axioms. In general our axiomatic understanding of the multi-attribute apportionment model is still not well-advanced.

In the final part of this section we will consider a model which takes into account both the voters' preferences over candidates, and external constraints based on attributes of the candidates. Instead of defining this model formally, we provide an illustrative example.

\begin{example}\label{ex:multi-attribute-pr}
Assume we want to select a representative committee. Such a committee should be gender-balanced, containing 50\% of male~(M) and 50\% of female~(F) committee members. Additionally, the committee should represent people from different educational backgrounds: at least 25\% and at most 50\% of its members should have a bachelor's degree~(B), between 40\% and 60\% should have an upper-secondary education~(U), and between 10\% and 25\%---a primary or lower-secondary education~(P). Finally, the selected committee should contain at least 25\% young people~(Y) and at least 50\% senior people~(S).
The pool of candidates from which we can select members of such a committee is given in the table below. Additionally, seven voters express their preferences via the following approval ballots.

\begin{center}
\minipage{0.45\textwidth}
\begin{center}
\begin{tabular}{c|ccc}
Name & Gender & Education & Age \\ \hline
$c_1$ & $\mathrm{F}$ & $\mathrm{B}$ & $\mathrm{Y}$ \\
$c_2$ & $\mathrm{M}$ & $\mathrm{U}$ & $\mathrm{Y}$ \\
$c_3$ & $\mathrm{M}$ & $\mathrm{U}$ & $\mathrm{S}$ \\
$c_4$ & $\mathrm{F}$ & $\mathrm{P}$ & $\mathrm{S}$ \\
$c_5$ & $\mathrm{M}$ & $\mathrm{U}$ & $\mathrm{Y}$ \\
$c_6$ & $\mathrm{M}$ & $\mathrm{U}$ & $\mathrm{Y}$ \\
$c_7$ & $\mathrm{M}$ & $\mathrm{U}$ & $\mathrm{Y}$ \\
$c_8$ & $\mathrm{F}$ & $\mathrm{B}$ & $\mathrm{S}$ 
\end{tabular}

\end{center}

\endminipage\hfill
\minipage{0.45\textwidth}
\begin{align*}
&A(1) = \{c_1, c_2, c_3\} \\
&A(2) = \{c_3, c_5\} \\
&A(3) = \{c_7, c_8\} \\
&A(4) = \{c_3, c_4, c_5, c_7\} \\
&A(5) = \{c_1, c_8\} \\
&A(6) = \{c_6\} \\
&A(7) = \{c_1, c_2, c_6\}
\end{align*}
\endminipage
\end{center} 

Assume we want to select $k = 4$ committee members. 
The winning committee according to AV would be $W_1 = \{c_1, c_3, c_7, c_8\}$ (for simplicity, we assume the ties are broken lexicographically $c_8 \succ c_7 \succ \ldots \succ c_1$), and according to PAV, the winning committee would be $W_2 = \{c_1, c_3, c_6, c_8\}$. However, each of these two committees violates the attribute-level constraints. The committee maximising the AV-score and the PAV-score subject to these constraints would be, respectively, $W_3 = \{c_1, c_3, c_4, c_7\}$ and $W_4 = \{c_3, c_4, c_6, c_8\}$. 
\end{example}

As can be seen in \Cref{ex:multi-attribute-pr}, score-based ABC rules (in particular Thiele methods) are suitable for this approach: the winning committee is the one with the highest score that satisfies all external constraints. Following this approach, \citet{conf/aaai/BredereckFILS18} and \citet{conf/ijcai/CelisHV18}  considered the model of multi-winner elections with external constraints, but where the qualities of the committees are assessed via a general set function $f$. The function $f$ may in particular depend on the voters' ballots, for example we can set $f(W) = \score{\av}(A, W)$. \citet{azi:committees-soft-constraints} studied a  similar model, but assuming there is a global ranking over $C$ that represents the objective qualities of the candidates. There, the goal is to select the lexicographically best committee subject to the multi-attribute constraints, which are treated more softly than in case of \citet{conf/aaai/BredereckFILS18} and \citet{conf/ijcai/CelisHV18}. Let us also mention that \citet{bei-liu-keu-wan:diversity-constraints} studied a related model, but there the goal is to select the committee of maximal cardinality that satisfies the attribute-level constraints. We will consider algorithmic aspects of these and related approaches in \Cref{sec:computational-prop}.

Note that this approach is not compatible with rules that do not naturally provide a ranking of committees by scores (e.g., seq-\phragmen{} or the Method of Equal Shares). It is an interesting question how to adapt these rules to the model with external constraints.

\bibliographystyle{abbrvnat}
\bibliography{main}
\chapter{Algorithms and Computational Complexity}\label{sec:algorithms}

\intro{In this chapter, we discuss computational problems related to ABC rules and algorithms that solve these problems. We consider algorithmic techniques such as integer linear programming,
fixed-parameter algorithms, approximation algorithms, and algorithms for structured domains.
We further consider algorithmic aspects of proportionality and strategic voting.}

In this chapter, we discuss computational problems related to ABC rules and algorithms that solve these problems. We start by discussing the computational complexity of ABC rules. As many ABC rules are computationally difficult, a thorough algorithmic analysis is paramount to a practical application of these rules.
We consider algorithmic techniques such as integer linear programming,
fixed-parameter algorithms, approximation algorithms, and algorithms for structured domains.
Moreover, we discuss computational questions related to proportionality and to strategic voting.


\section{Computational Complexity}\label{sec:computational-complexity}

How computationally expensive is it to find a winning committee according to a given ABC rule? 
Clearly, this question is of major importance for the practical use of an ABC rule.
Here, we distinguish only two types of complexity: 
ABC rules that are computationally easy, i.e., computable in polynomial time, and ABC rules that are computationally expensive, i.e., those that are $\np$-hard. Note that this is only a coarse dichotomy; we discuss its implications further below.

Let us first consider the class of Thiele methods. Out of the three most prominent Thiele methods, two are $\np$-hard (CC and PAV) and one is computable in polynomial time (AV).
A polynomial-time algorithm for AV is straightforward: for each alternative~$c$ we compute its approval score $\score{\av}(A, c) = |\{i\in N: c\in A(i)\}|$ and select the $k$ alternatives with the largest scores.
To be able to claim $\np$-hardness of an ABC rule~$\calR$, we have to fix a decision problem; we choose the following for rules based on scores:
given an approval profile, is there a committee with $\calR$-score at least $s$?
The $\np$-hardness of CC has been shown by \citet{complexityProportionalRepr}; the $\np$-hardness of PAV by \citet{owaWinner} and \citet{azi-gas-gud-mac-mat-wal:c:multiwinner-approval} (for different decision problems).
A more general result shows that a large class of Thiele methods is $\np$-hard:

\begin{theorem}[{\citet[Theorem~5]{owaWinner}}]
Let $w:\naturals\to\reals$ be a non-decreasing function for which $w(i)-w(i-1)> w(i+1)-w(i)$ for some $i\in \naturals$.
Given an approval profile profile~$A$, a committee size~$k$, and a bound $s$, it is $\np$-hard to decide whether there exists a committee of size $k$ with a $w$-score of at least $s$, i.e., $\score{w}(A, W) \geq s$. 
\end{theorem}

Note that this theorem does not apply to AV, which is indeed polynomial-time computable. Interestingly, a similar result also holds for 2D-Euclidean preferences. We say that an approval profile is 2D-Euclidean if the voters and the candidates can be represented in the two-dimensional Euclidean space so that for each voter $i$ the following holds:  if $i$ approves a candidate $c$, then she also approves all candidates that are closer to $i$ than $c$.
The following theorem applies, e.g., to PAV and CC.

\begin{theorem}[\citet{god-bat-sko-fal:c:2d-abc}]
Let $w:\naturals\to\reals$ be a non-linear and concave function. Given a 2D-Euclidean approval profile profile~$A$, a committee size~$k$, and a bound $s$,  it is $\np$-hard to decide if there is a $k$-size committee with a $w$-score of at least $s$.
\end{theorem}

Winning committees of sequential and reverse sequential Thiele methods can be computed in polynomial time; this follows immediately from their definitions. The same holds for Greedy Monroe, seq-\phragmen{}, the Method of Equal Shares, and SAV. In contrast, appropriate decision problems for Monroe's rule \cite{complexityProportionalRepr}, leximax-\phragmen{} \cite{aaai/BrillFJL17-phragmen}, and MAV \cite{leg-mar-meh:approx_multiwinner} are $\np$-complete. The $\np$-hardness for MAV also holds for 2D-Euclidean preferences \cite{god-bat-sko-fal:c:2d-abc}.
These complexity results are summarised in \Cref{tab:axioms_summary}.

To conclude, the complexity classification discussed here should not be misunderstood in implying that $\np$-hard ABC rules are impractical and should be avoided. There is a wide range of algorithmic techniques available to solve $\np$-hard problems, and many disciplines in computer science encounter (and routinely solve) computationally hard problems.
Instead the message here is the following: When using a polynomial-time computable rule, even very large instances can be expected to be solved quickly. For $\np$-hard rules, a more thorough analysis is necessary to determine how large instances can be solved (cf.~\Cref{sec:howtocompute}).

\section{How to Compute Winning Committees?}\label{sec:howtocompute}

The arguably most central algorithmic question is: how to compute winning committees for an ABC voting rule? Clearly, the answer significantly differs from rule to rule.
Rules that can be computed in polynomial time generally do not require sophisticated algorithms.
In particular, algorithms for AV, SAV, as well as for sequential and reverse sequential Thiele methods follow immediately from their corresponding definitions. 
Algorithms for \phragmen's sequential rule and the Method of Equal Shares are slightly more involved but also do not require more than a careful adaption of the corresponding mathematical definitions. (Note that for seq-\phragmen{} it is more convenient to implement its discrete formulation.)
For rules that are $\np$-hard to compute, we discuss four algorithmic methods in the following: integer linear programs, fixed-parameter algorithms, approximation algorithms, and algorithms for structured domains.

\subsubsection{Integer Linear Programs (ILPs)}

The most common approach to compute $\np$-hard ABC rules is to employ integer linear program (ILP) solvers, such as Gurobi or CPLEX.
These are fast, general purpose solvers used for hard optimisation problems.
To use such a solver, one has to encode an ABC rule as a integer linear program, i.e., a system of linear inequalities constraining a linear expression that is maximised or minimised.
We will see two examples of ILPs in the following.
Several ILPs (including these two) are available in the \textsf{abcvoting} Python library~\cite{abcvoting}.

The ILP displayed in \Cref{fig:ilp-thiele} shows how PAV can be expressed in such a form. This particular ILP formulation for PAV is taken from \citet{jair/spoc}.
Two types of variables are used here: $x_{i,\ell}$ intuitively encodes that voter~$i$ approves at least $\ell$ candidates in the committee, and $y_c$ encodes that candidate~$c$ is contained in the winning committee.
\begin{figure}
\centering
\noindent\fbox{%
    \parbox{0.8 \textwidth}{%
\begin{align}
\text{\textbf{maximise} \ } &\sum_{i=1}^n \sum_{\ell=1}^k \frac 1 \ell \cdot x_{i,\ell} &\label{ilp:eq:max}\\
\text{\textbf{subject to}: \ }& x_{i,\ell} \in \{0, 1\} & \text{for }i\in[n]\text{, }\ell\in [k]\\
& y_{c} \in\{0,1\} & \text{for }c\in C\\
& \sum_{c\in C} y_{c} = k  & \label{ilp:eq:k}\\
& \sum_{\ell=1}^k x_{i,\ell} = \sum_{c\in A(i)} y_{c}& \text{for }i\in[n]\label{ilp:eq:ell}
\end{align}}\qquad}
\caption{An ILP for computing PAV}\label{fig:ilp-thiele}
\end{figure}
Given an election instance $(A,k)$, this ILP maximises the PAV-score expressed in~\eqref{ilp:eq:max}. Further it ensures that exactly $k$ candidates are selected with Equation~\eqref{ilp:eq:k} and that $x_{i,\ell}$ indeed encodes that voter~$i$ approves at least $\ell$ candidates in the committee with Equation~\eqref{ilp:eq:ell}.
Note that it can occur that $x_{i,\ell}=0$ and $x_{i,\ell+1}=1$, but this is never an optimal solution since $\frac 1 \ell > \frac{1}{\ell+1}$.
It is easy to see that this ILP can be adapted for computing other Thiele methods by adjusting the optimisation goal in~\eqref{ilp:eq:max}.
Another ILP formulation is due to \citet{owaWinner}. This ILP is applicable to a larger class of multi-winner rules (OWA rules).

As a second example of an ILP encoding, we present one for MAV in \Cref{fig:ilp-mav}.
\begin{figure}
\centering
\noindent\fbox{%
    \parbox{0.8 \textwidth}{%
\begin{align}
\text{\textbf{minimise} \ } & D &\notag\\
\text{\textbf{subject to}: \ }& d_{i,c} \in \{0, 1\} & \text{for }i\in[n]\text{, }c\in C\notag\\
& y_{c} \in\{0,1\} & \text{for }c\in C\notag\\
& \sum_{c\in C} y_{c} = k  & \notag\\
& d_{i,c} = 1-y_c & \text{for }c\in A(i)\label{ilp:eq:d1}\\
& d_{i,c} = y_c & \text{for }c\in C \setminus A(i)\label{ilp:eq:d2}\\
& \sum_{c\in C} d_{i,c} \leq D  & \label{ilp:eq:D}
\end{align}}\qquad}
\caption{An ILP for computing MAV}\label{fig:ilp-mav}
\end{figure}
Here, $y_{c}$ encodes whether candidate~$c$ is contained in the winning committee, $d_{i,c}$ encodes whether voter~$i$ disagrees with the decision of whether $c$ is in the committee or not, and $D$ is the maximum Hamming distance between a voter and the chosen committee.
Constraints~\eqref{ilp:eq:d1} and~\eqref{ilp:eq:d2} fix the value of $d_{i,c}$, i.e., 
\[d_{i,c}=\begin{cases}0 & \text{if }(c\in A(i)\text{ and }y_c=1) \text{ or } (c\notin A(i)\text{ and }y_c=0),\\1 & \text{otherwise.}\end{cases}\]
Then, $\sum_{c\in C} d_{i,c}$ is the Hamming distance between the committee defined by $y_c$ and $A(i)$. Due to Constraint~\eqref{ilp:eq:D}, these sums are $\leq D$ for all voters. Hence, by minimising $D$, we minimise the maximum distance.

Lastly, for Monroe's rule, \citet{potthoff1998proportional} discuss ILP formulations, and for \lexphrag{} an ILP is due to \citet{aaai/BrillFJL17-phragmen}.

\subsubsection{Fixed-Parameter Algorithms}
Fixed-parameter algorithms have received some attention for ABC rules.
The main idea is to identify a parameter of the problem (ideally one that is small in practice) and search for algorithms that require only polynomial time when this parameter is constant.
A fixed-parameter tractable (FPT) algorithm for a parameter $p$ is one with a runtime of $O(f(p)\cdot \mathrm{poly}(m,n))$, where $f$ is an arbitrary, typically exponential function.
Let us mention three natural parameters in the context of multi-winner elections:
the number of candidates ($m$), the committee size $k$, and the number of voters $n$.

Let us first discuss the parameter $m$, i.e., the number of candidates.
As there are ${m \choose k}\leq m^m$ committees, it is possible to consider each possible committee in an FPT algorithm. This bound gives trivial (and uninteresting) FPT results for most NP-hard rules.
For example, for $w$-Thiele methods one can compute $\score{w}(A, W)$ for each committee $W$ and pick those with maximum score.
An interesting exception is Monroe, where it is not immediately obvious how to compute the Monroe score of a given committee in polynomial time. This is achievable via a reduction to the min-cost max-flow problem, described by \citet{complexityProportionalRepr}.

For the parameter committee size $k$, most results are negative:
First, \citet{fullyProportionalRepr} show for Monroe and CC that it is $\mathrm{W[2]}$-hard to verify whether a committee exists with at least a certain Monroe-/CC-score. These hardness results continue to hold even if the number of unrepresented voters is used as an additional parameter \cite{fullyProportionalRepr}.
Second, \citet{misra2015parameterized} show an analogous W[2]-hardness result for MAV.
Third, \citet{azi-gas-gud-mac-mat-wal:c:multiwinner-approval} show for all Thiele methods with $2w(1)>w(2)$ that testing whether a committee is winning is $\mathrm{coW[1]}$-hard.\footnote{The condition $2w(1)>w(2)$ excludes AV but is satisfied for PAV and CC.}
All these results imply that one cannot hope for an FPT algorithm computing these ABC rules, i.e, it is unlikely that an algorithm exists with a runtime of, e.g., $O(2^k\cdot \mathrm{poly}(m,n))$.

The parameter $n$, the number of voters, is a natural choice if multi-winner elections are conducted in small groups and leads to interesting algorithms.
\citet{fullyProportionalRepr} show that CC and Monroe can be solved in time $n^n\cdot \mathrm{poly}(m,n)$.
In a similar vein, \citet{fal-sko-sli-tal:c:top-k-counting} show an FPT result with respect to~$n$ for a large class of multi-winner voting rules (including Thiele methods).
Their algorithm is based on mixed integer linear programming and Lenstra's result~\cite{Len83} that (mixed) integer linear programs can be solved in FPT time with the number of variables as parameter.\footnote{We refer the interested reader to \citet{gavenvciak2020integer} for a general overview of how integer linear programming can be used to find FPT algorithms.}
The results from \citet{fal-sko-sli-tal:c:top-k-counting} have been substantially generalised by
\citet{bredereck2020parameterized}, including an FPT result for Thiele methods with weighted voters.

Moreover, \citet{fullyProportionalRepr} provide a thorough and detailed parameterized complexity analysis for CC and Monroe for further parameters (e.g., the number of unrepresented voters) but find mostly hardness results.
\citet{yang2018parameterized} give an overview of further parameterized results; however, the concrete results announced in this short paper are not published yet.

To conclude, let us report on a positive result for MAV:
MAV can be computed in time $O(d^{2d})$, where $d$ is the optimal MAV-score, as shown by \citet{misra2015parameterized}.\footnote{\citet{misra2015parameterized} claimed that the runtime of their algorithm is $d^d$; this was corrected later \cite{liu2016parameterized,cygan2018approximation}.}
This runtime is essentially optimal subject to a standard complexity theoretic assumption, as shown by \citet{cygan2018approximation}.

\subsubsection{Approximation Algorithms}\label{sec:compute-approx}

The most natural approximation algorithm for Thiele methods are their sequential variants, as described in \Cref{sec:seq-rules}. Sequential $w$-Thiele provides a very good approximation of $w$-Thiele \cite{owaWinner,budgetSocialChoice}; this follows directly from a more general approximation result for submodular set functions by \citet{submodular}.

\begin{theorem}[\citet{owaWinner,budgetSocialChoice}]\label{thm:general_approximation}
Sequential $w$-Thiele is a $0.63$-approximation algorithm for $w$-Thiele. More specifically, 
Sequential $w$-Thiele achieves a $w$-score of at least $1- (1 - \nicefrac 1 k)^k \geq 1 - \nicefrac{1}{e}\geq 0.63$ times the optimal $w$-score.
\end{theorem}

\citet{DudyczMMS20} designed an algorithm that gives stronger approximation guarantees than $(1 - \nicefrac{1}{e})$ for $w$-Thiele methods for which the derivatives of the defining $w$-function decrease slower than a geometric sequence. The algorithm is based on pipage rounding of the fractional solution returned by a linear program. \citet{bar-faw-gho-gur:ell-best-approx} provided a $\left(1 - \frac{\ell^\ell}{e^\ell \cdot \ell !}\right)$-approximation algorithm for the $w$-Thiele function with $w(x) = \min(x, \ell)$. \Cref{tab:approx_thiele_summary} summaries the guarantees of the best approximation algorithms for most prominent Thiele methods.
Notably, under standard assumptions, all these guarantees cannot be improved within the class of algorithms running in polynomial time.  

{\renewcommand{\arraystretch}{1.5}%
\begin{table}[!t]
        \footnotesize
	\centering
	\makebox[\textwidth][c]{
	\begin{tabular*}{\textwidth}{l @{\extracolsep{\fill}} lll}
		\toprule
                          & $w$-function & approximation ratio & reference \\
		\midrule
		CC & $w(x) = \min(x, 1)$  & $1-\nicefrac{1}{e}$ & \citet{budgetSocialChoice} \\
		$\ell$-best & $w(x) = \min(x, \ell)$ & $1 - \frac{\ell^\ell}{e^\ell \cdot \ell !}$  & \citet{bar-faw-gho-gur:ell-best-approx}\\
		PAV & $w(x) = \sum_{i=1}^x \frac{1}{i}$  & $0.7965$ & \citet{DudyczMMS20} \\
		SLAV & $w(x) = \sum_{i=1}^x \frac{1}{2i-1}$  & $0.7394$ & \citet{DudyczMMS20} \\
		Penrose & $w(x) = \sum_{i=1}^x \frac{1}{i^2}$  & $0.7084$ & \citet{DudyczMMS20} \\
		\bottomrule
	\end{tabular*}
	}
	\caption{Guarantees of the approximation algorithms for the most prominent Thiele methods. The approximation ratios of the algorithms of \citet{budgetSocialChoice} and \citet{DudyczMMS20} are tight unless $\ptime = \np$. They are also tight for the algorithms that run in $f(k) \cdot n^{o(k)}$ time assuming the Gap Exponential Time Hypothesis (Gap-ETH). The approximation ratio of the algorithm of \citet{bar-faw-gho-gur:ell-best-approx} is tight assuming Unique Games Conjecture.}\label{tab:approx_thiele_summary}
\end{table}}

One can also find approximation algorithms for the corresponding \emph{minimisation} problem:
for $w$-Thiele, instead of maximising the $w$-score, one can equivalently minimise the difference from the theoretical optimum of $n\cdot w(k)$, i.e., to minimise 
the $w$-loss defined as 
$\loss{w}(A, W)= n\cdot w(k) - \score{w}(A, W)$. The minimisation and the maximisation variants of the problem have the same optimal solutions, but they differ in terms of approximability. If the optimal committee $W$ has a high score, i.e., if $\score{w}(A, W)$ is close to $n\cdot w(k)$, then an approximation algorithm for the minimisation variant would be superior. For instance, if for the optimal committee $W$ we have $\score{w}(A, W) = 0.95 \cdot n\cdot w(k)$, then a $2$-approximation algorithm for the minimisation variant of the problem is guaranteed to return a solution with the score at least as high as $0.9 \cdot n\cdot w(k)$. On the other hand, a $\nicefrac{1}{2}$-approximation algorithm for the maximisation variant  may return a committee with score equal to $0.475 \cdot n\cdot w(k)$. Conversely, if the the optimal committee has a significantly lower score than $n\cdot w(k)$, then a good approximation algorithm for the maximisation variant of the problem will produce better committees. 

\citet{ByrkaSS18} present a $2.36$-approximation algorithm for PAV according to this $\loss{w}$ measure. This algorithm is based on dependent rounding of a linear program solution.
It is notable that this result does not hold for arbitrary weights; in particular, such an approximation algorithm does not exist for CC under the assumption that $\ptime\neq\np$~\cite{ByrkaSS18}. 
While seq-PAV can be viewed as a voting rule in its own right, this is more debatable for such a rounding-based algorithm. In particular, it cannot be expected to satisfy nice axiomatic properties such as committee monotonicity, and thus constitutes first and foremost an approximation of PAV.

\citet{skowron2017fpt} describes two alternative algorithms that for certain Thiele methods (including PAV and CC) can provide arbitrarily good approximation guarantees and that work in FPT time for the parameter $(k, t)$, where $t$ is the upper-bound on the number of candidates each voter approves. Thus, these algorithms are practical only when the desired size of the committee $k$ and the approval sets of the voters are all small. Moreover, \citet{skowron2017fpt} shows that if each voter approves sufficiently many candidates, then Sequential $w$-Thiele provides an even better approximation guarantee than 0.63. Analogous results, but with the focus on CC, are due to \citet{sko-fal:cc-fpt-approx}.

For MAV, stronger approximation results hold: \citet{BS14:minimax_approximation} and \citet{cygan2018approximation} present polynomial-time approximation schemes (PTAS) for MAV, i.e., polynomial-time approximation algorithms that achieves arbitrary (but fixed) precision; previous work established first a 3-approximation algorithm (\citet{leg-mar-meh:approx_multiwinner}) and then a 2-approximation algorithm (\citet{caragiannis2010approximation}).

\subsubsection{Algorithms for Structured Domains}\label{subsec:Algorithms_for_Structured_Domains}

The fourth and final algorithmic technique is to consider structured preference domains.
Here, the assumption is that preference profiles possess some combinatorial structure that
gives algorithmic advantages.
We refer the interested reader to a survey by \citet{ElkindEtAlTRENDS2017} that discusses this topic more broadly.
For our purpose here, we would like to discuss only two restrictions: candidate and voter interval (defined by \citet{ijcai/ElkindL15-dichpref}, based on previous work by \citet{fal-hem-hem-rot:j:sp,dietrich2010majority,list2003possibility}), but we note that many other restrictions exist and have been studied extensively \cite{ijcai/ElkindL15-dichpref,jair/incompletesp,terzopoulou2020restricted,yang2019tree,peters2017recognising,EFLO15}.

A profile $A$ belongs to the candidate interval (CI) domain if there exists a linear order of candidates such that for each voter $i\in N$, the set $A(i)$ appears contiguously on the linear order.
Similarly, a profile $A$ belongs to the voter interval (VI) domain if there exists a linear order of voters such that for each voter $c\in C$, the set $N(c)$ appears contiguously on the linear order.
The CI domain is closely related to the single-peaked domain for arbitrary ordinal preferences and the VI domain is similar to the single-crossing domain; this is analysed in more detail by \citet{ijcai/ElkindL15-dichpref}.

Under the assumption that preferences belong either to the CI or VI domain, the computational complexity can change dramatically:
MAV is solvable in polynomial time if the given approval profile belongs either to the CI or VI domain~\cite{liu2016parameterized}.
Further, Thiele methods (\citet{jair/spoc}) and Monroe's rule (\citet{fullyProportionalRepr}) can be solved in polynomial time if the given approval profile belongs to the CI domain.
It remains an open problem whether the same holds for the VI domain.

\section{The Algorithmic Perspective on Proportionality}\label{sec:computational-prop}

In this section, we briefly review the literature that deals with the computational problem of finding a proportional committee.

\subsubsection{Finding Proportional Committees for Cohesive Groups}\label{subsec:findingpropcomms}

We first look at the proportionality concepts that formalise the behaviour of rules with respect to cohesive groups of voters; see \Cref{sec:cohesive_groups}.

Note that even the problem of deciding whether in a given instance of election there exists an $\ell$-cohesive group of voters is $\np$-complete~\cite{proprank}. Similarly, given a committee~$W$, deciding whether~$W$ satisfies the EJR condition is $\conp$-complete~\cite{justifiedRepresentation}; the same holds for the problem of deciding whether $W$ satisfies the PJR condition~\cite{AEHLSS18}. Checking if a given committee $W$ satisfies JR is computationally easy---for each candidate one needs to check whether the group of voters approving this candidate is $1$-cohesive, and if so, to check if less than $\nicefrac{n}{k}$ voters from such a group are left without a representative in~$W$. Checking whether a given committee satisfies perfect representation (\Cref{def:perfect_representation}) is also computationally easy---the problem reduces to finding a perfect constrained matching in a bipartite graph~\cite{pjr17}. 

While the problem of checking if a given committee satisfies the EJR/PJR condition is computationally hard, for a given election instance one can find in polynomial time \emph{some} committee that satisfies the two conditions (e.g., through the Method of Equal Shares~\cite{pet-sko:laminar}, or through a local-search algorithm for PAV~\cite{AEHLSS18}). The situation is quite different for perfect representation (PR): it is $\np$-complete to check whether there exists a PR committee for a given election instance~\cite{pjr17}. Consequently, unless $\mathrm{P} = \np$, there exists no polynomial-time ABC rule that satisfies perfect representation. 

\subsubsection{Finding Committees with Attribute-Level Constraints}

Next, we move to the model with external attribute-level constraints from \Cref{sec:prop_external_attributes}.

We start by considering the model from \Cref{ex:multi-attribute-pr}, where we have a set of voters with approval-based preferences over the candidates, the candidates have attribute values (the attributes can be, e.g., gender, age group, education level, etc.) and for each attribute value we are given quotas specifying upper and lower limits on the number of committee members with this particular attribute value. Two recent works by \citet{conf/aaai/BredereckFILS18} and \citet{conf/ijcai/CelisHV18} considered algorithmic aspects of the problem of finding committees maximising a certain score, subject to given attribute-level constraints. The authors considered the problem from the perspective of approximation algorithms and parameterized complexity theory, as well as variants of the problem, where the attribute-level constraints have certain special structures. We do not describe their results in detail as the specific results are obtained for the ranking-based multi-winner model (see \Cref{sec:multiwinner-rank}).
However, it is worth mentioning that even the problem of finding a committee that satisfies the attribute-level constraints is computationally hard. 
Approximation and fixed-parameter tractable algorithms for this simpler problem were studied by \citet{LS16:multi-attribute}.  

A very similar model to the one from \Cref{ex:multi-attribute-pr} is \emph{constrained approval voting (CAP)} (\citet{brams1990constrained,potthoff90}). The main difference to the previously discussed model is that CAP uses constraints formulated for combinations of attributes. For example, a constraint can have the following form: ``the proportion of young (Y) males (M) with higher education (H) in the committee should not exceed 14\%''. Specifically, \citet{brams1990constrained} and \citet{potthoff90} suggest to pick the committee that maximises the AV score subject to the aforementioned combinatorial constraints. A direct translation of CAP into an ILP problem was given by \citet{straszak1993computer}. 
In general, the setting of constrained approval voting has not been thoroughly studied in its full generality, and the model is fairly unexplored from a computational perspective.

Finally, the computational problem of finding a committee subject to attribute-level constraints is related to the multidimensional knapsack problem (the main difference is that in the multidimensional knapsack the candidates can contribute more than a unit weight to each attribute-level constraint) and to the generic problem of optimising a submodular function subject to constraints (see, e.g., a survey by \citet{submodularOverview}). However, this literature usually deals with more general types of constraints, whereas the voting literature we discussed often concerns more specific approaches.

\section{The Algorithmic Perspective on Strategic Voting}

Other types of computational problems arise when one analyses how the results of ABC elections are affected by changes in voters' preferences. There are several reasons to study this type of computational problems, and we briefly summarise them below. Historically, the first motivation was to use the computational complexity as a shield protecting elections from strategic manipulations. The reasoning was the following: if we cannot construct a good rule that is strategyproof (e.g., due to known impossibility theorems; cf.~\Cref{sec:proportionality_and_strategyproofness}), then we could at least aim at proposing a rule for which it is computationally hard for a voter to come up with a successful strategic manipulation. This motivation originated in the context of single-winner elections, and was first proposed by \citet{bar-tov-tri:j:manipulating}.
This reasoning was later contested since the analysis of computational complexity is worst-case in spirit. Even for rules for which the problem of finding a successful strategic manipulation is $\np$-hard, such manipulations can be found easily in the average case, in particular for many real-life preference profiles. For a more detailed discussion of these arguments (but with a focus on single-winner elections), we refer the reader to a survey by \citet{fal-pro:j:manipulation}, a book by \citet{meir2018strategic}, and handbook chapters by \citet{Handbook-manipulation} and \citet{Handbook-bribery}.

In addition to the original motivation to study strategic voting, 
there are other, more ``positive'' applications that do not concern insincere behaviour.
For example, the question of whether one can stop eliciting preferences and safely determine the winners of an election is equivalent to asking whether a group of (undecided) voters can still change the outcome of an election.
These questions are captured by the manipulation problems discussed in \Cref{subsec:manipulation}.
Furthermore, the problem of deciding whether the result of an election is robust to small changes in the given preference profile can also be phrased as ``bribing'' voters to change their ballots so to change the election result.
We discuss the robustness problem in \Cref{subsec:robustness}.

Before we move further, we note that for the case of selecting a single winner ($k=1$) under approval-based preferences, an excellent overview of computational issues related to strategic voting is given by \citet{bau-erd-hem-hem-rot:b:computational-apects-of-approval-voting}. 

\subsubsection{Computational Complexity of Manipulation}\label{subsec:manipulation}

We first consider the computational problem of finding a successful manipulation. 
Recall that we write $A_{+X}$ to denote the profile $A$ with one additional voter approving $X$, i.e., $A_{+X}=(A(1),\dots,A(n),X)$.

\begin{definition}\label{def:manipulation_problem}
Consider an ABC rule $\calR$. 
In the \textsc{Utility-Manipulation} problem, we are given an election instance $(A, k)$, a utility function $u \colon C \to \reals$, and a threshold value $t \in \reals$. We ask whether 
whether there exists a profile $A'$ that extends $A$ by $r$ additional voters such that $\sum_{c \in W} u(c) \geq t$ for some $W\in \calR(A_{+X}, k)$. 

In the \textsc{Subset-Manipulation} problem, we are given an election $(A, k)$, a subset of candidates $L \subseteq C$, and a positive integer~$r$. We ask whether there exists a profile $A'$ that extends $A$ by $r$ additional voters such that $L \subseteq W$ for some $W\in \calR(A', k)$.
\end{definition}

Intuitively, in \textsc{Utility-Manipulation} we have manipulators with a utility function describing their level of appreciation for different candidates; the utility function is additive. The question is whether the manipulators can submit approval ballots such that they derive a utility of at least $t$ from the elected committee. In \textsc{Subset-Manipulation}, the goal is slightly different---the manipulators want to ensure that the candidates from a given set $L$ are all selected. 
For $r = 1$, \textsc{Subset-Manipulation} can be represented as \textsc{Utility-Manipulation}: we assign the utility of one to the candidates from $L$ and the utility of zero to the other candidates, and set $t = |L|$. Observe that it makes sense to consider \textsc{Utility-Manipulation} also in the context of AV---this is because AV is strategyproof only for approval preferences, while the definition of \textsc{Utility-Manipulation} assumes the manipulators have more fine-grained preferences. 

\citet{mei-pro-ros-zoh:multiwinner_strategic} studied \textsc{Utility-Manipulation} for $r=1$ and showed that it is solvable in polynomial time for Multi-Winner Approval Voting with adversarial tie-breaking,\footnote{Adversarial tie-breaking means that ties between candidates are broken in disfavour of the manipulators.} \citet{bau-den-rey:manipulation_multiwinner} proved that also  \textsc{Subset-Manipulation} is solvable in polynomial time for~AV.
(The main focus of both papers is on ranking-based multi-winner rules, cf.~\Cref{sec:multiwinner-rank}.)
\citet{azi-gas-gud-mac-mat-wal:c:multiwinner-approval} show that \textsc{Utility-Manipulation} is computationally hard for SAV and PAV with a given tie-breaking order on candidates. They further prove that \textsc{Subset-Manipulation} is $\np$-hard for SAV and $\conp$-hard for PAV. For PAV the problem stays hard even if there is only a single manipulator ($r=1$), while for SAV with a single manipulator the problem becomes computable in polynomial time.

\citet{bre-kac-nie:strategic_multiwinner} studied a more general version of \textsc{Utility-Manipulation}, where the goal is to check whether there exists a coalition of voters that could jointly perform a successful manipulation. The authors focused on the $\ell$-Bloc rule, which is a variant of Multi-Winner Approval Voting, where each voter approves exactly $\ell$ candidates. Then, the coalition-manipulation problem is computationally hard in its all variants studied by the authors. On the other hand, if we look at an egalitarian version of $\ell$-Bloc (maximising the number of candidates in the committee that are approved by the worst-off voter), then the problem becomes computationally tractable. 
Another problem related to \textsc{Utility-Manipulation} has been considered by \citet{bar-gou-lan-mon-ries}: given utility functions of all voters, is there an approval profile consistent with the utility functions in which a given committee wins.

\subsubsection{Computational Complexity of Robustness}\label{subsec:robustness}

The next computational problem that we look at is \textsc{Robustness}, introduced by \citet{bre-fal-kac-nie-sko-tal:robustness} and adapted to the ABC setting by \citet{gaw-fal:c:robustness_of_abc_voting}. In the definition below, we consider the following three operations: the operation $\mathrm{Add}$ adds a candidate to the approval set of some voter, $\mathrm{Remove}$ deletes a candidate from the approval set of a voter, and $\mathrm{Swap}$ is a combination of  $\mathrm{Add}$ and $\mathrm{Remove}$ applied simultaneously to the approval set of a single voter.

\begin{definition}[\citet{bre-fal-kac-nie-sko-tal:robustness,gaw-fal:c:robustness_of_abc_voting}]
Consider an ABC rule $\calR$ and an operation $\mathrm{Op} \in \{\mathrm{Add}, \mathrm{Remove}, \mathrm{Swap}\}$. In the $\mathrm{Op}$-\textsc{Robustness} problem we are given an election instance $(A,k)$ and an integer $b$. We ask whether there exist a sequence $S$ of $b$ operations of type $\mathrm{Op}$ such that $\calR(A, k) \neq \calR(A', k)$, where $A'$ is the preference profile obtained from $A$ by applying the operations from sequence~$S$.
\end{definition}

\citet{gaw-fal:c:robustness_of_abc_voting} have shown that the $\mathrm{Op}$-\textsc{Robustness} problem is computationally hard for PAV and CC, for each type of the three operations. On the other hand, the problem can be solved in polynomial time for AV and SAV. The authors also computed the robustness radius---a measure that says how much the result of an election can change in response to a single change in the preference profile---for several ABC rules. Notably, they show that for $w$-Thiele  methods with $2w(1)>w(2)$ (this class includes PAV and CC), a single $\mathrm{Add}$, $\mathrm{Remove}$, or $\mathrm{Swap}$ operation can lead to a completely different winning committee. 

\citet{gaw-fal:c:robustness_of_abc_voting} and \citet{conf/sofsem/MisraS19} also considered the parameterized complexity of the \textsc{Robustness} problem, and have designed several parameterized algorithms for natural parameters, such as the number of voters $n$ and the number of candidates $m$.
\citet{fal-sko-tal:multiwinner_bribery} considered a similar problem, but they asked whether, through a sequence of operations of a given type, one can make a particular candidate a member of a winning committee.
This question is particularly relevant if one wants to report to non-winners how close they were to being selected.
Finally, robustness of ABC rules has also been studied by \citet{caragiannis2022evaluating}; their analysis is based on a noise model assuming a ``ground truth'' (i.e., optimal) committee.

%
%

\bibliographystyle{abbrvnat}
\bibliography{main}
\chapter{Related Formalisms and Applications}\label{sec:related_formalisms}

\intro{In this chapter, we discuss connections of approval-based committee voting with a number of other applications and formalisms. These include other multi-winner voting formalisms, participatory budgeting, voting in combinatorial domains, and judgement aggregation.}

In this chapter, we discuss connections of approval-based committee voting with a number of other applications and formalisms.

\section{Ranking-Based Multi-Winner Elections}\label{sec:multiwinner-rank}

Besides ABC voting, the other classic multi-winner election model is when voters provide a ranking of candidates from the most to the least preferred one.
That is, in the ranking-based model a voter's preference is expressed as a linear order of all candidates instead of a subset of candidates, as it is the case in the ABC model.
As it is the case with approval-based multi-winner elections, also the ranking-based model has attracted 
much attention in recent years. Alas, at the point of writing this book, there does not exist a 
comprehensive overview of this field of research. 
However, a very helpful introduction to multi-winner voting in general (with a focus on the ranking-based model) can found in a book chapter by \citet{FSST-trends}.

When comparing approval-based and ranking-based multi-winner rules, 
it is worth mentioning that the class of ABC scoring rules (\Cref{def:abc-scoring-rules}) has a very close analogue in the ranking-based model, namely the class of committee scoring rules~\cite{elk-fal-sko-sli:c:multiwinner-rules}. 
Indeed, committee scoring rules admit a very similar axiomatic characterisation to the one given in \Cref{thm:characterizationWelfareFunctions} for ABC scoring rules~\cite{skowron2019axiomatic}.
The class of committee scoring rules has been explored in depth by \citet{fal-sko-sli-tal-tal:j:hierarchy-committee}. In particular, the subclass of OWA-based committee scoring rules corresponds to the class of Thiele methods in the approval-based model. Other subclasses of committee scoring rules can be analogously defined for approval ballots, but to the best of our knowledge they have not been considered in the context of approval-based elections.

The approval-based and ranking-based model can be generalised to the model where voters provide weak orders over candidates, i.e., ranking with ties.
In this model, approval ballots correspond to a ranking with two levels (approved and disapproved candidates).
This variant has been considered, e.g., by \citet{aziz2020expanding}, but generally attracted much less attention so far.
This is due to the fact that the concepts discussed in this book (e.g., notions of proportionality) do not easily generalise to this more expressive setting and require substantial conceptual developments. Further work is required to consolidate the literature from the approval-based and ranking-based model in a systematic and notationally concise form.

\section{Trichotomous Preferences and Incomplete Information}\label{sec:multiwinner-trich}

In this book we consider the variant of the multi-winner election model where agents vote by specifying sets of approved candidates. Several recent (mostly algorithmic) works study an extended variant of this model, where the ballots are trichotomous, i.e., where each voter can approve, disapprove or remain neutral with regard to a candidate.
This model is discussed in detail by \citet{brams1978approval} and \citet{lines1986approval}.
\citet{bau-den:mav_trichotomous} and \citet{bau-den-rey:manipulation_multiwinner} generalise AV and MAV to trichotomous votes and explore related algorithmic questions.
This line of work has been continued by \citet{liu2016parameterized}.
Further, \citet{bau-boh-rey-sch-sel:minimax_ell_blocks} extend MAV to the case where each voter assigns each candidate to one of $\ell$ predefined buckets, where $\ell$ is a parameter. 
\citet{zhou2019parameterized} introduce variants of CC, PAV, and SAV for trichotomous ballots and study questions regarding parameterized complexity.
Finally, \citet{talmon2021proportionality} define and study notions of proportionality in the trichotomous setting.
In general, many questions regarding the trichotomous model remain unanswered. In particular, an axiomatic analysis is mostly missing (with work of \citet{alcantud2014dis} and \citet{gonzalez2019dilemma} as notable exceptions).

A model closely related to trichotomous preferences arises if approval ballots are incomplete due to missing information. In this model, the middle, ``neutral'' option corresponds to ``unknown''.
In practice, voting rules often have to be computed given incomplete information (such as missing ballots or incomplete ballots; see the handbook chapter of \citet{brandt2015handbook-chapter10} for a broader discussion). 
For ABC rules, a first analysis with focus on AV is due to \citet{bar-gou-lan-mon-ries}.
A more comprehensive treatment by \citet{imber2021committee} considers the class of Thiele methods and focuses on computational problems related to incomplete information.
Apart from the three-valued model of incomplete information, as discussed here, they also propose models where ``unknown'' candidates are ordered by preference but it is unclear where to separate them in approved and disapproved candidates.
Finally, \citet{terzopoulou2020restricted} study structured preference domains (cf.\ \Cref{subsec:Algorithms_for_Structured_Domains}) in connection with incomplete information.

\section{A Variable Number of Winners}

Throughout this paper, we assume that the committee size is fixed. In the literature on 
multi-winner voting with a variable number of winners~\cite{kil-handbook,Kilgour16} (also known as \emph{social dichotomy functions} \cite{duddy2014social}),
this assumption is dropped and a voting rule can return an arbitrary number of candidates---depending on the given election instance.
An example for such a rule, based on approval ballots, is the \emph{mean rule}, which returns 
all candidates with an above-average number of approvals (introduced by \citet{duddy2016}, further analysed by \citet{brandl2019axiomatic}).
Another example is Minimax Approval Voting (MAV), as discussed in \Cref{sec:furtherabc}.
In this setting, MAV returns all candidate subsets that minimise the largest Hamming distance among all voters. Other ABC rules do not easily translate to this setting. For example, Thiele methods always achieve a maximum score for the complete set of all alternatives.
Consequently, the formulation of such voting rules often contains a penalty mechanism for larger sets.

More details, in particular a computational view point and an experimental evaluation, can be found in the work of \citet{Faliszewski17OA}.
Further, the special case of shortlisting rules has been analysed by \citet{approvalbased-shortlisting}; this work includes recommendations which voting rules are 
particularly suitable for shortlisting scenarios.
Shortlisting in a proportional fashion was studied by \citet{prop-mwwavnow}; their focus lies on proportionality guarantees (related to the ones introduced in \Cref{sec:cohesive_groups}) for variable-sized sets of candidates.
Finally, \citet{DBLP:journals/corr/abs-2201-06655} consider an epistemic scenario where a ``correct'' selection of candidates has to be identified; approval ballots are viewed as noisy estimates of a ground truth.

\section{Participatory Budgeting}\label{sec:part-budg}

In participatory budgeting (PB), we assume that candidates come with different costs, and that the sum of the costs of the selected candidates cannot exceed a given budget. Thus, multi-winner elections can be viewed as a special case of PB, where the costs of the candidates are all equal.
Typically, candidates correspond to projects in this setting, each of which has an associated cost to be implemented. 
For an overview of different models and approaches to PB, we refer the reader to a recent survey by~\citet{aziz2020participatory}.  

Participatory budgeting based on approval ballots is one of the standard models and is often used in real-world PB referenda. 
Knapsack voting suggested by~\citet{goel2015knapsack} closely resembles AV. 
\citet{pet-pie-sko:c:participatory-budgeting-cardinal} showed that the Method of Equal Shares preserves its proportionality properties in the setting of PB---it satisfies an adapted version of EJR, and a logarithmic approximation of the core.
\citet{aziz2018proportionally} provide a taxonomy of axioms aimed at formalising proportionality in PB; those axioms are adaptations of JR and PJR (see \Cref{sec:cohesive_groups}). \citet{talmon2019framework} study other axioms, mostly pertaining to different forms of monotonicity (see \Cref{sec:comm-mon,sec:cand_supp_monotonicity}) and through experiments provide visualisations of the kind of committees returned by different participatory budgeting rules.
\citet{BaumeisterBH21} consider the computational complexity of strategic voting.
Generally, the assumption is that projects are independent of each other; \citet{projectinteractions} study participatory budgeting without this assumption.
Finally, \citet{ReyEH20} connect participatory budgeting based on approval ballots with judgement aggregation (see \Cref{sec:judgment}), which offers another possibility to include constraints.

\section{Budget Division and Probabilistic Social Choice}

The goal of a probabilistic social choice function is to divide a single unit of a global resource between the candidates. Thus, multi-winner elections can be viewed as instances of probabilistic social choice with the additional requirement that each candidate gets either $\nicefrac{1}{k}$-th fraction of the global resource, or nothing. For an overview of results on probabilistic social choice functions, we refer to a book chapter by \citet{Bran17a}.

Several works~\cite{BMS05a,Dudd15a,FGM16a,ABM19:fair-mixing,MPS20:fair-mixing,BBP+19a} study probabilistic social choice functions for approval votes. The particular focus of some of these works is put on formalising the concepts of fairness and proportionality. Some of these concepts closely resemble the ones that we discussed in the context of approval-based multi-winner elections (\Cref{sec:proportionality}). For example, \citet{ABM19:fair-mixing} and \citet{FGM16a} study the concept of the core (\Cref{sec:the_core}), \citet{ABM19:fair-mixing} additionally consider the axioms of average fair share, group fair share, and individual fair share---the properties that closely resemble---respectively---proportionality degree, PJR, and JR (\Cref{sec:cohesive_groups}), \citet{MPS20:fair-mixing} show the relation between these fairness properties and the utilitarian welfare of outcomes (cf.~\Cref{sec:degressive_and_regressive_proportionality_quantifying}). \citet{BMS05a} focuses on mechanisms which are strategyproof, and \citet{Dudd15a} proves that strategyproofness is incompatible with certain forms of proportionality---an impossibility result similar to the ones that we discuss in \Cref{sec:proportionality_and_strategyproofness}.  

\section{Voting in Combinatorial Domains}\label{sec:combinatorial_voting}

Multi-winner rules output fixed-size subsets of available candidates. An alternative way of thinking of such rules is that
\begin{inparaenum}[(1)]
\item for each candidate $c$ they make a decision whether $c$ should be selected to the winning committee or not, and 
\item there is a constraint which specifies that exactly $k$ decisions must be positive. 
\end{inparaenum}
Thus, with $m$ candidates there are $m$ dependent binary decisions (each decision is of the form ``include a candidate in the winning committee or not'') that are made by a multi-winner rule. These decisions are dependent (related) because of the constraint on the number of positive decisions. 

The literature on voting in combinatorial domains studies a more general setup, where a number of decisions (not necessarily binary) need to be made, and where there exist (possibly complex) relations between the decisions. Similarly, the preferences of the voters might have complex forms. For example, consider two issues---$I_1$ with two possible decisions $Y_1$ and $N_1$, and $I_2$ with two possible decisions $Y_2$ and $N_2$. A voter might prefer decision $Y_2$ only if the decision with respect to issue $I_1$ is $Y_1$; otherwise this voter might prefer $N_2$ over $Y_2$ (see the work of \citet{BKZ:voting-on-referenda} for a detailed discussion of this example). Various languages have been studied that allow voters to express such complex combinatorial preferences. For example, in the context of approval-based multi-winner elections, some of these languages would allow voters to express the view that a certain group of candidates works particularly well together, so they should either be  all selected as members of the winning committee or none of them should be chosen, or the view that some candidates should never be chosen together. In the literature on multi-winner elections, on the other hand, it is assumed that the preferences of the voters are separable, thus the voters can only make statements about their levels of appreciation for different candidates.    
An interesting middle ground between very general forms of combinatorial preferences and simple (i.e., separable) preferences was proposed by \citet{BarrotL16}: conditional approval ballots allow voters to specify their approval ballots conditional on whether certain candidates are to be included in the committee.

A comprehensive overview of the literature on voting in combinatorial domains can be found in a book chapter by \citet{Lang2016VotinginCombinatorial}. We highlight three works from this literature that deal with models particularly related to the model of approval-based multi-winner elections. In public decision making, as studied by \citet{conitzer2017fair}, the decisions are not related, the preferences of the voters with respect to decisions on various issues are separable, thus the model closely resembles the one studied in this book. The main difference is that in the model for public decisions there is no constraint specifying the number of decisions that can be positive. There, the authors focus on designing fair (i.e., proportional) rules. The model of sub-committee elections, due to \citet{aziz2018sub}, generalises the ones of multi-winner elections and public decisions. There, it is assumed that the set of candidates is partitioned and for each group of candidates there is a threshold bounding the number of candidates selected from this group.
 
Another formalism closely related to ABC voting is \emph{perpetual voting}, introduced by \citet{aaai/perpetual}. Here, instead of a committee we have time steps and in each step one candidate is selected. Hence, after $k$ rounds $k$ candidates are picked, which can be viewed as a committee. The main difference is that the set of available candidates and voters' preferences can change each round. The goal is to provide proportionality over time, which requires that the decision in round $k$ is made under consideration of the voters' satisfaction in previous rounds. This formalism can be viewed as a special case of voting in combinatorial domains (with a very  specific sequentiality constraint). Further, due to the sequential structure imposed by time, perpetual voting rules have close connections with committee monotonic ABC rules (such as seq-\phragmen{} and seq-PAV). Similar questions in a utility-based model have been studied by \citet{Freeman2017Fair}.
A voting rule related to the setting of perpetual voting is due to Gottlob Frege\footnote{Gottlob Frege (1848--1925) was a German philosopher and logician.} \citep{frege:1918a,frege:2000a}. The main difference is that the set of candidates remains the same in each round and the goal is to achieve a proportionally fair outcome for candidates (instead of voters). An analysis of this voting system is due to \citet{hll-frege}.
 
\section{Judgment Aggregation and Propositional Belief Merging}\label{sec:judgment}

In judgment aggregation, we are given a set of logical propositions and a set of voters providing true/false valuations for these propositions; the goal is to find a collective, aggregated valuation. Sometimes it is also required that the collective valuation must be consistent with exogenous logical constraints. Multi-winner elections can be represented as instances of judgment aggregation, where for each candidate we have a single Boolean variable representing whether the candidate is elected or not; the exogenous constraints can be used to enforce that exactly $k$ from these variables are set true. A chapter by \citet{Handbook-JA} in the \emph{Handbook of Computational Social Choice}  discusses this framework in detail and reviews judgment aggregation rules; see also the survey by \citet{list2009judgment}.

Propositional belief merging \cite{KoniecznyM04,KoniecznyP02,KoniecznyP11} is a very general framework, which allows agents to aggregate their individual positions (beliefs, preferences, judgements, goals) on a set of issues.
Also here this combined, collective outcome has to satisfy given exogenous logical constraints.
Approval-based committee voting can be seen as a special case of propositional belief merging, although the focus of these two directions of research has little overlap: belief merging operators are analysed with respect to a set of postulates that are only partially relevant in a voting context.
A few works have made an explicit effort to connect voting and belief merging.
A particular focus in this regard has been the study of belief merging and strategyproofness \cite{chopra2006social,everaere2007strategy,haret2019manipulating}.
Further, 
\citet{haret2016beyond} consider classic axioms from social choice theory in the context of belief merging.
Finally, \citet{aaai/propbm} introduce and analysed \emph{proportional} belief merging operators.

\section{Proportional Rankings}

The theory of multi-winner elections can be applied in a seemingly unrelated setting, where the goal is to find a ranking of candidates based on voters' preferences. One can observe that every committee monotonic (\Cref{def:comm-mon}), resolute ABC rule $\calR$ can be used to obtain a ranking of candidates: we put in the first position in the ranking the candidate that $\calR$ returns for $k=1$; call this candidate $c$. Committee monotonicity guarantees that the set of two candidates returned by $\calR$ for $k=2$ contains $c$; the other candidate is put in the second position in the ranking, etc.

In particular, if we use a proportional committee-monotonic rule (for example, seq-\phragmen{} or seq-PAV) then the obtained ranking will proportionally reflect the views of the voters in the sense that each prefix of such a ranking, viewed as a committee, will be proportional; this idea has been studied in detail by~\citet{proprank}. Proportional rankings are desirable, e.g., when one wants to provide a list of recommendations or search results that accommodate different types of users (cf.\ \emph{diversifying search results} \cite{drosou2010search,journals/ftir/SantosMO15}), or in the context of liquid democracy \cite{BKNS14a}, where an ordered list of proposals is presented to voters for their consideration.

Proportional rankings in a dynamic setting, where the rankings also take previously selected (and now unavailable) alternatives into account, have been studied by \citet{israel2021dynamic}.
This setting arises, e.g., in dynamic Q\&A platforms, where questions are proposed and upvoted.
The authors argue that questions that already have been asked should be taken into account when choosing the next question(s).

\bibliographystyle{abbrvnat}
\bibliography{main}
\chapter{Outlook and Research Directions}\label{sec:outlook}

\intro{In this chapter, we provide a list of what we view as particularly important open problems and research directions.
For instance, axiomatic characterisations of many ABC rules are missing, the compatibility of committee monotonicity and proportionality is not known, and many questions regarding the core property remain unanswered. 
This is followed by a list of more specific or more technical open questions, e.g., regarding particular axiomatic properties of an ABC rule, its computational complexity, and algorithmic challenges.
}

We conclude this book with a list of what we view as particularly important open problems and research directions. This is followed by a list of more specific or more technical open questions.
These two lists are naturally far from being exhaustive; many more research directions remain to be explored.

\section{Main Open Problems and Research Questions}

\begin{enumerate}[label=Q\arabic*]

\item \textbf{Axiomatic characterisations:} So far, only few axiomatic characterisations of ABC rules are known. Specifically, such characterisations are known only for ABC scoring rules and Thiele methods. Yet, axiomatic characterisations are essential if one wants to choose an ABC rule in a principled way. It is thus one of the major open problems to characterise other ABC rules, in particular, sequential Thiele methods, seq-\phragmen{}, the Method of Equal Shares, Monroe's rule, Minimax Approval Voting, and Satisfaction Approval Voting.
Further, almost no satisfiable proportionality-related axioms are known for the multi-attribute model (\Cref{sec:prop_external_attributes}), let alone axiomatic characterisations. 

\item \textbf{Committee monotonicity and proportionality:} The current state of research suggests that committee monotonic ABC rules are limited in how proportional they are, but there is no precise impossibility result known as of now. The main open question is whether there exist ABC rules that satisfy EJR and committee monotonicity. Only partial answers are known to this question. For example, it is known that such a rule can be defined for approval-based party-list elections (see the work of \citet{BGPSW19,brill2022approval}; mentioned in \Cref{sec:the_core}), but there is no clear generalisation of this rule to the setting of ABC rules. In case such a rule does not exist, it might be easier to first show that committee monotonicity and the core property are incompatible.

\item \textbf{The core property:} Does there exist an ABC rule that satisfies the core property (\Cref{def:core})? Equivalently, is the core always non-empty? In case the core can be empty, what is a sensible ABC rule that outputs a committee in the core whenever it exists?
Can such a rule be computed in polynomial time?

\item \textbf{Analysis beyond the worst-case:} With a few notable exceptions, in \Cref{sec:basic} and \Cref{sec:proportionality} we discussed axiomatic properties which are worst-case in spirit. A voting rule fails such an axiom even if there exist only few very unnatural election instances for which the property is not satisfied. An alternative approach would be to test if the properties hold for randomly generated instance of elections, or for elections from datasets containing real-life instance~\cite{conf/aldt/MatteiW13}. However, many common distributions of voters' preferences are too simplistic and do not capture the complexity of the voters' reasoning processes; the real election instances are rather scarce, and are collected in specific contexts, e.g., assuming that the voters' know the election rule that will be used to select winners. It is an important task to develop intermediate approaches that allow for a more fine-grained analysis and allow to understand which of the rules exhibit most desired properties on election instances that are likely to occur in practice.   

\item \textbf{Relation between axiomatic properties and computability:} It is still unclear which combinations of axiomatic properties of ABC rules can be achieved in polynomial time. It is known that some rules are $\np$-hard to compute, but it is unclear which axiomatic properties of these rules cause computational hardness. For example, it is not known whether the axiom of FJR (see \Cref{def:fjr}) is satisfiable by a rule computable in polynomial time. Further, is there a polynomial-time computable ABC rule that is proportional (e.g., that satisfies PJR) and satisfies Pareto optimality? Or does there exist a polynomial-time rule that satisfies consistency and extends D'Hondt? (By \Cref{thm:pav_characterisation}, such a rule must violate either neutrality, anonymity, or continuity.)

\item \textbf{Preference data from distributions:} An important challenge is to prepare a representative database containing sample approval-based elections. Realistic probability distributions would allow for the automatic generation of synthetic (but meaningful) election instances, which are important for numerical simulations and performance tests of algorithms. In comparison to the ranking-based model, much fewer statistical models for generating approval-based elections are know. Further, it would be highly desirable to identify a set of distributions that are representative and that cover numerous potential types of voters and voting scenarios. A noteworthy attempt at creating such a representative collection of distributions has been made for the ranking-based model by~\citet{szufa2020drawing}. For ABC elections, this issue remains to be explored.

\end{enumerate}

\section{Further Open Problems}

We continue with more specific or more technical open problems.

\begin{enumerate}[label=Q\arabic*,resume]

\item The key feature of Monroe's rule is its underlying assumption that a committee member
can represent only $\nicefrac 1 k$-th of the voter population. 
Monroe's rule could thus be generalised to many optimisation-based multi-winner rules by imposing 
the additional restriction that committee members
can represent (i.e., derive score from) an $\alpha$-fraction of voters.
This idea resembles the group activity selection problem, where a set of
activities is chosen, each of which has a maximum number of participants, and agents are assigned to activities subject to their preferences; see the survey of \citet{DarmannLangTRENDS2017}.
More generally, adding this ``Monroe-style'' constraint
can  be seen as requiring a homogeneous representation load among chosen committee members.
This is a sensible assumption whenever candidates can satisfy only a limited number of voters
(e.g., if candidates represent consumable goods).
This idea of committees with homogeneous representation loads is largely unexplored.

\item Most axiomatic notions for proportionality are only applicable to ABC rules that extend apportionment methods satisfying lower quota (see \Cref{fig:relation_between_proportionality_concepts}). This excludes, e.g., ABC rules that extend the Sainte-Lagu\"e method. As the Sainte-Lagu\"e method is in certain aspects superior to the D'Hondt method (\citet{BaYo82a} discuss this in detail), it would be desirable to have notions of proportionality that are agnostic to the underlying apportionment method. 

\item What is the proportionality degree of rev-seq-PAV?

\item Does there exist an ABC rule that satisfies priceability and Pareto efficiency?

\item What is the computational complexity of verifying whether a given committee belongs to the core? Is it possible to find a committee in the core in polynomial time (if it exists)? In case of computational hardness, can the methods presented in \Cref{sec:algorithms} be used to obtain 
algorithms that are fast in practice?\footnote{In a very recent preprint, \citet{brill2022approval} show that it 
is coNP-complete to verify whether a committee is in the core. Note that this does not rule
out the the existence of a polynomial-time algorithm \emph{finding} a committee in the core, as it is the case for EJR and PJR (cf.\ \Cref{subsec:findingpropcomms}).}

\item We have seen in \Cref{sec:proportionality_and_strategyproofness} that proportionality and strategyproofness are typically incompatible. The corresponding impossibility result for arbitrary, i.e., irresolute, ABC rules \cite{KVVBE20:strategyproofness} relies on Pareto efficiency. Since this is a property that many sensible ABC rules do not satisfy (see \Cref{sec:pareto}) it would be desirable to strengthen this result by relaxing this condition, e.g., by replacing Pareto efficiency with weak efficiency. Is this possible or are there ABC rules that are irresolute, strategyproof, proportional, but not Pareto efficient?
Furthermore, both the result for irresolute~\cite{KVVBE20:strategyproofness} rules and for resolute rules~\citep{pet:prop-sp,dominikthesis} rest on the assumption that the committee size~$k$ divides the number of voters. This assumption is unlikely to hold for large~$k$ and thus removing this assumption would be desirable.

\item A question related to monotonicity was asked by \citet{sanchez2019monotonicity}: Is there an ABC rules that is proportional (even in a very weak sense, e.g., satisfying JR) and satisfies support monotonicity without additional voters (\Cref{def:monotonicity})? As of now, AV and SAV are the only rules known to satisfy this property and both are not proportional.

\item Another question related to monotonicity concerns the Method of Equal Shares: while this method exhibits very strong proportionality guarantees (in particular EJR and priceability), it fails candidate monotonicity with additional voters (as discussed in \Cref{sec:cand_supp_monotonicity}).
Is there an equally proportional ABC rule that also satisfies candidate monotonicity?

\item We mentioned in \Cref{sec:anon-neutr-resol} that ABC rules that require tiebreaking do not satisfy neutrality (e.g., sequential and reverse sequential Thiele methods, Greedy Monroe, seq-\phragmen{}, and the Method of Equal Shares are not neutral). These rules can be made neutral with \emph{parallel universes tiebreaking}: a committee is winning under the neutral variant if and only if it is winning for \emph{some} tiebreaking order under the original rule.
Parallel universes tiebreaking has been analysed for single-winner rules \cite{con-rog-xia:c:mle,freeman2015general,bri-fis:c:ranked-pairs} but not for multi-winner rules.
Such a modification will have an algorithmic impact (trying all permutations of candidates would require exponential time), but the exact computational complexity of these neutral rules is not settled. Further, under which conditions can these rules be computed in polynomial time?

\item In \Cref{sec:computational-complexity}, we presented a coarse analysis of the computational complexity of ABC rules. This analysis could be refined by considering the \textsc{Candidate Winner} problem: given an election instance $(A,k)$ and a candidate $c$, does there exist a winning committee $W$ that contains $c$? This problem has recently be shown to be $\Theta_2^p$-complete for Monroe and CC by \citet{SonarDM20}. A similar analysis for other computationally hard voting rules (such as PAV) is missing.

\item Sequential PAV approximates the optimal PAV-score by a factor of at least $1-\frac 1 e$. What is the factor for Reverse Sequential PAV? Is it better? The same question can be asked for other Thiele methods.

\item Several approximation algorithms and heuristics have been proposed for PAV, including seq-PAV, rev-seq-PAV, the approximation algorithm based on dependent rounding (\cite{ByrkaSS18}, discussed in \Cref{sec:compute-approx}), and a local-search algorithm used for finding EJR committees in polynomial time \cite{AEHLSS18}. The difference between these algorithms has not been investigated from a practical point of view. The main question is which of these algorithms should be chosen to approximate PAV given a very large election?

\item Is it possible to compute Thiele methods and Monroe's rule in polynomial time if the given preference profile belongs to the voter interval (VI) domain (see \Cref{subsec:Algorithms_for_Structured_Domains})?

\item The computation of some polynomial-time ABC rules can clearly be parallelised. For example, for AV each candidate can be processed independently of others. The framework of $\ptime$-completeness~\cite{greenlaw1995limits} can be used to determine which ABC rules are inherently sequential (by showing P-completeness) and which can be parallelised (by showing, e.g., $\mathrm{NL}$-containment).
Such work has been done for single-winner rules~\cite{ijcai/CsarLP-schulze,aaai/CsarLPS17-mapreduce,brandt2009computational} but not for multi-winner rules. 

\item In real-life elections, it is sometimes required that each voter can approve at most $k$ candidates. It is interesting to see what are the consequences of such a requirement in terms of qualities of the committees produced by various rules. Sometimes, it is even possible to distribute up to $k$ points to candidates, i.e., to approve candidates more than once. This is clearly beyond the ABC model, but some concepts and results may transfer to such voting systems.

\end{enumerate}

\bibliographystyle{abbrvnat}
\bibliography{main}
\appendix

\chapter{Additional Proofs}\label{app:proofs}

\intro{In this appendix chapter, we provide some proofs and counterexamples that we were not able to find in the published literature.}

In this appendix chapter, we provide some proofs and counterexamples that we were not able to find in the published literature.
By default, we use alphabetic tiebreaking for ABC rules that require a tiebreaking order among candidates.

\section{Additional Proofs from Chapter~3}

\begin{proposition}
All Thiele methods with strictly increasing $w$-function as well as SAV satisfy strong Pareto efficiency; CC and MAV fail strong Pareto efficiency.\label{prop:av-pav-sav-pareto}
\end{proposition}
\begin{proof}
Observe that if $W_1$ dominates $W_2$ then the $w$-score of $W_1$ is strictly larger than that of $W_2$, due to our assumption that $w$ is strictly increasing. Thus, $W_2$ is not a winning committee for these ABC rules. The same argument holds for SAV.

To see that CC fails strong Pareto efficiency, consider
consider the approval profile
	\begin{align*}
		&1 \times \{ a, c, d \} & & 1 \times \{ b,c, d\}\text{.}
	\end{align*}
For $k=2$, $\{a,b\}$ is a winning committee even though it is dominated by $\{c,d\}$.

To see that MAV fails strong Pareto efficiency, consider
consider the approval profile
	\begin{align*}
		&1 \times \{ a, c\} & & 1 \times \{ b,c\} & & 1 \times \{ d, e\}\text{.}
	\end{align*}
For $k=1$, there is always one voter with Hamming distance $3$ to any size-$1$ committee. Consequently, all size-$1$ committees are winning even though $\{c\}$ dominates $\{a\}$ and $\{b\}$.
\end{proof}

\begin{proposition}
CC, PAV, Monroe, Greedy Monroe, \lexphrag{}, the Method of Equal Shares, and MAV do not satisfy committee monotonicity.\label{prop:comm-mon}
\end{proposition}
\begin{proof}
All counterexamples are implemented (and verified) in the \textsf{abcvoting} library~\cite{abcvoting}.

First, let us consider the approval profile
	\begin{align*}
		&2 \times \{ a \} & & 3 \times \{ a,c\} && 3 \times \{ b,c \} & & 2 \times \{b\}\text{,}
	\end{align*}
CC, PAV, Monroe, \lexphrag{}, and MAV choose $\{c\}$ for $k=1$ and $\{a,b\}$ for $k=2$.

For Greedy Monroe, consider the approval profile $A$ defined as
\begin{align*}
& A(1)=\dots=A(6)=\{a\}, & & A(7)=\dots=A(10)=\{a,c\}, & & A(11)=A(12)=\{a,b,c\}\\  
& A(13)=A(14)=\{a\}, & & A(15)=\{a, d\}, && A(16)=\dots=A(18)=\{b, d\}\text{.}
\end{align*}
We assume that Greedy Monroe breaks ties between candidates in alphabetic order and between voters in increasing order.
For $k=2$ groups have a size of 9,  for $k=3$ groups have a size of 6.
Now, for $k=2$, Greedy Monroe first chooses $a$ and assigns voters 1--9 and then candidate $b$ assigning voters $\{11,12,16,17,18\}$.
For $k=3$, Greedy Monroe first chooses $a$ and assigns voters 1--6, then candidate $c$ assigning voters 7--12, and finally candidate $d$ assigning voters 15--18.
We see that $\{a,b\}$ is not a subset of $\{a,c,d\}$.

For the Method of Equal Shares, consider
	\begin{align*}
		&A(1) : \{a, d, e\} && A(2): \{a, c\} && A(3): \{b, e\} && A(4): \{c, d, f\}\text{.}
	\end{align*}
For $k=3$, the budget of voters is $0.75$. Candidate $a$ is selected in the first round (due to alphabetic tiebreaking), reducing the budget of voters 1 and 2 to $0.25$. Then candidate $c$ is added (again by tie-breaking); the budget of voter 2 and 4 is decreased to $0$. Only voters 1 and 3 have budget left. Candidate $e$ is chosen last as the only remaining candidate with sufficient support. We see that the Method of Equal Shares selects the committee $\{a, c, e\}$.

For $k=4$, the budget of voters is $1$. In the first three rounds, candidates $a$, $c$, $d$, and $e$ can all be chosen by two voters paying $0.5$. By alphabetic tie-breaking, the Method of Equal Shares chooses $a, c, d$. In the fourth round, the remaining budgets are $0, 0, 1, 0$ for voters 1--4, respectively.
Thus, in the last round, candidate~$b$ is chosen.

We see that the Method of Equal Shares selects $\{a,c, e\}$ and $\{a,b, c,d\}$ and is thus not committee monotone. Note that this example does not use the second phase of the Method of Equal Shares (based on seq-\phragmen) and thus works independently of the chosen method how to fill remaining committee seats (i.e., the second phase).
\end{proof}

\begin{proposition}\label{prop:candmon}
Thiele methods, rev-seq-PAV, MAV, and SAV satisfy support monotonicity with additional voters; seq-PAV, seq-CC, seq-\phragmen{}, and \lexphrag{} satisfy candidate monotonicity with additional voters but fail support monotonicity with additional voters. Further, Monroe, Greedy Monroe, and the Method of Equal Shares fail candidate monotonicity with additional voters.

AV and SAV satisfy support monotonicity without additional voters; PAV, CC, seq-PAV, seq-CC, rev-seq-PAV, Monroe, greedy-Monroe, seq-\phragmen{}, \lexphrag{}, the Method of Equal Shares, and MAV satisfy candidate monotonicity without additional voters; none of these satisfy support monotonicity without additional voters.
\end{proposition}

\begin{proof}
All counterexamples are additionally implemented (and verified) in the \textsf{abcvoting} library~\cite{abcvoting}.

\textbf{Support monotonicity with additional voters: }
\citet{sanchez2019monotonicity} show that Thiele methods, MAV, and SAV satisfy support monotonicity with additional voters (referred to as ``support monotonicity with population increase'' in their paper)

We prove that rev-seq-PAV satisfies support monotonicity with additional voters as well: 
Recall that rev-seq-PAV is resolute by definition. Let $X$ be a subset of the winning committee and assume we add a voter approving $X$.
We claim that exactly the same candidates are removed and in exactly the same order.
Let us prove this by induction and assume it holds for rounds $m, \dots, \ell$, where $\ell\leq m$ (recall that in rev-seq-PAV we count the rounds in the reverse order).
As in rounds $m, \dots, \ell$ the same candidates were removed, the marginal contribution of candidates outside of~$X$ is the same. The marginal contribution of candidates contained in $X$ is larger. Consequently, the candidate with the least marginal contribution is the same as it was in the original election and thus not a candidate in~$X$. We conclude that an additional voter approving $X$ does not change the winning committee.

\citet{Janson16arxiv} (based on \phragmen{} \cite{Phra96a}) proves that seq-PAV, rev-seq-PAV, and seq-\phragmen{} satisfy candidate monotonicity with additional voters.
Further, \lexphrag{} satisfies candidate monotonicity with additional voters; this is a consequence of the fact that it satisfies weak support monotonicity with population increase \cite{sanchez2019monotonicity}, and the proof for seq-PAV in this paper also holds for seq-CC.
A counterexample showing that seq-\phragmen{} fails support monotonicity with additional voters can be found in~\cite{MoOl15a,Janson16arxiv}.
Further, counterexamples for \lexphrag{} and seq-PAV can be found in~\cite{sanchez2019monotonicity}.

To see that seq-CC fails support monotonicity with additional voters, consider the following election instance:
\begin{align*}
& 3 \times \{a\} &&
 1 \times \{a, c, d\} &&
 1 \times \{b\} &&
 2 \times \{b, c\} \\
& 1 \times \{b, d\} &&
 2 \times \{c\} &&
 2 \times \{d\}\text{.}
\end{align*}
For $k=3$, the winning committee according to seq-CC is $\{a, c, d\}$ (in order $c,a,d$, assuming alphabetic tiebreaking). 
If an additional voter approves $\{a, d\}$, 
seq-CC returns $\{a, b, c\}$ (in order $a,b,c$, assuming alphabetic tiebreaking) and hence seq-CC fails support monotonicity with additional voters.

To see that Greedy Monroe fails candidate monotonicity with additional voters, consider the following election instance:
\begin{align*}
& 1 \times \{b, c, d\} &&
 1 \times \{a, c, f\} &&
 1 \times \{a, d, e\} &&
 1 \times \{c, e\} \\
& 1 \times \{a, b\} &&
 2 \times \{d, f\} &&
 1 \times \{b, e\} &&
 1 \times \{b, f\}\text{.}
\end{align*}
For $k=3$, the winning committee according to Greedy Monroe is $\{b, e, f\}$. 
If an additional voter approves $\{e\}$, the 
winning committees changes to $\{b, c, d\}$.
This committee does not contain $e$ and hence Greedy Monroe fails candidate monotonicity with additional voters.

For the Method of Equal Shares, consider the following instance:
\begin{align*}
& 1 \times \{b, d\} &&
 1 \times \{a, b\}&&
 1 \times \{b, d, e\}&&
 1 \times \{a, e\}\\
& 2 \times \{c, d, e\}
 1 \times \{c, e\}&&
 1 \times \{a, c, e\}&&
 1 \times \{b, c, d\}\text{.}
\end{align*}
For $k=3$, the winning committee according to the Method of Equal Shares is
$\{a, d, e\}$. If an additional voter approves $\{a\}$,
the 
winning committee changes to 
$\{b, c, e\}$. As this committee does not contain $a$, 
the Method of Equal Shares fails candidate monotonicity with additional voters.

The Method of Equal Shares also fails candidate monotonicity with additional voters if 
only the first phase of the method is considered (i.e., the method may return fewer than~$k$ candidates).
 For the profile
\begin{align*}
& 2 \times \{a, b, c\} &&
 1 \times \{a, g\} &&
 1 \times \{d, e\} &&
 1 \times \{b, d, f\} &&
 1 \times \{a, f\} &&
 1 \times \{h\} & \\ &
 1 \times \{a, h\} &&
 1 \times \{b, h\} &&
 1 \times \{b, d\} &&
 1 \times \{d, e, f\} &&
 1 \times \{c, e, h\}\text{.}
\end{align*}
The 
original winning committee is
 $\{a, b, e\}$. If an additional voter approves $\{e\}$,
the 
winning committee changes to 
$\{a, d, h\}$ (assuming alphabetic tiebreaking). Thus, 
Equal Shares without the 2nd phase also fails candidate monotonicity with additional voters.
 
Finally, an example showing that Monroe violates  candidate monotonicity with additional voters 
can be found in~\cite{sanchez2019monotonicity}.

\textbf{Support monotonicity without additional voters: }
AV and SAV satisfy support monotonicity without additional voters~\cite{sanchez2019monotonicity}.
PAV, CC, seq-PAV, seq-CC\footnote{The proof is only stated for seq-PAV but holds for seq-CC as well.}, rev-seq-PAV, Monroe, seq-\phragmen{}, \lexphrag{}, and MAV satisfy candidate monotonicity without additional voters~\cite{Janson16arxiv,sanchez2019monotonicity}.

To see that Greedy Monroe satisfies candidate monotonicity without additional voters, 
let $c$ be a candidate in the winning committee.
Now note that a voter additionally approving~$c$ can only lead to $c$ being added in an earlier round.
Hence, it is still contained in the winning committee (which may change, however). An analogous argument holds for the Method of Equal Shares as well.

PAV, CC, Monroe, \lexphrag{}, and MAV do not satisfy the stronger axiom, i.e., support monotonicity without additional voters, as shown by \citet{sanchez2019monotonicity}.
Also seq-\phragmen{} fails this axiom \cite{Janson16arxiv, MoOl15a}.

For seq-PAV, 
consider
\begin{align*}
& 1 \times \{c, d\}&&
 1 \times \{a, c\}&&
 1 \times \{a, d\}&&
 1 \times \{a, f\}\\
& 1 \times \{b, c\}&&
 2 \times \{ b, f\}&&
 1 \times \{c, e\}\text{.}
\end{align*}
For $k=3$, the winning committee according to seq-PAV is $\{a,c,f\}$. 
If the first voter changes her ballot from  $\{c, d\}$ to $\{a, c, d, f\}$, the
winning committee changes to
$\{a, b, c\}$ (using alphabetic tiebreaking). Thus
seq-PAV fails support monotonicity without additional voters.

For rev-seq-PAV, 
consider
\begin{align*}
& 2 \times \{a, e\}&&
 2 \times \{b, c, d\}&&
 1 \times \{d, e\}&&
 3 \times \{c, e\}&&
 1 \times \{b, d, e\}\\
& 1 \times \{a, b, c\}&&
 1 \times \{c, d, e\}&&
 2 \times \{a, d, e\}&&
 1 \times \{b, d\}&&
 1 \times \{a, b\}\\
& 1 \times \{a, d\}&&
 1 \times \{a, b, d\}&&
 1 \times \{b, c\}\text{.}
\end{align*}
For $k=3$, the winning committee according to rev-seq-PAV is
$\{c, d, e\}$.
If the first voter changes her ballot from $\{a, e\}$ to $\{a, c, d, e\}$,
the winning committee changes to $\{b, d, e\}$. As this 
committee does not contain $c$,
rev-seq-PAV fails support monotonicity without additional voters.

To see that seq-CC fails support monotonicity without additional voters,
consider the profile
\begin{align*}
 1 \times \{e\} &&
 1 \times \{a\} &&
 1 \times \{a, d\} &&
 3 \times \{b\} &&
 2 \times \{a, c\} &&
 1 \times \{b, c, d\} &&
 2 \times \{c\} &&
 2 \times \{d\}.
\end{align*}
The winning committee according to seq-CC is $\{b, c, d\}$. If the first voter
changes her ballot from $\{e\}$ to $\{b,d,e\}$, the winning
committee changes to $\{a, b, c\}$. Candidate~$d$ is no longer contained in the 
winning committee, hence seq-CC fails support monotonicity without additional voters.
 
For the Method of Equal Shares,
consider
\begin{align*}
 1 \times \{b\} &&
 1 \times \{a, b, e\} &&
 2 \times \{b, e\} &&
 1 \times \{c\} &&
 1 \times \{a, c\} &&
 1 \times \{a\}\text{.}
\end{align*}
The original winning committee is $\{a, b, e\}$.
If the first voter changes her ballot from $\{b\}$ to $\{a, b, e\}$,
the winning committee changes to 
$\{a, b, c\}$ (using alphabetic tiebreaking). This contradicts support monotonicity without additional voters.
 
 For Greedy~Monroe, consider $k=2$ and
\begin{align*}
& A(1): \{d\}&&
 A(2): \{c\}&&
 A(3): \{b\}&&
 A(4): \{a, c\}\text{.}
\end{align*}
 We assume alphabetic tiebreaking for candidates; for voters we assume that smaller numbers are selected first.
 The winning committee is $\{b, c\}$.
 If the first voter additionally approves $\{b, c\}$ (the new ballot is $\{b, c, d\}$,
 then $b$ is selected in the first round (tiebreaking between $b$ and $c$) and is assigned to voters $1$ and $2$. 
 In the second round there is a tie between $a$, $b$, and $c$, and thus $a$ is added to the committee.
 The winning committee is now $\{a, c\}$, which contradicts support monotonicity without additional voters.
\end{proof}

\begin{proposition}\label{prop:strategyproofness}
AV with a fixed tiebreaking order on candidates satisfies cardinality-strategyproofness and thus inclusion-strategyproofness.
CC, PAV, seq-PAV, seq-CC, rev-seq-PAV, Monroe, Greedy Monroe, seq-\phragmen{}, \lexphrag{}, the Method of Equal Shares, MAV, and SAV do not satisfy inclusion-strategyproofness.
\end{proposition}

\begin{proof}
To see that AV satisfies cardinality-strategyproofness, consider a fixed voter $i$. Observe that if $i$ disapproves one of the (truly) approved candidates, say $c$, then it may cause at most one additional candidate getting into the winning committee. However, this will happen only if $c$ is removed from the winning committee. In such a case, the satisfaction of $i$ cannot increase. If $i$ approves a not-yet approved candidate, then this might only cause that this candidate replaces some other candidate in the committee. Again, such a change cannot increase the satisfaction of the voter. Finally, a voter changing her ballot can be decomposed into a sequence of changes which consists of either approving a disliked candidate or disapproving a candidate that is actually liked. Each such a change cannot increase the satisfaction of the voter, as we have seen.  

All counterexamples are also implemented (and verified) in the \textsf{abcvoting} library~\cite{abcvoting}.
Note that inclusion-strategyproofness is defined for resolute rules; hence we assume lexicographic tie-breaking between committees for otherwise irresolute rules.
For tiebreaking between candidates we assume alphabetic tiebreaking, as usual.

For CC consider the following profile with 5 voters:
\begin{align*}
& 1 \times \{a, b\} && 3 \times \{a\} &&
  1 \times \{c\}\text{.}
\end{align*}
We assume an arbitrary tiebreaking between committees and without loss of generality we assume that a tie between committee $\{a, b\}$ and $\{a, c\}$ is resolved in favour of $\{a, b\}$.
For $k=2$, the winning committee according to CC is $\{a, c\}$ with a CC-score of $5$.
If the first voter changes her ballot from $\{a, b\}$ to $\{b\}$, committees $\{a, b\}$ and $\{a, c\}$ are tied with a CC-score of $4$. By lexicographic tiebreaking, committee $\{a, b\}$ wins and the voter benefited from the manipulation.

For PAV consider the following profile with 6 voters:
\begin{align*}
& 1 \times \{c, d, e\} &&
  1 \times \{a, b\} &&
  1 \times \{b, f\} &&
  1 \times \{a, c, d\} &&
  1 \times \{b, c, f\} &&
  1 \times \{c, e, f\}\text{.}
\end{align*}
For $k = 3$ the only winning committee is $\{b, c, f\}$. If the first voter submits $\{e\}$ instead of $\{c, d, e\}$, then $\{b, c, e\}$ will become the only winning committee.

For seq-PAV consider the following profile with 6 voters:
\begin{align*}
& 1 \times \{a, b\} &&
  1 \times \{b, d\} &&
  1 \times \{c, f\} &&
  1 \times \{a, b, f\} &&
  1 \times \{b, f\} &&
  1 \times \{b, c\}\text{.}
\end{align*}
For $k = 3$ the winning committee is $\{b, c, f\}$. The first voter can successfully manipulate by changing her ballot to $\{a\}$---then the winning committee changes to $\{a, b, f\}$. 

For seq-CC consider the following profile with 12 voters:
\begin{align*}
& 1 \times \{b, e, f\} &&
  1 \times \{a, b\} &&
  1 \times \{d, e, f\} &&
  1 \times \{d, e\} &&
  1 \times \{b, f\} &&
  2 \times \{c, d\} \\
& 1 \times \{a, b, c\} &&
  1 \times \{a, c\} &&
  1 \times \{a, b, e\} &&
  1 \times \{a, e, f\} &&
  1 \times \{b, c, d\}\text{.}
\end{align*}
For $k = 3$ the winning committee is $\{a, b, d\}$. The first voter can successfully manipulate by changing her ballot to $\{c\}$---then the winning committee changes to $\{b, c, e\}$.

For rev-seq-PAV consider the following profile with 5 voters:
\begin{align*}
& 1 \times \{a, b, c\} &&
  1 \times \{b, d\} &&
  1 \times \{b, c\} &&
  1 \times \{a, d, e\} &&
  1 \times \{b, e\}\text{.}
\end{align*}
For $k = 2$ the winning committee is $\{b, d\}$ (using alphabetic tiebreaking). If the first voter changes her ballot from $\{a, b, c\}$ to $\{a\}$, then $\{a, b\}$ will become the winning committee. The first voter prefers this committee to $\{b, d\}$, thus she has an incentive to misreport her preferences.

For Monroe consider the following profile with 12 voters:
\begin{align*}
& 1 \times \{b, d\} &&
  1 \times \{a, b, c\} &&
  1 \times \{b, e\} &&
  1 \times \{d, e\} &&
  1 \times \{e, f\} &&
  1 \times \{b, c, e\} \\
& 1 \times \{c, d, e\} &&
  1 \times \{b, c\} &&
  2 \times \{a, f\} &&
  1 \times \{b, c, d\} &&
  1 \times \{a, d\}\text{.}
\end{align*}
For $k = 3$ the only winning committee is $\{a, b, e\}$. If the first voter changes her ballot to $\{f\}$, the winning committee changes to $\{b, d, f\}$.

For Greedy Monroe consider the following profile with 4 voters:
\begin{align*}
& 1 \times \{a, b\} &&
  1 \times \{a, c, f\} &&
  1 \times \{a, c, d\} &&
  1 \times \{e, f\}\text{.}
\end{align*}
For $k = 2$ the winning committee is $\{a, c\}$. If the first voter changes her ballot to $\{b\}$, then $\{a, b\}$ becomes the winning committee.

For seq-\phragmen{} consider the following profile with 6 voters:
\begin{align*}
& 1 \times \{a, b, c\} &&
  1 \times \{a, b\} &&
  1 \times \{b, f\} &&
  1 \times \{c, e\} &&
  1 \times \{b, e, f\} &&
  1 \times \{b, d, f\}\text{.}
\end{align*}
For $k = 2$ the winning committee is $\{b, f\}$. If the first voter changes her ballot from $\{a, b, c\}$ to $\{c\}$, then the winning committee changes to $\{b, c\}$, an outcome that the voter strictly prefers to the original winning committee.

For \lexphrag{} consider the following profile:
\begin{align*}
& 1 \times \{a, b\} && 3 \times \{b, c, d\} \text{.}
\end{align*}
For $k=3$, committee $\{b, c, d\}$ is winning with a load of $0.75$ distributed to each voter.
If the first voter changes her ballot from $\{a, b\}$ to $\{a\}$, then all committees are tied with a maximum load of $1$.
Due to lexicographic tiebreaking $\{a, b, c\}$ wins, which this voter strictly prefers to the original winning committee.

For the Method of Equal Shares consider the following profile with 6 voters:
\begin{align*}
& 1 \times \{b, c, d\} &&
  1 \times \{a, b\} &&
  1 \times \{b, d\} &&
  1 \times \{c, d\} &&
  2 \times \{d, e\}\text{.}
\end{align*}
For $k = 3$ the winning committee is $\{b, d, e\}$. The first voter can successfully manipulate by changing her ballot to $\{c\}$---then the winning committee changes to $\{b, c, d\}$.

For MAV consider the following profile with 6 voters:
\begin{align*}
& 1 \times \{a, b, c\} &&
  1 \times \{b, d\} &&
  2 \times \{a, b, e\} &&
  1 \times \{a, b, d\} &&
  1 \times \{a, b\}\text{.}
\end{align*}
For $k = 3$ the unique winning committee is $\{a, b, d\}$. If the first voter changes her ballot to $\{c\}$, then $\{a, b, c\}$ becomes the only winning committee. 

For SAV consider the following profile with 2 voters:
\begin{align*}
& 1 \times \{a, b, c\} &&
  1 \times \{d, e\}\text{.}
\end{align*}
For $k=1$ the winning committees according to SAV are $\{d\}$ and $\{e\}$; committee $\{d\}$ is chosen due to lexicographic tiebreaking.
If the first voter changes her ballot to $\{a\}$, the winning committee
will change to $\{a\}$, an outcome which is preferred by the first voter.
\end{proof}

\section{Additional Proofs from Chapter~4}

\begin{proposition}\label{prop:greedy_monroe_apportionment}
If $k$ divides $n$, then Greedy Monroe extends the largest remainders method.
\end{proposition}
\begin{proof}
Consider an apportionment instance with $p$ political parties, $C_1, \ldots, C_p$, and let $n_i$ denote the number of votes cast on party $C_i$. Since $n$ is divisible by $k$, Greedy Monroe always tries to assign a candidate to $\frac{n}{k}$ voters. 
Observe that:
\begin{align*}
n_i - \frac{n}{k} < \left\lfloor k \cdot \frac{n_i}{n} \right\rfloor \cdot \frac{n}{k} \leq n_i \text{.}
\end{align*}

Let $k_1 = \sum_{i = 1}^p \lfloor k \cdot \nicefrac{n_i}{n} \rfloor$. In the first $k_1$ rounds Greedy Monroe assigns to each party $C_i$ exactly $\lfloor k \cdot \nicefrac{n_i}{n} \rfloor$ seats. This is consistent with the first phase of the largest remainders method. During these rounds, whenever Greedy Monroe  assigns a seat to a party, it removes $\nicefrac{n}{k}$ of its supporters. Then, each party $C_i$ is left with less than $\frac{n}{k}$ supporters. Specifically, party $C_i$ is left with the following number of supporters:
\begin{align*}
n_i - \left\lfloor k \cdot \frac{n_i}{n} \right\rfloor \cdot \frac{n}{k} = \frac{n}{k} \left(  k \cdot \frac{n_i}{n} - \left\lfloor k \cdot \frac{n_i}{n} \right\rfloor \right) \text{.}
\end{align*}  
Next, Greedy Monroe will assign the remaining seats to the parties in the order of decreasing values $k \cdot \nicefrac{n_i}{n} - \lfloor k \cdot \nicefrac{n_i}{n}\rfloor$, that is, it will proceed exactly as the largest remainders method.
\end{proof}

\begin{proposition}\label{prop:greedy_monroe_apportionment_k_does_not_divide_n}
In the general case (when $k$ does not have to divide $n$), Greedy Monroe and Monroe do not extend the largest remainders method.
\end{proposition}
\begin{proof}
Consider an apportionment instance with 2 parties with, respectively, 50 votes and 31 votes. Assume the committee size is $k = 4$. For this instance LRM gives 2 seats to each party. Greedy Monroe can proceed as follows. It starts by giving the second party a representative and removing the group of 21 voters. Next it can give 3 representatives to the first party (depending on tiebreaking). The Monroe rule can also select 3 candidates from the first party and one candidate from the second party.
\end{proof}

\begin{proposition}
Greedy Monroe satisfies justified representation (JR).\label{prop:greedy-monroe-jr}
\end{proposition}

\begin{proof}
\citet{pjr17} show that Greedy Monroe satisfies PJR if $k$ divides $n$, hence it also satisfies JR under this condition. However, Greedy Monroe satisfies JR also without this additional constraint. Assume towards a contradiction that Greedy Monroe fails JR for the election instance $(A, k)$ and let $W$ be the winning committee according to Greedy Monroe. As $W$ does not satisfy JR, there exists a group of voters $V$ of size at least $\nicefrac n k$ and a candidate $c\notin W$ approved by all of them.
Adding candidate~$c$ would have increased the Monroe score of the committee by at least $\nicefrac n k$ in all rounds. Hence, the candidates contained in $W$ also increased the score by at least $\nicefrac n k$ each. Thus, $W$ has a Monroe score of $n$, i.e., all voters have an approved candidate in $W$, which implies that JR is satisfied.
\end{proof}

\begin{proposition}
An ABC rule with a proportionality degree of $f_{\calR}(\ell) = \ell-1$ may fail EJR.\label{prop:ejr+proprank}
\end{proposition}

\begin{proof}
Consider the profile
	\begin{align*}
		&2 \times \{ a,b\} & & 1 \times \{ c, d\}
	\end{align*}
for $k=3$.
An ABC rule that selects the committee $\{a,c,d\}$ fails EJR, but may have a proportionality degree of $f_{\calR}(\ell) = \ell-1$. (To fully define such an ABC rule, it could behave as PAV on all other profiles.)
\end{proof}

\begin{proposition}
An ABC rule cannot satisfy both perfect representation and weak Pareto efficiency.\label{prop:pr-pareto}
\end{proposition}
\begin{proof}
Consider the profile
	\begin{align*}
		&2 \times \{ a,c\} & & 1 \times \{ a,c, d\} &&  1 \times \{ a, d\} &&  1 \times \{ b, d\}
		 && 3 \times \{ b, c\}\text{.}
	\end{align*}
	For $k=2$, there is exactly one committee that satisfies perfect representation: $W_1=\{a,b\}$. This committee, however, is dominated by $W_2=\{c,d\}$. An ABC rule $\calR$ satisfies PR if it exclusively returns committees satisfying PR; hence $W_1$ is the only winning committee and thus $\calR$ fails weak Pareto efficiency.
\end{proof}

\begin{proposition}\label{prop:prop_degree_of_sav_and_mav}
The proportionality degree of the Method of Equal Shares is between $\frac{\ell-1}{2}$ and $\frac{\ell+1}{2}$. The proportionality degree of SAV and MAV is $0$.
\end{proposition}
\begin{proof}
For SAV fix $\ell \in \naturals$, set the committee size to $k = 2\ell+1$, and consider the following profile with $m = 2k$ candidates and $n = k$ voters: the first $\ell$ voters approve candidates $a_1, \ldots, a_{k}$ and the next $k - \ell$ voters approve $b_1, \ldots, b_{k}$. SAV will select the committee $\{b_1, \ldots, b_k\}$. The group of the first $\ell$ voters is $\ell$-cohesive, but no voter gets any representative in the elected committee. 

For MAV fix $\ell \in \naturals$, set the committee size to $k = \ell+1$, and consider the following profile with $m = 4k+1$ candidates and $n = k$ voters: the first $\ell$ voters approve candidates $a_1, \ldots a_{k}$ and the next voter approves $b_1, \ldots, b_{3k+1}$. MAV will select a $k$-element subset of $\{b_1, \ldots, b_{3k+1}\}$. The group of the first $\ell$ voters is $\ell$-cohesive, but no voter gets any representative in the elected committee. 

Finally, we consider the Method of Equal Shares. Since the method satisfies EJR~\cite{pet-sko:laminar} and EJR implies a proportionality degree of at least $f(\ell) = \frac{\ell-1}{2}$~\cite{pjr17}, we get the lower-bound. For the upper bound consider the following instance. Fix $\ell \in \naturals$. We set $n = k = \frac{\ell(\ell+1)}{2}$ and $m = k+\ell$. The voters are divided into $\ell$ groups $N = N_1 \cup N_2 \cup \ldots \cup N_{\ell}$ such that $|N_i| = i$ for each $i \in [\ell]$. The set of the first $k$ candidates is also divided into $\ell$ groups $C = C_1 \cup C_2 \cup \ldots \cup C_{\ell}$ such that $|C_i| = i$ for each $i \in [\ell]$. The set of remaining $\ell$ candidates is denoted by $A$. The voters from $N_i$ approve $C_i$. Additionally the first voter from each group $N_i$ approves $A$. The Method of Equal Shares can select the candidates from $C_{\ell}$ first. Then the voters from $N_{\ell}$ have no money left. Next the candidates from $C_{\ell-1}$ are selected, etc. Consequently, the method can return committee $C_1 \cup C_2 \cup \ldots \cup C_{\ell}$. Consider the voters who approve $A$. They form an $\ell$-cohesive group, but the average number of representatives that they get equals $1 + 2 + \ldots + \ell = \frac{\ell(\ell+1)}{2}$. This completes the proof.  
\end{proof}

\bibliographystyle{abbrvnat}
\bibliography{main}

\begin{thebibliography}{34}
\providecommand{\natexlab}[1]{#1}
\providecommand{\url}[1]{\texttt{#1}}
\expandafter\ifx\csname urlstyle\endcsname\relax
  \providecommand{\doi}[1]{doi: #1}\else
  \providecommand{\doi}{doi: \begingroup \urlstyle{rm}\Url}\fi

\bibitem[Al{\'o}s-Ferrer and Grani{\'c}(2012)]{alos2012two}
C.~Al{\'o}s-Ferrer and {\DH}.-G. Grani{\'c}.
\newblock Two field experiments on {A}pproval {V}oting in {G}ermany.
\newblock \emph{Social Choice and Welfare}, 39\penalty0 (1):\penalty0 171--205,
  2012.

\bibitem[Ames(1995)]{ames1995ElectoralStrategies}
B.~Ames.
\newblock Electoral strategies under open-list proportional representation.
\newblock \emph{American Journal of Political Science}, 39\penalty0
  (2):\penalty0 406--433, 1995.

\bibitem[Andr\'{e} et~al.(2015)Andr\'{e}, Depauw, and
  Martin]{audrey2015electoral}
A.~Andr\'{e}, S.~Depauw, and S.~Martin.
\newblock Electoral systems and legislators' constituency effort: The mediating
  effect of electoral vulnerability.
\newblock \emph{Comparative Political Studies}, 48\penalty0 (4):\penalty0
  464--496, 2015.

\bibitem[Aragones et~al.(2011)Aragones, Gilboa, and Weiss]{aragones2011making}
E.~Aragones, I.~Gilboa, and A.~Weiss.
\newblock Making statements and approval voting.
\newblock \emph{Theory and decision}, 71\penalty0 (4):\penalty0 461--472, 2011.

\bibitem[Brams and Fishburn(1983)]{brams1983approval}
S.~J. Brams and P.~C. Fishburn.
\newblock \emph{Approval Voting}.
\newblock Birkh\"auser, 1983.

\bibitem[Brams and Fishburn(2010)]{brams2010}
S.~J. Brams and P.~C. Fishburn.
\newblock Going from theory to practice: The mixed success of approval voting.
\newblock In J.-F. Laslier and M.~R. Sanver, editors, \emph{Handbook on
  Approval Voting}, pages 19--37. Springer, 2010.

\bibitem[Brams and Herschbach(2001)]{bramsScience}
S.~J. Brams and D.~R. Herschbach.
\newblock The science of elections.
\newblock \emph{Science}, 292\penalty0 (5521):\penalty0 1449, 2001.

\bibitem[Brandt et~al.(2016)Brandt, Conitzer, Endriss, Lang, and
  Procaccia]{Handbook-COMSOC}
F.~Brandt, V.~Conitzer, U.~Endriss, J.~Lang, and A.~D. Procaccia.
\newblock \emph{Handbook of Computational Social Choice}.
\newblock Cambridge University Press, New York, NY, USA, 1st edition, 2016.

\bibitem[Burdges et~al.(2020)Burdges, Cevallos, Czaban, Habermeier, Hosseini,
  Lama, Alper, Luo, Shirazi, Stewart, and Wood]{burdges2020overview}
J.~Burdges, A.~Cevallos, P.~Czaban, R.~Habermeier, S.~Hosseini, F.~Lama, H.~K.
  Alper, X.~Luo, F.~Shirazi, A.~Stewart, and G.~Wood.
\newblock Overview of {P}olkadot and its design considerations.
\newblock \emph{CoRR}, abs/2005.13456, 2020.
\newblock URL \url{https://arxiv.org/abs/2005.13456}.

\bibitem[Cevallos and Stewart(2021)]{cevallos2020verifiably}
A.~Cevallos and A.~Stewart.
\newblock A verifiably secure and proportional committee election rule.
\newblock In \emph{Proceedings of the 3rd ACM Conference on Advances in
  Financial Technologies}, pages 29--42, 2021.

\bibitem[Chakraborty et~al.(2019)Chakraborty, Patro, Ganguly, Gummadi, and
  Loiseau]{ChakrabortyPGGL19}
A.~Chakraborty, G.~K. Patro, N.~Ganguly, K.~P. Gummadi, and P.~Loiseau.
\newblock Equality of voice: Towards fair representation in crowdsourced top-k
  recommendations.
\newblock In \emph{Proceedings of the Conference on Fairness, Accountability,
  and Transparency, FAT* 2019, Atlanta, GA, USA, January 29-31, 2019}, pages
  129--138. {ACM}, 2019.
\newblock \doi{10.1145/3287560.3287570}.
\newblock URL \url{https://doi.org/10.1145/3287560.3287570}.

\bibitem[Chang(2005)]{chang2005ElectoralIncentives}
E.~Chang.
\newblock Electoral incentives for political corruption under open-list
  proportional representation.
\newblock \emph{The Journal of Politics}, 67\penalty0 (3):\penalty0 716--730,
  2005.

\bibitem[Colomer(2011)]{personalRepresentation}
J.~M. Colomer.
\newblock \emph{Personal Representation: The Neglected Dimension of Electoral
  Systems}.
\newblock ECPR Press, Colchester, 2011.

\bibitem[Faliszewski et~al.(2017{\natexlab{a}})Faliszewski, Sawicki, Schaefer,
  and Smolka]{FaliszewskiSSS17}
P.~Faliszewski, J.~Sawicki, R.~Schaefer, and M.~Smolka.
\newblock Multiwinner voting in genetic algorithms.
\newblock \emph{{IEEE} Intell. Syst.}, 32\penalty0 (1):\penalty0 40--48,
  2017{\natexlab{a}}.

\bibitem[Faliszewski et~al.(2017{\natexlab{b}})Faliszewski, Skowron, Slinko,
  and Talmon]{FSST-trends}
P.~Faliszewski, P.~Skowron, A.~Slinko, and N.~Talmon.
\newblock Multiwinner voting: {A} new challenge for social choice theory.
\newblock In U.~Endriss, editor, \emph{Trends in Computational Social Choice},
  chapter~2, pages 27--47. AI Access, 2017{\natexlab{b}}.

\bibitem[Farahani and Hekmatfar(2009)]{far-hek:b:facility-location}
F.~Z. Farahani and M.~Hekmatfar, editors.
\newblock \emph{Facility Location: {Concepts}, Models, and Case Studies}.
\newblock Springer, 2009.

\bibitem[Farrell(2011)]{Far11}
D.~Farrell.
\newblock \emph{Electoral systems: {A} comparative introduction}.
\newblock Palgrave Macmillan, 2011.

\bibitem[Gawron and Faliszewski(2021)]{gawronusing}
G.~Gawron and P.~Faliszewski.
\newblock Using multiwinner voting to search for movies.
\newblock In \emph{Proceedings of the 3rd Games, Agents, and Incentives
  Workshop (GAIW 2021)}. {IFAAMAS}, 2021.

\bibitem[Goel et~al.(2015)Goel, Krishnaswamy, Sakshuwong, and
  Aitamurto]{goel2015knapsack}
A.~Goel, A.~K. Krishnaswamy, S.~Sakshuwong, and T.~Aitamurto.
\newblock Knapsack voting.
\newblock \emph{Collective Intelligence}, 1, 2015.

\bibitem[Grofman(2016)]{grofmanProspectives}
B.~Grofman.
\newblock Perspectives on the comparative study of electoral systems.
\newblock \emph{Annual Review of Political Science}, 19:\penalty0 1--23, 2016.

\bibitem[Israel and Brill(2021)]{israel2021dynamic}
J.~Israel and M.~Brill.
\newblock Dynamic proportional rankings.
\newblock In \emph{Proceedings of the 30th International Joint Conference on
  Artificial Intelligence (IJCAI-2021)}, pages 261--267, 2021.

\bibitem[Kilgour(2010)]{kil-handbook}
D.~M. Kilgour.
\newblock Approval balloting for multi-winner elections.
\newblock In J.-F. Laslier and M.~R. Sanver, editors, \emph{Handbook on
  Approval Voting}, pages 105--124. Springer, 2010.

\bibitem[Kilgour and Marshall(2012)]{kil-mar:j:minimax-approval}
D.~M. Kilgour and E.~Marshall.
\newblock Approval balloting for fixed-size committees.
\newblock In D.~S. Felsenthal and M.~Machover, editors, \emph{Electoral
  Systems: Paradoxes, Assumptions, and Procedures}, Studies in Choice and
  Welfare, chapter~12, pages 305--326. Springer, 2012.

\bibitem[Lackner et~al.(2021)Lackner, Regner, Krenn, and Forster]{abcvoting}
M.~Lackner, P.~Regner, B.~Krenn, and S.~S. Forster.
\newblock {abcvoting: A Python library of approval-based committee voting
  rules}, 2021.
\newblock URL \url{https://doi.org/10.5281/zenodo.3904466}.
\newblock Current version: \url{https://github.com/martinlackner/abcvoting}.

\bibitem[Laslier and Sanver(2010)]{av-handbook}
J.-F. Laslier and M.~R. Sanver, editors.
\newblock \emph{Handbook on Approval Voting}.
\newblock Springer, 2010.

\bibitem[Laslier and {V}an~der Straeten(2008)]{las-str:j:approval-experiment}
J.-F. Laslier and K.~{V}an~der Straeten.
\newblock A live experiment on approval voting.
\newblock \emph{Experimental Economics}, 11\penalty0 (1):\penalty0 97--105,
  2008.

\bibitem[Laslier and {Van~der~Straeten}(2016)]{laslier2016StrategicVoting}
J.-F. Laslier and K.~{Van~der~Straeten}.
\newblock Strategic voting in multi-winners elections with approval balloting:
  a theory for large electorates.
\newblock \emph{Social Choice and Welfare}, 47\penalty0 (3):\penalty0 559--587,
  2016.

\bibitem[Lijphart and Grofman(1984)]{grofmanChoosingElectoral}
A.~Lijphart and B.~Grofman.
\newblock \emph{Choosing an Electoral System: Issues and Alternatives}.
\newblock Praeger, New York, 1984.

\bibitem[Lu and Boutilier(2011)]{budgetSocialChoice}
T.~Lu and C.~Boutilier.
\newblock Budgeted social choice: {F}rom consensus to personalized decision
  making.
\newblock In \emph{Proceedings of the 22nd International Joint Conference on
  Artificial Intelligence (IJCAI-2011)}, pages 280--286, 2011.

\bibitem[Lu and Boutilier(2015)]{bou-lu:c:value-directed}
T.~Lu and C.~Boutilier.
\newblock Value directed compression of large-scale assignment problems.
\newblock In \emph{Proceedings of the 29th Conference on Artificial
  Intelligence (AAAI-2015)}, pages 1182--1190, 2015.

\bibitem[Renwick and Pilet(2016)]{renwick16}
A.~Renwick and J.~B. Pilet.
\newblock \emph{Faces on the Ballot: {T}he Personalization of Electoral Systems
  in Europe}.
\newblock Oxford University Press, 2016.

\bibitem[Skowron et~al.(2016)Skowron, Faliszewski, and Lang]{owaWinner}
P.~Skowron, P.~Faliszewski, and J.~Lang.
\newblock Finding a collective set of items: {F}rom proportional
  multirepresentation to group recommendation.
\newblock \emph{Artificial Intelligence}, 241:\penalty0 191--216, 2016.

\bibitem[Skowron et~al.(2017)Skowron, Lackner, Brill, Peters, and
  Elkind]{proprank}
P.~Skowron, M.~Lackner, M.~Brill, D.~Peters, and E.~Elkind.
\newblock Proportional rankings.
\newblock In \emph{Proceedings of the 26th International Joint Conference on
  Artificial Intelligence (IJCAI-2017)}, pages 409--415, 2017.

\bibitem[{V}an~der Straeten et~al.(2018){V}an~der Straeten, Lachat, and
  Laslier]{LasStr17}
K.~{V}an~der Straeten, R.~Lachat, and J.-F. Laslier.
\newblock Strategic voting in multi-winner elections with approval balloting:
  {A}n application to the 2011 regional government election in {Z}urich.
\newblock In J.~Aldrich, A.~Blais, and L.~Stephenson, editors, \emph{The Many
  Faces of Strategic Voting}. CBS, 2018.

\end{thebibliography}


\begin{thebibliography}{48}
\providecommand{\natexlab}[1]{#1}
\providecommand{\url}[1]{\texttt{#1}}
\expandafter\ifx\csname urlstyle\endcsname\relax
  \providecommand{\doi}[1]{doi: #1}\else
  \providecommand{\doi}{doi: \begingroup \urlstyle{rm}\Url}\fi

\bibitem[Amanatidis et~al.(2015)Amanatidis, Barrot, Lang, Markakis, and
  Ries]{amanatidis2015multiple}
G.~Amanatidis, N.~Barrot, J.~Lang, E.~Markakis, and B.~Ries.
\newblock Multiple referenda and multiwinner elections using {H}amming
  distances: {C}omplexity and manipulability.
\newblock In \emph{Proceedings of the 14th International Conference on
  Autonomous Agents and Multiagent Systems (AAMAS-2015)}, pages 715--723, 2015.

\bibitem[Aziz and Lee(2020)]{aziz2020expanding}
H.~Aziz and B.~E. Lee.
\newblock The expanding approvals rule: Improving proportional representation
  and monotonicity.
\newblock \emph{Social Choice and Welfare}, 54\penalty0 (1):\penalty0 1--45,
  2020.

\bibitem[Aziz et~al.(2017)Aziz, Brill, Conitzer, Elkind, Freeman, and
  Walsh]{justifiedRepresentation}
H.~Aziz, M.~Brill, V.~Conitzer, E.~Elkind, R.~Freeman, and T.~Walsh.
\newblock Justified representation in approval-based committee voting.
\newblock \emph{Social Choice and Welfare}, 48\penalty0 (2):\penalty0 461--485,
  2017.

\bibitem[Aziz et~al.(2018)Aziz, Faliszewski, Grofman, Slinko, and
  Talmon]{AFGST-egalitarian}
H.~Aziz, P.~Faliszewski, B.~Grofman, A.~Slinko, and N.~Talmon.
\newblock Egalitarian committee scoring rules.
\newblock In \emph{Proceedings of the 27th International Joint Conference on
  Artificial Intelligence (IJCAI-2018)}, pages 56--62, 2018.

\bibitem[Betzler et~al.(2013)Betzler, Slinko, and
  Uhlmann]{fullyProportionalRepr}
N.~Betzler, A.~Slinko, and J.~Uhlmann.
\newblock On the computation of fully proportional representation.
\newblock \emph{Journal of Artificial Intelligence Research}, 47:\penalty0
  475--519, 2013.

\bibitem[Brams and Kilgour(2014)]{BrKi14a}
S.~J. Brams and D.~M. Kilgour.
\newblock Satisfaction approval voting.
\newblock In \emph{Voting Power and Procedures}, Studies in Choice and Welfare,
  pages 323--346. Springer, 2014.

\bibitem[Brams et~al.(2007)Brams, Kilgour, and Sanver]{minimaxProcedure}
S.~J. Brams, D.~M. Kilgour, and M.~R. Sanver.
\newblock A minimax procedure for electing committees.
\newblock \emph{Public Choice}, 132\penalty0 (3--4):\penalty0 401--420, 2007.

\bibitem[Brill et~al.(2017)Brill, Freeman, Janson, and
  Lackner]{aaai/BrillFJL17-phragmen}
M.~Brill, R.~Freeman, S.~Janson, and M.~Lackner.
\newblock Phragm{\'e}n's voting methods and justified representation.
\newblock In \emph{Proceedings of the 31st Conference on Artificial
  Intelligence (AAAI-2017)}, pages 406--413, 2017.
\newblock Extended version at \url{https://arxiv.org/abs/2102.12305}.

\bibitem[Camps et~al.(2019)Camps, Mora, and Saumell]{camps2019method}
R.~Camps, X.~Mora, and L.~Saumell.
\newblock The method of {E}nestr\"om and {P}hragm\'en for parliamentary
  elections by means of approval voting.
\newblock \emph{arXiv preprint arXiv:1907.10590}, 2019.
\newblock URL \url{https://arxiv.org/abs/1907.10590}.

\bibitem[Carleman(1938)]{carleman1938phragmen}
T.~Carleman.
\newblock L.~{E}.~{P}hragm{\'e}n in memoriam.
\newblock \emph{Acta Mathematica}, 69:\penalty0 XXXI--XXXIII, 1938.

\bibitem[Cevallos and Stewart(2021)]{cevallos2020verifiably}
A.~Cevallos and A.~Stewart.
\newblock A verifiably secure and proportional committee election rule.
\newblock In \emph{Proceedings of the 3rd ACM Conference on Advances in
  Financial Technologies}, pages 29--42, 2021.

\bibitem[Chamberlin and C{ourant}(1983)]{ccElection}
B.~Chamberlin and P.~C{ourant}.
\newblock Representative deliberations and representative decisions:
  {P}roportional representation and the {B}orda rule.
\newblock \emph{American Political Science Review}, 77\penalty0 (3):\penalty0
  718--733, 1983.

\bibitem[Elkind and Ismaili(2015)]{elk-ism:c:owa-egalitarian-utilitarian}
E.~Elkind and A.~Ismaili.
\newblock {OWA}-based extensions of the {Chamberlin-Courant} rule.
\newblock In \emph{Proceedings of the 4th International Conference on
  Algorithmic Decision Theory (ADT-2015)}, pages 486--502, 2015.

\bibitem[Elkind et~al.(2017{\natexlab{a}})Elkind, Faliszewski, Laslier,
  Skowron, Slinko, and Talmon]{2dpictures}
E.~Elkind, P.~Faliszewski, J.~Laslier, P.~Skowron, A.~Slinko, and N.~Talmon.
\newblock What do multiwinner voting rules do? {A}n experiment over the
  two-dimensional euclidean domain.
\newblock In \emph{Proceedings of the 31st Conference on Artificial
  Intelligence (AAAI-2017)}, pages 494--501, 2017{\natexlab{a}}.

\bibitem[Elkind et~al.(2017{\natexlab{b}})Elkind, Faliszewski, Skowron, and
  Slinko]{elk-fal-sko-sli:c:multiwinner-rules}
E.~Elkind, P.~Faliszewski, P.~Skowron, and A.~Slinko.
\newblock Properties of multiwinner voting rules.
\newblock \emph{Social Choice and Welfare}, 48\penalty0 (3):\penalty0 599--632,
  2017{\natexlab{b}}.

\bibitem[Enestr\"om(1896)]{enestrom:1896}
G.~H. Enestr\"om.
\newblock Om aritmetiska och statistiska metoder f\"or proportionella val.
\newblock \emph{\"Ofversigt af Kongliga Vetenskaps-Akademiens F\"orhandlingar},
  53:\penalty0 543--570, 1896.

\bibitem[Faliszewski et~al.(2017)Faliszewski, Skowron, Slinko, and
  Talmon]{FaliszewskiSST17}
P.~Faliszewski, P.~Skowron, A.~Slinko, and N.~Talmon.
\newblock Multiwinner rules on paths from $k$-{B}orda to
  {C}hamberlin-{C}ourant.
\newblock In \emph{Proceedings of the 26th International Joint Conference on
  Artificial Intelligence (IJCAI-2017)}, pages 192--198. ijcai.org, 2017.

\bibitem[Faliszewski et~al.(2018)Faliszewski, Lackner, Peters, and
  Talmon]{FLPT-effective}
P.~Faliszewski, M.~Lackner, D.~Peters, and N.~Talmon.
\newblock Effective heuristics for committee scoring rules.
\newblock In \emph{Proceedings of the 32nd Conference on Artificial
  Intelligence (AAAI-2018)}, pages 1023--1030, 2018.

\bibitem[Fishburn and Peke{\v{c}}(2004)]{fishburn2004approval}
P.~C. Fishburn and A.~Peke{\v{c}}.
\newblock Approval voting for committees: {T}hreshold approaches.
\newblock Technical report, 2004.
\newblock Technical Report.

\bibitem[Godziszewski et~al.(2021)Godziszewski, Batko, Skowron, and
  Faliszewski]{god-bat-sko-fal:c:2d-abc}
M.~Godziszewski, P.~Batko, P.~Skowron, and P.~Faliszewski.
\newblock An analysis of approval-based committee rules for {2D}-{E}uclidean
  elections.
\newblock In \emph{Proceedings of the 35th Conference on Artificial
  Intelligence (AAAI-2021)}, pages 5448--5455, 2021.

\bibitem[Janson(2012)]{Jans12a}
S.~Janson.
\newblock Proportionella valmetoder.
\newblock Technical report, 2012.
\newblock Available at http://www2.math.uu.se/\textasciitilde
  svante/papers/sjV6.pdf.

\bibitem[Janson(2016)]{Janson16arxiv}
S.~Janson.
\newblock Phragm{\'{e}}n's and {T}hiele's election methods.
\newblock \emph{CoRR}, abs/1611.08826, 2016.
\newblock URL \url{http://arxiv.org/abs/1611.08826}.

\bibitem[Kilgour and Marshall(2012)]{kil-mar:j:minimax-approval}
D.~M. Kilgour and E.~Marshall.
\newblock Approval balloting for fixed-size committees.
\newblock In D.~S. Felsenthal and M.~Machover, editors, \emph{Electoral
  Systems: Paradoxes, Assumptions, and Procedures}, Studies in Choice and
  Welfare, chapter~12, pages 305--326. Springer, 2012.

\bibitem[Konieczny and P{\'e}rez(2011)]{konieczny2011logic}
S.~Konieczny and R.~P. P{\'e}rez.
\newblock Logic based merging.
\newblock \emph{Journal of Philosophical Logic}, 40\penalty0 (2):\penalty0
  239--270, 2011.

\bibitem[Lackner and Skowron(2020)]{lac-sko2019}
M.~Lackner and P.~Skowron.
\newblock Utilitarian welfare and representation guarantees of approval-based
  multiwinner rules.
\newblock \emph{Artificial Intelligence}, 288:\penalty0 103366, 2020.

\bibitem[Lauritzen(2002)]{lauritzen2002thiele}
S.~L. Lauritzen.
\newblock \emph{Thiele: {P}ioneer in Statistics}.
\newblock Clarendon Press, 2002.

\bibitem[Monroe(1995)]{monroeElection}
B.~Monroe.
\newblock Fully proportional representation.
\newblock \emph{American Political Science Review}, 89\penalty0 (4):\penalty0
  925--940, 1995.

\bibitem[Mora and Oliver(2015)]{MoOl15a}
X.~Mora and M.~Oliver.
\newblock Eleccions mitjan{\c c}ant el vot d'aprovaci{\'o}. {E}l m{\`e}tode de
  {P}hragm{\'e}n i algunes variants.
\newblock \emph{Butllet{\'\i} de la Societat Catalana de Matem{\`a}tiques},
  30\penalty0 (1):\penalty0 57--101, 2015.

\bibitem[Moulin(1988)]{moulinAxioms}
H.~Moulin.
\newblock \emph{Axioms of Cooperative Decision Making}.
\newblock Cambridge University Press, 1988.

\bibitem[O'Connor and Robertson(2003)]{thiele-history}
J.~J. O'Connor and E.~F. Robertson.
\newblock {T}horvald {N}icolai {T}hiele.
\newblock \emph{MacTutor History of Mathematics}, 2003.
\newblock URL
  \url{http://mathshistory.st-andrews.ac.uk/Biographies/Thiele.html}.

\bibitem[O'Connor and Robertson(2011)]{phrag-history}
J.~J. O'Connor and E.~F. Robertson.
\newblock {L}ars {E}dvard {P}hragm\'en.
\newblock \emph{MacTutor History of Mathematics}, 2011.
\newblock URL
  \url{http://mathshistory.st-andrews.ac.uk/Biographies/Phragmen.html}.

\bibitem[Peters and Skowron(2020)]{pet-sko:laminar}
D.~Peters and P.~Skowron.
\newblock Proportionality and the limits of welfarism.
\newblock In \emph{Proceedings of the 2020 ACM Conference on Economics and
  Computation (ACM-EC-2020)}, pages 793--794, 2020.
\newblock Extended version at \url{https://arxiv.org/abs/1911.11747}.

\bibitem[Peters et~al.(2021)Peters, Pierczynski, and
  Skowron]{pet-pie-sko:c:participatory-budgeting-cardinal}
D.~Peters, G.~Pierczynski, and P.~Skowron.
\newblock Proportional participatory budgeting with additive utilities.
\newblock In \emph{Proceedings of the Thirty-fifth Conference on Neural
  Information Processing Systems (NeurIPS-2021)}, pages 12726--12737, 2021.

\bibitem[Phragm{\'e}n(1893)]{Phrag93}
E.~Phragm{\'e}n.
\newblock Om proportionella val.
\newblock \emph{Stockholms Dagblad}, 14 March 1893, 1893.
\newblock Summary of a public lecture published in a newspaper.

\bibitem[Phragm{\'e}n(1894)]{Phra94a}
E.~Phragm{\'e}n.
\newblock Sur une m{\'e}thode nouvelle pour r{\'e}aliser, dans les
  {\'e}lections, la repr{\'e}sentation proportionnelle des partis.
\newblock \emph{\"Ofversigt af Kongliga Vetenskaps-Akademiens F\"orhandlingar},
  51\penalty0 (3):\penalty0 133--137, 1894.

\bibitem[Phragm{\'e}n(1895)]{Phra95a}
E.~Phragm{\'e}n.
\newblock \emph{Proportionella val. En valteknisk studie}.
\newblock Svenska sp{\"o}rsm{\aa}l 25. Lars H{\"o}kersbergs f{\"o}rlag,
  Stockholm, 1895.

\bibitem[Phragm{\'e}n(1896)]{Phra96a}
E.~Phragm{\'e}n.
\newblock Sur la th{\'e}orie des {\'e}lections multiples.
\newblock \emph{{\"O}fversigt af Kongliga Vetenskaps-Akademiens
  F{\"o}rhandlingar}, 53:\penalty0 181--191, 1896.

\bibitem[Phragm{\'e}n(1899)]{Phra99a}
E.~Phragm{\'e}n.
\newblock Till fr{\aa}gan om en proportionell valmetod.
\newblock \emph{Statsvetenskaplig Tidskrift}, 2\penalty0 (2):\penalty0
  297--305, 1899.

\bibitem[Phragm{\'e}n and Lindel{\"o}f(1908)]{phragmen1908extension}
E.~Phragm{\'e}n and E.~Lindel{\"o}f.
\newblock Sur une extension d'un principe classique de l'analyse et sur
  quelques propri{\'e}t{\'e}s des fonctions monog{\`e}nes dans le voisinage
  d'un point singulier.
\newblock \emph{Acta Mathematica}, 31\penalty0 (1):\penalty0 381--406, 1908.

\bibitem[S{\'a}nchez-Fern{\'a}ndez et~al.(2021)S{\'a}nchez-Fern{\'a}ndez,
  Fern{\'a}ndez, Fisteus, and Brill]{sanchez2021maximin}
L.~S{\'a}nchez-Fern{\'a}ndez, N.~Fern{\'a}ndez, J.~A. Fisteus, and M.~Brill.
\newblock The maximin support method: {A}n extension of the {D’Hondt} method
  to approval-based multiwinner elections.
\newblock In \emph{Proceedings of the 35th Conference on Artificial
  Intelligence (AAAI-2021)}, pages 5690--5697, 2021.

\bibitem[Skowron et~al.(2015)Skowron, Faliszewski, and
  Slinko]{sko-fal-sli:j:multiwinner}
P.~Skowron, P.~Faliszewski, and A.~Slinko.
\newblock Achieving fully proportional representation: {Approximability}
  result.
\newblock \emph{Artificial Intelligence}, 222:\penalty0 67--103, 2015.

\bibitem[Skowron et~al.(2016)Skowron, Faliszewski, and Lang]{owaWinner}
P.~Skowron, P.~Faliszewski, and J.~Lang.
\newblock Finding a collective set of items: {F}rom proportional
  multirepresentation to group recommendation.
\newblock \emph{Artificial Intelligence}, 241:\penalty0 191--216, 2016.

\bibitem[Stubhaug(2010)]{stubhaug2010gosta}
A.~Stubhaug.
\newblock \emph{G{\"o}sta Mittag-Leffler: A man of conviction}.
\newblock Springer Science \& Business Media, 2010.

\bibitem[Thiele(1895)]{Thie95a}
T.~N. Thiele.
\newblock Om flerfoldsvalg.
\newblock In \emph{Oversigt over det Kongelige Danske Videnskabernes Selskabs
  Forhandlinger}, pages 415--441. 1895.

\bibitem[Thiele(1903)]{thiele1903theory}
T.~N. Thiele.
\newblock \emph{Theory of observations}.
\newblock Charles \& Edwin Layton, 1903.

\bibitem[Thiele(1909)]{thiele1909interpolationsrechnung}
T.~N. Thiele.
\newblock \emph{Interpolationsrechnung}.
\newblock Teubner, Leipzig, 1909.

\bibitem[Xiao et~al.(2019)Xiao, Deng, Lu, and Wu]{xiao2019optimal}
Y.~Xiao, H.~Deng, X.~Lu, and J.~Wu.
\newblock Optimal ballot-length in approval balloting-based multi-winner
  elections.
\newblock \emph{Decision Support Systems}, 118:\penalty0 1--9, 2019.

\bibitem[Zwicker(2016)]{Handbook-comsocintro}
W.~S. Zwicker.
\newblock Introduction to the theory of voting.
\newblock In F.~Brandt, V.~Conitzer, U.~Endriss, J.~Lang, and A.~D. Procaccia,
  editors, \emph{Handbook of Computational Social Choice}, pages 23--56.
  Cambridge University Press, New York, NY, USA, 1st edition, 2016.

\end{thebibliography}


\begin{thebibliography}{36}
\providecommand{\natexlab}[1]{#1}
\providecommand{\url}[1]{\texttt{#1}}
\expandafter\ifx\csname urlstyle\endcsname\relax
  \providecommand{\doi}[1]{doi: #1}\else
  \providecommand{\doi}{doi: \begingroup \urlstyle{rm}\Url}\fi

\bibitem[Arrow(1963)]{arrow1963}
K.~J. Arrow.
\newblock \emph{Social Choice and Individual Values}.
\newblock Wiley, New York, 2nd edition, 1963.

\bibitem[Aziz and Lee(2020)]{aziz2020expanding}
H.~Aziz and B.~E. Lee.
\newblock The expanding approvals rule: Improving proportional representation
  and monotonicity.
\newblock \emph{Social Choice and Welfare}, 54\penalty0 (1):\penalty0 1--45,
  2020.

\bibitem[Aziz and Monnot(2020)]{aziz2020computing}
H.~Aziz and J.~Monnot.
\newblock Computing and testing {P}areto optimal committees.
\newblock \emph{Autonomous Agents and Multi-Agent Systems}, 34\penalty0
  (1):\penalty0 1--20, 2020.

\bibitem[Barber{\`a} et~al.(2004)Barber{\`a}, Bossert, and
  Pattanaik]{barbera2004ranking}
S.~Barber{\`a}, W.~Bossert, and P.~K. Pattanaik.
\newblock Ranking sets of objects.
\newblock In \emph{Handbook of utility theory}, pages 893--977. Springer, 2004.

\bibitem[Campbell and Kelly(2015)]{campbell2015finer}
D.~E. Campbell and J.~S. Kelly.
\newblock The finer structure of resolute, neutral, and anonymous social choice
  correspondences.
\newblock \emph{Economics Letters}, 132:\penalty0 109--111, 2015.

\bibitem[Darmann(2013)]{Darmann13_condorcet_committees}
A.~Darmann.
\newblock How hard is it to tell which is a {C}ondorcet committee?
\newblock \emph{Mathematical Social Sciences}, 66\penalty0 (3):\penalty0
  282--292, 2013.

\bibitem[Dellis and Oak(2006)]{Dellis200647}
A.~Dellis and M.~P. Oak.
\newblock Approval voting with endogenous candidates.
\newblock \emph{Games and Economic Behavior}, 54\penalty0 (1):\penalty0 47--76,
  2006.

\bibitem[Eisermann(2001)]{eisermann2001pareto}
G.~Eisermann.
\newblock {P}areto, {V}ilfredo (1848--1923).
\newblock In \emph{International Encyclopedia of the Social \& Behavioral
  Sciences}, pages 11048--11051. Elsevier, 2001.

\bibitem[Elkind et~al.(2017)Elkind, Faliszewski, Skowron, and
  Slinko]{elk-fal-sko-sli:c:multiwinner-rules}
E.~Elkind, P.~Faliszewski, P.~Skowron, and A.~Slinko.
\newblock Properties of multiwinner voting rules.
\newblock \emph{Social Choice and Welfare}, 48\penalty0 (3):\penalty0 599--632,
  2017.

\bibitem[Faliszewski et~al.(2019)Faliszewski, Skowron, Slinko, and
  Talmon]{fal-sko-sli-tal-tal:j:hierarchy-committee}
P.~Faliszewski, P.~Skowron, A.~Slinko, and N.~Talmon.
\newblock Committee scoring rules: {Axiomatic} characterization and hierarchy.
\newblock \emph{ACM Transactions on Economics and Computation}, 6\penalty0
  (1):\penalty0 Article~3, 2019.

\bibitem[G{\"a}rdenfors(1979)]{gardenfors1979definitions}
P.~G{\"a}rdenfors.
\newblock On definitions of manipulation of social choice functions.
\newblock In J.-J. Laffont, editor, \emph{Aggregation and Revelation of
  Preferences}. North-Holland, 1979.

\bibitem[Gibbard(1973)]{gib:j:polsci:manipulation}
A.~Gibbard.
\newblock Manipulation of voting schemes.
\newblock \emph{Econometrica}, 41\penalty0 (4):\penalty0 587--601, 1973.

\bibitem[Janson(2016)]{Janson16arxiv}
S.~Janson.
\newblock Phragm{\'{e}}n's and {T}hiele's election methods.
\newblock \emph{CoRR}, abs/1611.08826, 2016.
\newblock URL \url{http://arxiv.org/abs/1611.08826}.

\bibitem[Kilgour and Marshall(2012)]{kil-mar:j:minimax-approval}
D.~M. Kilgour and E.~Marshall.
\newblock Approval balloting for fixed-size committees.
\newblock In D.~S. Felsenthal and M.~Machover, editors, \emph{Electoral
  Systems: Paradoxes, Assumptions, and Procedures}, Studies in Choice and
  Welfare, chapter~12, pages 305--326. Springer, 2012.

\bibitem[Kluiving et~al.(2020)Kluiving, Vries, Vrijbergen, Boixel, and
  Endriss]{KVVBE20:strategyproofness}
B.~Kluiving, A.~Vries, P.~Vrijbergen, A.~Boixel, and U.~Endriss.
\newblock Analysing irresolute multiwinner voting rules with approval ballots
  via {SAT} solving.
\newblock In \emph{Proceedings of the 24th European Conference on Artificial
  Intelligence (ECAI-2020)}, volume 325 of \emph{Frontiers in Artificial
  Intelligence and Applications}, pages 131--138. {IOS} Press, 2020.

\bibitem[Lackner and Skowron(2018)]{lac-sko:t:multiwinner-strategyproofness}
M.~Lackner and P.~Skowron.
\newblock Approval-based multi-winner rules and strategic voting.
\newblock In \emph{Proceedings of the 27th International Joint Conference on
  Artificial Intelligence (IJCAI-2018)}, pages 340--436, 2018.

\bibitem[Lackner and Skowron(2020)]{lac-sko2019}
M.~Lackner and P.~Skowron.
\newblock Utilitarian welfare and representation guarantees of approval-based
  multiwinner rules.
\newblock \emph{Artificial Intelligence}, 288:\penalty0 103366, 2020.

\bibitem[Lackner and Skowron(2021)]{jet-consistentabc}
M.~Lackner and P.~Skowron.
\newblock Consistent approval-based multi-winner rules.
\newblock \emph{Journal of Economic Theory}, 192:\penalty0 105173, 2021.

\bibitem[Laslier and {Van~der~Straeten}(2016)]{laslier2016StrategicVoting}
J.-F. Laslier and K.~{Van~der~Straeten}.
\newblock Strategic voting in multi-winners elections with approval balloting:
  a theory for large electorates.
\newblock \emph{Social Choice and Welfare}, 47\penalty0 (3):\penalty0 559--587,
  2016.

\bibitem[May(1952)]{mayAxiomatic1952}
K.~O. May.
\newblock A set of independent necessary and sufficient conditions for simple
  majority decision.
\newblock \emph{Econometrica}, 20\penalty0 (4):\penalty0 680--684, 1952.

\bibitem[Meir et~al.(2008)Meir, Procaccia, Rosenschein, and
  Zohar]{mei-pro-ros-zoh:multiwinner_strategic}
R.~Meir, A.~D. Procaccia, J.~S. Rosenschein, and A.~Zohar.
\newblock Complexity of strategic behavior in multi-winner elections.
\newblock \emph{Journal of Artificial Intelligence Research}, 33:\penalty0
  149--178, 2008.

\bibitem[Mora and Oliver(2015)]{MoOl15a}
X.~Mora and M.~Oliver.
\newblock Eleccions mitjan{\c c}ant el vot d'aprovaci{\'o}. {E}l m{\`e}tode de
  {P}hragm{\'e}n i algunes variants.
\newblock \emph{Butllet{\'\i} de la Societat Catalana de Matem{\`a}tiques},
  30\penalty0 (1):\penalty0 57--101, 2015.

\bibitem[Moulin(1988)]{moulinAxioms}
H.~Moulin.
\newblock \emph{Axioms of Cooperative Decision Making}.
\newblock Cambridge University Press, 1988.

\bibitem[Niemi(1984)]{niemi1984problem}
R.~G. Niemi.
\newblock The problem of strategic behavior under approval voting.
\newblock \emph{American Political Science Review}, 78\penalty0 (4):\penalty0
  952--958, 1984.

\bibitem[Ozkes and Sanver(2021)]{ozkes2021anonymous}
A.~I. Ozkes and M.~R. Sanver.
\newblock Anonymous, neutral, and resolute social choice revisited.
\newblock \emph{Social Choice and Welfare}, 57\penalty0 (1):\penalty0 97--113,
  2021.

\bibitem[Peters(2018)]{pet:prop-sp}
D.~Peters.
\newblock Proportionality and strategyproofness in multiwinner elections.
\newblock In \emph{Proceedings of the 17th International Conference on
  Autonomous Agents and Multiagent Systems (AAMAS-2018)}, pages 1549--1557,
  2018.

\bibitem[S{\'a}nchez-Fern{\'a}ndez and Fisteus(2019)]{sanchez2019monotonicity}
L.~S{\'a}nchez-Fern{\'a}ndez and J.~A. Fisteus.
\newblock Monotonicity axioms in approval-based multi-winner voting rules.
\newblock In \emph{Proceedings of the 18th International Conference on
  Autonomous Agents and Multiagent Systems (AAMAS-2019)}, pages 485--493.
  International Foundation for Autonomous Agents and Multiagent Systems, 2019.

\bibitem[Satterthwaite(1975)]{sat:j:polsci:manipulation}
M.~Satterthwaite.
\newblock Strategy-proofness and {Arrow's} conditions: Existence and
  correspondence theorems for voting procedures and social welfare functions.
\newblock \emph{Journal of Economic Theory}, 10\penalty0 (2):\penalty0
  187--217, 1975.

\bibitem[Scheuerman et~al.(2020)Scheuerman, Harman, Mattei, and
  Venable]{sch-har-mat-ven:heuristic-strategic}
J.~Scheuerman, J.~Harman, N.~Mattei, and K.~B. Venable.
\newblock Heuristic strategies in uncertain approval voting environments.
\newblock In \emph{Proceedings of the 19th International Conference on
  Autonomous Agents and Multiagent Systems (AAMAS-2020)}, pages 1993--1995,
  2020.

\bibitem[Scheuerman et~al.(2021)Scheuerman, Harman, Mattei, and
  Venable]{ScheuermanHMV21}
J.~Scheuerman, J.~L. Harman, N.~Mattei, and K.~B. Venable.
\newblock Modeling voters in multi-winner approval voting.
\newblock In \emph{Proceedings of the 35th Conference on Artificial
  Intelligence (AAAI-2021)}, pages 5709--5716, 2021.

\bibitem[Sinopoli et~al.(2006)Sinopoli, Dutta, and Laslier]{DeSinopoli2006}
F.~D. Sinopoli, B.~Dutta, and J.-F. Laslier.
\newblock Approval voting: three examples.
\newblock \emph{International Journal of Game Theory}, 35\penalty0
  (1):\penalty0 27--38, 2006.

\bibitem[Skowron et~al.(2017)Skowron, Lackner, Brill, Peters, and
  Elkind]{proprank}
P.~Skowron, M.~Lackner, M.~Brill, D.~Peters, and E.~Elkind.
\newblock Proportional rankings.
\newblock In \emph{Proceedings of the 26th International Joint Conference on
  Artificial Intelligence (IJCAI-2017)}, pages 409--415, 2017.

\bibitem[Skowron et~al.(2019)Skowron, Faliszewski, and
  Slinko]{skowron2019axiomatic}
P.~Skowron, P.~Faliszewski, and A.~Slinko.
\newblock Axiomatic characterization of committee scoring rules.
\newblock \emph{Journal of Economic Theory}, 180:\penalty0 244--273, 2019.

\bibitem[Smith(1973)]{smi:j:scoring-rules}
J.~Smith.
\newblock Aggregation of preferences with variable electorate.
\newblock \emph{Econometrica}, 41\penalty0 (6):\penalty0 1027--1041, 1973.

\bibitem[Taylor(2005)]{taylor2005social}
A.~D. Taylor.
\newblock \emph{Social choice and the mathematics of manipulation}.
\newblock Cambridge University Press, 2005.

\bibitem[Young(1974)]{young74}
H.~P. Young.
\newblock A note on preference aggregation.
\newblock \emph{Econometrica}, 42\penalty0 (6):\penalty0 1129--1131, 1974.

\end{thebibliography}


\begin{thebibliography}{74}
\providecommand{\natexlab}[1]{#1}
\providecommand{\url}[1]{\texttt{#1}}
\expandafter\ifx\csname urlstyle\endcsname\relax
  \providecommand{\doi}[1]{doi: #1}\else
  \providecommand{\doi}{doi: \begingroup \urlstyle{rm}\Url}\fi

\bibitem[Aziz(2019)]{azi:committees-soft-constraints}
H.~Aziz.
\newblock A rule for committee selection with soft diversity constraints.
\newblock \emph{Group Decision and Negotiation}, 28:\penalty0 1193–--1200,
  2019.

\bibitem[Aziz et~al.(2015)Aziz, Gaspers, Gudmundsson, Mackenzie, Mattei, and
  Walsh]{azi-gas-gud-mac-mat-wal:c:multiwinner-approval}
H.~Aziz, S.~Gaspers, J.~Gudmundsson, S.~Mackenzie, N.~Mattei, and T.~Walsh.
\newblock Computational aspects of multi-winner approval voting.
\newblock In \emph{Proceedings of the 14th International Conference on
  Autonomous Agents and Multiagent Systems (AAMAS-2015)}, pages 107--115, 2015.

\bibitem[Aziz et~al.(2017)Aziz, Brill, Conitzer, Elkind, Freeman, and
  Walsh]{justifiedRepresentation}
H.~Aziz, M.~Brill, V.~Conitzer, E.~Elkind, R.~Freeman, and T.~Walsh.
\newblock Justified representation in approval-based committee voting.
\newblock \emph{Social Choice and Welfare}, 48\penalty0 (2):\penalty0 461--485,
  2017.

\bibitem[Aziz et~al.(2018)Aziz, Elkind, Huang, Lackner,
  S{\'a}nchez-Fern{\'a}ndez, and Skowron]{AEHLSS18}
H.~Aziz, E.~Elkind, S.~Huang, M.~Lackner, L.~S{\'a}nchez-Fern{\'a}ndez, and
  P.~Skowron.
\newblock On the complexity of extended and proportional justified
  representation.
\newblock In \emph{Proceedings of the 32nd Conference on Artificial
  Intelligence (AAAI-2018)}, pages 902--909, 2018.

\bibitem[Balinski and Young(1982)]{BaYo82a}
M.~Balinski and H.~P. Young.
\newblock \emph{Fair Representation: {M}eeting the Ideal of One Man, One Vote}.
\newblock Yale University Press, 1982.
\newblock (2nd Edition by Brookings Institution Press, 2001).

\bibitem[Barrot et~al.(2017)Barrot, Lang, and
  Yokoo]{bar-lan-mak:multiwinner_manipulation}
N.~Barrot, J.~Lang, and M.~Yokoo.
\newblock Manipulation of {H}amming-based approval voting for multiple
  referenda and committee elections.
\newblock In \emph{Proceedings of the 16th International Conference on
  Autonomous Agents and Multiagent Systems (AAMAS-2017)}, pages 597--605, 2017.

\bibitem[Bei et~al.(2020)Bei, Liu, Poon, and
  Wang]{bei-liu-keu-wan:diversity-constraints}
X.~Bei, S.~Liu, C.~K. Poon, and H.~Wang.
\newblock Candidate selections with proportional fairness constraints.
\newblock In \emph{Proceedings of the 19th International Conference on
  Autonomous Agents and Multiagent Systems (AAMAS-2020)}, pages 150--158, 2020.

\bibitem[Benad{\`{e}} et~al.(2019)Benad{\`{e}}, G{\"{o}}lz, and
  Procaccia]{ben-gol-pro:stratification}
G.~Benad{\`{e}}, P.~G{\"{o}}lz, and A.~Procaccia.
\newblock No stratification without representation.
\newblock In \emph{Proceedings of the 2019 ACM Conference on Economics and
  Computation (ACM-EC-2019)}, pages 281--314, 2019.

\bibitem[Botan(2021)]{botan2021manipulability}
S.~Botan.
\newblock Manipulability of {T}hiele methods on party-list profiles.
\newblock In \emph{Proceedings of the 20th International Conference on
  Autonomous Agents and Multiagent Systems (AAMAS-2021)}, pages 223--231, 2021.

\bibitem[Bredereck et~al.(2018)Bredereck, Faliszewski, Igarashi, Lackner, and
  Skowron]{conf/aaai/BredereckFILS18}
R.~Bredereck, P.~Faliszewski, A.~Igarashi, M.~Lackner, and P.~Skowron.
\newblock Multiwinner elections with diversity constraints.
\newblock In \emph{Proceedings of the 32nd Conference on Artificial
  Intelligence (AAAI-2018)}, pages 933--940, 2018.

\bibitem[Bredereck et~al.(2019)Bredereck, Faliszewski, Kaczmarczyk, and
  Niedermeier]{bre-fal-kac-nie2019:experimental_ejr}
R.~Bredereck, P.~Faliszewski, A.~Kaczmarczyk, and R.~Niedermeier.
\newblock An experimental view on committees providing justified
  representation.
\newblock In \emph{Proceedings of the 28th International Joint Conference on
  Artificial Intelligence (IJCAI-2019)}, pages 109--115, 2019.

\bibitem[Brill et~al.(2017)Brill, Freeman, Janson, and
  Lackner]{aaai/BrillFJL17-phragmen}
M.~Brill, R.~Freeman, S.~Janson, and M.~Lackner.
\newblock Phragm{\'e}n's voting methods and justified representation.
\newblock In \emph{Proceedings of the 31st Conference on Artificial
  Intelligence (AAAI-2017)}, pages 406--413, 2017.
\newblock Extended version at \url{https://arxiv.org/abs/2102.12305}.

\bibitem[Brill et~al.(2018)Brill, Laslier, and
  Skowron]{bri-las-sko:c:apportionment}
M.~Brill, J.-F. Laslier, and P.~Skowron.
\newblock Multiwinner approval rules as apportionment methods.
\newblock \emph{Journal of Theoretical Politics}, 30\penalty0 (3):\penalty0
  358--382, 2018.

\bibitem[Brill et~al.(2019)Brill, Faliszewski, Sommer, and
  Talmon]{bri-fal-som-tal:balanced-cc}
M.~Brill, P.~Faliszewski, F.~Sommer, and N.~Talmon.
\newblock Approximation algorithms for {BalancedCC} multiwinner rules.
\newblock In \emph{Proceedings of the 18th International Conference on
  Autonomous Agents and Multiagent Systems (AAMAS-2019)}, pages 494--502, 2019.

\bibitem[Brill et~al.(2020)Brill, G{\"{o}}lz, Peters, Schmidt{-}Kraepelin, and
  Wilker]{BGPSW19}
M.~Brill, P.~G{\"{o}}lz, D.~Peters, U.~Schmidt{-}Kraepelin, and K.~Wilker.
\newblock Approval-based apportionment.
\newblock In \emph{Proceedings of the 34th Conference on Artificial
  Intelligence (AAAI-2020)}, pages 1854--1861, 2020.
\newblock Extended version at \url{http://arxiv.org/abs/1911.08365}.

\bibitem[Brill et~al.(2022)Brill, Israel, Micha, and
  Peters]{individual-representation}
M.~Brill, J.~Israel, E.~Micha, and J.~Peters.
\newblock Individual representation in approval-based committee voting.
\newblock In \emph{Proceedings of the 36th Conference on Artificial
  Intelligence (AAAI-2022)}, pages 4892--4899, 2022.

\bibitem[Bummel(2010)]{Bumm10a}
A.~Bummel.
\newblock \emph{The Composition of a Parliamentary Assembly at the United
  Nations}.
\newblock Committee for a Democratic U.N., 2010.

\bibitem[Celis et~al.(2018)Celis, Huang, and Vishnoi]{conf/ijcai/CelisHV18}
L.~E. Celis, L.~Huang, and N.~K. Vishnoi.
\newblock Multiwinner voting with fairness constraints.
\newblock In \emph{Proceedings of the 27th International Joint Conference on
  Artificial Intelligence (IJCAI-2018)}, pages 144--151, 2018.

\bibitem[Cembrano et~al.(2022)Cembrano, Correa, and
  Verdugo]{cembrano2022multidimensional}
J.~Cembrano, J.~Correa, and V.~Verdugo.
\newblock Multidimensional political apportionment.
\newblock \emph{Proceedings of the National Academy of Sciences}, 119\penalty0
  (15):\penalty0 e2109305119, 2022.

\bibitem[Cevallos and Stewart(2021)]{cevallos2020verifiably}
A.~Cevallos and A.~Stewart.
\newblock A verifiably secure and proportional committee election rule.
\newblock In \emph{Proceedings of the 3rd ACM Conference on Advances in
  Financial Technologies}, pages 29--42, 2021.

\bibitem[Chalkiadakis et~al.(2011)Chalkiadakis, Elkind, and
  Wooldridge]{chalkiadakis2011computational}
G.~Chalkiadakis, E.~Elkind, and M.~Wooldridge.
\newblock Computational aspects of cooperative game theory.
\newblock \emph{Synthesis Lectures on Artificial Intelligence and Machine
  Learning}, 5\penalty0 (6):\penalty0 1--168, 2011.

\bibitem[Chamberlin and C{ourant}(1983)]{ccElection}
B.~Chamberlin and P.~C{ourant}.
\newblock Representative deliberations and representative decisions:
  {P}roportional representation and the {B}orda rule.
\newblock \emph{American Political Science Review}, 77\penalty0 (3):\penalty0
  718--733, 1983.

\bibitem[Chen et~al.(2019)Chen, Fain, Lyu, and
  Munagala]{DBLP:conf/icml/ChenFLM19}
X.~Chen, B.~Fain, L.~Lyu, and K.~Munagala.
\newblock Proportionally fair clustering.
\newblock In \emph{Proceedings of the 36th International Conference on Machine
  Learning (ICML 2019)}, volume~97, pages 1032--1041, 2019.

\bibitem[Cheng et~al.(2019)Cheng, Jiang, Munagala, and Wang]{cheng2019group}
Y.~Cheng, Z.~Jiang, K.~Munagala, and K.~Wang.
\newblock Group fairness in committee selection.
\newblock In \emph{Proceedings of the 2019 ACM Conference on Economics and
  Computation (ACM-EC-2019)}, pages 263--279. ACM, 2019.

\bibitem[D'Hondt(1878)]{d1878question}
V.~D'Hondt.
\newblock \emph{Question {\'e}lectorale. La repr{\'e}sentation proportionnelle
  des partis}.
\newblock Bruylant, Bruxelles, 1878.

\bibitem[D'Hondt(1885)]{d1885expose}
V.~D'Hondt.
\newblock \emph{Expos{\'e} du syst{\`e}me pratique de repr{\'e}sentation
  proportionnelle adopt{\'e} par le Comit{\'e} de l'Associaton R{\'e}formiste
  Belge}.
\newblock Van der Haegen, Ghent, 1885.

\bibitem[Diss and Steffen(2017)]{diss2017distribution}
M.~Diss and F.~Steffen.
\newblock {The Distribution of Power in the Lebanese Parliament Revisited}.
\newblock Technical report, HAL, 2017.
\newblock Working Paper.

\bibitem[Duddy(2014)]{duddy2014electing}
C.~Duddy.
\newblock Electing a representative committee by approval ballot: An
  impossibility result.
\newblock \emph{Economics Letters}, 124\penalty0 (1):\penalty0 14--16, 2014.

\bibitem[Elkind et~al.(2017)Elkind, Faliszewski, Laslier, Skowron, Slinko, and
  Talmon]{2dpictures}
E.~Elkind, P.~Faliszewski, J.~Laslier, P.~Skowron, A.~Slinko, and N.~Talmon.
\newblock What do multiwinner voting rules do? {A}n experiment over the
  two-dimensional euclidean domain.
\newblock In \emph{Proceedings of the 31st Conference on Artificial
  Intelligence (AAAI-2017)}, pages 494--501, 2017.

\bibitem[Elkind et~al.(2021)Elkind, Faliszewski, Igarashi, Manurangsi,
  Schmidt{-}Kraepelin, and Suksompong]{corr/abs-2112-05994}
E.~Elkind, P.~Faliszewski, A.~Igarashi, P.~Manurangsi, U.~Schmidt{-}Kraepelin,
  and W.~Suksompong.
\newblock The price of justified representation.
\newblock \emph{CoRR}, abs/2112.05994, 2021.
\newblock URL \url{https://arxiv.org/abs/2112.05994}.

\bibitem[Fain et~al.(2018)Fain, Munagala, and Shah]{fain2018fair}
B.~Fain, K.~Munagala, and N.~Shah.
\newblock Fair allocation of indivisible public goods.
\newblock In \emph{Proceedings of the 2018 ACM Conference on Economics and
  Computation (ACM-EC-2018)}, pages 575--592. ACM, 2018.
\newblock Extended version at \url{http://arxiv.org/abs/1805.03164}.

\bibitem[Faliszewski and Talmon(2018)]{fal-tal:balancing-cc}
P.~Faliszewski and N.~Talmon.
\newblock Between proportionality and diversity: {B}alancing district sizes
  under the {Chamberlin-Courant} rule.
\newblock In \emph{Proceedings of the 17th International Conference on
  Autonomous Agents and Multiagent Systems (AAMAS-2018)}, pages 14--22, 2018.

\bibitem[Faliszewski et~al.(2017{\natexlab{a}})Faliszewski, Skowron, Slinko,
  and Talmon]{FSST-trends}
P.~Faliszewski, P.~Skowron, A.~Slinko, and N.~Talmon.
\newblock Multiwinner voting: {A} new challenge for social choice theory.
\newblock In U.~Endriss, editor, \emph{Trends in Computational Social Choice},
  chapter~2, pages 27--47. AI Access, 2017{\natexlab{a}}.

\bibitem[Faliszewski et~al.(2017{\natexlab{b}})Faliszewski, Skowron, Slinko,
  and Talmon]{FaliszewskiSST17}
P.~Faliszewski, P.~Skowron, A.~Slinko, and N.~Talmon.
\newblock Multiwinner rules on paths from $k$-{B}orda to
  {C}hamberlin-{C}ourant.
\newblock In \emph{Proceedings of the 26th International Joint Conference on
  Artificial Intelligence (IJCAI-2017)}, pages 192--198. ijcai.org,
  2017{\natexlab{b}}.

\bibitem[G{\"a}rdenfors(1979)]{gardenfors1979definitions}
P.~G{\"a}rdenfors.
\newblock On definitions of manipulation of social choice functions.
\newblock In J.-J. Laffont, editor, \emph{Aggregation and Revelation of
  Preferences}. North-Holland, 1979.

\bibitem[Godziszewski et~al.(2021)Godziszewski, Batko, Skowron, and
  Faliszewski]{god-bat-sko-fal:c:2d-abc}
M.~Godziszewski, P.~Batko, P.~Skowron, and P.~Faliszewski.
\newblock An analysis of approval-based committee rules for {2D}-{E}uclidean
  elections.
\newblock In \emph{Proceedings of the 35th Conference on Artificial
  Intelligence (AAAI-2021)}, pages 5448--5455, 2021.

\bibitem[Janson(2016)]{Janson16arxiv}
S.~Janson.
\newblock Phragm{\'{e}}n's and {T}hiele's election methods.
\newblock \emph{CoRR}, abs/1611.08826, 2016.
\newblock URL \url{http://arxiv.org/abs/1611.08826}.

\bibitem[Jaworski and Skowron(2022)]{jaw-sko:phragmen-degressive-regressive}
M.~Jaworski and P.~Skowron.
\newblock Phragm\'{e}n rules for degressive and regressive proportionality.
\newblock In \emph{Proceedings of the 31st International Joint Conference on
  Artificial Intelligence (IJCAI-2022)}, pages 328--334, 2022.

\bibitem[Jiang et~al.(2020)Jiang, Munagala, and Wang]{jiang2019approx}
Z.~Jiang, K.~Munagala, and K.~Wang.
\newblock Approximately stable committee selection.
\newblock In \emph{Proceedings of the 52nd Symposium on Theory of Computing
  (STOC-2020)}, pages 463--472. {ACM}, 2020.

\bibitem[Kagita et~al.(2021)Kagita, Pujari, Padmanabhan, Aziz, and
  Kumar]{kagita2021committee}
V.~R. Kagita, A.~K. Pujari, V.~Padmanabhan, H.~Aziz, and V.~Kumar.
\newblock Committee selection using attribute approvals.
\newblock In \emph{Proceedings of the 20th International Conference on
  Autonomous Agents and Multiagent Systems (AAMAS-2021)}, pages 683--691, 2021.

\bibitem[Kluiving et~al.(2020)Kluiving, Vries, Vrijbergen, Boixel, and
  Endriss]{KVVBE20:strategyproofness}
B.~Kluiving, A.~Vries, P.~Vrijbergen, A.~Boixel, and U.~Endriss.
\newblock Analysing irresolute multiwinner voting rules with approval ballots
  via {SAT} solving.
\newblock In \emph{Proceedings of the 24th European Conference on Artificial
  Intelligence (ECAI-2020)}, volume 325 of \emph{Frontiers in Artificial
  Intelligence and Applications}, pages 131--138. {IOS} Press, 2020.

\bibitem[Koriyama et~al.(2013)Koriyama, Laslier, Mac\'e, and
  Treibich]{RePEc:ucp:jpolec:doi:10.1086/670380}
Y.~Koriyama, J.-F. Laslier, A.~Mac\'e, and R.~Treibich.
\newblock {Optimal Apportionment}.
\newblock \emph{Journal of Political Economy}, 121\penalty0 (3):\penalty0
  584--608, 2013.

\bibitem[Lackner and Skowron(2018)]{lac-sko:t:multiwinner-strategyproofness}
M.~Lackner and P.~Skowron.
\newblock Approval-based multi-winner rules and strategic voting.
\newblock In \emph{Proceedings of the 27th International Joint Conference on
  Artificial Intelligence (IJCAI-2018)}, pages 340--436, 2018.

\bibitem[Lackner and Skowron(2020)]{lac-sko2019}
M.~Lackner and P.~Skowron.
\newblock Utilitarian welfare and representation guarantees of approval-based
  multiwinner rules.
\newblock \emph{Artificial Intelligence}, 288:\penalty0 103366, 2020.

\bibitem[Lackner and Skowron(2021)]{jet-consistentabc}
M.~Lackner and P.~Skowron.
\newblock Consistent approval-based multi-winner rules.
\newblock \emph{Journal of Economic Theory}, 192:\penalty0 105173, 2021.

\bibitem[Lang and Skowron(2018)]{LS16:multi-attribute}
J.~Lang and P.~Skowron.
\newblock Multi-attribute proportional representation.
\newblock \emph{Artificial Intelligence}, 263:\penalty0 74--106, 2018.

\bibitem[Lari et~al.(2014)Lari, Ricca, and Scozzari]{LariRS14}
I.~Lari, F.~Ricca, and A.~Scozzari.
\newblock Bidimensional allocation of seats via zero-one matrices with given
  line sums.
\newblock \emph{Annals {OR}}, 215\penalty0 (1):\penalty0 165--181, 2014.

\bibitem[Laslier(2012)]{laslier2012WhyNotProportional}
J.-F. Laslier.
\newblock Why not proportional?
\newblock \emph{Mathematical Social Sciences}, 63\penalty0 (2):\penalty0
  90--93, 2012.

\bibitem[Mac\'{e} and Treibich(2012)]{MT12}
A.~Mac\'{e} and R.~Treibich.
\newblock Computing the optimal weights in a utilitarian model of
  apportionment.
\newblock \emph{Mathematical Social Sciences}, 63\penalty0 (2):\penalty0
  141--151, 2012.

\bibitem[Micha and Shah(2020)]{DBLP:conf/icalp/Micha020}
E.~Micha and N.~Shah.
\newblock Proportionally fair clustering revisited.
\newblock In \emph{Proceedings of the 47th International Colloquium on
  Automata, Languages, and Programming (ICALP 2020)}, pages 85:1--85:16, 2020.

\bibitem[Monroe(1995)]{monroeElection}
B.~Monroe.
\newblock Fully proportional representation.
\newblock \emph{American Political Science Review}, 89\penalty0 (4):\penalty0
  925--940, 1995.

\bibitem[Munagala et~al.(2022)Munagala, Shen, Wang, and
  Wang]{mun-she-wan-wan:approximate-core}
K.~Munagala, Y.~Shen, K.~Wang, and Z.~Wang.
\newblock Approximate core for committee selection via multilinear extension
  and market clearing.
\newblock In \emph{Proceedings of the 33rd ACM-SIAM Symposium on Discrete
  Algorithms (SODA-2022)}, pages 2229--2252, 2022.

\bibitem[Osborne and Rubinstein(1994)]{RePEc:mtp:titles:0262650401}
J.~M. Osborne and A.~Rubinstein.
\newblock \emph{A Course in Game Theory}, volume~1 of \emph{MIT Press Books}.
\newblock The MIT Press, 1994.

\bibitem[Penrose(1946)]{Penr46a}
L.~S. Penrose.
\newblock The elementary statistics of majority voting.
\newblock \emph{Journal of the Royal Statistical Society}, 109\penalty0
  (1):\penalty0 53--57, 1946.

\bibitem[Peters(2018)]{pet:prop-sp}
D.~Peters.
\newblock Proportionality and strategyproofness in multiwinner elections.
\newblock In \emph{Proceedings of the 17th International Conference on
  Autonomous Agents and Multiagent Systems (AAMAS-2018)}, pages 1549--1557,
  2018.

\bibitem[Peters(2019)]{dominikthesis}
D.~Peters.
\newblock \emph{Fair Division of the Commons}.
\newblock PhD thesis, University of Oxford, 9 2019.

\bibitem[Peters and Skowron(2020)]{pet-sko:laminar}
D.~Peters and P.~Skowron.
\newblock Proportionality and the limits of welfarism.
\newblock In \emph{Proceedings of the 2020 ACM Conference on Economics and
  Computation (ACM-EC-2020)}, pages 793--794, 2020.
\newblock Extended version at \url{https://arxiv.org/abs/1911.11747}.

\bibitem[Peters et~al.(2021{\natexlab{a}})Peters, Pierczynski, and
  Skowron]{pet-pie-sko:c:participatory-budgeting-cardinal}
D.~Peters, G.~Pierczynski, and P.~Skowron.
\newblock Proportional participatory budgeting with additive utilities.
\newblock In \emph{Proceedings of the Thirty-fifth Conference on Neural
  Information Processing Systems (NeurIPS-2021)}, pages 12726--12737,
  2021{\natexlab{a}}.

\bibitem[Peters et~al.(2021{\natexlab{b}})Peters, Pierczyski, Shah, and
  Skowron]{pet-pie-sha-sko:c:stable-priceability}
D.~Peters, G.~Pierczyski, N.~Shah, and P.~Skowron.
\newblock Market-based explanations of collective decisions.
\newblock In \emph{Proceedings of the 35th Conference on Artificial
  Intelligence (AAAI-2021)}, pages 5656--5663, 2021{\natexlab{b}}.

\bibitem[Phragm{\'e}n(1895)]{Phra95a}
E.~Phragm{\'e}n.
\newblock \emph{Proportionella val. En valteknisk studie}.
\newblock Svenska sp{\"o}rsm{\aa}l 25. Lars H{\"o}kersbergs f{\"o}rlag,
  Stockholm, 1895.

\bibitem[Pierczynski and Skowron(2021)]{abs-2108-01987}
G.~Pierczynski and P.~Skowron.
\newblock Core-stable committees under restricted domains.
\newblock \emph{CoRR}, abs/2108.01987, 2021.
\newblock URL \url{https://arxiv.org/abs/2108.01987}.

\bibitem[Pukelsheim(2014)]{Puke14a:biApportionment}
F.~Pukelsheim.
\newblock \emph{Proportional Representation: Apportionment Methods and Their
  Applications}, chapter Representing Districts and Parties: Double
  Proportionality.
\newblock Springer, 2014.

\bibitem[Pukelsheim(2017)]{Puke17}
F.~Pukelsheim.
\newblock \emph{Proportional Representation: Apportionment Methods and Their
  Applications}.
\newblock Springer International Publishing, 2nd edition, 2017.

\bibitem[Pukelsheim et~al.(2012)Pukelsheim, Ricca, Simeone, Scozzari, and
  Serafini]{PukelsheimRSSS12}
F.~Pukelsheim, F.~Ricca, B.~Simeone, A.~Scozzari, and P.~Serafini.
\newblock Network flow methods for electoral systems.
\newblock \emph{Networks}, 59\penalty0 (1):\penalty0 73--88, 2012.

\bibitem[Rose(2013)]{rose2013RepresentingEuropeans}
R.~Rose.
\newblock \emph{Representing Europeans: a pragmatic approach}.
\newblock Oxford University Press, Oxford, 2013.

\bibitem[Sainte-Lagu{\"e}(1910)]{sainte1910representation}
A.~Sainte-Lagu{\"e}.
\newblock La repr{\'e}sentation proportionnelle et la m{\'e}thode des moindres
  carr{\'e}s.
\newblock In \emph{Annales scientifiques de l'{\'e}cole Normale
  Sup{\'e}rieure}, volume~27, pages 529--542, 1910.

\bibitem[S{\'a}nchez-Fern{\'a}ndez et~al.(2017)S{\'a}nchez-Fern{\'a}ndez,
  Elkind, Lackner, Fern{\'a}ndez, Fisteus, {Basanta Val}, and Skowron]{pjr17}
L.~S{\'a}nchez-Fern{\'a}ndez, E.~Elkind, M.~Lackner, N.~Fern{\'a}ndez, J.~A.
  Fisteus, P.~{Basanta Val}, and P.~Skowron.
\newblock Proportional justified representation.
\newblock In \emph{Proceedings of the 31st Conference on Artificial
  Intelligence (AAAI-2017)}, pages 670--676, 2017.

\bibitem[Serafini and Simeone(2012)]{SerafiniS12}
P.~Serafini and B.~Simeone.
\newblock Parametric maximum flow methods for minimax approximation of target
  quotas in biproportional apportionment.
\newblock \emph{Networks}, 59\penalty0 (2):\penalty0 191--208, 2012.

\bibitem[Skowron(2021)]{skowron:prop-degree}
P.~Skowron.
\newblock Proportionality degree of multiwinner rules.
\newblock In \emph{Proceedings of the 2021 ACM Conference on Economics and
  Computation (ACM-EC-2021)}, pages 820--840, 2021.
\newblock Extended version at \url{https://arxiv.org/abs/1810.08799}.

\bibitem[Skowron et~al.(2017)Skowron, Lackner, Brill, Peters, and
  Elkind]{proprank}
P.~Skowron, M.~Lackner, M.~Brill, D.~Peters, and E.~Elkind.
\newblock Proportional rankings.
\newblock In \emph{Proceedings of the 26th International Joint Conference on
  Artificial Intelligence (IJCAI-2017)}, pages 409--415, 2017.

\bibitem[S{\l}omczy{\'n}ski and {\.Z}yczkowski(2006)]{SlZy06a}
W.~S{\l}omczy{\'n}ski and K.~{\.Z}yczkowski.
\newblock Penrose voting system and optimal quota.
\newblock \emph{Acta Physica Polonica}, B~37:\penalty0 3133--3143, 2006.

\bibitem[S{\l}omczy{\'n}ski and
  {\.Z}yczkowski(2017)]{slomczynski2017degressive}
W.~S{\l}omczy{\'n}ski and K.~{\.Z}yczkowski.
\newblock Degressive proportionality in the european union.
\newblock In \emph{The Composition of the European Parliament. Workshop 30
  January 2017}, pages 37--48, 2017.

\bibitem[Subiza and Peris(2017)]{subiza2017representative}
B.~Subiza and J.~E. Peris.
\newblock A representative committee by approval balloting.
\newblock \emph{Group Decision and Negotiation}, 26\penalty0 (5):\penalty0
  1029--1040, 2017.

\bibitem[Thiele(1895)]{Thie95a}
T.~N. Thiele.
\newblock Om flerfoldsvalg.
\newblock In \emph{Oversigt over det Kongelige Danske Videnskabernes Selskabs
  Forhandlinger}, pages 415--441. 1895.

\end{thebibliography}


\begin{thebibliography}{64}
\providecommand{\natexlab}[1]{#1}
\providecommand{\url}[1]{\texttt{#1}}
\expandafter\ifx\csname urlstyle\endcsname\relax
  \providecommand{\doi}[1]{doi: #1}\else
  \providecommand{\doi}{doi: \begingroup \urlstyle{rm}\Url}\fi

\bibitem[Aziz et~al.(2015)Aziz, Gaspers, Gudmundsson, Mackenzie, Mattei, and
  Walsh]{azi-gas-gud-mac-mat-wal:c:multiwinner-approval}
H.~Aziz, S.~Gaspers, J.~Gudmundsson, S.~Mackenzie, N.~Mattei, and T.~Walsh.
\newblock Computational aspects of multi-winner approval voting.
\newblock In \emph{Proceedings of the 14th International Conference on
  Autonomous Agents and Multiagent Systems (AAMAS-2015)}, pages 107--115, 2015.

\bibitem[Aziz et~al.(2017)Aziz, Brill, Conitzer, Elkind, Freeman, and
  Walsh]{justifiedRepresentation}
H.~Aziz, M.~Brill, V.~Conitzer, E.~Elkind, R.~Freeman, and T.~Walsh.
\newblock Justified representation in approval-based committee voting.
\newblock \emph{Social Choice and Welfare}, 48\penalty0 (2):\penalty0 461--485,
  2017.

\bibitem[Aziz et~al.(2018)Aziz, Elkind, Huang, Lackner,
  S{\'a}nchez-Fern{\'a}ndez, and Skowron]{AEHLSS18}
H.~Aziz, E.~Elkind, S.~Huang, M.~Lackner, L.~S{\'a}nchez-Fern{\'a}ndez, and
  P.~Skowron.
\newblock On the complexity of extended and proportional justified
  representation.
\newblock In \emph{Proceedings of the 32nd Conference on Artificial
  Intelligence (AAAI-2018)}, pages 902--909, 2018.

\bibitem[Barman et~al.(2020)Barman, Fawzi, Ghoshal, and
  G{\"u}rp{\i}nar]{bar-faw-gho-gur:ell-best-approx}
S.~Barman, O.~Fawzi, S.~Ghoshal, and E.~G{\"u}rp{\i}nar.
\newblock Tight approximation bounds for maximum multi-coverage.
\newblock In \emph{Proceedings of the 2020 International Conference on Integer
  Programming and Combinatorial Optimization (IPCO-2020)}, pages 66--77, 2020.

\bibitem[Barrot et~al.(2013)Barrot, Gourv{\`e}s, Lang, Monnot, and
  Ries]{bar-gou-lan-mon-ries}
N.~Barrot, L.~Gourv{\`e}s, J.~Lang, J.~Monnot, and B.~Ries.
\newblock Possible winners in approval voting.
\newblock In \emph{Proceedings of the 3rd International Conference on
  Algorithmic Decision Theory (ADT-2013)}, pages 57--70, 2013.

\bibitem[{{Bartholdi}} et~al.(1989){{Bartholdi}}, Tovey, and
  Trick]{bar-tov-tri:j:manipulating}
J.~{{Bartholdi}}, III, C.~Tovey, and M.~Trick.
\newblock The computational difficulty of manipulating an election.
\newblock \emph{Social Choice and Welfare}, 6\penalty0 (3):\penalty0 227--241,
  1989.

\bibitem[Baumeister et~al.(2010)Baumeister, Erd\'{e}lyi, Hemaspaandra,
  Hemaspaandra, and
  Rothe]{bau-erd-hem-hem-rot:b:computational-apects-of-approval-voting}
D.~Baumeister, G.~Erd\'{e}lyi, E.~Hemaspaandra, L.~Hemaspaandra, and J.~Rothe.
\newblock Computational aspects of approval voting.
\newblock In J.~F. Laslier and M.~R. Sanver, editors, \emph{Handbook of
  Approval Voting}, pages 199--251. Springer, 2010.

\bibitem[Baumeister et~al.(2015)Baumeister, Dennisen, and
  Rey]{bau-den-rey:manipulation_multiwinner}
D.~Baumeister, S.~Dennisen, and L.~Rey.
\newblock Winner determination and manipulation in minisum and minimax
  committee elections.
\newblock In \emph{Proceedings of the 4th International Conference on
  Algorithmic Decision Theory (ADT-2015)}, pages 469--485, 2015.

\bibitem[Betzler et~al.(2013)Betzler, Slinko, and
  Uhlmann]{fullyProportionalRepr}
N.~Betzler, A.~Slinko, and J.~Uhlmann.
\newblock On the computation of fully proportional representation.
\newblock \emph{Journal of Artificial Intelligence Research}, 47:\penalty0
  475--519, 2013.

\bibitem[Brams(1990)]{brams1990constrained}
S.~J. Brams.
\newblock Constrained approval voting: A voting system to elect a governing
  board.
\newblock \emph{Interfaces}, 20\penalty0 (5):\penalty0 67--80, 1990.

\bibitem[Bredereck et~al.(2017)Bredereck, Kaczmarczyk, and
  Niedermeier]{bre-kac-nie:strategic_multiwinner}
R.~Bredereck, A.~Kaczmarczyk, and R.~Niedermeier.
\newblock On coalitional manipulation for multiwinner elections:
  {S}hortlisting.
\newblock In \emph{Proceedings of the 26th International Joint Conference on
  Artificial Intelligence (IJCAI-2017)}, pages 887--893, 2017.

\bibitem[Bredereck et~al.(2018)Bredereck, Faliszewski, Igarashi, Lackner, and
  Skowron]{conf/aaai/BredereckFILS18}
R.~Bredereck, P.~Faliszewski, A.~Igarashi, M.~Lackner, and P.~Skowron.
\newblock Multiwinner elections with diversity constraints.
\newblock In \emph{Proceedings of the 32nd Conference on Artificial
  Intelligence (AAAI-2018)}, pages 933--940, 2018.

\bibitem[Bredereck et~al.(2020)Bredereck, Faliszewski, Kaczmarczyk, Knop, and
  Niedermeier]{bredereck2020parameterized}
R.~Bredereck, P.~Faliszewski, A.~Kaczmarczyk, D.~Knop, and R.~Niedermeier.
\newblock Parameterized algorithms for finding a collective set of items.
\newblock In \emph{Proceedings of the 34th Conference on Artificial
  Intelligence (AAAI-2020)}, pages 1838--1845, 2020.

\bibitem[Bredereck et~al.(2021)Bredereck, Faliszewski, Kaczmarczyk,
  Niedermeier, Skowron, and Talmon]{bre-fal-kac-nie-sko-tal:robustness}
R.~Bredereck, P.~Faliszewski, A.~Kaczmarczyk, R.~Niedermeier, P.~Skowron, and
  N.~Talmon.
\newblock Robustness among multiwinner voting rules.
\newblock \emph{Artificial Intelligence}, 290:\penalty0 103403, 2021.

\bibitem[Brill et~al.(2017)Brill, Freeman, Janson, and
  Lackner]{aaai/BrillFJL17-phragmen}
M.~Brill, R.~Freeman, S.~Janson, and M.~Lackner.
\newblock Phragm{\'e}n's voting methods and justified representation.
\newblock In \emph{Proceedings of the 31st Conference on Artificial
  Intelligence (AAAI-2017)}, pages 406--413, 2017.
\newblock Extended version at \url{https://arxiv.org/abs/2102.12305}.

\bibitem[Byrka and Sornat(2014)]{BS14:minimax_approximation}
J.~Byrka and K.~Sornat.
\newblock {PTAS} for minimax approval voting.
\newblock In \emph{Proceedings of the 10th International Conference on Web and
  Internet Economics (WINE-2014)}, pages 203--217, 2014.

\bibitem[Byrka et~al.(2018)Byrka, Skowron, and Sornat]{ByrkaSS18}
J.~Byrka, P.~Skowron, and K.~Sornat.
\newblock Proportional approval voting, harmonic k-median, and negative
  association.
\newblock In \emph{45th International Colloquium on Automata, Languages, and
  Programming, {ICALP} 2018}, volume 107 of \emph{LIPIcs}, pages 26:1--26:14.
  Schloss Dagstuhl - Leibniz-Zentrum f{\"{u}}r Informatik, 2018.

\bibitem[Caragiannis et~al.(2010)Caragiannis, Kalaitzis, and
  Markakis]{caragiannis2010approximation}
I.~Caragiannis, D.~Kalaitzis, and E.~Markakis.
\newblock Approximation algorithms and mechanism design for minimax approval
  voting.
\newblock In \emph{Proceedings of the 24th Conference on Artificial
  Intelligence (AAAI-2010)}, pages 737--742, 2010.

\bibitem[Caragiannis et~al.(2022)Caragiannis, Kaklamanis, Karanikolas, and
  Krimpas]{caragiannis2022evaluating}
I.~Caragiannis, C.~Kaklamanis, N.~Karanikolas, and G.~A. Krimpas.
\newblock Evaluating approval-based multiwinner voting in terms of robustness
  to noise.
\newblock \emph{Autonomous Agents and Multi-Agent Systems}, 36\penalty0
  (1):\penalty0 1--22, 2022.

\bibitem[Celis et~al.(2018)Celis, Huang, and Vishnoi]{conf/ijcai/CelisHV18}
L.~E. Celis, L.~Huang, and N.~K. Vishnoi.
\newblock Multiwinner voting with fairness constraints.
\newblock In \emph{Proceedings of the 27th International Joint Conference on
  Artificial Intelligence (IJCAI-2018)}, pages 144--151, 2018.

\bibitem[Conitzer and Walsh(2016)]{Handbook-manipulation}
V.~Conitzer and T.~Walsh.
\newblock Barriers to manipulation in voting.
\newblock In F.~Brandt, V.~Conitzer, U.~Endriss, J.~Lang, and A.~D. Procaccia,
  editors, \emph{Handbook of Computational Social Choice}, pages 127--145.
  Cambridge University Press, New York, NY, USA, 1st edition, 2016.

\bibitem[Cygan et~al.(2018)Cygan, Kowalik, Soca{\l}a, and
  Sornat]{cygan2018approximation}
M.~Cygan, {\L}.~Kowalik, A.~Soca{\l}a, and K.~Sornat.
\newblock Approximation and parameterized complexity of minimax approval
  voting.
\newblock \emph{Journal of Artificial Intelligence Research}, 63:\penalty0
  495--513, 2018.

\bibitem[Dietrich and List(2010)]{dietrich2010majority}
F.~Dietrich and C.~List.
\newblock Majority voting on restricted domains.
\newblock \emph{Journal of Economic Theory}, 145\penalty0 (2):\penalty0
  512--543, 2010.

\bibitem[Dudycz et~al.(2020)Dudycz, Manurangsi, Marcinkowski, and
  Sornat]{DudyczMMS20}
S.~Dudycz, P.~Manurangsi, J.~Marcinkowski, and K.~Sornat.
\newblock Tight approximation for {Proportional Approval Voting}.
\newblock In \emph{Proceedings of the 29th International Joint Conference on
  Artificial Intelligence (IJCAI-2020)}, pages 276--282, 2020.

\bibitem[Elkind and Lackner(2015)]{ijcai/ElkindL15-dichpref}
E.~Elkind and M.~Lackner.
\newblock Structure in dichotomous preferences.
\newblock In \emph{Proceedings of the 24th International Joint Conference on
  Artificial Intelligence (IJCAI-2015)}, pages 2019--2025. ijcai.org, 2015.

\bibitem[Elkind et~al.(2015)Elkind, Faliszewski, Lackner, and
  Obraztsova]{EFLO15}
E.~Elkind, P.~Faliszewski, M.~Lackner, and S.~Obraztsova.
\newblock The complexity of recognizing incomplete single-crossing preferences.
\newblock In \emph{Proceedings of the 29th Conference on Artificial
  Intelligence (AAAI-2015)}, pages 865--871, 2015.

\bibitem[Elkind et~al.(2017)Elkind, Lackner, and Peters]{ElkindEtAlTRENDS2017}
E.~Elkind, M.~Lackner, and D.~Peters.
\newblock Structured preferences.
\newblock In U.~Endriss, editor, \emph{Trends in Computational Social Choice},
  chapter~10, pages 187--207. AI Access, 2017.

\bibitem[Faliszewski and Procaccia(2010)]{fal-pro:j:manipulation}
P.~Faliszewski and A.~D. Procaccia.
\newblock {AI}'s war on manipulation: {A}re we winning?
\newblock \emph{AI Magazine}, 31\penalty0 (4):\penalty0 52--64, 2010.

\bibitem[Faliszewski and Rothe(2016)]{Handbook-bribery}
P.~Faliszewski and J.~Rothe.
\newblock Control and bribery in voting.
\newblock In F.~Brandt, V.~Conitzer, U.~Endriss, J.~Lang, and A.~D. Procaccia,
  editors, \emph{Handbook of Computational Social Choice}, pages 146--168.
  Cambridge University Press, New York, NY, USA, 1st edition, 2016.

\bibitem[Faliszewski et~al.(2011)Faliszewski, Hemaspaandra, Hemaspaandra, and
  Rothe]{fal-hem-hem-rot:j:sp}
P.~Faliszewski, E.~Hemaspaandra, L.~Hemaspaandra, and J.~Rothe.
\newblock The shield that never was: {Societies} with single-peaked preferences
  are more open to manipulation and control.
\newblock \emph{Information and Computation}, 209\penalty0 (2):\penalty0
  89--107, 2011.

\bibitem[Faliszewski et~al.(2017)Faliszewski, Skowron, and
  Talmon]{fal-sko-tal:multiwinner_bribery}
P.~Faliszewski, P.~Skowron, and N.~Talmon.
\newblock Bribery as a measure of candidate success: Complexity results for
  approval-based multiwinner rules.
\newblock In \emph{Proceedings of the 16th International Conference on
  Autonomous Agents and Multiagent Systems (AAMAS-2017)}, pages 6--14, 2017.

\bibitem[Faliszewski et~al.(2018)Faliszewski, Skowron, Slinko, and
  Talmon]{fal-sko-sli-tal:c:top-k-counting}
P.~Faliszewski, P.~Skowron, A.~Slinko, and N.~Talmon.
\newblock Multiwinner analogues of the plurality rule: {Axiomatic} and
  algorithmic views.
\newblock \emph{Social Choice and Welfare}, 51\penalty0 (3):\penalty0 513--550,
  2018.

\bibitem[Fitzsimmons and Lackner(2020)]{jair/incompletesp}
Z.~Fitzsimmons and M.~Lackner.
\newblock Incomplete preferences in single-peaked electorates.
\newblock \emph{Journal of Artificial Intelligence Research}, 67:\penalty0
  797--833, 2020.

\bibitem[Gaven{\v{c}}iak et~al.(2020)Gaven{\v{c}}iak, Kouteck{\`y}, and
  Knop]{gavenvciak2020integer}
T.~Gaven{\v{c}}iak, M.~Kouteck{\`y}, and D.~Knop.
\newblock Integer programming in parameterized complexity: Five miniatures.
\newblock \emph{Discrete Optimization}, page 100596, 2020.

\bibitem[Gawron and Faliszewski(2019)]{gaw-fal:c:robustness_of_abc_voting}
G.~Gawron and P.~Faliszewski.
\newblock Robustness of approval-based multiwinner voting rules.
\newblock In \emph{Proceedings of the 6th International Conference on
  Algorithmic Decision Theory (ADT-2019)}, pages 17--31, 2019.

\bibitem[Godziszewski et~al.(2021)Godziszewski, Batko, Skowron, and
  Faliszewski]{god-bat-sko-fal:c:2d-abc}
M.~Godziszewski, P.~Batko, P.~Skowron, and P.~Faliszewski.
\newblock An analysis of approval-based committee rules for {2D}-{E}uclidean
  elections.
\newblock In \emph{Proceedings of the 35th Conference on Artificial
  Intelligence (AAAI-2021)}, pages 5448--5455, 2021.

\bibitem[Krause and Golovin(2012)]{submodularOverview}
A.~Krause and D.~Golovin.
\newblock Submodular function maximization.
\newblock Technical report, 2012.

\bibitem[Lackner et~al.(2021)Lackner, Regner, Krenn, and Forster]{abcvoting}
M.~Lackner, P.~Regner, B.~Krenn, and S.~S. Forster.
\newblock {abcvoting: A Python library of approval-based committee voting
  rules}, 2021.
\newblock URL \url{https://doi.org/10.5281/zenodo.3904466}.
\newblock Current version: \url{https://github.com/martinlackner/abcvoting}.

\bibitem[Lang and Skowron(2018)]{LS16:multi-attribute}
J.~Lang and P.~Skowron.
\newblock Multi-attribute proportional representation.
\newblock \emph{Artificial Intelligence}, 263:\penalty0 74--106, 2018.

\bibitem[LeGrand et~al.(2007)LeGrand, Markakis, and
  Mehta]{leg-mar-meh:approx_multiwinner}
R.~LeGrand, E.~Markakis, and A.~Mehta.
\newblock Some results on approximating the minimax solution in approval
  voting.
\newblock In \emph{Proceedings of the 6th International Conference on
  Autonomous Agents and Multiagent Systems (AAMAS-2007)}, pages 198:1--198:3,
  2007.

\bibitem[Lenstra(1983)]{Len83}
H.~W. Lenstra.
\newblock Integer programming with a fixed number of variables.
\newblock \emph{Mathematics of Operations Research}, 8\penalty0 (4):\penalty0
  538--548, 1983.

\bibitem[List(2003)]{list2003possibility}
C.~List.
\newblock A possibility theorem on aggregation over multiple interconnected
  propositions.
\newblock \emph{Mathematical Social Sciences}, 45\penalty0 (1):\penalty0 1--13,
  2003.

\bibitem[Liu and Guo(2016)]{liu2016parameterized}
H.~Liu and J.~Guo.
\newblock Parameterized complexity of winner determination in minimax committee
  elections.
\newblock In \emph{Proceedings of the 15th International Conference on
  Autonomous Agents and Multiagent Systems (AAMAS-2016)}, pages 341--349, 2016.

\bibitem[Lu and Boutilier(2011)]{budgetSocialChoice}
T.~Lu and C.~Boutilier.
\newblock Budgeted social choice: {F}rom consensus to personalized decision
  making.
\newblock In \emph{Proceedings of the 22nd International Joint Conference on
  Artificial Intelligence (IJCAI-2011)}, pages 280--286, 2011.

\bibitem[Meir(2018)]{meir2018strategic}
R.~Meir.
\newblock \emph{Strategic voting}.
\newblock Synthesis Lectures on Artificial Intelligence and Machine Learning.
  Morgan \& Claypool Publishers, 2018.

\bibitem[Meir et~al.(2008)Meir, Procaccia, Rosenschein, and
  Zohar]{mei-pro-ros-zoh:multiwinner_strategic}
R.~Meir, A.~D. Procaccia, J.~S. Rosenschein, and A.~Zohar.
\newblock Complexity of strategic behavior in multi-winner elections.
\newblock \emph{Journal of Artificial Intelligence Research}, 33:\penalty0
  149--178, 2008.

\bibitem[Misra and Sonar(2019)]{conf/sofsem/MisraS19}
N.~Misra and C.~Sonar.
\newblock Robustness radius for {C}hamberlin-{C}ourant on restricted domains.
\newblock In \emph{Proceedings of the 45th International Conference on Current
  Trends in Theory and Practice of Computer Science (SOFSEM-2019)}, pages
  341--353, 2019.

\bibitem[Misra et~al.(2015)Misra, Nabeel, and Singh]{misra2015parameterized}
N.~Misra, A.~Nabeel, and H.~Singh.
\newblock On the parameterized complexity of minimax approval voting.
\newblock In \emph{Proceedings of the 14th International Conference on
  Autonomous Agents and Multiagent Systems (AAMAS-2015)}, pages 97--105, 2015.

\bibitem[Nemhauser et~al.(1978)Nemhauser, Wolsey, and Fisher]{submodular}
G.~Nemhauser, L.~Wolsey, and M.~Fisher.
\newblock An analysis of approximations for maximizing submodular set
  functions.
\newblock \emph{Mathematical Programming}, 14\penalty0 (1):\penalty0 265--294,
  1978.

\bibitem[Peters(2017)]{peters2017recognising}
D.~Peters.
\newblock Recognising multidimensional {E}uclidean preferences.
\newblock In \emph{Proceedings of the 31st Conference on Artificial
  Intelligence (AAAI-2017)}, pages 642--648, 2017.

\bibitem[Peters and Lackner(2020)]{jair/spoc}
D.~Peters and M.~Lackner.
\newblock Preferences single-peaked on a circle.
\newblock \emph{Journal of Artificial Intelligence Research}, 68:\penalty0
  463--502, 2020.

\bibitem[Peters and Skowron(2020)]{pet-sko:laminar}
D.~Peters and P.~Skowron.
\newblock Proportionality and the limits of welfarism.
\newblock In \emph{Proceedings of the 2020 ACM Conference on Economics and
  Computation (ACM-EC-2020)}, pages 793--794, 2020.
\newblock Extended version at \url{https://arxiv.org/abs/1911.11747}.

\bibitem[Potthoff(1990)]{potthoff90}
R.~Potthoff.
\newblock Use of linear programming for constrained approval voting.
\newblock \emph{Interfaces}, 20\penalty0 (5):\penalty0 79--80, 1990.

\bibitem[Potthoff and Brams(1998)]{potthoff1998proportional}
R.~F. Potthoff and S.~J. Brams.
\newblock Proportional representation: Broadening the options.
\newblock \emph{Journal of Theoretical Politics}, 10\penalty0 (2):\penalty0
  147--178, 1998.

\bibitem[Procaccia et~al.(2008)Procaccia, Rosenschein, and
  Zohar]{complexityProportionalRepr}
A.~D. Procaccia, J.~S. Rosenschein, and A.~Zohar.
\newblock On the complexity of achieving proportional representation.
\newblock \emph{Social Choice and Welfare}, 30\penalty0 (3):\penalty0 353--362,
  2008.

\bibitem[S{\'a}nchez-Fern{\'a}ndez et~al.(2017)S{\'a}nchez-Fern{\'a}ndez,
  Elkind, Lackner, Fern{\'a}ndez, Fisteus, {Basanta Val}, and Skowron]{pjr17}
L.~S{\'a}nchez-Fern{\'a}ndez, E.~Elkind, M.~Lackner, N.~Fern{\'a}ndez, J.~A.
  Fisteus, P.~{Basanta Val}, and P.~Skowron.
\newblock Proportional justified representation.
\newblock In \emph{Proceedings of the 31st Conference on Artificial
  Intelligence (AAAI-2017)}, pages 670--676, 2017.

\bibitem[Skowron(2017)]{skowron2017fpt}
P.~Skowron.
\newblock {FPT} approximation schemes for maximizing submodular functions.
\newblock \emph{Information and Computation}, 257:\penalty0 65--78, 2017.

\bibitem[Skowron and Faliszewski(2017)]{sko-fal:cc-fpt-approx}
P.~Skowron and P.~Faliszewski.
\newblock {Chamberlin-Courant} rule with approval ballots: Approximating the
  maxcover problem with bounded frequencies in {FPT} time.
\newblock \emph{Journal of Artificial Intelligence Research}, 60:\penalty0
  687--716, 2017.

\bibitem[Skowron et~al.(2016)Skowron, Faliszewski, and Lang]{owaWinner}
P.~Skowron, P.~Faliszewski, and J.~Lang.
\newblock Finding a collective set of items: {F}rom proportional
  multirepresentation to group recommendation.
\newblock \emph{Artificial Intelligence}, 241:\penalty0 191--216, 2016.

\bibitem[Skowron et~al.(2017)Skowron, Lackner, Brill, Peters, and
  Elkind]{proprank}
P.~Skowron, M.~Lackner, M.~Brill, D.~Peters, and E.~Elkind.
\newblock Proportional rankings.
\newblock In \emph{Proceedings of the 26th International Joint Conference on
  Artificial Intelligence (IJCAI-2017)}, pages 409--415, 2017.

\bibitem[Straszak et~al.(1993)Straszak, Libura, Sikorski, and
  Wagner]{straszak1993computer}
A.~Straszak, M.~Libura, J.~Sikorski, and D.~Wagner.
\newblock Computer-assisted constrained approval voting.
\newblock \emph{Group Decision and Negotiation}, 2\penalty0 (4):\penalty0
  375--385, 1993.

\bibitem[Terzopoulou et~al.(2020)Terzopoulou, Karpov, and
  Obraztsova]{terzopoulou2020restricted}
Z.~Terzopoulou, A.~Karpov, and S.~Obraztsova.
\newblock Restricted domains of dichotomous preferences with possibly
  incomplete information.
\newblock In \emph{Proceedings of the 19th International Conference on
  Autonomous Agents and Multiagent Systems (AAMAS-2020)}, pages 2023--2025,
  2020.

\bibitem[Yang(2019)]{yang2019tree}
Y.~Yang.
\newblock On the tree representations of dichotomous preferences.
\newblock In \emph{Proceedings of the 28th International Joint Conference on
  Artificial Intelligence (IJCAI-2019)}, pages 10--16, 2019.

\bibitem[Yang and Wang(2018)]{yang2018parameterized}
Y.~Yang and J.~Wang.
\newblock Parameterized complexity of multi-winner determination: {M}ore effort
  towards fixed-parameter tractability.
\newblock In \emph{Proceedings of the 17th International Conference on
  Autonomous Agents and Multiagent Systems (AAMAS-2018)}, pages 2142--2144,
  2018.

\end{thebibliography}


\begin{thebibliography}{68}
\providecommand{\natexlab}[1]{#1}
\providecommand{\url}[1]{\texttt{#1}}
\expandafter\ifx\csname urlstyle\endcsname\relax
  \providecommand{\doi}[1]{doi: #1}\else
  \providecommand{\doi}{doi: \begingroup \urlstyle{rm}\Url}\fi

\bibitem[Alcantud and Laruelle(2014)]{alcantud2014dis}
J.~C.~R. Alcantud and A.~Laruelle.
\newblock Dis\&approval voting: a characterization.
\newblock \emph{Social Choice and Welfare}, 43\penalty0 (1):\penalty0 1--10,
  2014.

\bibitem[Allouche et~al.(2022)Allouche, Lang, and
  Yger]{DBLP:journals/corr/abs-2201-06655}
T.~Allouche, J.~Lang, and F.~Yger.
\newblock Multi-winner approval voting goes epistemic.
\newblock \emph{CoRR}, abs/2201.06655, 2022.
\newblock URL \url{https://arxiv.org/abs/2201.06655}.

\bibitem[Aziz and Lee(2018)]{aziz2018sub}
H.~Aziz and B.~E. Lee.
\newblock Sub-committee approval voting and generalized justified
  representation axioms.
\newblock In \emph{Proceedings of the 2018 AAAI/ACM Conference on AI, Ethics,
  and Society}, pages 3--9, 2018.

\bibitem[Aziz and Lee(2020)]{aziz2020expanding}
H.~Aziz and B.~E. Lee.
\newblock The expanding approvals rule: Improving proportional representation
  and monotonicity.
\newblock \emph{Social Choice and Welfare}, 54\penalty0 (1):\penalty0 1--45,
  2020.

\bibitem[Aziz and Shah(2021)]{aziz2020participatory}
H.~Aziz and N.~Shah.
\newblock Participatory budgeting: Models and approaches.
\newblock In \emph{Pathways Between Social Science and Computational Social
  Science}, pages 215--236. Springer, 2021.

\bibitem[Aziz et~al.(2018)Aziz, Lee, and Talmon]{aziz2018proportionally}
H.~Aziz, B.~E. Lee, and N.~Talmon.
\newblock Proportionally representative participatory budgeting: Axioms and
  algorithms.
\newblock In \emph{Proceedings of the 16th International Conference on
  Autonomous Agents and Multiagent Systems (AAMAS-2017)}, pages 23--31.
  International Foundation for Autonomous Agents and Multiagent Systems, 2018.

\bibitem[Aziz et~al.(2019)Aziz, Bogomolnaia, and Moulin]{ABM19:fair-mixing}
H.~Aziz, A.~Bogomolnaia, and H.~Moulin.
\newblock Fair mixing: the case of dichotomous preferences.
\newblock In \emph{Proceedings of the 2019 ACM Conference on Economics and
  Computation (ACM-EC-2019)}, pages 753--781, 2019.

\bibitem[Barrot and Lang(2016)]{BarrotL16}
N.~Barrot and J.~Lang.
\newblock Conditional and sequential approval voting on combinatorial domains.
\newblock In \emph{Proceedings of the 25th International Joint Conference on
  Artificial Intelligence (IJCAI-2016)}, pages 88--94, 2016.

\bibitem[Barrot et~al.(2013)Barrot, Gourv{\`e}s, Lang, Monnot, and
  Ries]{bar-gou-lan-mon-ries}
N.~Barrot, L.~Gourv{\`e}s, J.~Lang, J.~Monnot, and B.~Ries.
\newblock Possible winners in approval voting.
\newblock In \emph{Proceedings of the 3rd International Conference on
  Algorithmic Decision Theory (ADT-2013)}, pages 57--70, 2013.

\bibitem[Baumeister and Dennisen(2015)]{bau-den:mav_trichotomous}
D.~Baumeister and S.~Dennisen.
\newblock Voter dissatisfaction in committee elections.
\newblock In \emph{Proceedings of the 14th International Conference on
  Autonomous Agents and Multiagent Systems (AAMAS-2015)}, pages 1707--1708,
  2015.

\bibitem[Baumeister et~al.(2015)Baumeister, Dennisen, and
  Rey]{bau-den-rey:manipulation_multiwinner}
D.~Baumeister, S.~Dennisen, and L.~Rey.
\newblock Winner determination and manipulation in minisum and minimax
  committee elections.
\newblock In \emph{Proceedings of the 4th International Conference on
  Algorithmic Decision Theory (ADT-2015)}, pages 469--485, 2015.

\bibitem[Baumeister et~al.(2016)Baumeister, B{\"{o}}hnlein, Rey, Schaudt, and
  Selker]{bau-boh-rey-sch-sel:minimax_ell_blocks}
D.~Baumeister, T.~B{\"{o}}hnlein, L.~Rey, O.~Schaudt, and A.~Selker.
\newblock Minisum and minimax committee election rules for general preference
  types.
\newblock In \emph{Proceedings of the 22nd European Conference on Artificial
  Intelligence (ECAI-2016)}, pages 1656--1657, 2016.

\bibitem[Baumeister et~al.(2021)Baumeister, Boes, and
  Hillebrand]{BaumeisterBH21}
D.~Baumeister, L.~Boes, and J.~Hillebrand.
\newblock Complexity of manipulative interference in participatory budgeting.
\newblock In \emph{Proceedings of the 7th International Conference on
  Algorithmic Decision Theory (ADT-2021)}, volume 13023 of \emph{Lecture Notes
  in Computer Science}, pages 424--439. Springer, 2021.

\bibitem[Behrens et~al.(2014)Behrens, Kistner, Nitsche, and Swierczek]{BKNS14a}
J.~Behrens, A.~Kistner, A.~Nitsche, and B.~Swierczek.
\newblock \emph{The Principles of LiquidFeedback}.
\newblock Interaktive Demokratie e.V. Berlin, 2014.

\bibitem[Bogomolnaia et~al.(2005)Bogomolnaia, Moulin, and Stong]{BMS05a}
A.~Bogomolnaia, H.~Moulin, and R.~Stong.
\newblock Collective choice under dichotomous preferences.
\newblock \emph{Journal of Economic Theory}, 122\penalty0 (2):\penalty0
  165--184, 2005.

\bibitem[Boutilier and Rosenschein(2016)]{brandt2015handbook-chapter10}
C.~Boutilier and J.~S. Rosenschein.
\newblock Incomplete information and communication in voting.
\newblock In F.~Brandt, V.~Conitzer, U.~Endriss, J.~Lang, and A.~Procaccia,
  editors, \emph{Handbook of Computational Social Choice}. Cambridge University
  Press, 2016.

\bibitem[Brams and Fishburn(1978)]{brams1978approval}
S.~J. Brams and P.~C. Fishburn.
\newblock Approval voting.
\newblock \emph{American Political Science Review}, 72\penalty0 (3):\penalty0
  831--847, 1978.

\bibitem[Brams et~al.(1997)Brams, Kilgour, and
  Zwicker]{BKZ:voting-on-referenda}
S.~J. Brams, D.~M. Kilgour, and W.~Zwicker.
\newblock Voting on referenda: the separability problem and possible solutions.
\newblock \emph{Electoral Studies}, 16\penalty0 (3):\penalty0 359--377, 1997.

\bibitem[Brandl and Peters(2019)]{brandl2019axiomatic}
F.~Brandl and D.~Peters.
\newblock An axiomatic characterization of the {B}orda mean rule.
\newblock \emph{Social choice and welfare}, 52\penalty0 (4):\penalty0 685--707,
  2019.

\bibitem[Brandl et~al.(2020)Brandl, Brandt, Peters, Stricker, and
  Suksompong]{BBP+19a}
F.~Brandl, F.~Brandt, D.~Peters, C.~Stricker, and W.~Suksompong.
\newblock Funding public projects: {A} case for the {N}ash product rule, 2020.
\newblock Working paper.

\bibitem[Brandt(2017)]{Bran17a}
F.~Brandt.
\newblock Rolling the dice: {R}ecent results in probabilistic social choice.
\newblock In U.~Endriss, editor, \emph{Trends in Computational Social Choice},
  chapter~1, pages 3--26. AI Access, 2017.

\bibitem[Chopra et~al.(2006)Chopra, Ghose, and Meyer]{chopra2006social}
S.~Chopra, A.~Ghose, and T.~Meyer.
\newblock Social choice theory, belief merging, and strategy-proofness.
\newblock \emph{Information Fusion}, 7\penalty0 (1):\penalty0 61--79, 2006.

\bibitem[Conitzer et~al.(2017)Conitzer, Freeman, and Shah]{conitzer2017fair}
V.~Conitzer, R.~Freeman, and N.~Shah.
\newblock Fair public decision making.
\newblock In \emph{Proceedings of the 2017 ACM Conference on Economics and
  Computation}, pages 629--646, 2017.

\bibitem[Drosou and Pitoura(2010)]{drosou2010search}
M.~Drosou and E.~Pitoura.
\newblock Search result diversification.
\newblock \emph{ACM SIGMOD Record}, 39\penalty0 (1):\penalty0 41--47, 2010.

\bibitem[Duddy(2015)]{Dudd15a}
C.~Duddy.
\newblock Fair sharing under dichotomous preferences.
\newblock \emph{Mathematical Social Sciences}, 73:\penalty0 1--5, 2015.

\bibitem[Duddy et~al.(2014)Duddy, Houy, Lang, Piggins, and
  Zwicker]{duddy2014social}
C.~Duddy, N.~Houy, J.~Lang, A.~Piggins, and W.~S. Zwicker.
\newblock Social dichotomy functions.
\newblock Extended abstract for presentation at the 2014 meeting of the Society
  for Social Choice and Welfare, 2014.

\bibitem[Duddy et~al.(2016)Duddy, Piggins, and Zwicker]{duddy2016}
C.~Duddy, A.~Piggins, and W.~S. Zwicker.
\newblock Aggregation of binary evaluations: {A} {B}orda-like approach.
\newblock \emph{Social Choice and Welfare}, 46\penalty0 (2):\penalty0 301--333,
  2016.

\bibitem[Elkind et~al.(2017)Elkind, Faliszewski, Skowron, and
  Slinko]{elk-fal-sko-sli:c:multiwinner-rules}
E.~Elkind, P.~Faliszewski, P.~Skowron, and A.~Slinko.
\newblock Properties of multiwinner voting rules.
\newblock \emph{Social Choice and Welfare}, 48\penalty0 (3):\penalty0 599--632,
  2017.

\bibitem[Endriss(2016)]{Handbook-JA}
U.~Endriss.
\newblock Judgment aggregation.
\newblock In F.~Brandt, V.~Conitzer, U.~Endriss, J.~Lang, and A.~D. Procaccia,
  editors, \emph{Handbook of Computational Social Choice}, pages 399--426.
  Cambridge University Press, New York, NY, USA, 1st edition, 2016.

\bibitem[Everaere et~al.(2007)Everaere, Konieczny, and
  Marquis]{everaere2007strategy}
P.~Everaere, S.~Konieczny, and P.~Marquis.
\newblock The strategy-proofness landscape of merging.
\newblock \emph{Journal of Artificial Intelligence Research}, 28:\penalty0
  49--105, 2007.

\bibitem[Fain et~al.(2016)Fain, Goel, and Munagala]{FGM16a}
B.~Fain, A.~Goel, and K.~Munagala.
\newblock The core of the participatory budgeting problem.
\newblock In \emph{Proceedings of the 12th International Conference on Web and
  Internet Economics (WINE-2016)}, pages 384--399, 2016.

\bibitem[Faliszewski et~al.(2017)Faliszewski, Skowron, Slinko, and
  Talmon]{FSST-trends}
P.~Faliszewski, P.~Skowron, A.~Slinko, and N.~Talmon.
\newblock Multiwinner voting: {A} new challenge for social choice theory.
\newblock In U.~Endriss, editor, \emph{Trends in Computational Social Choice},
  chapter~2, pages 27--47. AI Access, 2017.

\bibitem[Faliszewski et~al.(2019)Faliszewski, Skowron, Slinko, and
  Talmon]{fal-sko-sli-tal-tal:j:hierarchy-committee}
P.~Faliszewski, P.~Skowron, A.~Slinko, and N.~Talmon.
\newblock Committee scoring rules: {Axiomatic} characterization and hierarchy.
\newblock \emph{ACM Transactions on Economics and Computation}, 6\penalty0
  (1):\penalty0 Article~3, 2019.

\bibitem[Faliszewski et~al.(2020)Faliszewski, Slinko, and
  Talmon]{Faliszewski17OA}
P.~Faliszewski, A.~Slinko, and N.~Talmon.
\newblock Multiwinner rules with variable number of winners.
\newblock In \emph{Proceedings of the 24th European Conference on Artificial
  Intelligence (ECAI-2020)}, pages 67--74, 2020.
\newblock URL \url{https://doi.org/10.3233/FAIA200077}.

\bibitem[Freeman et~al.(2017)Freeman, Zahedi, and Conitzer]{Freeman2017Fair}
R.~Freeman, S.~M. Zahedi, and V.~Conitzer.
\newblock Fair and efficient social choice in dynamic settings.
\newblock In \emph{Proceedings of the 26th International Joint Conference on
  Artificial Intelligence (IJCAI-2017)}, pages 4580--4587. ijcai.org, 2017.

\bibitem[Freeman et~al.(2020)Freeman, Kahng, and Pennock]{prop-mwwavnow}
R.~Freeman, A.~Kahng, and D.~M. Pennock.
\newblock Proportionality in approval-based elections with a variable number of
  winners.
\newblock In \emph{Proceedings of the 29th International Joint Conference on
  Artificial Intelligence (IJCAI-2020)}, pages 132--138, 2020.

\bibitem[Frege(1918)]{frege:1918a}
G.~Frege.
\newblock {Vorschl\"age f\"ur ein Wahlgesetz}.
\newblock Original typescript at the Th\"uringer Universit\"ats- und
  Landesbibliothek (ThULB), 1918.
\newblock URL
  \url{https://archive.thulb.uni-jena.de/collections/receive/HisBest_cbu_00022979}.

\bibitem[Frege(2000)]{frege:2000a}
G.~Frege.
\newblock {Vorschl\"age f\"ur ein Wahlgesetz}.
\newblock In G.~Gabriel and U.~Dathe, editors, \emph{{Gottlob Frege: Werk und
  Wirkung.} {Mit den unver{\"o}ffentlichten Vorschl{\"a}gen f{\"u}r ein
  Wahlgesetz von Gottlob Frege}}, pages 297--313. Mentis, 2000.

\bibitem[Goel et~al.(2015)Goel, Krishnaswamy, Sakshuwong, and
  Aitamurto]{goel2015knapsack}
A.~Goel, A.~K. Krishnaswamy, S.~Sakshuwong, and T.~Aitamurto.
\newblock Knapsack voting.
\newblock \emph{Collective Intelligence}, 1, 2015.

\bibitem[Gonzalez et~al.(2019)Gonzalez, Laruelle, and
  Solal]{gonzalez2019dilemma}
S.~Gonzalez, A.~Laruelle, and P.~Solal.
\newblock Dilemma with approval and disapproval votes.
\newblock \emph{Social Choice and Welfare}, 53\penalty0 (3):\penalty0 497--517,
  2019.

\bibitem[Haret and Wallner(2019)]{haret2019manipulating}
A.~Haret and J.~P. Wallner.
\newblock Manipulating skeptical and credulous consequences when merging
  beliefs.
\newblock In \emph{European Conference on Logics in Artificial Intelligence
  (JELIA-2019)}, pages 133--150, 2019.

\bibitem[Haret et~al.(2016)Haret, Pfandler, and Woltran]{haret2016beyond}
A.~Haret, A.~Pfandler, and S.~Woltran.
\newblock Beyond {IC} postulates: Classification criteria for merging
  operators.
\newblock In \emph{ECAI}, pages 372--380, 2016.

\bibitem[Haret et~al.(2020)Haret, Lackner, Pfandler, and Wallner]{aaai/propbm}
A.~Haret, M.~Lackner, A.~Pfandler, and J.~P. Wallner.
\newblock Proportional belief merging.
\newblock In \emph{Proceedings of the 34th AAAI Conference on Artificial
  Intelligence (AAAI 2020)}, pages 2822--2829, 2020.

\bibitem[Harrenstein et~al.(2020)Harrenstein, Lackner, and Lackner]{hll-frege}
P.~Harrenstein, M.~Lackner, and M.~Lackner.
\newblock A mathematical analysis of an election system proposed by {G}ottlob
  {F}rege.
\newblock \emph{Erkenntnis}, 2020.

\bibitem[Imber et~al.(2022)Imber, Israel, Brill, and
  Kimelfeld]{imber2021committee}
A.~Imber, J.~Israel, M.~Brill, and B.~Kimelfeld.
\newblock Approval-based committee voting under incomplete information.
\newblock In \emph{Proceedings of the 36th Conference on Artificial
  Intelligence (AAAI-2022)}, pages 5076--5083, 2022.

\bibitem[Israel and Brill(2021)]{israel2021dynamic}
J.~Israel and M.~Brill.
\newblock Dynamic proportional rankings.
\newblock In \emph{Proceedings of the 30th International Joint Conference on
  Artificial Intelligence (IJCAI-2021)}, pages 261--267, 2021.

\bibitem[Jain et~al.(2020)Jain, Sornat, and Talmon]{projectinteractions}
P.~Jain, K.~Sornat, and N.~Talmon.
\newblock Participatory budgeting with project interactions.
\newblock In \emph{Proceedings of the 29th International Joint Conference on
  Artificial Intelligence (IJCAI-2020)}, pages 386--392. ijcai.org, 2020.

\bibitem[Kilgour(2010)]{kil-handbook}
D.~M. Kilgour.
\newblock Approval balloting for multi-winner elections.
\newblock In J.-F. Laslier and M.~R. Sanver, editors, \emph{Handbook on
  Approval Voting}, pages 105--124. Springer, 2010.

\bibitem[Kilgour(2016)]{Kilgour16}
D.~M. Kilgour.
\newblock Approval elections with a variable number of winners.
\newblock \emph{Theory and Decision}, 81, 02 2016.

\bibitem[Konieczny and {Pino P{\'{e}}rez}(2002)]{KoniecznyP02}
S.~Konieczny and R.~{Pino P{\'{e}}rez}.
\newblock Merging information under constraints: {A} logical framework.
\newblock \emph{Journal of Logic and Computation}, 12\penalty0 (5):\penalty0
  773--808, 2002.

\bibitem[Konieczny and {Pino P{\'{e}}rez}(2011)]{KoniecznyP11}
S.~Konieczny and R.~{Pino P{\'{e}}rez}.
\newblock Logic based merging.
\newblock \emph{Journal of Philosophical Logic}, 40\penalty0 (2):\penalty0
  239--270, 2011.

\bibitem[Konieczny et~al.(2004)Konieczny, Lang, and Marquis]{KoniecznyM04}
S.~Konieczny, J.~Lang, and P.~Marquis.
\newblock {DA}\({}^{\mbox{2}}\) merging operators.
\newblock \emph{Artificial Intelligence}, 157\penalty0 (1-2):\penalty0 49--79,
  2004.

\bibitem[Lackner(2020)]{aaai/perpetual}
M.~Lackner.
\newblock Perpetual voting: {F}airness in long-term decision making.
\newblock In \emph{Proceedings of the 34th AAAI Conference on Artificial
  Intelligence (AAAI 2020)}, pages 2103--2110, 2020.

\bibitem[Lackner and Maly(2021)]{approvalbased-shortlisting}
M.~Lackner and J.~Maly.
\newblock Approval-based shortlisting.
\newblock In \emph{Proceedings of the 20th International Conference on
  Autonomous Agents and Multiagent Systems (AAMAS-2021)}, pages 1566--1568,
  2021.

\bibitem[Lang and Xia(2016)]{Lang2016VotinginCombinatorial}
J.~Lang and L.~Xia.
\newblock Voting in combinatorial domains.
\newblock In F.~Brandt, V.~Conitzer, U.~Endriss, J.~Lang, and A.~D. Procaccia,
  editors, \emph{Handbook of Computational Social Choice}. Cambridge University
  Press, 2016.

\bibitem[Lines(1986)]{lines1986approval}
M.~Lines.
\newblock Approval voting and strategy analysis: {A} venetian example.
\newblock \emph{Theory and Decision}, 20\penalty0 (2):\penalty0 155--172, 1986.

\bibitem[List and Puppe(2009)]{list2009judgment}
C.~List and C.~Puppe.
\newblock Judgment aggregation: {A} survey.
\newblock In P.~Anand, P.~Pattanaik, and C.~Puppe, editors, \emph{The Handbook
  of Rational and Social Choice}. Oxford University Press, 2009.

\bibitem[Liu and Guo(2016)]{liu2016parameterized}
H.~Liu and J.~Guo.
\newblock Parameterized complexity of winner determination in minimax committee
  elections.
\newblock In \emph{Proceedings of the 15th International Conference on
  Autonomous Agents and Multiagent Systems (AAMAS-2016)}, pages 341--349, 2016.

\bibitem[Michorzewski et~al.(2020)Michorzewski, Peters, and
  Skowron]{MPS20:fair-mixing}
M.~Michorzewski, D.~Peters, and P.~Skowron.
\newblock Price of fairness in budget division and probabilistic social choice.
\newblock In \emph{Proceedings of the 34th Conference on Artificial
  Intelligence (AAAI-2020)}, pages 2184--2191, 2020.

\bibitem[Peters et~al.(2021)Peters, Pierczynski, and
  Skowron]{pet-pie-sko:c:participatory-budgeting-cardinal}
D.~Peters, G.~Pierczynski, and P.~Skowron.
\newblock Proportional participatory budgeting with additive utilities.
\newblock In \emph{Proceedings of the Thirty-fifth Conference on Neural
  Information Processing Systems (NeurIPS-2021)}, pages 12726--12737, 2021.

\bibitem[Rey et~al.(2020)Rey, Endriss, and de~Haan]{ReyEH20}
S.~Rey, U.~Endriss, and R.~de~Haan.
\newblock Designing participatory budgeting mechanisms grounded in judgment
  aggregation.
\newblock In \emph{Proceedings of the 17th International Conference on
  Principles of Knowledge Representation and Reasoning, ({KR}-2020)}, pages
  692--702, 2020.

\bibitem[Santos et~al.(2015)Santos, MacDonald, and
  Ounis]{journals/ftir/SantosMO15}
R.~L.~T. Santos, C.~MacDonald, and I.~Ounis.
\newblock Search result diversification.
\newblock \emph{Foundations and Trends in Information Retrieval}, 9\penalty0
  (1):\penalty0 1--90, 2015.

\bibitem[Skowron et~al.(2017)Skowron, Lackner, Brill, Peters, and
  Elkind]{proprank}
P.~Skowron, M.~Lackner, M.~Brill, D.~Peters, and E.~Elkind.
\newblock Proportional rankings.
\newblock In \emph{Proceedings of the 26th International Joint Conference on
  Artificial Intelligence (IJCAI-2017)}, pages 409--415, 2017.

\bibitem[Skowron et~al.(2019)Skowron, Faliszewski, and
  Slinko]{skowron2019axiomatic}
P.~Skowron, P.~Faliszewski, and A.~Slinko.
\newblock Axiomatic characterization of committee scoring rules.
\newblock \emph{Journal of Economic Theory}, 180:\penalty0 244--273, 2019.

\bibitem[Talmon and Faliszewski(2019)]{talmon2019framework}
N.~Talmon and P.~Faliszewski.
\newblock A framework for approval-based budgeting methods.
\newblock In \emph{Proceedings of the 33rd Conference on Artificial
  Intelligence (AAAI-2019)}, volume~33, pages 2181--2188, 2019.

\bibitem[Talmon and Page(2021)]{talmon2021proportionality}
N.~Talmon and R.~Page.
\newblock Proportionality in committee selection with negative feelings.
\newblock \emph{CoRR}, abs/2101.01435, 2021.
\newblock URL \url{https://arxiv.org/abs/2101.01435}.

\bibitem[Terzopoulou et~al.(2020)Terzopoulou, Karpov, and
  Obraztsova]{terzopoulou2020restricted}
Z.~Terzopoulou, A.~Karpov, and S.~Obraztsova.
\newblock Restricted domains of dichotomous preferences with possibly
  incomplete information.
\newblock In \emph{Proceedings of the 19th International Conference on
  Autonomous Agents and Multiagent Systems (AAMAS-2020)}, pages 2023--2025,
  2020.

\bibitem[Zhou et~al.(2019)Zhou, Yang, and Guo]{zhou2019parameterized}
A.~Zhou, Y.~Yang, and J.~Guo.
\newblock Parameterized complexity of committee elections with dichotomous and
  trichotomous votes.
\newblock In \emph{Proceedings of the 18th International Conference on
  Autonomous Agents and Multiagent Systems (AAMAS-2019)}, pages 503--510.
  International Foundation for Autonomous Agents and Multiagent Systems, 2019.

\end{thebibliography}


\begin{thebibliography}{19}
\providecommand{\natexlab}[1]{#1}
\providecommand{\url}[1]{\texttt{#1}}
\expandafter\ifx\csname urlstyle\endcsname\relax
  \providecommand{\doi}[1]{doi: #1}\else
  \providecommand{\doi}{doi: \begingroup \urlstyle{rm}\Url}\fi

\bibitem[Aziz et~al.(2018)Aziz, Elkind, Huang, Lackner,
  S{\'a}nchez-Fern{\'a}ndez, and Skowron]{AEHLSS18}
H.~Aziz, E.~Elkind, S.~Huang, M.~Lackner, L.~S{\'a}nchez-Fern{\'a}ndez, and
  P.~Skowron.
\newblock On the complexity of extended and proportional justified
  representation.
\newblock In \emph{Proceedings of the 32nd Conference on Artificial
  Intelligence (AAAI-2018)}, pages 902--909, 2018.

\bibitem[Balinski and Young(1982)]{BaYo82a}
M.~Balinski and H.~P. Young.
\newblock \emph{Fair Representation: {M}eeting the Ideal of One Man, One Vote}.
\newblock Yale University Press, 1982.
\newblock (2nd Edition by Brookings Institution Press, 2001).

\bibitem[Brandt et~al.(2009)Brandt, Fischer, and
  Harrenstein]{brandt2009computational}
F.~Brandt, F.~Fischer, and P.~Harrenstein.
\newblock The computational complexity of choice sets.
\newblock \emph{Mathematical Logic Quarterly}, 55\penalty0 (4):\penalty0
  444--459, 2009.

\bibitem[Brill and Fischer(2012)]{bri-fis:c:ranked-pairs}
M.~Brill and F.~Fischer.
\newblock The price of neutrality for the ranked pairs method.
\newblock In \emph{Proceedings of the 26th Conference on Artificial
  Intelligence (AAAI-2012)}, pages 1299--1305, 2012.

\bibitem[Brill et~al.(2020)Brill, G{\"{o}}lz, Peters, Schmidt{-}Kraepelin, and
  Wilker]{BGPSW19}
M.~Brill, P.~G{\"{o}}lz, D.~Peters, U.~Schmidt{-}Kraepelin, and K.~Wilker.
\newblock Approval-based apportionment.
\newblock In \emph{Proceedings of the 34th Conference on Artificial
  Intelligence (AAAI-2020)}, pages 1854--1861, 2020.
\newblock Extended version at \url{http://arxiv.org/abs/1911.08365}.

\bibitem[Byrka et~al.(2018)Byrka, Skowron, and Sornat]{ByrkaSS18}
J.~Byrka, P.~Skowron, and K.~Sornat.
\newblock Proportional approval voting, harmonic k-median, and negative
  association.
\newblock In \emph{45th International Colloquium on Automata, Languages, and
  Programming, {ICALP} 2018}, volume 107 of \emph{LIPIcs}, pages 26:1--26:14.
  Schloss Dagstuhl - Leibniz-Zentrum f{\"{u}}r Informatik, 2018.

\bibitem[Conitzer et~al.(2009)Conitzer, Rognlie, and Xia]{con-rog-xia:c:mle}
V.~Conitzer, M.~Rognlie, and L.~Xia.
\newblock Preference functions that score rankings and maximum likelihood
  estimation.
\newblock In \emph{Proceedings of the 21st International Joint Conference on
  Artificial Intelligence (IJCAI-2009)}, pages 109--115, 2009.

\bibitem[Csar et~al.(2017)Csar, Lackner, Pichler, and
  Sallinger]{aaai/CsarLPS17-mapreduce}
T.~Csar, M.~Lackner, R.~Pichler, and E.~Sallinger.
\newblock Winner determination in huge elections with mapreduce.
\newblock In \emph{Proceedings of the 31st Conference on Artificial
  Intelligence (AAAI-2017)}, pages 451--458, 2017.

\bibitem[Csar et~al.(2018)Csar, Lackner, and Pichler]{ijcai/CsarLP-schulze}
T.~Csar, M.~Lackner, and R.~Pichler.
\newblock Computing the {S}chulze method for large-scale preference data sets.
\newblock In \emph{Proceedings of the 27th International Joint Conference on
  Artificial Intelligence (IJCAI 2018)}, pages 180--187. ijcai.org, 2018.

\bibitem[Darmann and Lang(2017)]{DarmannLangTRENDS2017}
A.~Darmann and J.~Lang.
\newblock Group activity selection problems.
\newblock In U.~Endriss, editor, \emph{Trends in Computational Social Choice},
  chapter~5, pages 87--103. AI Access, 2017.

\bibitem[Freeman et~al.(2015)Freeman, Brill, and Conitzer]{freeman2015general}
R.~Freeman, M.~Brill, and V.~Conitzer.
\newblock General tiebreaking schemes for computational social choice.
\newblock In \emph{Proceedings of the 14th International Conference on
  Autonomous Agents and Multiagent Systems (AAMAS-2015)}, pages 1401--1409,
  2015.

\bibitem[Greenlaw et~al.(1995)Greenlaw, Hoover, and Ruzzo]{greenlaw1995limits}
R.~Greenlaw, H.~J. Hoover, and W.~L. Ruzzo.
\newblock \emph{Limits to parallel computation: P-completeness theory}.
\newblock Oxford University Press, 1995.

\bibitem[Kluiving et~al.(2020)Kluiving, Vries, Vrijbergen, Boixel, and
  Endriss]{KVVBE20:strategyproofness}
B.~Kluiving, A.~Vries, P.~Vrijbergen, A.~Boixel, and U.~Endriss.
\newblock Analysing irresolute multiwinner voting rules with approval ballots
  via {SAT} solving.
\newblock In \emph{Proceedings of the 24th European Conference on Artificial
  Intelligence (ECAI-2020)}, volume 325 of \emph{Frontiers in Artificial
  Intelligence and Applications}, pages 131--138. {IOS} Press, 2020.

\bibitem[Mattei and Walsh(2013)]{conf/aldt/MatteiW13}
N.~Mattei and T.~Walsh.
\newblock Preflib: A library for preferences http: //www.preflib.org.
\newblock In \emph{Proceedings of the 3rd International Conference on
  Algorithmic Decision Theory (ADT-2013)}, pages 259--270, 2013.

\bibitem[Peters(2018)]{pet:prop-sp}
D.~Peters.
\newblock Proportionality and strategyproofness in multiwinner elections.
\newblock In \emph{Proceedings of the 17th International Conference on
  Autonomous Agents and Multiagent Systems (AAMAS-2018)}, pages 1549--1557,
  2018.

\bibitem[Peters(2019)]{dominikthesis}
D.~Peters.
\newblock \emph{Fair Division of the Commons}.
\newblock PhD thesis, University of Oxford, 9 2019.

\bibitem[S{\'a}nchez-Fern{\'a}ndez and Fisteus(2019)]{sanchez2019monotonicity}
L.~S{\'a}nchez-Fern{\'a}ndez and J.~A. Fisteus.
\newblock Monotonicity axioms in approval-based multi-winner voting rules.
\newblock In \emph{Proceedings of the 18th International Conference on
  Autonomous Agents and Multiagent Systems (AAMAS-2019)}, pages 485--493.
  International Foundation for Autonomous Agents and Multiagent Systems, 2019.

\bibitem[Sonar et~al.(2020)Sonar, Dey, and Misra]{SonarDM20}
C.~Sonar, P.~Dey, and N.~Misra.
\newblock On the complexity of winner verification and candidate winner for
  multiwinner voting rules.
\newblock In \emph{Proceedings of the 29th International Joint Conference on
  Artificial Intelligence (IJCAI-2020)}, pages 89--95. ijcai.org, 2020.

\bibitem[Szufa et~al.(2020)Szufa, Faliszewski, Skowron, Slinko, and
  Talmon]{szufa2020drawing}
S.~Szufa, P.~Faliszewski, P.~Skowron, A.~Slinko, and N.~Talmon.
\newblock Drawing a map of elections in the space of statistical cultures.
\newblock In \emph{Proceedings of the 19th International Conference on
  Autonomous Agents and Multiagent Systems (AAMAS-2020)}, pages 1341--1349,
  2020.

\end{thebibliography}


\begin{thebibliography}{7}
\providecommand{\natexlab}[1]{#1}
\providecommand{\url}[1]{\texttt{#1}}
\expandafter\ifx\csname urlstyle\endcsname\relax
  \providecommand{\doi}[1]{doi: #1}\else
  \providecommand{\doi}{doi: \begingroup \urlstyle{rm}\Url}\fi

\bibitem[Janson(2016)]{Janson16arxiv}
S.~Janson.
\newblock Phragm{\'{e}}n's and {T}hiele's election methods.
\newblock \emph{CoRR}, abs/1611.08826, 2016.
\newblock URL \url{http://arxiv.org/abs/1611.08826}.

\bibitem[Lackner et~al.(2021)Lackner, Regner, Krenn, and Forster]{abcvoting}
M.~Lackner, P.~Regner, B.~Krenn, and S.~S. Forster.
\newblock {abcvoting: A Python library of approval-based committee voting
  rules}, 2021.
\newblock URL \url{https://doi.org/10.5281/zenodo.3904466}.
\newblock Current version: \url{https://github.com/martinlackner/abcvoting}.

\bibitem[Mora and Oliver(2015)]{MoOl15a}
X.~Mora and M.~Oliver.
\newblock Eleccions mitjan{\c c}ant el vot d'aprovaci{\'o}. {E}l m{\`e}tode de
  {P}hragm{\'e}n i algunes variants.
\newblock \emph{Butllet{\'\i} de la Societat Catalana de Matem{\`a}tiques},
  30\penalty0 (1):\penalty0 57--101, 2015.

\bibitem[Peters and Skowron(2020)]{pet-sko:laminar}
D.~Peters and P.~Skowron.
\newblock Proportionality and the limits of welfarism.
\newblock In \emph{Proceedings of the 2020 ACM Conference on Economics and
  Computation (ACM-EC-2020)}, pages 793--794, 2020.
\newblock Extended version at \url{https://arxiv.org/abs/1911.11747}.

\bibitem[Phragm{\'e}n(1896)]{Phra96a}
E.~Phragm{\'e}n.
\newblock Sur la th{\'e}orie des {\'e}lections multiples.
\newblock \emph{{\"O}fversigt af Kongliga Vetenskaps-Akademiens
  F{\"o}rhandlingar}, 53:\penalty0 181--191, 1896.

\bibitem[S{\'a}nchez-Fern{\'a}ndez and Fisteus(2019)]{sanchez2019monotonicity}
L.~S{\'a}nchez-Fern{\'a}ndez and J.~A. Fisteus.
\newblock Monotonicity axioms in approval-based multi-winner voting rules.
\newblock In \emph{Proceedings of the 18th International Conference on
  Autonomous Agents and Multiagent Systems (AAMAS-2019)}, pages 485--493.
  International Foundation for Autonomous Agents and Multiagent Systems, 2019.

\bibitem[S{\'a}nchez-Fern{\'a}ndez et~al.(2017)S{\'a}nchez-Fern{\'a}ndez,
  Elkind, Lackner, Fern{\'a}ndez, Fisteus, {Basanta Val}, and Skowron]{pjr17}
L.~S{\'a}nchez-Fern{\'a}ndez, E.~Elkind, M.~Lackner, N.~Fern{\'a}ndez, J.~A.
  Fisteus, P.~{Basanta Val}, and P.~Skowron.
\newblock Proportional justified representation.
\newblock In \emph{Proceedings of the 31st Conference on Artificial
  Intelligence (AAAI-2017)}, pages 670--676, 2017.

\end{thebibliography}

\end{document}